\documentclass[a4paper,aps,twocolumn,superscriptaddress,floatfix]{revtex4}

\usepackage{graphicx}
\usepackage[breaklinks]{hyperref}
\usepackage{etoolbox}
\appto\UrlBreaks{\do\-}
\usepackage{amsmath, amssymb, amsthm}
\usepackage{mathtools}
\usepackage{color}
\usepackage{braket}
\usepackage{bbm}
\usepackage{algorithm, algpseudocode}
\floatname{algorithm}{Protocol}

\DeclareMathOperator*{\argmin}{argmin}

\DeclareMathOperator{\polylog}{polylog}
\allowdisplaybreaks[4]

\theoremstyle{definition}
\newtheorem{theorem}{Theorem}

\newtheorem{proposition}[theorem]{Proposition}
\theoremstyle{definition}
\newtheorem{definition}[theorem]{Definition}

\theoremstyle{remark}
\newtheorem{conjecture}{Conjecture}
\newtheorem{remark}[conjecture]{Remark}

\begin{document}

\title{Polylog-overhead highly fault-tolerant measurement-based quantum computation: all-Gaussian implementation with Gottesman-Kitaev-Preskill code}

\author{Hayata Yamasaki}
\email{hayata.yamasaki@gmail.com}
\affiliation{Photon Science Center, Graduate School of Engineering, The University of Tokyo, 7--3--1 Hongo, Bunkyo-ku, Tokyo 113--8656, Japan}
\author{Kosuke Fukui}
\affiliation{Department of Applied Physics, Graduate School of Engineering, The University of Tokyo, 7--3--1  Hongo, Bunkyo-ku, Tokyo 113--8656, Japan}
\author{Yuki Takeuchi}
\affiliation{NTT Communication Science Laboratories, NTT Corporation, 3--1 Morinosato-Wakamiya, Atsugi, Kanagawa 243--0198, Japan}
\author{Seiichiro Tani}
\affiliation{NTT Communication Science Laboratories, NTT Corporation, 3--1 Morinosato-Wakamiya, Atsugi, Kanagawa 243--0198, Japan}
\author{Masato Koashi}
\affiliation{Department of Applied Physics, Graduate School of Engineering, The University of Tokyo, 7--3--1  Hongo, Bunkyo-ku, Tokyo 113--8656, Japan}
\affiliation{Photon Science Center, Graduate School of Engineering, The University of Tokyo, 7--3--1 Hongo, Bunkyo-ku, Tokyo 113--8656, Japan}

\date{\today}

\begin{abstract}
  Scalability of flying photonic quantum systems in generating quantum entanglement offers a potential for implementing large-scale fault-tolerant quantum computation, especially by means of measurement-based quantum computation (MBQC). However, existing protocols for MBQC inevitably impose a polynomial overhead cost in implementing quantum computation due to geometrical constraints of entanglement structures used in the protocols, and the polynomial overhead potentially cancels out useful polynomial speedups in quantum computation. To implement quantum computation without this cancellation, we construct a protocol for photonic MBQC that achieves as low as \textit{poly-logarithmic overhead}, by introducing an entanglement structure for low-overhead qubit permutation. Based on this protocol, we design a fault-tolerant photonic MBQC protocol that can be performed by experimentally tractable \textit{homodyne detection and Gaussian entangling operations} combined with the Gottesman-Kitaev-Preskill (GKP) quantum error-correcting code, which we concatenate with the $7$-qubit code. Our fault-tolerant protocol achieves the threshold $7.8$~dB in terms of the squeezing level of the GKP code, outperforming $8.3$~dB of the best existing protocol for fault-tolerant quantum computation with the GKP surface code. Thus, bridging a gap between theoretical progress on MBQC and photonic experiments towards implementing MBQC, our results open a new way towards realization of a large class of quantum speedups including those polynomial.
\end{abstract}

\maketitle

\section{\label{sec:introduction}Introduction}

\begin{figure}[t]
    \centering
    \includegraphics[width=3.4in]{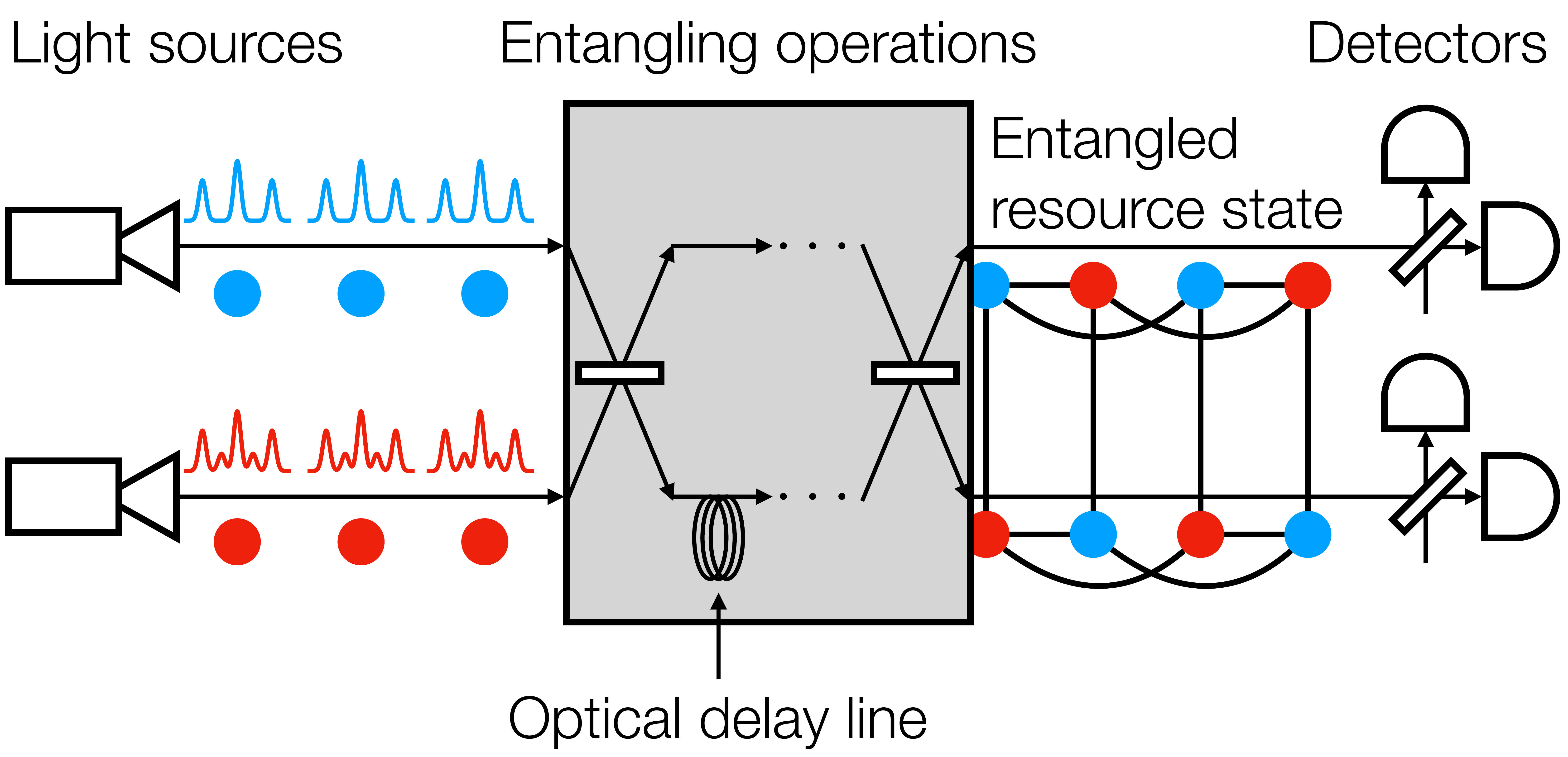}
    \caption{\label{fig:introduction}Measurement-based quantum computation (MBQC) through the time-domain multiplexing (TDM) approach. Light sources continuously emit optical modes.
      Shifting the timing by optical delay lines, we use different time domains of the light from the sources as subsystems (blue and red circles).
      We generate a resource state for MBQC by entangling these subsystems.
    Detectors immediately measure the resource state before decoherence. The resource state is independent of what to compute, and the MBQC is conducted by switching between different measurements. In our MBQC protocol, measurements are two types of homodyne detection in the position and momentum quadratures, and entangling operations can be Gaussian operations, while we concentrate the technological challenge of introducing non-Gaussianity on light sources, in the mode of which we encode qubits by means of the Gottesman-Kitaev-Preskill (GKP) code.}
\end{figure}

\begin{figure}[t]
    \centering
    \includegraphics[width=3.4in]{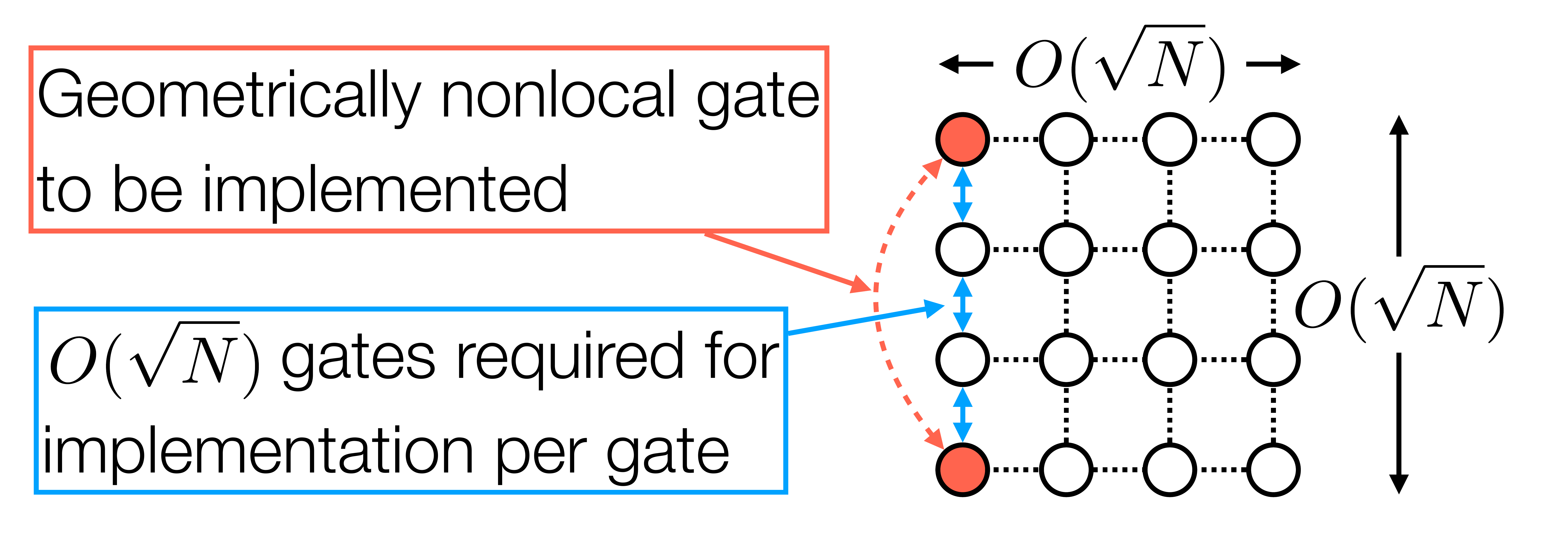}
    \caption{\label{fig:overhead}An example of polynomial overhead per implementation of a multiqubit gate on architectures with the constraints of geometrically local interactions, such as superconducting qubits. Suppose that $N$ qubits depicted as circles are allowed to interact only between those neighboring with respect to a $2$-dimensional geometrical constraint given by the square lattice in the figure. One way to implement a geometrically nonlocal two-qubit gate (red arrow) between distant qubits (red circles) is to implement $\textsc{SWAP}$ gates repeatedly between neighboring qubits along the blue arrows, so that the state of the two distant qubits can be brought to two neighboring qubits on which the two-qubit gate is implementable using the geometrically local interaction. In this case, the overhead of implementation per gate is polynomially large $O(\sqrt{N})$ in $N$ and potentially cancels out polynomial speedups in quantum computation such as those based on Grover search.}
\end{figure}

Photonic quantum technologies~\cite{C2,B8,doi:10.1063/1.5100160} provide a platform for generating multipartite quantum entanglement, correlation that is characteristic to multipartite quantum systems~\cite{H4,A12,Y5}, on a large scale that can be useful for quantum computation.
Quantum computation offers advantages over conventional classical computation in terms of speedup in computational time~\cite{W5,H3} and stronger security~\cite{B14,B15},
where whether the quantum speedup is exponential or polynomial may depend on computational tasks.
Owing to the nature of photonic systems that fly in space, photonic architectures can achieve the large-scale entanglement generation through a time-domain multiplexing (TDM) approach~\cite{Y3},
where different time domains of light emitted continuously from a light source are used many times, which are potentially unlimited, as subsystems for an entangled state to be generated.
Using the TDM approach, the experiment in Ref.~\cite{Y4} has generated a continuous-variable (CV) $1$-dimensional cluster state on the scale of $1.2\times 10^6$ entangled modes, and the experiments in Refs.~\cite{A7,Larsen369} have also generated large-scale $2$-dimensional CV cluster states.

Multipartite entanglement on this large scale matters to measurement-based quantum computation (MBQC)~\cite{B2,R3} and quantum error correction (QEC)~\cite{G,D,T10,B}.
MBQC, also known as one-way quantum computation, is a model of quantum computation implementable by generating a fixed multipartite entangled state as a computational resource, followed by sequentially measuring each subsystem of this resource entangled state.
In MBQC, the resource state is independent of what to compute, the computation is conducted by switching between measurements in different bases.
To achieve the quantum computation in a fault-tolerant way, QEC is a key technique for suppressing noise-induced errors in computation by redundantly representing a state of a logical quantum bit (qubit) as an entangled state of an adequate number of physical qubits.
Besides several strategies to implement photonic quantum computation such as the Knill, Laflamme, and Milburn scheme~\cite{K1} and the loop-based architecture~\cite{T1},
the scalability of the photonic architectures suggests a promising framework of MBQC through the TDM approach, using the large-scale entanglement generation for preparing the resource state, followed by immediate measurements of the resource to execute the computation~\cite{doi:10.1063/1.5100160,Y3,Y4,A7,Larsen369}, as illustrated in Fig.~\ref{fig:introduction}.

In this paper, we construct a \textit{low-overhead} protocol for implementing the fault-tolerant quantum computation using the photonic MBQC\@.
Algorithms for quantum computation are conventionally represented as quantum circuits where quantum logic gates may act on arbitrary qubits~\cite{N4},
but implementation of the algorithms in experiments may incur an overhead computational cost due to architectural constraints as well as QEC\@.
For example, matter-based superconducting qubits~\cite{W3,K8} arranged on a $2$-dimensional chip surface impose the constraints of geometrically local interactions only between neighboring qubits on the surface.
Within such an architectural constraint on the geometry, the overhead in applying a multiqubit gate to geometrically distant qubits can be polynomially large in terms of the number of qubits to mediate the distant interactions on the chip surface~\cite{B16}, as illustrated in Fig.~\ref{fig:overhead}.
While we may still achieve universal quantum computation only by geometrically local two-qubit gates, rewriting a geometrically nonlocal quantum circuit into a geometry-respecting circuit may increase the size of the circuit for interacting geometrically distant qubits;
\textit{e.g.}, as shown in Fig.~\ref{fig:overhead}, the implementation of the $\textsc{SWAP}$ gates repeated many times to mediate interactions at a distance incurs additional computational steps in the rewriting.
The additional computational steps (\textit{i.e.}, the additional circuit size) in the implementation of quantum computation, compared to the original geometrically nonlocal circuit, are referred to as the \textit{overhead}.

As for MBQC, whereas various multipartite entangled states are known to serve as resources~\cite{B1,R1,R2,B3,M9,R8,R7,R5,R6,W1,W2,K3,N1,N5,M2,M6,G4,Y1}, MBQC using the existing resource states similarly produces polynomial overheads since the resource states for MBQC are conventionally designed to be generated within the constraints of geometrically local interactions~\footnote{To be precise, in contrast with the other known resource states, the resource state in Ref.~\cite{Y1} cannot be generated within $2$-dimensional geometrical constraints, but the MBQC protocol in Ref.~\cite{Y1} incur a polynomial overhead cost.},
which would be relevant to the matter-based architectures, such as the superconducting qubits, or to spin systems appearing in condensed matter physics.
For example, the cluster state~\cite{B1,R1,R2}, that is, a graph state represented as a $2$-dimensional square lattice, can be generated by geometrically local interaction with respect to a square lattice.
To implement a given geometrically nonlocal circuit by MBQC using the cluster state,
we need to rewrite the given circuit into a circuit only with the geometrically local two-qubit gates that are implementable using the cluster state.
Due to this rewriting,
similarly to Fig.~\ref{fig:overhead},
MBQC using the cluster state incurs a polynomial overhead.
In general, the overhead arising from $2$-dimensional, or $d$-dimensional, geometrical constraints scales polynomially $O(N^\alpha)$ in terms of the number $N$ of the qubits in the circuit~\cite{B16}, where $\alpha>0$, and hence causes polynomial overhead in implementing quantum computation.

These polynomial overheads in implementing quantum computation are problematic because the polynomial overheads, totally or partially, cancel out promising polynomial speedups of quantum computation that have potential social impacts, such as those for \textsf{NP}-hard combinatorial optimization~\cite{H14,A8}, recommendation systems~\cite{K9,T3,A9}, machine learning~\cite{B17,C7,A10}, and a general class of quantum speedups based on Grover search~\cite{G7}.
Indeed, it has been an actual issue in Refs.~\cite{C5,C6,J2} how to compile a given geometrically nonlocal circuit into a circuit respecting the given geometry for an implementation at a small (yet polynomial) overhead cost, while the optimization in this compilation is computationally hard in general~\cite{J1,Botea2018OnTC,Siraichi2018,Tan2020}.
In contrast with the matter-based architectures with the geometrical constraints, the flying photonic systems are essentially free from the geometrical constraints.
However, as long as we use the resource states for MBQC that are generated by the geometry-respecting interactions, the photonic MBQC suffers from the same problem of the polynomial overheads.

To overcome this problem,
we here construct an MBQC protocol optimized for the photonic MBQC and achieving as low as a poly-logarithmic (polylog) overhead in implementing any quantum circuit composed of a computationally universal gate set,
where no geometrical constraint is assumed on the gates.
In our protocol, the resource state for MBQC is a multiqubit hypergraph state that can be encoded into optical modes using the Gottesman-Kitaev-Preskill (GKP) code~\cite{G1}, which encodes a discrete qubit into a CV mode.
This MBQC protocol is also implementable using discrete qubits, not only using the CV mode with the GKP code.
We refer to the logical qubit encoded into the CV mode via the GKP code as a \textit{GKP qubit}, and to the physical CV state of the GKP code as a \textit{GKP state}.

While MBQC requires adaptive switching of measurements to conduct the quantum computation, our MBQC protocol switches only between two types of \textit{homodyne detection} of the optical mode in the position quadrature and in the momentum quadrature.
Significantly, we can reliably perform this under current technologies by applying a phase shift to the local oscillator used in implementing the homodyne detection.
Furthermore, the resource state can be generated by experimentally tractable Gaussian operations for entangling the optical modes from light source, given that the light source can emit a GKP-codeword state denoted by $\Ket{0}$ and a GKP magic state denoted by $\Ket{\frac{\pi}{8}}$.
Our strategy concentrates the challenge in implementing quantum computation on realization of the states $\Ket{0}$ and $\Ket{\frac{\pi}{8}}$ of the GKP code, which are non-Gaussian.
Note that rather than realizing both $\Ket{0}$ and $\Ket{\frac{\pi}{8}}$, we can use a single source that emits only one of the states $\Ket{0}$ and $\Ket{\frac{\pi}{8}}$ to prepare the other by Gaussian operations~\cite{B12,Yamasaki2019}.
Non-Gaussian states and operations are technologically more costly to use than Gaussian states and operations,
but non-Gaussianity is necessary for universal quantum computation.
The non-Gaussianity is also indispensable for correcting errors in CV quantum computation~\cite{PhysRevLett.102.120501}, where the use of the GKP code is well motivated due to its high performance in QEC~\cite{F1,V1,N2}.
Advantageously, our resource state can be generated mostly using Gaussian operations and GKP $\Ket{0}$s similarly to the conventional multiqubit cluster states encoded by the GKP code; that is, the required number of logical non-Clifford gates on the GKP code for generating the resource state is kept small so as to simplify the resource state preparation.
Note that as long as we take the advantage of implementing MBQC only with homodyne detection, which performs Pauli-basis measurements of GKP qubits for our resource state, it is impossible to remove all the non-Clifford gates from our resource state due to the Gottesman-Knill theorem~\cite{G9,A6}.
The required logical non-Clifford gate, \textit{i.e.}, the controlled-controlled-$Z$ ($CCZ$) gate in our case, can be implemented by Gaussian operations on at most two optical modes combined with GKP magic states $\Ket{\frac{\pi}{8}}$s.
This implementation of a $CCZ$ gate may require a constant additional implementation cost of non-Gaussianity, but the cost of preparing a GKP magic state $\Ket{\frac{\pi}{8}}$ is comparable to that of non-Gaussian $\Ket{0}$ of the GKP code~\cite{Yamasaki2019}.
Owing to the periodicity of our resource state for MBQC, we can generate the resource state by a \textit{non-adaptive} optical circuit of the Gaussian operations.
Therefore, our MBQC protocol can be implemented feasibly in the all-Gaussian way as long as we can additionally realize the GKP code as the source of non-Gaussianity, where no other non-Gaussianity than the GKP code is necessary.

Our overhead reduction in MBQC is achieved by introducing an idea of sorting networks~\cite{K4} into the entanglement structure of our resource states for the MBQC, so that the multiqubit state during the computation can be rearranged in an arbitrary order at a \textit{polylog} overhead cost.
As a result, the required number of non-Gaussian GKP qubits for implementing our polylog-overhead MBQC protocol is significantly small compared to implementing the existing MBQC protocol with the polynomial overhead.
Moreover,
owing to the entanglement structure for efficient qubit permutation,
our MBQC protocol does not require the computationally hard compilation of a given geometrically nonlocal circuit into a geometry-respecting circuit.

Furthermore, based on this MBQC protocol, we construct a fault-tolerant MBQC protocol that leads to a QEC threshold outperforming that of the existing leading-edge protocol for CV quantum computation.
Following Refs.~\cite{M4,F2,F6,N3},
we evaluate the threshold in terms of the squeezing level of each non-Gaussian GKP state, $\Ket{0}$ and $\Ket{\frac{\pi}{8}}$, given initially by the source.
A smaller squeezing level of the GKP code means a worse initial quality of the GKP code, leading to more noise.
In contrast with the existing state-of-the-art technique for fault-tolerant CV quantum computation that combines topological QEC with the GKP code~\cite{F2,F6,N3,V4,Hanggli2020},
we use the concatenated $7$-qubit code~\cite{PhysRevLett.77.793,S3,N4} to represent our GKP-encoded multiqubit resource state for MBQC, discovering that not only the topological QEC but also the concatenated QEC can provide a promising scheme for fault-tolerant CV quantum computation.
Significantly, our numerical calculation shows that our protocol achieves the threshold $7.8$~dB in terms of the required squeezing level and outperforms the previous leading-edge protocol with threshold $8.3$~dB~\cite{F6} in the same noise model.
With this threshold, our fault-tolerant MBQC protocol is applicable to MBQC using an \textit{arbitrary} multiqubit resource state at a logical level, in contrast with the existing fault-tolerant MBQC protocols~\cite{R8,R7,R5,R6,N1,N5} requiring a \textit{fixed} geometry of interaction, such as the $3$D square lattice for MBQC using the $3$D cluster state.
Applying this protocol to our resource hypergraph state, we can achieve fault-tolerant MBQC at a polylog overhead cost.

To improve the threshold,
we introduce two new QEC techniques for reducing errors that we need to suppress in fault-tolerant computation with the GKP code.
The first technique reduces CV errors arising from variances of Gaussian functions representing a finitely squeezed GKP state of a physical CV mode.
Instead of a conventional variance-reduction algorithm, we introduce a technique to enhance performance of the variance reduction by adjusting variances of an auxiliary GKP qubit used in the algorithm.
The second technique reduces bit- and phase-flip errors of GKP qubits at the logical level of the GKP code.
To achieve this, we establish a low-overhead error-detection technique for post-selecting constant-size building blocks for generating our resource state in high fidelity, by using the $7$-qubit code not only as an error-correcting code but also as an error-detecting code.
Our protocol begins with preparing high-fidelity building-block states by post-selection to attain high fault tolerance,
and then, we switch into deterministically transforming these building blocks into our resource state to achieve low overhead asymptotically.
Each post-selection of the constant-size building block may incur at most a constant overhead, and our MBQC protocol including QEC is designed to achieve the \textit{polylog} overall overhead cost.
Remarkably, the asymptotic scaling of the overhead in QEC is the same as that of the existing low-overhead protocols for fault-tolerant computation with the $7$-qubit code~\cite{G6,Chamberland2019faulttolerantmagic} up to a constant factor.
Moreover, we can reduce the constant factor of the overhead by controlling a parameter of our fault-tolerant protocol depending on how much better squeezing level of the GKP code we can realize than the threshold $7.8$ dB.

Consequently, the combination of these MBQC protocols establishes a framework for realizing polylog-overhead fault-tolerant quantum computation, without canceling out a large class of quantum speedups including those polynomial.
Furthermore, the generality of our protocols suggests a wide range of applications beyond implementing quantum computation; \textit{e.g.}, our polylog-overhead MBQC protocol can be useful for achieving polylog overhead in blind quantum computation based on MBQC~\cite{B14,B15}, and our fault-tolerant resource preparation protocol is applicable to preparing multipartite entangled states that serve as resources for distributed quantum information processing over optical networks~\cite{Wehner2018,Y14,Y6,Y5}.
Our results open up a route of taking flight from the geometrically constrained architectures, towards realizing quantum information processing without canceling out not only exponential but also polynomial quantum advantages in terms of speedup.

The rest of this paper is structured as follows.
In Sec.~\ref{sec:preliminaries}, we provide preliminaries to the GKP code, universal MBQC, techniques for fault-tolerant CV MBQC, and overhead cost in implementing quantum computation.
In Sec.~\ref{sec:resource_state}, we introduce and analyze an MBQC protocol using a multiqubit resource hypergraph state for implementing quantum computation at a polylog overhead cost.
Based on this MBQC protocol, we construct a protocol for achieving fault-tolerant MBQC in Sec.~\ref{sec:fault_tolerant}.
Our conclusion is given in Sec.~\ref{sec:conclusion}.

\section{\label{sec:preliminaries}Preliminaries}

In this section, we provide preliminaries to this paper.
In Sec.~\ref{sec:gkp}, we recall the Gottesman-Kitaev-Preskill (GKP) code and fix notations for continuous-variable (CV) quantum computation.
In Sec.~\ref{sec:mbqc}, we recall the idea of measurement-based quantum computation (MBQC) and clarify the definition of the universality of MBQC\@.
In Sec.~\ref{sec:qec_gkp},
we summarize techniques for quantum error correction (QEC) and fault-tolerant CV MBQC\@.
In Sec.~\ref{sec:overhead},
we define the overhead cost in implementing quantum computation.
The readers who are familiar with the GKP code can skip Sec.~\ref{sec:gkp}, and those who are familiar with MBQC can skip Sec.~\ref{sec:mbqc}.
The readers who are interested only in Sec.~\ref{sec:resource_state} on our results of polylog-overhead MBQC can skip Secs.~\ref{sec:gkp} and~\ref{sec:qec_gkp}.

\subsection{\label{sec:gkp}Gottesman-Kitaev-Preskill (GKP) code}

In this section, we recall the GKP code. We first introduce the notion of ideal GKP code, and then proceed to introducing the approximate GKP code that we consider in this paper.

\textbf{Ideal GKP code}:
The GKP code~\cite{G1} is a CV code for encoding qubits
into oscillators' position quadrature $\hat{q}$ and momentum quadrature $\hat{p}$,
\begin{align}
  &\hat{q}=\frac{1}{\sqrt{2}}\left(\hat{a}+\hat{a}^\dag\right),&&\hat{p}=\frac{1}{\sqrt{2}\mathrm{i}}\left(\hat{a}-\hat{a}^\dag\right),\\
  \Leftrightarrow\quad&\hat{a}=\frac{1}{\sqrt{2}}\left(\hat{q}+\mathrm{i}\hat{p}\right),&&\hat{a}^\dag=\frac{1}{\sqrt{2}}\left(\hat{q}-\mathrm{i}\hat{p}\right),
\end{align}
and $\hat{a}$ and $\hat{a}^\dag$  are the annihilation and creation operators, respectively, of a bosonic mode representing a CV system~\cite{C2,B8}.
The GKP code is designed to correct displacement errors in the $\hat{q}$ and $\hat{p}$ quadratures, and this method of quantum error correction (QEC) can achieve nearly optimal performance against Gaussian errors~\cite{F1,V1,N2}.
Regarding notations on qubit-based quantum computation~\cite{N4},
Pauli operators on a qubit are written as
\begin{align}
    X&\coloneqq\Ket{0}\Bra{1}+\Ket{1}\Bra{0},\\
    Y&\coloneqq\mathrm{i}\Ket{1}\Bra{0}-\mathrm{i}\Ket{0}\Bra{1},\\
    Z&\coloneqq\Ket{0}\Bra{0}-\Ket{1}\Bra{1}.
\end{align}
The $X$ basis refers to
\begin{equation}
    \left\{\Ket{+}\coloneqq\frac{1}{\sqrt{2}}\left(\Ket{0}+\Ket{1}\right),\Ket{-}\coloneqq\frac{1}{\sqrt{2}}\left(\Ket{0}-\Ket{1}\right)\right\},
\end{equation}
the $Y$ basis refers to
\begin{equation}
    \left\{\frac{1}{\sqrt{2}}\left(\Ket{0}+\mathrm{i}\Ket{1}\right),\frac{1}{\sqrt{2}}\left(\Ket{0}-\mathrm{i}\Ket{1}\right)\right\},
\end{equation}
and the $Z$ basis refers to the computational basis
\begin{equation}
    \left\{\Ket{0},\Ket{1}\right\}.
\end{equation}
A one-qubit measurement in the $X$, $Y$, and $Z$ bases is called the $X$, $Y$, and $Z$ measurement, respectively, where the bit values $0$ and $1$ of the measurement outcome refer to those corresponding to $\Ket{0}$ and $\Ket{1}$ in the $Z$ measurement, respectively, and to $\Ket{+}$ and $\Ket{-}$ in the $X$ measurement.
The identity operator is denoted by $\mathbbm{1}$.

Each codeword of a GKP code is ideally a superposition of infinitely many Gaussian states that are infinitely squeezed.
The simplest class of the GKP codes is the one-mode square lattice GKP code, and its orthogonal codewords corresponding to the $Z$ basis $\left\{\Ket{0},\Ket{1}\right\}$ are represented in terms of eigenstates $\Ket{q}$ of $\hat{q}$ as
\begin{align}
    \label{eq:gkp_codeword_0}
    \Ket{0}&\propto\sum_{s=-\infty}^{\infty}\Ket{q=2s\sqrt{\pi}},\\
    \label{eq:gkp_codeword_1}
    \Ket{1}&\propto\sum_{s=-\infty}^{\infty}\Ket{q=\left(2s+1\right)\sqrt{\pi}},
\end{align}
where we have $\hat{q}\Ket{q=\tilde{q}}=\tilde{q}\Ket{q=\tilde{q}}$ for any $\tilde{q}\in\mathbb{R}$.
In this case, the codewords corresponding to the $X$ basis $\left\{\Ket{+},\Ket{-}\right\}$ are represented in terms of eigenstates $\Ket{p}$ of $\hat{p}$ as
\begin{align}
    \Ket{+}&\coloneqq\frac{1}{\sqrt{2}}\left(\Ket{0}+\Ket{1}\right)\propto\sum_{s=-\infty}^{\infty}\Ket{p=2s\sqrt{\pi}},\\
    \Ket{-}&\coloneqq\frac{1}{\sqrt{2}}\left(\Ket{0}-\Ket{1}\right)\propto\sum_{s=-\infty}^{\infty}\Ket{p=\left(2s+1\right)\sqrt{\pi}},
\end{align}
where we have $\hat{p}\Ket{p=\tilde{p}}=\tilde{p}\Ket{p=\tilde{p}}$ for any $\tilde{p}\in\mathbb{R}$.
We refer to the logical qubit encoded in this GKP code as a GKP qubit,
and a physical state of the GKP code as a GKP state.
As shown in Ref.~\cite{G1}, Gaussian operations~\cite{C2,B8} suffice to perform Clifford gates on GKP qubits, such as the $H$, $S$, and controlled-$Z$ ($CZ$) gates
\begin{align}
  H&\coloneqq\frac{1}{\sqrt{2}}\left(\Ket{0}\Bra{0}+\Ket{0}\Bra{1}+\Ket{1}\Bra{0}-\Ket{1}\Bra{1}\right),\\
  S&\coloneqq\Ket{0}\Bra{0}+\mathrm{i}\Ket{1}\Bra{1},\\
  \label{eq:cz}
  CZ&\coloneqq\Ket{0}\Bra{0}\otimes\mathbbm{1}+\Ket{1}\Bra{1}\otimes Z\nonumber\\
    &=\sum_{b_1,b_2=0}^{1}{\left(-1\right)}^{b_1 b_2}\Ket{b_1}\Bra{b_1}\otimes\Ket{b_2}\Bra{b_2}.
\end{align}
These gates are implemented by symplectic transformations of quadratures
\begin{align}
  H:\,&\hat{q}\to\hat{p},&&\hat{p}\to -\hat{q},\\
  S:\,&\hat{q}\to\hat{q},&&\hat{p}\to \hat{p}-\hat{q},\\
  \label{eq:cz_quadrature}
  CZ:\,&\hat{q}_1\to\hat{q}_1,&&\hat{p}_1\to \hat{p}_1+\hat{q}_2,\nonumber\\
       &\hat{q}_2\to \hat{q}_2,&&\hat{p}_2\to \hat{p}_2+\hat{q}_1,
\end{align}
where $\hat{q}_1, \hat{p}_1$ and $\hat{q}_2, \hat{p}_2$ are quadratures of the control and target modes, respectively.
Note that to implement Clifford unitary transformations, instead of the $CZ$ gate, we can use the controlled-$\textsc{NOT}$ ($\textsc{CNOT}$) gate, which we also call the controlled-$X$ ($CX$) gate
\begin{equation}
  \textsc{CNOT}\coloneqq\Ket{0}\Bra{0}\otimes\mathbbm{1}+\Ket{1}\Bra{1}\otimes X,
\end{equation}
and the $\textsc{CNOT}$ gate on GKP qubits is implemented by
\begin{equation}
  \label{eq:cx_quadrature}
  \begin{aligned}
    \textsc{CNOT}:\,&\hat{q}_1\to \hat{q}_1,&&\hat{p}_1\to \hat{p}_1-\hat{p}_2,\\
           &\hat{q}_2\to \hat{q}_1+\hat{q}_2,&&\hat{p}_2\to \hat{p}_2,
  \end{aligned}
\end{equation}
where the notations are the same as those in~\eqref{eq:cz_quadrature}.

\begin{figure}[t]
    \centering
    \includegraphics[width=3.4in]{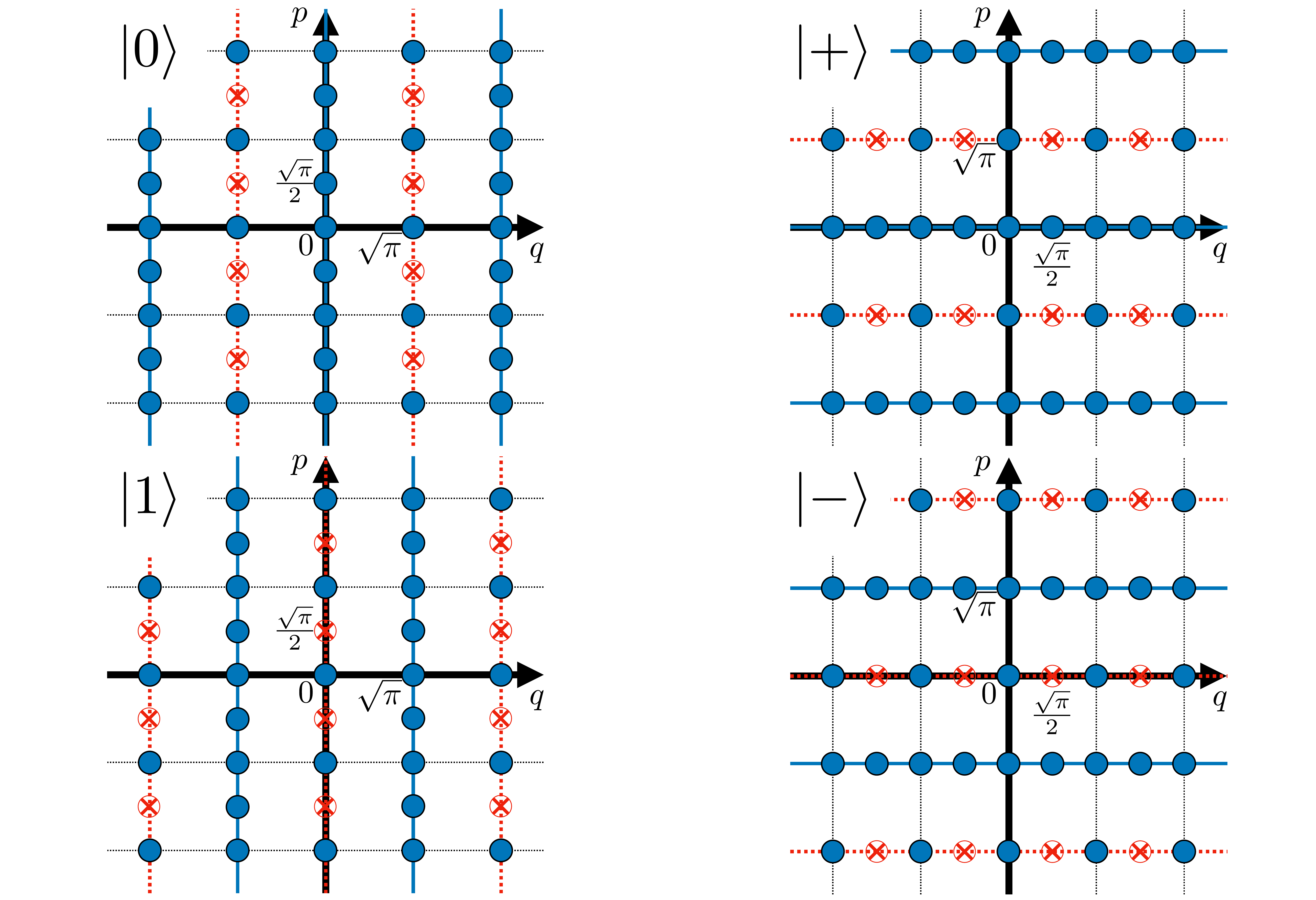}
    \caption{\label{fig:wigner_function}Wigner functions of GKP states in the $Z$ and $X$ bases of a GKP qubit, where each blue filled circle represents a positive delta function $\delta$, and each red circled X represents the corresponding negative delta function $-\delta$; \textit{e.g.}, $\Ket{0}$ is represented by $W(q,p)\propto\sum_{s,t\in\mathbb{Z}}{(-1)}^{st}\delta\left(q-\sqrt{\pi}s\right)\delta\left(p-\frac{\sqrt{\pi}}{2}t\right)$. The $Z$ basis $\left\{\Ket{0},\Ket{1}\right\}$ can be identified by measuring $\hat{q}$, where the blue solid lines at $2\sqrt{\pi}$ intervals represent the peaks of the probability density function in measuring $\hat{q}$, while the probability vanishes on red dotted lines due to interference. The $X$ basis $\left\{\Ket{+},\Ket{-}\right\}$ can be identified by measuring $\hat{p}$, similarly to the $Z$ basis.}
\end{figure}

The $Z$ and $X$ measurements of the GKP qubits can be performed by measuring the quadratures $\hat{q}$ and $\hat{p}$, respectively, and hence can be implemented by the most common Gaussian measurement, homodyne detection~\cite{B8}.
The ideal GKP states in the $Z$ and $X$ basis are illustrated in Fig.~\ref{fig:wigner_function} as Wigner functions, where the Winger function of a density operator $\hat\psi$ is defined as
\begin{equation}
  W\left(q,p\right)\coloneqq\frac{1}{2\pi}\int_{-\infty}^{\infty}dx\,\mathrm{e}^{\mathrm{i}xp}\Braket{q-\frac{x}{2}|\hat{\psi}|q+\frac{x}{2}},
\end{equation}
and the probability density functions in measuring $\hat{q}$ and $\hat{p}$ are represented, respectively, as
\begin{align}
    \Braket{q|\hat{\psi}|q}&=\int_{-\infty}^{\infty}W\left(q,p\right)dp,\\
    \Braket{p|\hat{\psi}|p}&=\int_{-\infty}^{\infty}W\left(q,p\right)dq.
\end{align}
The $Z$ measurement of a GKP qubit is implemented by homodyne detection of $\hat{q}$,
and the $X$ measurement by that of $\hat{p}$.
In the ideal case of infinitely squeezed GKP qubits for $\Ket{0}$ in~\eqref{eq:gkp_codeword_0} and $\Ket{1}$ in~\eqref{eq:gkp_codeword_1}, from the measurement outcome $\tilde{q}\in\mathbb{R}$ of the homodyne detection, we can obtain the logical bit value $\tilde{b}\in\left\{0,1\right\}$ of the $Z$ measurement of the GKP qubit by
\begin{equation}
  \label{eq:true_bit}
  \tilde{b}=\begin{cases}
    0&\text{if }\tilde{q}\in\left\{2s\sqrt{\pi}:s\in\mathbb{Z}\right\},\\
    1&\text{if }\tilde{q}\in\left\{\left(2s+1\right)\sqrt{\pi}:s\in\mathbb{Z}\right\}.
  \end{cases}
\end{equation}

\begin{figure}[t]
  \centering
  \includegraphics[width=2.0in]{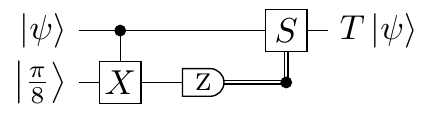}
  \caption{\label{fig:gate_teleportation}A state-injection quantum circuit for applying the $T$ gate to any input one-qubit state $\Ket{\psi}$. The circuit is composed only of Clifford gates except the auxiliary qubit prepared in the $\frac{\pi}{8}$ state $\Ket{\frac{\pi}{8}}$ defined as~\eqref{eq:t_state}, where the Clifford gates on GKP qubits can be implemented by Gaussian operations.}
\end{figure}

As for non-Clifford gates on GKP qubits, Ref.~\cite{G1} provides protocols using Gaussian operations and an auxiliary mode prepared in a GKP Hadamard eigenstate, a GKP $\frac{\pi}{8}$ state, or a cubic phase state for deterministically performing $T$ gates on a GKP qubit
\begin{equation}
    T\coloneqq\Ket{0}\Bra{0}+\mathrm{e}^{\mathrm{i}\frac{\pi}{4}}\Ket{1}\Bra{1}.
\end{equation}
In particular, the GKP $\frac{\pi}{8}$ state is defined as
\begin{equation}
  \label{eq:t_state}
  \Ket{\frac{\pi}{8}}\coloneqq T\Ket{+}=\frac{1}{\sqrt{2}}\left(\Ket{0}+\mathrm{e}^{\mathrm{i}\frac{\pi}{4}}\Ket{1}\right),
\end{equation}
and the $T$ gate can be applied to any input state $\Ket{\psi}$ by performing Gaussian operations that implement the quantum circuit in Fig.~\ref{fig:gate_teleportation} on GKP qubits.
Using the protocol for implementing $T$, we can deterministically implement the controlled-controlled-$Z$ ($CCZ$) gate, another non-Clifford gate, on GKP qubits
\begin{equation}
    \label{eq:ccz}
    CCZ\coloneqq\sum_{b_1,b_2,b_3=0}^{1}{\left(-1\right)}^{b_1 b_2 b_3}\Ket{b_1}\Bra{b_1}\otimes\Ket{b_2}\Bra{b_2}\otimes\Ket{b_3}\Bra{b_3},
\end{equation}
by a quantum circuit composed of Clifford gates (\textit{e.g.}, $H$, $S$, and $CZ$), four $T$ gates, and one auxiliary GKP qubit prepared in $\Ket{0}$~\cite{S2,C1,M1} (as we will see later in Fig.~\ref{fig:preparation_ccz} in Sec.~\ref{sec:resource_generation}),
or by a quantum circuit composed of Clifford gates and seven $T$ gates without any auxiliary GKP qubit~\cite{A4,G8}.
Similarly to $\Ket{\frac{\pi}{8}}$, we may write
\begin{align}
  \label{eq:ccz_state}
  \Ket{CCZ}&\coloneqq CCZ\Ket{+}^{\otimes 3}\nonumber\\
           &=\frac{1}{2\sqrt{2}}\sum_{b_1,b_2,b_3=0}^{1}{\left(-1\right)}^{b_1 b_2 b_3}\Ket{b_1}\otimes\Ket{b_2}\otimes\Ket{b_3}.
\end{align}

In this paper, in addition to Gaussian operations, light sources that can emit $\Ket{0}$ and $\Ket{\frac{\pi}{8}}$ of GKP qubits are assumed to be available, up to an approximation shown in the following.
Since $\Ket{0}$ and $\Ket{\frac{\pi}{8}}$ of the GKP code can be generated by Gaussian operations from $\Ket{\frac{\pi}{8}}$s and $\Ket{0}$s respectively~\cite{B12,Yamasaki2019},
one of these light sources are actually optional for universal quantum computation, but to simplify the presentation, we assume that both are available.

\textbf{Approximate GKP code}:
The codewords of the ideal GKP codes, such as the right-hand sides of~\eqref{eq:gkp_codeword_0} and~\eqref{eq:gkp_codeword_1}, are non-normalized and hence unphysical,
but this normalization problem can be circumvented by considering a superposition of finitely squeezed states weighted by a Gaussian envelope~\cite{G1,M10}.
In particular, we substitute $\Ket{q}$ in~\eqref{eq:gkp_codeword_0} and~\eqref{eq:gkp_codeword_1} with a normalized finitely squeezed vacuum state
\begin{equation}
  \label{eq:squeezed_vacuum}
  \hat{S}\left(-\ln\sqrt{2\sigma^2}\right)\Ket{\mathrm{vac}}={\left(\frac{1}{2\pi\sigma^2}\right)}^{\frac{1}{4}}\int_{-\infty}^{\infty}dq\,\mathrm{e}^{-\frac{1}{2\left(2\sigma^2\right)}q^2}\Ket{q},
\end{equation}
where $\Ket{\mathrm{vac}}$ is the vacuum state, $\hat{S}$ is the squeezing operator defined as
\begin{equation}
    \hat{S}\left(r\right)\coloneqq \mathrm{e}^{\frac{r}{2}\left({\hat{a}}^2 - {\hat{a}}^{\dag 2}\right)},\,r\in\mathbb{R},
\end{equation}
and $\sigma^2$ represents the variance of the Gaussian probability density function obtained by measuring the state~\eqref{eq:squeezed_vacuum} in the $\hat{q}$ quadrature.
Using the variance of $\Ket{\mathrm{vac}}$, \textit{i.e.}, $\frac{1}{2}$, as the reference,
we may represent the squeezing level corresponding to $\sigma^2$ in terms of the decibel (dB)
\begin{equation}
  \label{eq:squeezing_level}
  -10\log_{10}\left(2\sigma^2\right),
\end{equation}
where the negative sign is by convention~\cite{M4}.
To obtain approximate GKP codewords,
we consider the superposition of the squeezed vacuum states~\eqref{eq:squeezed_vacuum} displaced and weighted by a Gaussian envelope in such a way that the code space spanned by the codewords is symmetric, \textit{i.e.}, invariant under reversal of $\hat{q}$ and $\hat{p}$.
As a result, we have the following standard form~\cite{M10} of approximately orthogonal codewords of the symmetric approximate GKP code that converges to those of the ideal GKP code as $\sigma\to 0$
\begin{widetext}
  \begin{align}
    \label{eq:gkp_approximate_codeword_0}
    &\begin{aligned}
      \Ket{0}&\propto\sum_{s=-\infty}^{\infty}\int_{-\infty}^{\infty}{dq}\,\mathrm{e}^{-\frac{2\sigma^2}{2\left(1-4\sigma^4\right)}{\left(2s\sqrt{\left(1-4\sigma^4\right)\pi}\right)}^2}\mathrm{e}^{-\frac{1}{2\left(2\sigma^2\right)}{\left(q-2s\sqrt{\left(1-4\sigma^4\right)\pi}\right)}^2}\Ket{q}\\
             &=\sum_{s=-\infty}^{\infty}\mathrm{e}^{-\frac{2\sigma^2}{2\left(1-4\sigma^4\right)}{\left(2s\sqrt{\left(1-4\sigma^4\right)\pi}\right)}^2}\hat{V}\left(2s\sqrt{\left(1-4\sigma^4\right)\pi},0\right)\hat{S}\left(-\ln\sqrt{2\sigma^2}\right)\Ket{\mathrm{vac}}\\
             &=\hat{S}\left(-\ln\sqrt{1-4\sigma^4}\right)\sum_{s=-\infty}^{\infty}\mathrm{e}^{-\frac{2\sigma^2}{2}{\left(2s\sqrt{\pi}\right)}^2}\hat{V}\left(2s\sqrt{\pi},0\right)\hat{S}\left(-\ln\sqrt{\frac{2\sigma^2}{1-4\sigma^4}}\right)\Ket{\mathrm{vac}},\\
    \end{aligned}\\
    \label{eq:gkp_approximate_codeword_1}
    &\begin{aligned}
      \Ket{1}&\propto\sum_{s=-\infty}^{\infty}\int_{-\infty}^{\infty}{dq}\,\mathrm{e}^{-\frac{2\sigma^2}{2\left(1-4\sigma^4\right)}{\left(\left(2s+1\right)\sqrt{\left(1-4\sigma^4\right)\pi}\right)}^2}\mathrm{e}^{-\frac{1}{2\left(2\sigma^2\right)}{\left(q-\left(2s+1\right)\sqrt{\left(1-4\sigma^4\right)\pi}\right)}^2}\Ket{q}\\
             &=\sum_{s=-\infty}^{\infty}\mathrm{e}^{-\frac{2\sigma^2}{2\left(1-4\sigma^4\right)}{\left(\left(2s+1\right)\sqrt{\left(1-4\sigma^4\right)\pi}\right)}^2}\hat{V}\left(\left(2s+1\right)\sqrt{\left(1-4\sigma^4\right)\pi},0\right)\hat{S}\left(-\ln\sqrt{2\sigma^2}\right)\Ket{\mathrm{vac}}\\
             &=\hat{S}\left(-\ln\sqrt{1-4\sigma^4}\right)\sum_{s=-\infty}^{\infty}\mathrm{e}^{-\frac{2\sigma^2}{2}{\left(\left(2s+1\right)\sqrt{\pi}\right)}^2}\hat{V}\left(\left(2s+1\right)\sqrt{\pi},0\right)\hat{S}\left(-\ln\sqrt{\frac{2\sigma^2}{1-4\sigma^4}}\right)\Ket{\mathrm{vac}},
    \end{aligned}
    \end{align}
\end{widetext}
where $\hat{V}\left(r_q,r_p\right)$ is the generalized Weyl-Heisenberg displacement operator~\cite{M10} defined as
\begin{align}
  \label{eq:translation}
  &\hat{V}\left(r_q,r_p\right)\coloneqq\mathrm{e}^{-\frac{\mathrm{i}}{2}r_q r_p}\mathrm{e}^{\mathrm{i}r_p \hat{q}}\mathrm{e}^{-\mathrm{i}r_q \hat{p}},\,r_q,r_p\in\mathbb{R},\\
  &\hat{V}\left(r_q,0\right)\Ket{q=0}=\Ket{q=r_q},\\
  &\hat{V}\left(0,r_p\right)\Ket{p=0}=\Ket{p=r_p},
\end{align}
$\hat{S}\left(-\ln\sqrt{1-4\sigma^4}\right)$ contributes to making the approximate GKP code symmetric, and by convention, we will approximate $1-4\sigma^4\approx 1$ to simplify the presentation in this paper.
Note that $\hat{V}\left(r_q,r_p\right)$ can be represented in terms of the conventional displacement operator $\hat{D}\left(\alpha\right)$ as
\begin{equation}
  \hat{V}\left(r_q,r_p\right)=\hat{D}\left(\frac{r_q+\mathrm{i}r_p}{\sqrt{2}}\right),
\end{equation}
where $\hat{D}\left(\alpha\right)$ is defined as
\begin{equation}
  \hat{D}\left(\alpha\right)\coloneqq \mathrm{e}^{\alpha {\hat{a}}^\dag - \alpha^\ast \hat{a}},\,\alpha\in\mathbb{C}.
\end{equation}
When we measure a GKP qubit in the $\hat{q}$ quadrature,
each finitely squeezed vacuum state in the superposition of the approximate GKP codewords yields a Gaussian peak in the probability density function of the real-valued measurement outcome.
For the symmetric GKP code, a measurement in the $\hat{p}$ quadrature yields Gaussian peaks in the same positions as those in $\hat{q}$.
In this regard, the variances of the GKP qubit in $\hat{q}$ and $\hat{p}$ quadratures refer to the variances of the squeezed vacuum state corresponding to each of these Gaussian peaks.
In particular, for the symmetric approximate GKP code~\eqref{eq:gkp_approximate_codeword_0} and~\eqref{eq:gkp_approximate_codeword_1},
the squeezing level of the GKP code with the variance $\sigma^2$ is defined as~\eqref{eq:squeezing_level}.
Note that as we will discuss in Sec.~\ref{sec:qec_gkp}, even if initialized symmetrically, a GKP qubit during the computation may have different variances in $\hat{q}$ and $\hat{p}$ quadratures, which we write as $\sigma_q^2$ and $\sigma_p^2$, respectively.

\subsection{\label{sec:mbqc}Measurement-based quantum computation (MBQC)}

In this section,
we provide preliminaries to MBQC\@.
We first define MBQC and hypergraph states.
Then, we give a definition of complexity and universality of MBQC used in this paper, where our definition of complexity is based on the size of quantum circuits, rather than the depth of the circuits.
Finally, we describe properties of MBQC protocols that may affect implementability of MBQC, especially when Clifford operations are more tractable than non-Clifford operations as is the case of the GKP code.

\textbf{MBQC and hypergraph states}:
MBQC~\cite{B1,R1,R2} is a model of quantum computation that is performed by initially preparing a multipartite entangled resource state, followed by sequential quantum measurements of each subsystem for this resource state~\cite{B2,R3}.
The resource state for MBQC is independent of given computational tasks,
and the computation is driven by the measurements whose bases and temporal order can be adaptively chosen depending on the outcomes of previous measurements.
This adaptation of the measurements is necessary for correcting byproduct operators, that is, randomly applied unitary transformations caused by previous measurements.
The classical process of determining next measurements from the previous measurement outcomes is called feed-forward.

Several protocols for MBQC are known, and different protocols may use different families of resource states, such as the cluster state~\cite{B1,R1,R2}, the brickwork state~\cite{B3}, the graph states on the triangular lattice~\cite{M9}, the $3$D cluster states~\cite{R8,R7,R5,R6}, the Affleck-Kennedy-Lieb-Tasaki (AKLT) state~\cite{W1,W2}, the parity-phase (weighted) graph states~\cite{K3}, and the crystal-structure cluster states~\cite{N1,N5}.
References~\cite{M2,M6,G4,Y1} have shown that there also exist MBQC protocols using hypergraph states~\cite{Q1,R4,G5} as a resource, which will play a key role in our construction of MBQC protocols.
Hypergraph states are multiqubit entangled states that are defined using mathematical structure of hypergraphs, including graph states~\cite{H13} as special cases.
While a graph consists of vertices and edges connecting two vertices,
a hypergraph represents a generalized structure consisting of vertices and hyperedges connecting any number of vertices.
We may refer to a hyperedge between two vertices of a hypergraph as an edge between the two.
Let
\begin{equation}
    G=\left(V,E\right)
\end{equation}
be a hypergraph of $M$ vertices, where
\begin{equation}
    V=\left\{v_1,\ldots,v_M\right\}
\end{equation}
is the set of vertices, and
\begin{equation}
    E=\left\{e_1,e_2,\ldots\right\}
\end{equation}
is the set of hyperedges satisfying for each $e\in E$
\begin{equation}
    e=\left\{v_{m_1},v_{m_2},\ldots\right\}\subseteq V.
\end{equation}
The hypergraph state for $G$ is an $M$-qubit state denoted by
\begin{equation}
    \Ket{G}
\end{equation}
and obtained as follows:
first, for each $v_m\in V$, a qubit labeled $v_m$ is initialized as $\Ket{+}^{v_m}$, and then, for each hyperedge $e=\left\{v_{m_1},\ldots,v_{m_k}\right\}\in E$ and $k=\left|e\right|$, we apply a generalized controlled-$Z$ gate $C^{k}Z$~\cite{R4} to $k$ qubits corresponding to the $k$ vertices $v_{m_1},\ldots,v_{m_k}$ in the hyperedge $e$, where $\left|e\right|$ denotes the cardinality of $e$, and
\begin{equation}
    \label{eq:generalized_cz}
    C^{k}Z\coloneqq\mathbbm{1}^{v_{m_1},\ldots,v_{m_k}}-2\Ket{11\cdots 1}\Bra{11\cdots 1}^{v_{m_1},\ldots,v_{m_k}}.
\end{equation}
The generalized controlled-$Z$ gate $C^{k}Z$ is symmetric under permutation of qubits, and reduces to the $Z$ gate, the $CZ$ gate, and the $CCZ$ gate in the case of $k=1,2,3$, respectively.
In the same way as the hypergraph states,
graph states are defined for graphs, \textit{i.e.}, a special class of hypergraphs, where all the edges of the graphs correspond to $CZ$ gates.
While graph states can be prepared using only Clifford gates, that is, performing $CZ$ gates on $\Ket{+}$s,
a preparation of a hypergraph state may require non-Clifford gates since the generalized controlled-$Z$ gate in~\eqref{eq:generalized_cz} can be a non-Clifford gate such as the $CCZ$ gate.
As graphs are special hypergraphs, we may also refer to graph states as hypergraph states.

\textbf{Universality and complexity of MBQC}:
In this paper, especially in Sec.~\ref{sec:resource_state}, we will analyze the complexity of our MBQC protocol defined based on sequential implementation.
This definition especially matters to photonic MBQC based on the TDM approach using sequential resource state preparation and measurements.
We will also analyze required depths for resource state preparation and measurements in Appendices~\ref{sec:depth} and~\ref{sec:parallelizability}, respectively.

We refer to the number of computational steps of an MBQC protocol as the \textit{complexity} of the MBQC protocol, which consists of the preparation complexity, the quantum complexity, and the classical complexity~\cite{V2}.
Our definition of complexity is based on the size of quantum circuits, rather than the depth of the circuits.
The size of a quantum circuit composed of a certain gate set refers to the total number of elementary gates contained in the circuit, that is, the time steps required for performing the circuit under the condition that elementary gates in the circuit are applied sequentially.
The depth of a quantum circuit refers to the minimal time steps required for performing the circuit under the condition that at each time step, only at most a single elementary gate can be applied to each qubit,
but elementary gates on disjoint sets of qubits can be applied in parallel.

In particular, we define the complexity of the MBQC protocol as follows.
The preparation complexity refers to the required size of the circuit for preparing the resource state for the MBQC protocol from a fixed product state $\Ket{0}\otimes\Ket{0}\otimes\cdots$, where the resource state preparation in this paper is performed using a gate set that can exactly implement the unitary transformations $H$, $CZ$, and $CCZ$ using a fixed number of the elementary gates, such as the gate sets
\begin{align}
  &\left\{H, CZ, T\right\},\\
  \label{eq:H_S_CZ_CCZ}
  &\left\{H, S, CZ, CCZ\right\}.
\end{align}
The quantum complexity of the MBQC protocol is the required total size of quantum circuits for performing measurements in appropriately adapted bases, that is, for implementing each of these measurements in MBQC by a unitary transformation followed by a measurement in the computational basis.
The classical complexity of the MBQC protocol consists of the required total size of classical circuits for describing the classical processes in the protocol, which include preprocessing for obtaining the description of the quantum circuit for implementing the MBQC protocol from the given computational problem, and also include feed-forwards for providing a temporal order and bases of next measurements (\textit{i.e.}, the measurement patterns) from the sequence of the previous measurement outcomes.

To discuss the notion of universal MBQC, recall that universality in the circuit model of quantum computation can be defined in two different ways, \textit{i.e.}, strict universality and computational universality.
In the circuit model, a fixed initial $N$-qubit state $\Ket{0}^{\otimes N}$ in the computational basis, \textit{i.e.}, the $Z$ basis $\left\{\Ket{0},\Ket{1}\right\}$, is transformed by a sequence of quantum gates in a given gate set, and classical outcomes of the computation are obtained from the final state of this unitary transformation by single-qubit measurements in the computational basis.
A gate set is strictly universal if it can implement an arbitrary $N$-qubit unitary transformation with an arbitrarily high precision.
A gate set is computationally universal if for any output probability distribution of an $N$-qubit quantum circuit composed of $t$ gates in a strictly universal gate set, we can use a quantum circuit composed of the computationally universal gate set to sample the classical outcomes according to this probability distribution up to an error $\epsilon$ in the total variation distance within at most poly-logarithmic overhead in terms of $N$, $t$, and $\frac{1}{\epsilon}$~\cite{A1}.
For example, the gate set $\left\{H,CCZ\right\}$ is not strictly universal since it cannot generate a quantum state having imaginary coefficients with respect to the computational basis, but $\left\{H,CCZ\right\}$ is computationally universal~\cite{S1}.
We can rewrite any quantum circuit composed of the $H$ gate and the controlled-$S$ ($CS$) gate, which is strictly universal, as the corresponding circuit composed of $\left\{H,CCZ\right\}$ at a constant overhead cost~\cite{A1}.
Using this fact, we can rewrite any quantum circuit composed Clifford gates and the $T$ gate as the corresponding circuit composed of $\left\{H,CCZ\right\}$ within a polylog overhead in terms of $N$, $t$, and $\frac{1}{\epsilon}$,
while the gate sets $\left\{H,CS\right\}$ and $\left\{H,CCZ\right\}$ cannot exactly implement the $T$ gate at a constant overhead cost~\cite{Beverland2019}.

As for universality in MBQC, we use a definition of universality that corresponds to the computational universality in the circuit model.
We call a family of quantum states a \textit{universal} resource for MBQC if there exists a protocol for MBQC using a state in this family for simulating any $N$-qubit $D$-depth quantum circuit composed of a computationally universal gate set within a polynomial number of computational steps in terms of $N$ and $D$~\cite{G2,G3}, where the simulation of the circuit in this paper refers to sampling the classical outcomes from the output probability distribution of the circuit, and the resource states used in the protocol depend only on $N$ and $D$ and are independent of the $N$-qubit $D$-depth quantum circuit.
This definition of universality is called efficient CC-universality in Ref.~\cite{V2}, where CC refers to classical inputs and classical outputs,
while universality in the sense of performing arbitrary unitary transformations~\cite{V2,V3} is different from this definition of universality.
Note that this definition of universal resources based on deterministic and exact simulation of quantum circuits suffices for this paper, while generalization to probabilistic or approximate simulations is straightforward.

\textbf{Implementability of MBQC}:
The choice of resource states for MBQC crucially affects implementability of the MBQC protocols.
In particular, it is not necessary for universal MBQC with multiqubit resource states to realize arbitrary single-qubit rotation in implementing measurements, but a restricted subclass of the rotation may suffice.
For example, it suffices to use measurements in bases on the $XY$-plane for MBQC with the cluster state~\cite{M5} and the brickwork state~\cite{B3}, to use Pauli-$X$, $Y$, and $Z$ measurements for MBQC with the hypergraph states in Refs.~\cite{M2,G4}, and to use Pauli-$X$ and $Z$ measurements for MBQC with the parity-phase (weighted) graph states~\cite{K3} and the hypergraph state in Ref.~\cite{Y1}.
A resource state for universal MBQC only with measurements in Pauli bases is called Pauli universal.
Although resource states can be prepared by a fixed quantum circuit before we decide what to compute, measurement bases in MBQC are needed to be adapted by feed-forward processes and cannot be fixed before the computation.
In this regard, measurements are online operations in MBQC, and the required measurement bases for MBQC should be suitably chosen so that these measurements can be faithfully realized in experiments.
Such suitable measurement bases may depend on a given physical system used for implementing MBQC\@; \textit{e.g.}, for superconducting qubits where the $Z$ basis is the energy eigenbasis, $X$ measurements are not particularly easier to implement than those in other bases on the $XY$-plane.
But in cases of photonic GKP qubits, experimentally tractable homodyne detection corresponds to Pauli measurements such as $Z$ and $X$ measurements on GKP qubits.
Thus, aiming at MBQC using GKP qubits, we will provide universal resources for MBQC only with $Z$ and $X$ measurements in Sec.~\ref{sec:resource_state}.

In addition to the measurement bases, there is another factor affecting implementability of MBQC protocols, that is, the required number of \textit{costly elementary gates to realize in experiments} for generating the resource state.
Whether a gate is costly or not may depend on the physical system.
In MBQC using photonic GKP qubits,
the costly gates on the GKP qubits in preparing resources are logical non-Clifford gates on the GKP qubits.
Note that non-Clifford gates cannot be completely removed from a quantum circuit for preparing Pauli universal resources~\cite{M2}; otherwise, \textit{i.e.}, if only Clifford gates were used for both resource state preparation and measurements in MBQC, the whole MBQC could be efficiently simulated by a classical computation due to the Gottesman-Knill theorem~\cite{G9,A6}.
Hence in this paper, we will also discuss an advantage of our resource states in terms of the required number of logical non-Clifford gates, in particular, the $CCZ$ gates, for the resource preparation.

\subsection{\label{sec:qec_gkp}Quantum error correction (QEC) for photonic MBQC}

In this section, we present preliminaries to quantum error correction (QEC) in this paper.
We first fix our noise model and describe how to perform fault-tolerant MBQC using the GKP code.
In fault-tolerant computation using GKP qubits, we need to suppress two types of errors; one is variances of finitely squeezed GKP qubits that may accumulate as computation proceeds, and the other is the bit- and phase-flip errors in the GKP qubits as a result of this finite squeezing.
Then, we recall the definition of $7$-qubit code, which we use to correct the bit- and phase-flip errors of the GKP qubit.
QEC for GKP qubits is different from that for discrete qubits in that we may correct the bit- and phase-flip errors of GKP qubits using analog information that is given by real-valued outcomes of homodyne detection, in addition to digital information of discrete bit values that can be estimated using the strategy~\eqref{eq:decision} below.
The use of analog information leads to an improved QEC performance.
Thus finally, we review the following techniques for GKP error correction that we will use in combination:
\begin{itemize}
  \item \textit{Single-qubit quantum error correction (SQEC)}, which reduces accumulated variances for a data GKP qubit during computation using an auxiliary GKP qubit~\cite{G1,F6};
  \item \textit{Highly reliable measurements (HRMs)}, which uses post-selection to decrease the probability of misidentifying the bit value of the measurement outcome of a GKP qubit~\cite{F2,F6};
  \item \textit{Analog quantum error correction (AQEC)}, which uses analog information to improve the performance of QEC in deciding the logical bit value of the outcome of a logical measurement on a multiqubit quantum error-correcting code concatenated with GKP qubits~\cite{F1}.
\end{itemize}

\textbf{Noise model and QEC in fault-tolerant MBQC using the GKP code}:
To realize scalable quantum computation, it is indispensable to perform quantum error correction (QEC)~\cite{G,D,T10,B} for attaining fault tolerance.
One way to achieve fault-tolerant MBQC is to represent the resource state for MBQC as a logical state of a quantum error-correcting code.
For example, MBQC protocols using the $3$D cluster state~\cite{R7,R5,R6} combine the cluster state with the surface code for topological QEC~\cite{B6,K2,F3}, where syndrome measurements for QEC can be performed by measuring qubits for this resource state.
This technique of embedding a quantum error-correcting code in resource states is called foliation~\cite{B7}.
In QEC, the failure probability refers to the probability for a QEC protocol to fail to recover from the errors induced by the noise.
The noise model determines how the errors may occur,
and the threshold is the amount of noise below which we can arbitrarily suppress the failure probability in QEC\@.
Notice that the threshold should be regarded as a \textit{rough} standard for quantifying how highly fault-tolerant a protocol is, since the fault-tolerant protocols usually incur impractically large overheads near the threshold; however, the threshold is useful for comparing the performance of different protocols in terms of fault tolerance.

In this paper, we use the same noise model as the existing works~\cite{M4,F2,F6,N3},
where the amount of noise is quantified by the squeezing level of initial states of a GKP qubit given by the approximate GKP code shown in Sec.~\ref{sec:gkp}.
Note that Refs.~\cite{M4,F2,F6,N3} and this paper use the same approximate GKP code based on the one-mode square lattice GKP code.
In this noise model, the smaller the squeezing level, the worse the quality of the approximate GKP code, which leads to more noise.
The threshold in this paper is defined as the infimum of the squeezing level of the GKP qubit that a fault-tolerant protocol can use to reduce the failure probability to as small as desired.

In particular, we assume that we can initialize each GKP qubit in $\Ket{0}$ or $\Ket{\frac{\pi}{8}}\propto \Ket{0}+\mathrm{e}^{\mathrm{i}\frac{\pi}{4}}\Ket{1}$ using the approximate GKP code with a fixed nonzero variance $\sigma^2$ in $\hat{q}$ and $\hat{p}$ quadratures, where $\Ket{0}$ and $\Ket{1}$ are defined as~\eqref{eq:gkp_approximate_codeword_0} and~\eqref{eq:gkp_approximate_codeword_1}, respectively.
We regard the variance $\sigma^2$ as a main source of noise for simplicity, abstracting effects of other possible noise sources such as loss and imperfection in Gaussian operations and homodyne detections, while these effects may also be taken into account in terms of an increase of $\sigma^2$~\cite{F6}.
Even though Gaussian operations are noiseless in this model,
the variances of GKP qubits accumulate as the computation proceeds with implementing logical gates using the Gaussian operations.
To calculate the accumulation of the variances of GKP qubits feasibly,
in place of a GKP state consisting of multiple Gaussian peaks,
we may perform the calculation approximately for a single Gaussian function representing each peak in the same way as Ref.~\cite{M4,F2,F6,N3}; \textit{e.g.}, given two GKP qubits labeled $1$ and $2$ with variances $\left(\sigma_{1,q}^2,\sigma_{1,p}^2\right)$ and $\left(\sigma_{2,q}^2,\sigma_{2,p}^2\right)$, respectively, the variances of the two GKP qubits changes by applying the $CZ$ gate according to~\eqref{eq:cz_quadrature} as
\begin{equation}
  \label{eq:variance_accumulation_cz}
  \begin{aligned}
    \left(\sigma_{1,q}^2,\sigma_{1,p}^2\right)&\to\left(\sigma_{1,q}^2,\sigma_{1,p}^2+\sigma_{2,q}^2\right),\\
    \left(\sigma_{2,q}^2,\sigma_{2,p}^2\right)&\to\left(\sigma_{2,q}^2,\sigma_{2,p}^2+\sigma_{1,q}^2\right),
  \end{aligned}
\end{equation}
and by applying the $\textsc{CNOT}$ gate according to~\eqref{eq:cx_quadrature} as
\begin{equation}
  \label{eq:variance_accumulation_cx}
  \begin{aligned}
  \left(\sigma_{1,q}^2,\sigma_{1,p}^2\right)&\to\left(\sigma_{1,q}^2,\sigma_{1,p}^2+\sigma_{2,p}^2\right),\\
  \left(\sigma_{2,q}^2,\sigma_{2,p}^2\right)&\to\left(\sigma_{1,q}^2+\sigma_{2,q}^2,\sigma_{2,p}^2\right).
  \end{aligned}
\end{equation}
Note that this noise model is different from the code capacity model analyzed in Refs.~\cite{F2,V4,Hanggli2020}, where the accumulation of the variances in preparing the resource state for MBQC is ignored.
In contrast to the code capacity model, a $CZ$ gate on GKP qubits in our noise model doubles the variance $\sigma^2$ in $\hat{p}$ of the GKP qubits due to~\eqref{eq:variance_accumulation_cz}, and hence in the resource state preparation, we should use a technique for variance reduction, \textit{i.e.}, SQEC, as we review later in this section.

In a measurement of the GKP qubit,
the accumulated variances of a GKP qubit cause a nonzero probability of misidentifying the approximately orthogonal codewords of the GKP qubit; \textit{e.g.}, when a finitely squeezed GKP state $\Ket{0}$ is measured in the logical $Z$ basis of the GKP code, there is a nonzero probability of obtaining a measurement outcome corresponding to $\Ket{1}$.
In the cases of finite squeezing,
in contrast to the ideal case~\eqref{eq:true_bit},
the real-valued measurement outcome $\tilde{q}\in\mathbb{R}$ of homodyne detection may deviate from $2s\sqrt{\pi}$ and $(2s+1)\sqrt{\pi}$ due to the nonzero variances of the Gaussian functions of approximate GKP codewords~\eqref{eq:gkp_approximate_codeword_0} and~\eqref{eq:gkp_approximate_codeword_1}.
Then, one way to decide an estimate $\tilde{b}\in\left\{0,1\right\}$ of the logical bit value from $\tilde{q}$ is to estimate it by minimizing the absolute value of the deviation $\tilde{\Delta}$ from $2s\sqrt{\pi}$ and $(2s+1)\sqrt{\pi}$, that is,
\begin{align}
  \label{eq:decision}
  \tilde{b}&\coloneqq\argmin_{b\in\left\{0,1\right\}}\min_{s\in\mathbb{Z}}\left\{\left|\tilde{q}-q\right|:q=\left(2s+b\right)\sqrt{\pi}\right\},\\
  \tilde{s}&\coloneqq\argmin_{s\in\mathbb{Z}}\left\{\left|\tilde{q}-q\right|:q=\left(2s+\tilde{b}\right)\sqrt{\pi}\right\},\\
  \label{eq:peak_estimate}
  q_{\tilde{b},\tilde{s}}&\coloneqq\left(2\tilde{s}+\tilde{b}\right)\sqrt{\pi},\\
  \label{eq:deviation_estimate}
  \tilde{\Delta}&\coloneqq \tilde{q}-q_{\tilde{b},\tilde{s}}\in\left[-\frac{\sqrt{\pi}}{2},\frac{\sqrt{\pi}}{2}\right].
\end{align}
The effects of nonzero variances $\sigma_q^2$ and $\sigma_p^2$ of a GKP qubit are analogous to bit- and phase-flip errors at the logical bit level of the GKP code.
Note that these flips are not necessarily uniform over $Z$ and $X$ bases of a GKP qubit, since it may hold that $\sigma_q^2\neq\sigma_p^2$ as the variances accumulate in the computation.

We can attain fault tolerance by concatenating the GKP code with a multiqubit error-correcting code, so that the bit- and phase-flip errors caused by the accumulated variances of the GKP qubits can be suppressed.
While a quantum error-correcting code may require a large number of quantum systems to represent a logical state, large-scale entanglement generation through the TDM approach illustrated in Fig.~\ref{fig:introduction} serves as a promising candidate for achieving fault-tolerant MBQC in a scalable way.
As for protocols for fault-tolerant MBQC using photonic CV systems,
Ref.~\cite{M4} has first introduced a protocol using the Gaussian CV cluster state~\cite{Z1} as a resource state, where errors in this protocol arising from finite squeezing are corrected using auxiliary GKP qubits prepared in a codeword of a multiqubit error-correcting code.
Another direction of fault-tolerant photonic MBQC has later established in Refs.~\cite{F2,F6,N3,V4,Hanggli2020}, achieving a better threshold than that in Ref.~\cite{M4} by directly preparing GKP qubits in a multiqubit resource state to avoid errors arising from Gaussian CV states.
References~\cite{F2,F6,N3,V4,Hanggli2020} use a multiqubit $3$D cluster state represented by the GKP qubits, and errors caused by finite squeezing are corrected using topological QEC\@.
To achieve a good threshold, our protocol will be based on the latter approach of representing a multiqubit resource state using GKP qubits; however, our approach will be different from Ref.~\cite{F2,F6,N3,V4,Hanggli2020} in that the resource state is a multiqubit hypergraph state represented by the concatenated $7$-qubit code, and the errors will be corrected using concatenated QEC rather than topological QEC\@.

In addition to the choice of a quantum error-correcting code, the threshold for QEC in the given noise model is determined by how to detect and recover from errors using the quantum error-correcting code.
Techniques using post-selection are crucial for achieving better thresholds.
For example, magic state distillation~\cite{B11} and Knill's error-correcting $C_4 / C_6$ architecture~\cite{K5,K6} use post-selection for probabilistically discarding low-fidelity quantum states, so as to attain high fault tolerance in implementing quantum computation.
As for post-selection in fault-tolerant MBQC, higher fidelity in the preparation of resource states encoded by a quantum error-correcting code leads to better accuracy in MBQC\@.
When a success probability of a post-selection is lower bounded by $p$, repetitions of the post-selection $\frac{1}{p}$ times succeed at least once in expectation.
However, a post-selection of the whole resource state of $n$ qubits in high fidelity would succeed with an exponentially small probability as $n$ increases, and hence would incur an exponentially large computational cost in terms of $n$.
In contrast, Refs.~\cite{F4,F2} suggest a scalable protocol for resource state preparation with post-selection, which probabilistically select a constant number of qubits prepared in a graph state to guarantee high fidelity, incurring only a constant overhead cost.
Then, the constant-size high-fidelity graph states are deterministically entangled to obtain the whole resource state for MBQC\@.
In particular, Ref.~\cite{F4} combines the $7$-qubit code with the post-selection to achieve a comparable threshold to those of the surface code and Knill's error-correcting $C_4 / C_6$ architecture.
While Ref.~\cite{F4} uses the $7$-qubit code as a quantum error-correcting code, the novelty of our fault-tolerant protocol is to improve the threshold further by using the $7$-qubit code not only as the error-correcting code but also as an error-detecting code.
Moreover, in contrast to Knill's error-correcting $C_4 / C_6$ architecture~\cite{K5,K6} using an error-detecting code at all the levels of concatenation,  our fault-tolerant protocol is advantageous since we can switch the $7$-qubit code from the error-detecting code to the error-correcting code so as to reduce the overhead cost of post-selection at large concatenation levels.

\textbf{Steane's $7$-qubit code}:
We summarize definition and properties of the $7$-qubit code~\cite{PhysRevLett.77.793,S3,N4}.
Each codeword of the $7$-qubit code at a concatenation level $L\in\left\{0,1,\ldots\right\}$ consists of $7^L$ physical qubits, where a logical qubit at the concatenation level $0$ refers to a physical qubit and corresponds to a GKP qubit in our case, and for each $l\in\left\{1,\ldots,L\right\}$, a level-$l$ logical qubit consists of $7$ level-$\left(l-1\right)$ logical qubits.
The codewords at the concatenation level $l$ are represented in terms of $7$ level-$\left(l-1\right)$ qubits as
\begin{align}
  \label{eq:seven_qubit_code_0}
  \Ket{0^{\left(l\right)}}=\frac{1}{\sqrt{8}}(&\Ket{0000000}+\Ket{1010101}+\Ket{0110011} \nonumber \\
  +&\Ket{1100110}+\Ket{0001111}+\Ket{1011010} \nonumber \\
  +&\Ket{0111100}+\Ket{1101001}),\\
  \label{eq:seven_qubit_code_1}
  \Ket{1^{\left(l\right)}}=\frac{1}{\sqrt{8}}(&\Ket{1111111}+\Ket{0101010}+\Ket{1001100}\nonumber \\
  +&\Ket{0011001}+\Ket{1110000}+\Ket{0100101}\nonumber \\
  +&\Ket{1000011}+\Ket{0010110}),
\end{align}
where the state of each of the seven qubits on the right-hand sides is encoded using level-$\left(l-1\right)$ logical qubits.

The stabilizers of the $7$-qubit code are generated by
\begin{equation}
  \label{eq:syndromes}
  \begin{aligned}
  &\mathbbm{1}\otimes\mathbbm{1}\otimes\mathbbm{1}\otimes Z\otimes Z\otimes Z\otimes Z,\\
  &Z\otimes\mathbbm{1}\otimes Z\otimes\mathbbm{1}\otimes Z\otimes\mathbbm{1}\otimes Z,\\
  &Z\otimes Z\otimes\mathbbm{1}\otimes\mathbbm{1}\otimes Z\otimes Z\otimes\mathbbm{1},\\
  &\mathbbm{1}\otimes\mathbbm{1}\otimes\mathbbm{1}\otimes X\otimes X\otimes X\otimes X,\\
  &X\otimes\mathbbm{1}\otimes X\otimes\mathbbm{1}\otimes X\otimes\mathbbm{1}\otimes X,\\
  &X\otimes X\otimes\mathbbm{1}\otimes\mathbbm{1}\otimes X\otimes X\otimes\mathbbm{1},
  \end{aligned}
\end{equation}
and logical $Z$ and $X$ operators of the $7$-qubit code can be given respectively by
\begin{equation}
  \label{eq:logical}
  \begin{aligned}
  &Z_\mathrm{L}\coloneqq Z\otimes\mathbbm{1}\otimes\mathbbm{1}\otimes\mathbbm{1}\otimes\mathbbm{1}\otimes Z\otimes Z,\\
  &X_\mathrm{L}\coloneqq X\otimes\mathbbm{1}\otimes\mathbbm{1}\otimes\mathbbm{1}\otimes\mathbbm{1}\otimes X\otimes X,
  \end{aligned}
\end{equation}
which are equivalent to $Z^{\otimes 7}$ and $X^{\otimes 7}$ respectively up to multiplication of the stabilizers~\eqref{eq:syndromes}.
Logical Clifford gates of the $7$-qubit code have transversal implementations~\cite{D}; that is, a level-$l$ logical Clifford gate can be implemented at the $l-1$ level by performing a Clifford gate on all the physical qubits in parallel.
In particular, the logical $H$ and $S$ gates of the $7$-qubit code at the concatenation level $1$ can be implemented respectively at the physical level by
\begin{equation}
  \label{eq:transversal_implementation}
  \begin{aligned}
  &H^{\otimes 7},\\
  &S^{\dag \otimes 7}.
  \end{aligned}
\end{equation}
The logical $CZ$ and $\textsc{CNOT}$ gates can be implemented in the same way as the implementations~\eqref{eq:logical} of the logical $Z$ and $X$ gates, that is, by performing $CZ$ and $\textsc{CNOT}$, respectively, on $1$st, $6$th, and $7$th pairs of $2\times 7=14$ physical qubits for the control logical qubit and the target logical qubit.
Using the logical operators $Z^{\otimes 7}$ and $X^{\otimes 7}$ equivalent to~\eqref{eq:logical},
logical $Z$ and $X$ measurements of the $7$-qubit code can be implemented transversally by the $Z$ and $X$ measurements of all the physical qubits, respectively.
Note that logical non-Clifford gates for the $7$-qubit code cannot have transversal implementations even up to a reasonably small approximation~\cite{E1,F8,W4}, but we can use the state injection in Fig.~\ref{fig:gate_teleportation} to implement a logical $T$ gate, using the logical Clifford operations and the logical magic state $\Ket{\frac{\pi}{8}}$.

The $7$-qubit code can be used either as an error-correcting code or as an error-detecting code.
In the logical $Z$ and $X$ measurements of the $7$-qubit code,
we can correct up to one $Z$ error and up to one $X$ error, or can detect up to two $Z$ errors and up to two $X$ errors in the given state.
Two $Z$ errors and two $X$ errors cannot be corrected because the measurement outcomes of a $7$-qubit codeword with two errors are the same as another codeword with one error;
if two $Z$ errors and two $X$ errors occur, the operation for correcting the one error would cause a logical $Z$ and $X$ error, respectively.
Thus, to detect up to two errors using the $7$-qubit code, we perform the measurement in the same way as the error correction with the $7$-qubit code, and if we detect any error, we discard the post-measurement state.
In this paper, we use the $7$-qubit code as an error-correcting code, unless stated otherwise.

\begin{figure}[t]
    \centering
    \includegraphics[width=3.4in]{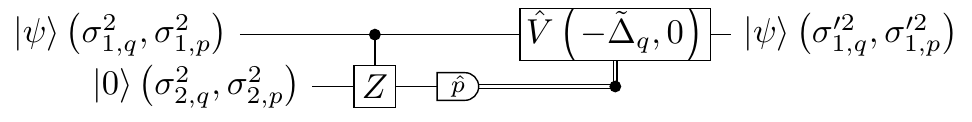}
    \includegraphics[width=3.4in]{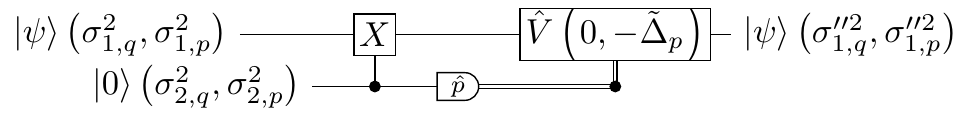}
    \caption{\label{fig:sqec}Single-qubit quantum error correction (SQEC) for reducing the variance in the $\hat{q}$ quadrature of a data GKP qubit $1$ with variances $\left(\sigma_{1,q}^2,\sigma_{1,p}^2\right)$ using an auxiliary GKP qubit $2$ with variances $\left(\sigma_{2,q}^2,\sigma_{2,p}^2\right)$ at the top ($\hat{q}$-SQEC), and that in the $\hat{p}$ quadrature at the bottom ($\hat{p}$-SQEC). Using the measurement outcome of the homodyne detection, the $\hat{q}$-SQEC estimates the deviation in $\hat{q}$ of the data GKP qubit $1$ as $\tilde{\Delta}_q$ given by~\eqref{eq:Delta_q}, and the $\hat{p}$-SQEC estimates that in $\hat{p}$ as $\tilde{\Delta}_p$ given by~\eqref{eq:Delta_p}, followed by the displacement $\hat{V}$ defined as~\eqref{eq:translation}. As a result, for any input state $\Ket{\psi}$ of the data GKP qubit, the output is the same state $\Ket{\psi}$ of a GKP qubit with variances $\left(\sigma_{1,q}^{\prime 2},\sigma_{1,p}^{\prime 2}\right)$ given by~\eqref{eq:sqec_q} in $\hat{q}$-SQEC and with $\left(\sigma_{1,q}^{\prime\prime 2},\sigma_{1,p}^{\prime\prime 2}\right)$ given by~\eqref{eq:sqec_p} in $\hat{p}$-SQEC\@.}
\end{figure}

\textbf{Single-qubit quantum error correction (SQEC)}:
SQEC aims to reduce variances of a GKP qubit that accumulate during computation as shown in~\eqref{eq:variance_accumulation_cz} and~\eqref{eq:variance_accumulation_cx}.

Given a data GKP qubit $1$ used for computation and an auxiliary GKP qubit $2$ prepared in $\Ket{0}$,
the $\hat{q}$ quadrature of these GKP qubits is denoted by $\hat{q}_1$ and $\hat{q}_2$,
the $\hat{p}$ quadrature by $\hat{p}_1$ and $\hat{p}_2$,
and the variances by $\left(\sigma_{1,q}^2,\sigma_{1,p}^2\right)$ and $\left(\sigma_{2,q}^2,\sigma_{2,p}^2\right)$, respectively.
The SQEC in the $\hat{q}$ quadrature ($\hat{q}$-SQEC) is performed by measuring $\hat{p}_2+\hat{q}_1$ using a quantum circuit shown in the upper part of Fig.~\ref{fig:sqec} for estimating the deviation of $\hat{q}_1$ for the data GKP qubit $1$, and the SQEC in the $\hat{p}$ quadrature ($\hat{p}$-SQEC) by measuring $\hat{p}_2-\hat{p}_1$ using a quantum circuit shown in the lower part of Fig.~\ref{fig:sqec} for estimating that of $\hat{p}_1$.
As established in Ref.~\cite{F6}, we can use the Gauss-Markov theorem, which is widely known in statistics, to perform a maximum-likelihood estimation of the deviations of a GKP qubit, which leads to reduced variances compared to the original SQEC~\cite{G1} without the maximum-likelihood estimation.

In the SQECs~\cite{F6} using the maximum-likelihood estimation, we first calculate $\tilde{\Delta}$ from the measurement outcome of the auxiliary GKP qubit $2$ in the same way as~\eqref{eq:deviation_estimate}, and then estimate the deviation of $\hat{q}_1$ or $\hat{p}_1$ in the data GKP qubit $1$, taking into account the variances of the auxiliary GKP qubit $2$.
To describe how the maximum-likelihood estimation works in the SQECs, instead of GKP states consisting of multiple Gaussian peaks, simply consider the data mode $1$ as one Gaussian peak with variances $\left(\sigma_{1,q}^2,\sigma_{1,p}^2\right)$ in $\hat{q}_1$ and $\hat{p}_1$, and the auxiliary mode $2$ as another Gaussian peak with variances $\left(\sigma_{2,q}^2,\sigma_{2,p}^2\right)$ in $\hat{q}_2$ and $\hat{p}_2$.
In the maximum-likelihood estimation, the deviations in $\hat{q}_1$ and $\hat{p}_1$ are respectively estimated by
\begin{align}
  \label{eq:Delta_q}
  \tilde{\Delta}_q&\coloneqq\frac{\sigma_{1,q}^2}{\sigma_{1,q}^2+\sigma_{2,p}^2}\tilde{\Delta},\\
  \label{eq:Delta_p}
  \tilde{\Delta}_p&\coloneqq\frac{\sigma_{1,p}^2}{\sigma_{1,p}^2+\sigma_{2,p}^2}\tilde{\Delta}.
\end{align}
If we perform the displacements in Fig.~\ref{fig:sqec} using these estimates $\tilde{\Delta}_q$ and $\tilde{\Delta}_p$,
in the case of one Gaussian peak,
the variances $\left(\sigma_{1,q}^{\prime 2},\sigma_{1,p}^{\prime 2}\right)$ of the mode $1$ after the $\hat{q}$-SQEC and $\left(\sigma_{1,q}^{\prime\prime 2},\sigma_{1,p}^{\prime\prime 2}\right)$ after the $\hat{p}$-SQEC are respectively given by
\begin{align}
  \label{eq:sqec_q}
  \left(\sigma_{1,q}^{\prime 2},\sigma_{1,p}^{\prime 2}\right)=\left(\frac{\sigma_{1,q}^2\sigma_{2,p}^2}{\sigma_{1,q}^2+\sigma_{2,p}^2},\sigma_{1,p}^2+\sigma_{2,q}^2\right),\\
  \label{eq:sqec_p}
  \left(\sigma_{1,q}^{\prime\prime 2},\sigma_{1,p}^{\prime\prime 2}\right)=\left(\sigma_{1,q}^2+\sigma_{2,q}^2,\frac{\sigma_{1,p}^2\sigma_{2,p}^2}{\sigma_{1,p}^2+\sigma_{2,p}^2}\right),
\end{align}
which reduce $\sigma_{1,q}^2$ and $\sigma_{1,p}^2$ to $\sigma_{1,q}^{\prime 2}$ and $\sigma_{1,p}^{\prime\prime 2}$, respectively.
While actual GKP states may be in superposition of multiple Gaussian peaks, we perform the SQECs for GKP qubits with the maximum-likelihood estimation in the same way.
We consider the variances of the data GKP qubit $1$ after the $\hat{q}$-SQEC and the $\hat{p}$-SQEC to be the same as~\eqref{eq:sqec_q} and~\eqref{eq:sqec_p}, respectively.

This approximate description of SQECs for GKP qubits using one Gaussian peak is reasonable as long as the estimation of $\tilde{\Delta}$ is correct.
However, since GKP states consist of multiple Gaussian peaks, the measurements in the SQECs may cause bit- or phase-flip errors in $\tilde{b}$ at the logical bit level of the GKP code, and unlike measuring one Gaussian peak, these errors may lead to errors in estimating $\tilde{\Delta}$.
Because of this trade-off relation between the reduction of the variances and the errors caused by logical bit or phase flips in SQECs,
we cannot arbitrarily reduce the variance of GKP qubits even if we repeat SQECs many times.
Thus in Sec.~\ref{sec:fault_tolerant}, we will develop an optimized SQEC algorithm that adjusts the squeezing levels of auxiliary GKP qubits, so as to enhance the performance of variance reduction per SQEC compared to the existing techniques.

\begin{figure}[t]
    \centering
    \includegraphics[width=3.4in]{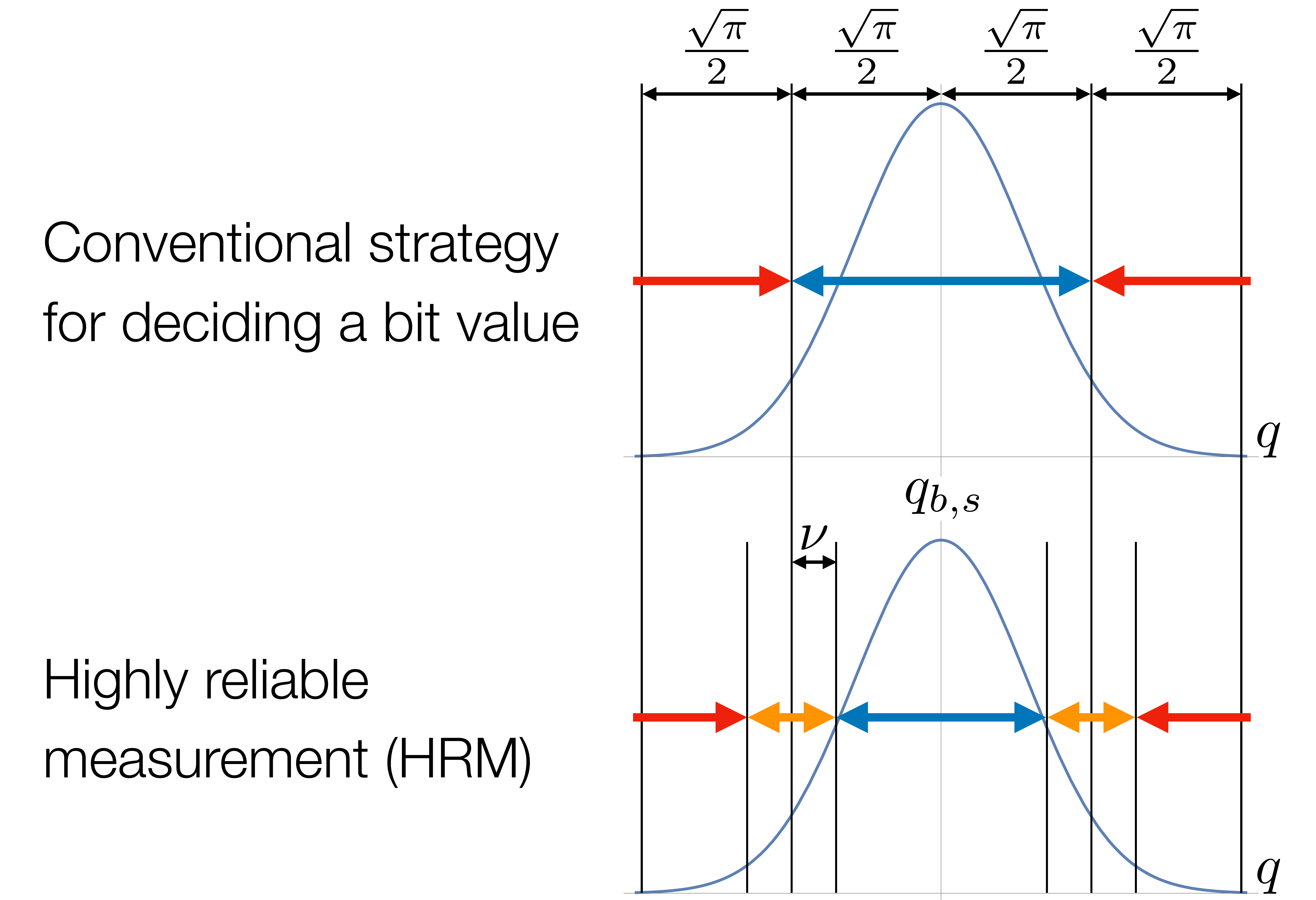}
    \caption{\label{fig:hrm}A conventional strategy for deciding a bit value of the outcome of a logical $Z$ or $X$ measurement of a GKP qubit from the real-valued outcome of homodyne detection at the top, and a highly-reliable measurement (HRM) using post-selection for decreasing the probability of misidentifying the bit value in this decision. In the conventional strategy~\eqref{eq:decision}, the decision is correct if the deviation $\Delta$ is in the blue region~\eqref{eq:upper_deviation}, and may be incorrect if $\Delta$ is in the red region~\eqref{eq:lower_deviation}. In contrast, the HRM introduces a safety margin $\nu$ and discards an unreliable decision if an estimate $\tilde{\Delta}$ of the deviation is not in a smaller region~\eqref{eq:hrm_condition} than the blue region of the conventional strategy by this margin. In the HRM, the decision is correct if $\Delta$ is in the blue region~\eqref{eq:hrm_upper}, is discarded if in the yellow region~\eqref{eq:hrm_discard}, and may be incorrect if in the red region~\eqref{eq:hrm_lower}, which is smaller than the red region in the conventional strategy.}
\end{figure}

\textbf{Highly reliable measurements (HRMs)}:
HRMs~\cite{F2,F6} aim to decrease the probability of misidentification in deciding a bit value of a GKP qubit from an outcome of homodyne detection compared to the conventional strategy~\eqref{eq:decision}.
We explain the HRM in the $Z$ basis of a GKP qubit, while that in the $X$ basis works in the same way.
For a bit value $b\in\left\{0,1\right\}$, consider a case where the GKP qubit to be measured is in state $\Ket{b}$.
If we perform homodyne detection of $\Ket{b}$ in the $\hat{q}$ quadrature,
the probability distribution of the measurement outcome $\tilde{q}\in\mathbb{R}$ consists of multiple Gaussian peaks centered at $\left\{q_{b,s}\coloneqq\left(2s+b\right)\sqrt{\pi}:s\in\mathbb{Z}\right\}$.
Suppose that $\tilde{q}$ is in the interval $\left[q_{b,s}-\sqrt{\pi},q_{b,s}+\sqrt{\pi}\right]$ for some $s$.
We here write the deviation $\Delta\coloneqq \tilde{q}-q_{b,s}\in\left[-\sqrt{\pi},\sqrt{\pi}\right]$.
As illustrated in the upper part of Fig.~\ref{fig:hrm},
if the deviation $\Delta$ satisfies
\begin{equation}
  \label{eq:upper_deviation}
  \left|\Delta\right|<\frac{\sqrt{\pi}}{2},
\end{equation}
then the decision $\tilde{b}$ in the conventional strategy~\eqref{eq:decision} on the bit value of the GKP qubit is correct.
On the other hand, if $\Delta$ satisfies
\begin{equation}
  \label{eq:lower_deviation}
  \left|\Delta\right|\geqq \frac{\sqrt{\pi}}{2},
\end{equation}
then $\tilde{b}$ may be incorrect.

In contrast to this conventional strategy, instead of the upper bound $\frac{\sqrt{\pi}}{2}$ in~\eqref{eq:upper_deviation}, HRM introduces a fixed safety margin denoted by
\begin{equation}
  \nu\in\left(0,\frac{\sqrt{\pi}}{2}\right),
\end{equation}
and discards the decision $\tilde{b}$ if the estimated deviation $\tilde{\Delta}$ does not satisfy
\begin{equation}
  \label{eq:hrm_condition}
  \left|\tilde{\Delta}\right|<\frac{\sqrt{\pi}}{2}-\nu.
\end{equation}
As illustrated in the lower part of Fig.~\ref{fig:hrm},
if $\Delta$ satisfies
\begin{equation}
  \label{eq:hrm_upper}
  \left|\Delta\right|<\frac{\sqrt{\pi}}{2}-\nu,
\end{equation}
then the decision $\tilde{b}$ in HRM is correct.
If $\Delta$ satisfies
\begin{equation}
  \label{eq:hrm_discard}
  \frac{\sqrt{\pi}}{2}-\nu\leqq\left|\Delta\right|\leqq\frac{\sqrt{\pi}}{2}+\nu,
\end{equation}
then this decision is considered to be unreliable and hence is discarded.
If $\Delta$ satisfies
\begin{equation}
  \label{eq:hrm_lower}
  \left|\Delta\right|>\frac{\sqrt{\pi}}{2}+\nu,
\end{equation}
then $\tilde{b}$ after the post-selection may still be incorrect,
but this region is smaller than that in~\eqref{eq:lower_deviation}.
The HRM can detect and avoid an unreliable decision of $\tilde{b}$ using this post-selection.
In this paper, we set $\nu$ for HRM
\begin{equation}
  \nu=\frac{2\sqrt{\pi}}{5}.
\end{equation}
In our protocol, we require that the number of HRMs should be upper bounded by a constant, and other measurements than these HRMs should be performed without post-selection, so that the post-selection of HRM may incur at most a constant overhead cost in implementing the computation.

\begin{figure}[t]
    \centering
    \includegraphics[width=3.4in]{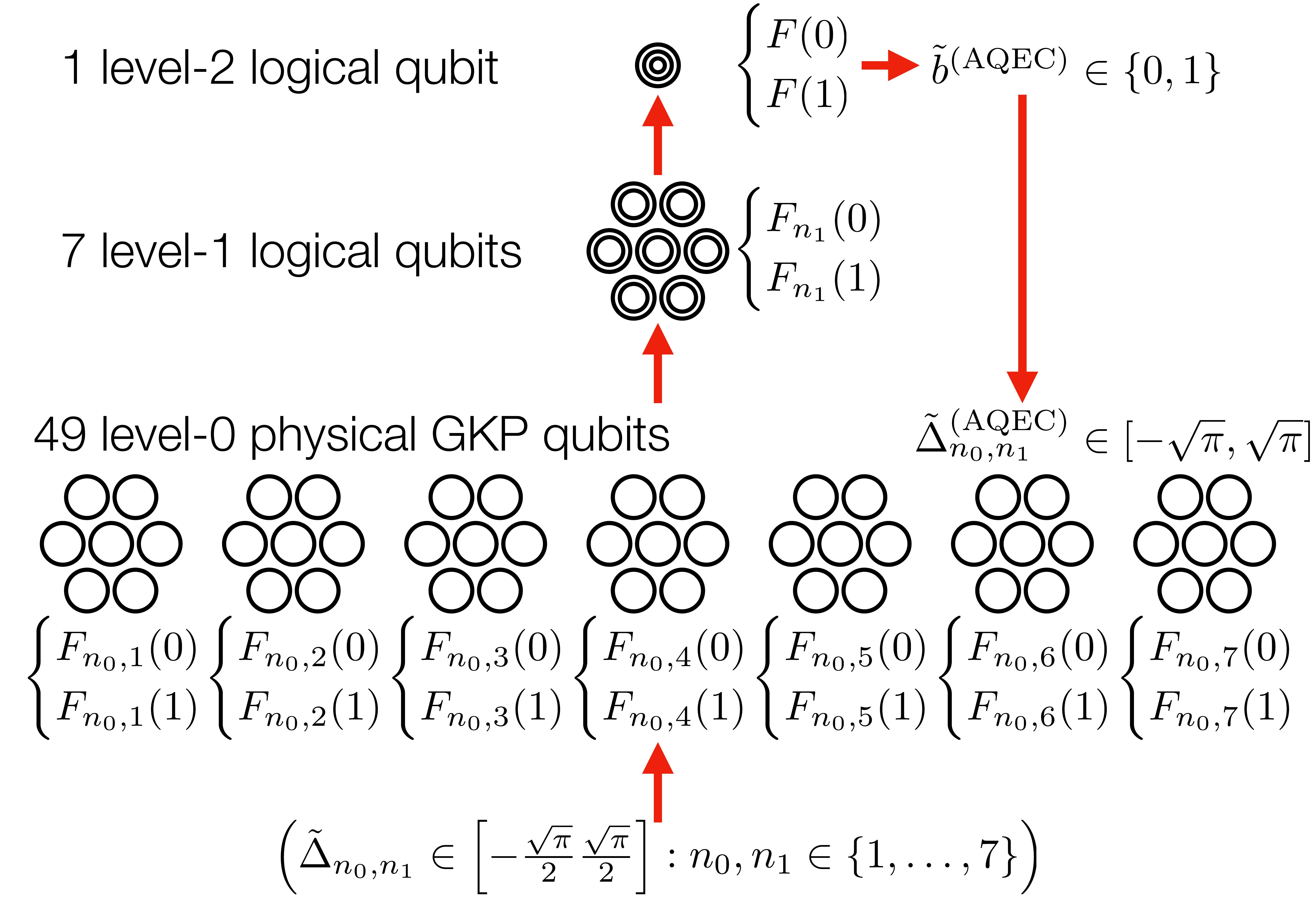}
    \caption{\label{fig:analog_qec}Analog quantum error correction (AQEC) for deciding a bit-value outcome $0,1$ of a logical $Z$ and $X$ measurement for the $7$-qubit code at the concatenation level $2$ from measurement outcomes of homodyne detection of all the physical GKP qubits in the $\hat{q}$ and $\hat{p}$ quadrature, respectively. Let $\left(n_0,n_1\right)$ for each $n_0,n_1\in\left\{1,\ldots,7\right\}$ label each of the $7\times 7=49$ physical GKP qubits comprising the level-$2$ logical qubit, where the $n_0$th of the $7$ physical GKP qubits comprises the $n_1$th level-$1$ logical qubit. Let $\tilde{\Delta}_{n_0,n_1}\in\left[-\frac{\sqrt{\pi}}{2},\frac{\sqrt{\pi}}{2}\right]$ denote an estimate of the deviation defined as~\eqref{eq:deviation_estimate} in measuring the $n_0$th physical GKP qubit comprising the $n_1$th level-$1$ logical qubit. For each level-$0$ physical GKP qubit $\left(n_0,n_1\right)$, we calculate the likelihood $F_{n_0,n_1}(0)$ of the level-$0$ bit value $0$ and  $F_{n_0,n_1}(1)$ of $1$, using the real-valued estimate $\tilde{\Delta}_{n_0,n_1}$ of the deviation as shown in~\eqref{eq:F_level_0_0} and~\eqref{eq:F_level_0_1}. For each level-$1$ logical qubit $n_1$, we calculate the likelihood $F_{n_1}(0)$ of the level-$1$ logical bit value $0$ and $F_{n_1}(1)$ of $1$, using the level-$0$ likelihoods $F_{n_0,n_1}(0)$ and $F_{n_0,n_1}(1)$ as shown in~\eqref{eq:F_level_1_0} and~\eqref{eq:F_level_1_1}. Recursively, we calculate the likelihood $F(0)$ of the logical bit value $0$ at the top level $2$, and $F(1)$ of $1$, using the level-$1$ likelihoods $F_{n_1}(0)$ and $F_{n_1}(1)$ as shown in~\eqref{eq:F_level_2_0} and~\eqref{eq:F_level_2_1}. We decide an estimate $\tilde{b}^{(\mathrm{AQEC})}\in\left\{0,1\right\}$ of the logical bit value using the top-level likelihoods $F(0)$ and $F(1)$ as shown in~\eqref{eq:b_aqec}, and these likelihoods also yield a more likely estimate $\tilde{\Delta}_{n_0,n_1}^{(\mathrm{AQEC})}\in\left[-\sqrt{\pi},\sqrt{\pi}\right]$ of the deviation defined as~\eqref{eq:deviation_estimate_aqec} than $\tilde{\Delta}_{n_0,n_1}$.}
\end{figure}

\textbf{Analog quantum error correction (AQEC)}:
AQEC~\cite{F1} improves the QEC performance using the analog information of homodyne detection of GKP qubits.
As shown in~\eqref{eq:decision}, in the $Z$ measurement of a GKP qubit, we can make a decision on an estimate $\tilde{b}\in\left\{0,1\right\}$ of the bit value from the measurement outcome $\tilde{q}\in\mathbb{R}$ of homodyne detection of the GKP qubit in the $\hat{q}$ quadrature.
In QEC using a concatenated quantum error-correcting code for discrete qubits,
given digital information of binary outcomes of syndrome measurements,
Ref.~\cite{P5} provides an efficient and optimal way of deciding how to recover from errors by calculating a likelihood function from the digital information.
When we concatenate the multiqubit quantum error-correcting code with GKP qubits, the likelihood function for this decision could be calculated from the bit values decided by the outcomes of homodyne detection using the strategy~\eqref{eq:decision}.
In contrast to this strategy based on digital information, the AQEC calculates likelihoods from analog information of real-valued outcomes $\tilde{q}\in\mathbb{R}$ of homodyne detection.
For a GKP qubit with variance $\sigma^2$, this likelihood is approximated using a Gaussian function
\begin{equation}
  f\left(x\right) \coloneqq \frac{1}{\sqrt{2\pi\sigma^{2}}} \mathrm{e}^{-\frac{x^2}{2\sigma^{2}}},
\end{equation}
where $x$ is the distance between $\tilde{q}$ and the closest Gaussian peak of the GKP state.
In particular, in the measurement of the logical $Z$ basis $\left\{\Ket{b}\right\}$ of the GKP code for $b\in\left\{0,1\right\}$,
we approximate the likelihood of $b=\tilde{b}$ using
\begin{align}
  \label{eq1}
  f_\mathrm{high}\left(\tilde{\Delta}\right)&\coloneqq f\left(\tilde{\Delta}\right)\nonumber\\
                                            &=\frac{1}{\sqrt{2\pi\sigma^{2}}} \mathrm{e}^{-\frac{\tilde{\Delta}^2}{2\sigma^{2}}},
\end{align}
and that of $b\neq\tilde{b}$ using
\begin{align}
  \label{eq2}
f_\mathrm{low}\left(\tilde{\Delta}\right)&\coloneqq f\left(\left|\sqrt{\pi}-\left|\tilde{\Delta}\right|\right|\right)\nonumber\\
                                           &= \frac{1}{\sqrt{2\pi\sigma^{2}}} \mathrm{e}^{-\frac{{\left(\sqrt{\pi}-\left|\tilde{\Delta}\right|\right)}^2}{2\sigma^{2}}}.
\end{align}
In the AQEC, we can reduce incorrect decisions by considering the likelihood of the joint event of the deviations that occur on multiple GKP qubits, and choosing the most likely candidate of the deviations.

In the following, we describe the AQEC applied to a concatenated code, specifically, the Steane's $7$-qubit code~\cite{PhysRevLett.77.793,S3,N4} introduced above.
As depicted in Fig.~\ref{fig:analog_qec},
we describe how to perform logical $Z$ measurement of the $7$-qubit code with AQEC, which can be used both for error correction and error detection.
When a measurement in the logical $Z$ basis of a level-$L$ logical qubit is performed, the AQEC decides an estimate of the logical bit value of the measurement outcome by calculating the likelihood using the real-valued outcomes of homodyne detection of all the $7^L$ physical GKP qubits in the $\hat{q}$ quadrature.
In particular, we explain the calculation of the likelihood of a level-$2$ bit value from outcomes of homodyne detection in the $\hat{q}$ quadrature, \textit{i.e.}, $L=2$, while the likelihood of the level-$L$ bit value for $L\geqq 3$ in general can be calculated from the likelihood of the level-$(L-1)$ bit values in a recursive manner.
Note that a logical bit value corresponding to a measurement in the logical $X$ basis is decided in the same way by substituting the $\hat{q}$ quadrature with the $\hat{p}$ quadrature.

A level-$2$ logical qubit consists of $7$ level-$1$ logical qubits where each consists of $7$ GKP qubits, \textit{i.e.}, $49$ GKP qubits in total, where the physical GKP qubits are considered to be at level $0$.
We label these $49$ GKP qubits as
\begin{equation}
  \label{eq:label}
  \left(n_0,n_1\right),
\end{equation}
where $n_1\in\left\{1,\ldots,7\right\}$ denotes a label representing each of the $7$ level-$1$ logical qubits that comprise the level-$2$ logical qubit, and for each $n_1$, $n_0\in\left\{1,\ldots,7\right\}$ denotes a label representing each of the $7$ GKP qubits that comprise the level-$1$ logical qubit labeled $n_1$.
Measuring these $7\times 7$ GKP qubits, we obtain the $49$ measurement outcomes $\tilde{q}_{n_0,n_1}$ and hence $49$ estimates $\tilde{b}_{n_0,n_1}$ of the bit values given by~\eqref{eq:decision} and $49$ estimates $\tilde{\Delta}_{n_0,n_1}$ of the deviations given by~\eqref{eq:deviation_estimate}, where subscripts represents the label of the GKP qubits.

For the GKP qubit labeled $\left(n_0,n_1\right)$,
the likelihood of the level-$0$ bit value $0$, \textit{i.e.}, $\Ket{0}$ of the approximate GKP qubit at the physical level, is given by
\begin{equation}
  \label{eq:F_level_0_0}
  F_{n_0,n_1}\left(0\right)=\begin{cases}
    f_{\mathrm{high}}\left(\tilde{\Delta}_{n_0,n_1}\right),&\text{if }\tilde{b}_{n_0,n_1}=0,\\
    f_{\mathrm{low}}\left(\tilde{\Delta}_{n_0,n_1}\right),&\text{if }\tilde{b}_{n_0,n_1}=1,
  \end{cases}
\end{equation}
and that of $1$, \textit{i.e.}, $\Ket{1}$, is given by
\begin{equation}
  \label{eq:F_level_0_1}
  F_{n_0,n_1}\left(1\right)=\begin{cases}
    f_{\mathrm{low}}\left(\tilde{\Delta}_{n_0,n_1}\right),&\text{if }\tilde{b}_{n_0,n_1}=0,\\
    f_{\mathrm{high}}\left(\tilde{\Delta}_{n_0,n_1}\right),&\text{if }\tilde{b}_{n_0,n_1}=1,
  \end{cases}
\end{equation}
where $f_{\mathrm{high}}$ and $f_{\mathrm{low}}$ are defined as~\eqref{eq1} and~\eqref{eq2}, respectively.
To calculate the likelihoods of a level-$1$ logical bit value, consider the $7$ GKP qubits $\left(1,n_1\right),\ldots,\left(7,n_1\right)$ comprising a level-$1$ logical qubit labeled $n_1$.
Suppose that the bit sequence $\left(\tilde{b}_{1,n_1},\ldots,\tilde{b}_{7,n_1}\right)$ that we obtain as the estimated bit values is one of $\left(0000000\right)$, $\left(1010101\right)$, $\left(0110011\right)$, $\left(1100110\right)$,$\left(0001111\right)$, $\left(1011010\right)$, $\left(0111100\right)$, and $\left(1101001\right)$, and then, it is likely that the logical bit value of the level-$1$ logical qubit $n_1$ is $0$ due to~\eqref{eq:seven_qubit_code_0}, where there is no errors on the $7$ physical GKP qubits for $n_1$.
In contrast, if the bit sequence is $\left(1000000\right)$, then we consider $8$ error patterns for each of the $2$ possible level-$1$ logical bit values $0$ and $1$ of $n_1$, where each of these $8\times 2=16$ patterns corresponds to each term on the right-hand side of~\eqref{eq:seven_qubit_code_0} and~\eqref{eq:seven_qubit_code_1}.
For the level-$1$ logical bit value $0$, the first error pattern is a single error on the physical GKP qubits that changes the state of a term on the right-hand side of~\eqref{eq:seven_qubit_code_0} from $\left(0000000\right)$ to $\left(1000000\right)$,
and the second pattern is triple errors that change $\left(1010101\right)$ to $\left(1000000\right)$.
The other patterns for the logical bit value $0$ can be described in the same way as the first and second patterns using each term on the right-hand side of~\eqref{eq:seven_qubit_code_0}.
Then, using~\eqref{eq1} and~\eqref{eq2}, we calculate the likelihoods of the first pattern ${F}^{(1)}_{n_1}\left(0\right)$ and the second error pattern ${F}^{(2)}_{n_1}\left(0\right)$ for the logical bit value $0$ of $n_1$ as
\begin{align}
  {F}^{(1)}_{n_1}(0)&=\prod_{n_0=1}^{7}F_{n_0,n_1}\left(0\right)\nonumber\\
  &=f_{\mathrm{low}}(\tilde{\Delta}_{1,n_1}) f_{\mathrm{high}}(\tilde{\Delta}_{2,n_1})f_{\mathrm{high}}(\tilde{\Delta}_{3,n_1})\nonumber \\
  &\quad\times f_{\mathrm{high}}(\tilde{\Delta}_{4,n_1})f_{\mathrm{high}}(\tilde{\Delta}_{5,n_1})\nonumber\\
  &\quad\times f_{\mathrm{high}}(\tilde{\Delta}_{6,n_1})f_{\mathrm{high}}(\tilde{\Delta}_{7,n_1}),\\
  {F}^{(2)}_{n_1}(0)&=F_{1,n_1}(1) F_{2,n_1}(0) F_{3,n_1}(1)\nonumber\\
         &\quad\times F_{4,n_1}(0) F_{5,n_1}(1) F_{6,n_1}(0) F_{7,n_1}(1)\nonumber\\
         &=f_{\mathrm{high}}(\tilde{\Delta}_{1,n_1})f_{\mathrm{high}}(\tilde{\Delta}_{2,n_1})f_{\mathrm{low}}(\tilde{\Delta}_{3,n_1})\nonumber \\
  &\quad\times f_{\mathrm{high}}(\tilde{\Delta}_{4,n_1})f_{\mathrm{low}}(\tilde{\Delta}_{5,n_1})\nonumber\\
  &\quad\times f_{\mathrm{high}}(\tilde{\Delta}_{6,n_1})f_{\mathrm{low}}(\tilde{\Delta}_{7,n_1}),
\end{align}
and those of the other patterns ${F}^{(3)}_{n_1}(0),\ldots,{F}^{(8)}_{n_1}(0)$ are also calculated in the same way.
The likelihood of the level-$1$ logical value $0$, that is, $\ket{0^{\left(1\right)}}$, for $n_1$ is given by
\begin{equation}
  \label{eq:F_level_1_0}
  {F}_{n_1}\left(0\right)={F}^{(1)}_{n_1}\left(0\right)+\cdots+{F}^{(8)}_{n_1}\left(0\right).
\end{equation}
In the same way, using~\eqref{eq:seven_qubit_code_1} instead of~\eqref{eq:seven_qubit_code_0}, the likelihood of the level-$1$ logical value $1$, that is, $\Ket{1^{\left(1\right)}}$, for $n_1$ is given by
\begin{equation}
  \label{eq:F_level_1_1}
  {F}_{n_1}\left(1\right)={F}^{(1)}_{n_1}\left(1\right)+\cdots+{F}^{(8)}_{n_1}\left(1\right).
\end{equation}

As for the likelihoods of the level-$2$ logical bit values $\Ket{0^{\left(2\right)}}$ and $\Ket{1^{\left(2\right)}}$, we consider the eight error patterns for the level-$1$ logical states corresponding to each term on the right-hand sides of~\eqref{eq:seven_qubit_code_0} and~\eqref{eq:seven_qubit_code_1}.
For example, the likelihoods corresponding to the first and second terms in~\eqref{eq:seven_qubit_code_0} are calculated respectively by
\begin{align}
  F^{(1)}(0)=&\prod_{n_1=1}^7 F_{n_1}(0),\\
  F^{(2)}(0)=&F_{1}(1) F_{2}(0) F_{3}(1)\nonumber\\
         &\times F_{4}(0) F_{5}(1) F_{6}(0) F_{7}(1).
\end{align}
Then, the likelihood of the level-$2$ logical bit value $0$ corresponding to $\Ket{0^{\left(2\right)}}$ and that of $1$ corresponding to $\Ket{1^{\left(2\right)}}$ are given by
\begin{align}
  \label{eq:F_level_2_0}
  F(0)&=F^{(1)}(0)+\cdots+F^{(8)}(0),\\
  \label{eq:F_level_2_1}
  F(1)&=F^{(1)}(1)+\cdots+F^{(8)}(1).
\end{align}
More generally, the likelihoods of the level-$l$ logical bit value for any $l\geqq 3$ can be obtained from the likelihoods of the level-$(l-1)$ bit values in a similar manner by substituting $n_1$ in the above calculation at level $2$ with $n_{l-1}$.

Consequently, given the outcomes of homodyne detection of all the $7^L$ GKP qubits, the AQEC decides an estimate $\tilde{b}^{(\mathrm{AQEC})}$ of the level-$L$ logical bit value by calculating the likelihoods $F(0)$ and $F(1)$ at the top level of concatenation, followed by choosing
\begin{equation}
  \label{eq:b_aqec}
  \tilde{b}^{(\mathrm{AQEC})}=\begin{cases}
    0&\text{if }F(0)\geqq F(1),\\
    1&\text{if }F(0)< F(1).
  \end{cases}
\end{equation}
Once $\tilde{b}^{(\mathrm{AQEC})}$ is obtained, we can also decide estimates of level-$l$ logical bit values for each $l\in\left\{L-1,\ldots,0\right\}$ using the likelihoods, as we explain for $L=2$ in the following.
Suppose that we have $\tilde{b}^{(\mathrm{AQEC})}=0$.
Then, we decide an estimate $\tilde{b}_{n_1}^{(\mathrm{AQEC})}$ of a level-$1$ logical bit value of each level-$1$ logical qubit labeled $n_1$ using the bit values maximizing the likelihoods among the eight patterns $\left\{F^{(1)}(0),\ldots,F^{(8)}(0)\right\}$ corresponding to~\eqref{eq:seven_qubit_code_0}.
For example, if $F^{(1)}(0)$ is the maximal, we decide $\left(\tilde{b}_{1}^{(\mathrm{AQEC})},\ldots,\tilde{b}_{7}^{(\mathrm{AQEC})}\right)$ as $(0000000)$, and if $F^{(1)}(1)$ is the maximal, we decide as $(1010101)$.
Suppose that we have $\tilde{b}^{(\mathrm{AQEC})}=1$.
Then, we decide $\tilde{b}_{n_1}^{(\mathrm{AQEC})}$ as those maximizing $\left\{F^{(1)}(1),\ldots,F^{(8)}(1)\right\}$ corresponding to~\eqref{eq:seven_qubit_code_1}.
By repeating this estimation in a recursive manner from level $L$ to level $0$, for each level-$0$ GKP qubit labeled $(n_0,n_1)$, we can decide an estimate $\tilde{b}_{n_0,n_1}^{(\mathrm{AQEC})}$ of a level-$0$ logical bit value, which may be different from the estimate $\tilde{b}_{n_0,n_1}$ given by~\eqref{eq:decision} without AQEC\@.
This estimate $\tilde{b}_{n_0,n_1}^{(\mathrm{AQEC})}$ is more likely than $\tilde{b}_{n_0,n_1}$ in terms of the likelihoods with respect to the $7$-qubit code~\eqref{eq:seven_qubit_code_0} and~\eqref{eq:seven_qubit_code_1}.
Then, similarly to the estimate $\tilde{\Delta}_{n_0,n_1}$ of the deviation of the GKP qubit labeled $(n_0,n_1)$ given by~\eqref{eq:deviation_estimate} without AQEC, we can use $\tilde{b}_{n_0,n_1}^{(\mathrm{AQEC})}$ to obtain a more likely estimate $\tilde{\Delta}_{n_0,n_1}^{(\mathrm{AQEC})}$ than $\tilde{\Delta}_{n_0,n_1}$, that is,
\begin{align}
  \tilde{s}_{n_0,n_1}^{(\mathrm{AQEC})}&\coloneqq\argmin_{s\in\mathbb{Z}}\left\{\left|\tilde{q}_{n_0,n_1}-q\right|:\right.\nonumber\\
                                       &\quad \left.q=\left(2s+\tilde{b}_{n_0,n_1}^{(\mathrm{AQEC})}\right)\sqrt{\pi}\right\},\\
                                       \label{eq:peak_estimate_aqec}
  q_{\tilde{b}_{n_0,n_1}^{(\mathrm{AQEC})},\tilde{s}_{n_0,n_1}^{(\mathrm{AQEC})}}&\coloneqq\left(2\tilde{s}_{n_0,n_1}^{(\mathrm{AQEC})}+\tilde{b}_{n_0,n_1}^{(\mathrm{AQEC})}\right)\sqrt{\pi},\\
  \label{eq:deviation_estimate_aqec}
  \tilde{\Delta}_{n_0,n_1}^{(\mathrm{AQEC})}&\coloneqq \tilde{q}_{n_0,n_1}-q_{\tilde{b}_{n_0,n_1}^{(\mathrm{AQEC})},\tilde{s}_{n_0,n_1}^{(\mathrm{AQEC})}}\nonumber\\
                                            &\in\left[-\sqrt{\pi},\sqrt{\pi}\right],
\end{align}
where the range $\left[-\sqrt{\pi},\sqrt{\pi}\right]$ of $\tilde{\Delta}_{n_0,n_1}^{(\mathrm{AQEC})}$ is larger than that of $\tilde{\Delta}_{n_0,n_1}$ since $q_{\tilde{b}_{n_0,n_1}^{(\mathrm{AQEC})},\tilde{s}_{n_0,n_1}^{(\mathrm{AQEC})}}$ is not necessarily the closest Gaussian peak to the outcome $\tilde{q}_{n_0,n_1}$.
In Sec.~\ref{sec:fault_tolerant}, we will combine $\tilde{\Delta}_{n_0,n_1}^{(\mathrm{AQEC})}$ obtained from the AQEC with the SQEC to suppress errors caused by bit or phase flips in the SQEC\@.

Given $n=7^L$ physical qubits for QEC\@, the AQEC calculates $O(n)$ likelihood functions in total, and if each likelihood function can be calculated within a constant time using a numerical library, the calculation of these likelihoods to decide the level-$L$ logical bit value can be performed \textit{efficiently} within time
\begin{equation}
  \label{eq:aqec_complexity}
  O\left(n\right).
\end{equation}

\subsection{\label{sec:overhead}Overhead in implementing quantum computation}

In this section, we define overhead in implementing quantum computation, and give examples of the overhead.
We first define the overhead.
Then, we discuss overheads that appear in MBQC and QEC\@.

\textbf{Definition of overhead in implementing quantum computation}:
Given a quantum circuit representing quantum computation,
an overhead in implementing the quantum computation refers to the extra computational steps required for its implementation compared to the size of the original quantum circuit.
Algorithms for the quantum computation are conventionally represented as quantum circuits where elementary gates may act on arbitrary qubits in the circuits~\cite{N4}, and we call these circuits and gates \textit{geometrically nonlocal} circuits and gates, respectively.
In contrast to the geometrically nonlocal circuits, it is also possible to consider circuits consisting of one-qubit gates and multiqubit gates acting only on neighboring qubits with respect to a certain geometrical alignment of the qubits, and such circuits and gates are called \textit{geometrically local} circuits and gates, respectively.
It is important to reduce overheads in implementing a given geometrically nonlocal quantum circuit since polynomial overheads may cancel out some polynomial quantum speedups such as those shown in Sec.~\ref{sec:introduction}.
Whether we can realize polynomial quantum speedups also matters to security analysis of post-quantum cryptography~\cite{10.1007/978-3-319-79063-3_9,10.1007/978-3-030-26948-7_2}.
The geometrical constraints, as well as quantum error correction (QEC), may cause overheads in the implementation, as we discuss in this section.

The overhead of an MBQC protocol refers to the extra computational steps of the MBQC protocol compared to the sizes of the quantum circuits to be simulated by the MBQC\@.
In particular, an $N$-qubit $D$-depth quantum circuit may consist of $O\left(DN\right)$ elementary gates.
In this case, let $t\left(D,N\right)$ be the complexity of an MBQC protocol for simulating the $N$-qubit $D$-depth quantum circuit.
Then, the overhead refers to the function
\begin{equation}
  \label{eq:def_overhead}
  \frac{t\left(D,N\right)}{DN},
\end{equation}
where the denominator $DN$ is the maximal number of elementary gates in $N$-qubit $D$-depth quantum circuits to be simulated, and we are interested in the scaling of an upper bound of this function in simulating arbitrary $N$-qubit $D$-depth quantum circuits.
In case an MBQC protocol probabilistically succeeds in simulating the circuit, we repeat the MBQC protocol until we succeed, and the overhead of the MBQC protocol may refer to the function~\eqref{eq:def_overhead} to succeed once in expectation.
Photonic MBQC illustrated in Fig.~\ref{fig:introduction} performs resource state preparation and measurements sequentially,
and the above definition of overhead based on the size of quantum circuits especially matters to this sequential implementation of quantum computation using the photonic architecture.
We also discuss the depth of quantum circuits to be simulated and MBQC implementing the circuit using parallel measurements in Appendix~\ref{sec:parallelizability}.

\textbf{Examples of overhead}:
An overhead in MBQC typically arises from geometrical constraints.
As discussed in Sec.~\ref{sec:introduction},
a resource state for MBQC that can be prepared only by geometrically local multiqubit gates, \textit{e.g.}, those with respect to a lattice on a $2D$ plane, leads to the polynomial overheads.
However, whether such geometrical locality and planarity of interactions are necessary or not may depend on physical quantum systems for realizing MBQC\@.
Some matter-based systems such as superconducting qubits may use qubits arranged on a $2D$ plane,
but in the case of photonic MBQC, these geometrical locality and planarity are not essential since the photonic GKP qubits can be easily moved in space compared to the matter-based systems.
Our goal here is to avoid the polynomial overhead in MBQC and achieve an implementation of geometrically nonlocal circuits within a \textit{poly-logarithmic} overhead,
by taking advantage of the property of the photonic GKP qubits moving in space.

In contrast to the polynomial overheads caused by the geometrical constraints, QEC using a concatenated quantum error-correcting code conventionally incur a poly-logarithmic overhead cost, unless we use too many post-selections.
One way of implementing a given circuit in a fault-tolerant way is to replace the initialization of a physical qubit in $\Ket{0}$ with a preparation of a logical qubit in $\Ket{0}$ of a concatenated quantum error-correcting code, and to replace each physical elementary gate in the circuit with an implementation of the corresponding logical elementary gate on the code followed by QEC\@.
While the overhead cost of this implementation may depend on the concatenation level of the code, the required concatenation level $L$ for implementing a circuit of size $t$ within a failure probability $\epsilon$ is known to scale~\cite{N4}
\begin{equation}
  \label{eq:overhead_qec}
  L=O\left(\log\left(\polylog\left(\frac{t}{\epsilon}\right)\right)\right).
\end{equation}
For each gate in the original circuit,
the required number of physical gates for implementing the corresponding logical gate in the fault-tolerant circuit is
\begin{equation}
  O\left(\polylog\left(\frac{t}{\epsilon}\right)\right).
\end{equation}
For implementing the fault-tolerant circuit as a whole,
the required number of physical gates is
\begin{equation}
  O\left(t\times\polylog\left(\frac{t}{\epsilon}\right)\right),
\end{equation}
which is only poly-logarithmically larger than the original circuit.
If we use post-selection in QEC, the post-selection can be another factor causing the overhead.
As discussed in Sec.~\ref{sec:qec_gkp}, let $p$ be a lower bound of the success probability of a post-selection, and each of the post-selection may incur an overhead cost of order $\frac{1}{p}$ in expectation.
Hence, to achieve polylog overhead including QEC, we will design the fault-tolerant protocol to keep the sum of the overhead costs caused by the post-selections at most poly-logarithmic.

\section{\label{sec:resource_state}Polylog-overhead universal MBQC protocol for photonic systems}

In this section, we introduce a family of hypergraph states that serve as resources for MBQC achieving poly-logarithmic overhead in simulating an arbitrary $N$-qubit $D$-depth quantum circuit only by $Z$ and $X$ measurements, which can be implemented by homodyne detection on GKP qubits as discussed in Sec.~\ref{sec:preliminaries}.
These hypergraph states, defined in Sec.~\ref{sec:definition}, are denoted by
\begin{equation}
    \label{eq:resources}
    \left\{\Ket{G_\mathrm{res}^{\left(N,D\right)}}:N,D\in\left\{1,2,\ldots\right\}\right\},
\end{equation}
where $G_\mathrm{res}^{\left(N,D\right)}$ is the corresponding hypergraph defined later as~\eqref{eq:g_res} and illustrated in Fig.~\ref{fig:resource}.
To prove the universality, we discuss the preparation complexity of these resource hypergraph states in Sec.~\ref{sec:preparation_complexity}.
Then in Sec.~\ref{sec:quantum_complexity}, we show our protocol for MBQC using $\Ket{G_\mathrm{res}^{\left(N,D\right)}}$ as a resource for simulating an arbitrary $N$-qubit $D$-depth quantum circuit composed of a computationally universal gate set $\left\{H,CCZ\right\}$.
How this protocol works will be illustrated in Fig.~\ref{fig:circuit_n_qubit_d_depth}, while the patterns of the $Z$ and $X$ measurements for implementing gates in a quantum circuit will be summarized in Fig.~\ref{fig:measurement}.
Note that our construction of the resource hypergraph states can be used for MBQC on arbitrary multi-qubit systems, while optimized based on the properties of GKP qubits.
In the following, the floor and ceiling functions are denoted by $\lfloor{}\cdots{}\rfloor$ and $\lceil{}\cdots{}\rceil$, respectively.

\subsection{\label{sec:definition}Definition of resource states}

In this subsection, we define a hypergraph $G_\mathrm{res}^{\left(N,D\right)}$ (shown later in Fig.~\ref{fig:resource}) representing a hypergraph state $\Ket{G_\mathrm{res}^{\left(N,D\right)}}$ which serves as a universal resource state for MBQC to simulate an arbitrary $N$-qubit $D$-depth quantum circuit composed of a computationally universal gate set $\left\{H,CCZ\right\}$.
For two hypergraphs $G_1=\left(V_1,E_1\right)$ and $G_2=\left(V_2,E_2\right)$,
their union $G_1\cup G_2$ is defined as a hypergraph
\begin{equation}
    \label{eq:union}
    G_1\cup G_2\coloneqq \left(V_1\cup V_2,E_1\cup E_2\right),
\end{equation}
where $V_1$ and $V_2$ may share vertices with the same label that are identified with each other.

\begin{figure}[t]
    \centering
    \includegraphics[width=3.4in]{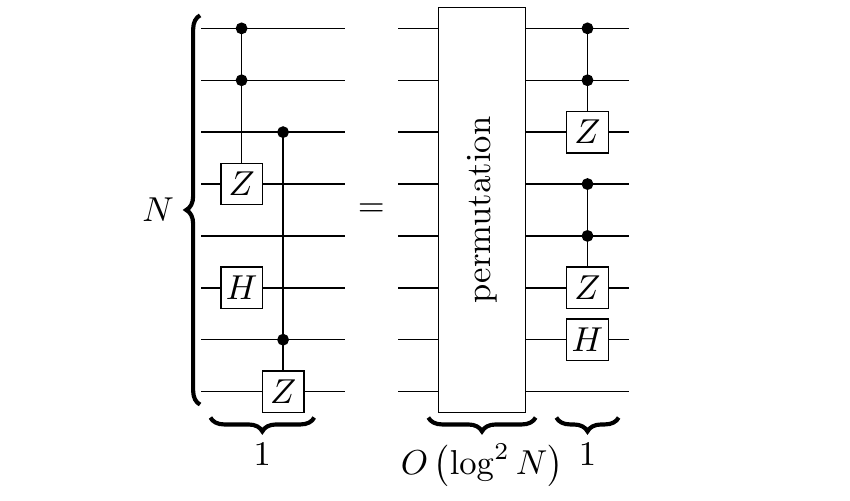}
    \caption{\label{fig:circuit_sorted}Rewriting an $N$-qubit $1$-depth quantum circuit composed of $\left\{H,CCZ\right\}$ in terms of geometrically local $CCZ$ gates using permutation of $N$ qubits. We implement this permutation using a sorting network, \textit{e.g.}, the odd-even merging networks of depth $O\left(\log^2 N\right)$, as shown in the main text. Note that equivalence of the circuits in the figure is up to permutation of the output.}
\end{figure}

\begin{figure}[t]
    \centering
    \includegraphics[width=3.4in]{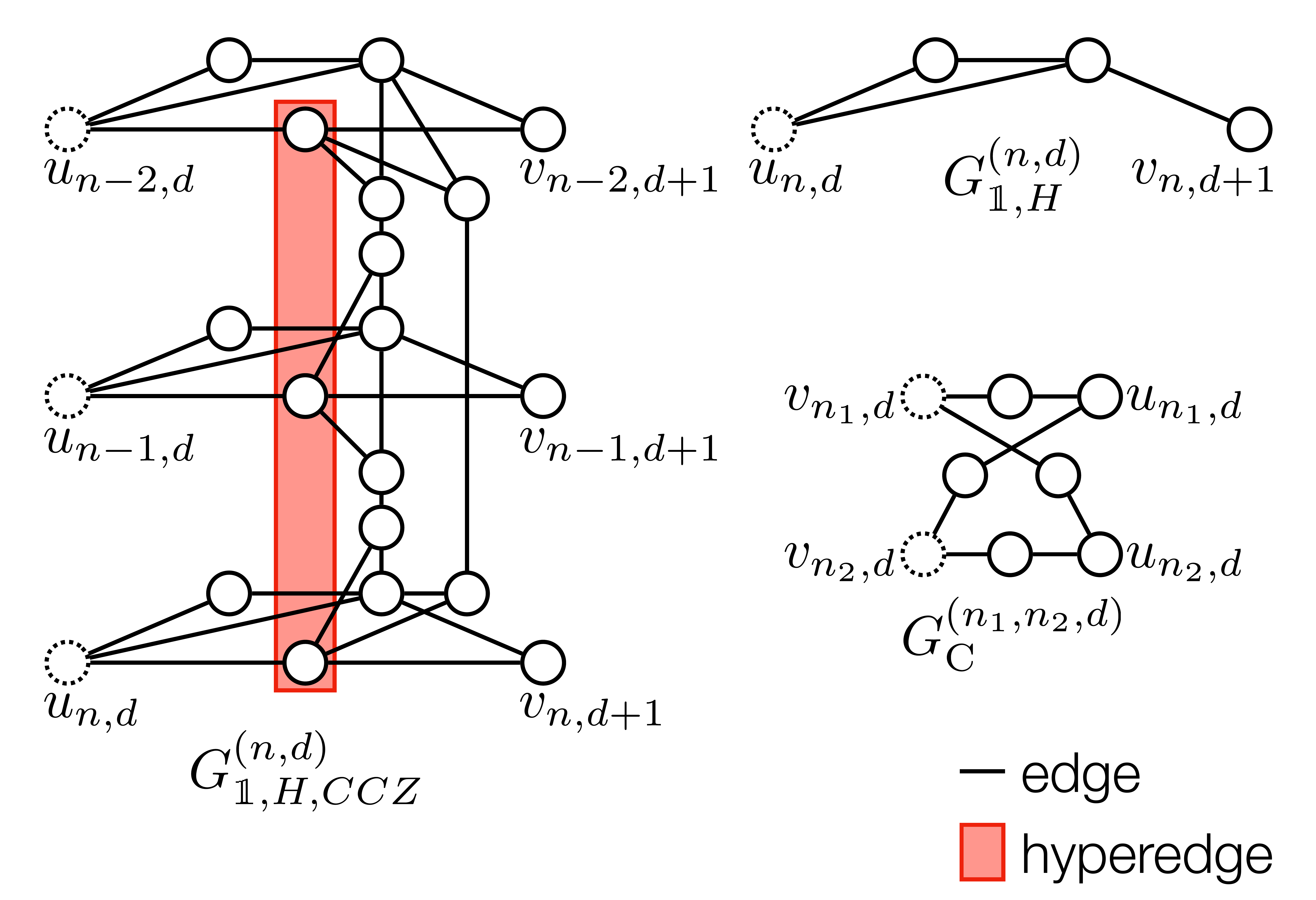}
    \caption{\label{fig:parts}Definitions of hypergraphs $G_{\mathbbm{1},H,CCZ}^{\left(n,d\right)}$, $G_{\mathbbm{1},H}^{\left(n,d\right)}$ and $G_\textup{C}^{\left(n_1,n_2,d\right)}$. Vertices are represented as circles, and an edge between two vertices and a hyperedge among three vertices are represented as a black line and a red rectangle, respectively. We may collectively refer to edges and hyperedges as hyperedges for brevity. Some of the vertices are labeled as shown in the figure, and vertices represented as dotted circles of each hypergraph will be identified with another hypergraph's vertices represented as solid circles having the same label when these hypergraphs are pasted together to obtain the union defined as~\eqref{eq:union}. Vertices and hyperedges that are not explicitly labeled will never be identified with any other when we take the union of hypergraphs.}
\end{figure}

\begin{figure}[t]
    \centering
    \includegraphics[width=3.4in]{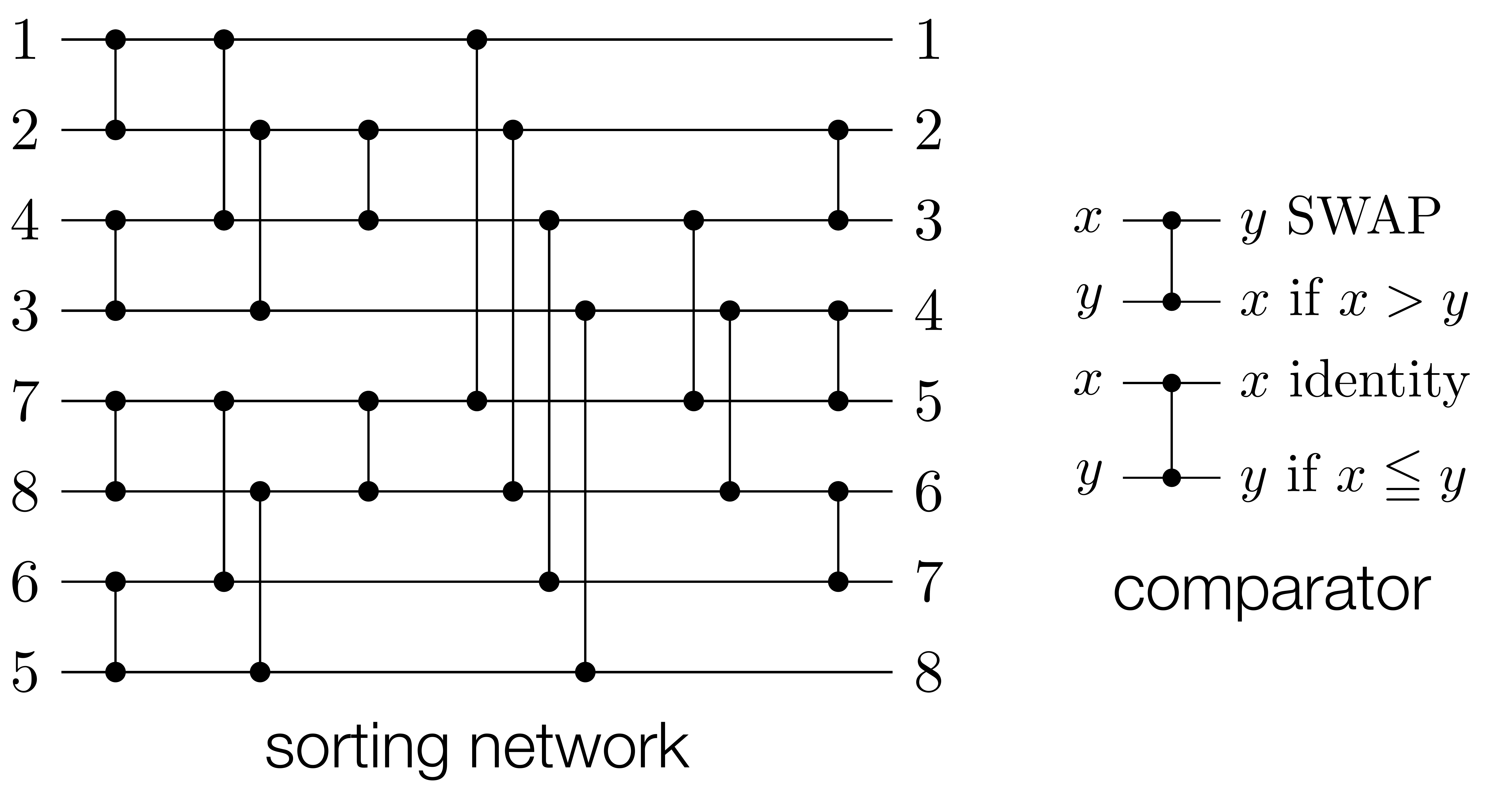}
    \caption{\label{fig:sorting_network}A sorting network of $N$ wires and the comparators represented as two dots connected by a vertical line, where the figure in particular shows the odd-even merging network for $N=8$. A sorting network is a fixed circuit composed of wires and the comparators acting as a $\textsc{SWAP}$ or identity gate to rearrange any input in ascending order, implementing arbitrary permutation.}
    \includegraphics[width=3.4in]{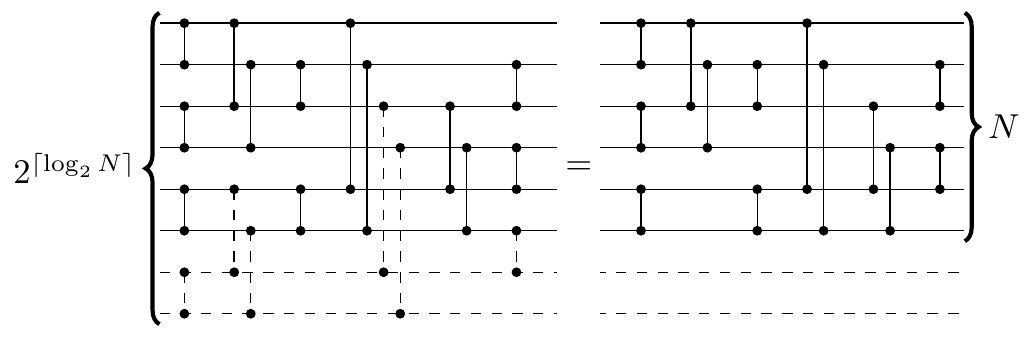}
    \caption{\label{fig:sort_arbitrary}Construction of a sorting network of $N$ wires from that of $2^{\left\lceil\log_2 N\right\rceil}$ wires, by removing the dashed wires and the dashed comparators involving the removed wires. For any $N$, we may use a sorting network of $N$ wires that can be obtained from the odd-even merging network of $2^{\left\lceil\log_2 N\right\rceil}$ wires by removing the last $2^{\left\lceil\log_2 N\right\rceil}-N$ wires.}
\end{figure}

Our construction of the hypergraph $G_\mathrm{res}^{\left(N,D\right)}$ consists of $D$ sub-hypergraphs $G_\textup{depth}^{\left(N,1\right)},G_\textup{depth}^{\left(N,2\right)},\ldots,G_\textup{depth}^{\left(N,D\right)}$ (shown later in Fig.~\ref{fig:one_depth}), where $G_\textup{depth}^{\left(N,d\right)}$ for each $d\in\left\{1,\ldots,D\right\}$ represents a hypergraph state for simulating an arbitrary $N$-qubit $1$-depth quantum circuit,
so that simulations of an $N$-qubit $D$-depth quantum circuit using $G_\mathrm{res}^{\left(N,D\right)}$ can be achieved by simulating arbitrary $N$-qubit $1$-depth quantum circuits repeatedly $D$ times.
As illustrated in Fig.~\ref{fig:circuit_sorted}, any $N$-qubit $1$-depth quantum circuit, which may contain geometrically nonlocal $CCZ$ gates, can be implemented (up to permutation of the output) by first permuting the $N$ qubits appropriately, followed by performing an $N$-qubit $1$-depth quantum circuit whose $CCZ$ gates are geometrically local.
Our construction of $G_\textup{depth}^{\left(N,d\right)}$ consists of a graph $G_\textup{S}^{\left(N,d\right)}$ (shown later in Fig.~\ref{fig:sort_graph}) for arbitrarily sorting the $N$ qubits in the circuit, hypergraphs $G_{\mathbbm{1},H,CCZ}^{\left(n,d\right)}$ for performing an arbitrary $3$-qubit $1$-depth circuit composed of $\left\{H,CCZ\right\}$ (including $\mathbbm{1}^{\otimes 3}$) on the $(n-2)$th, $(n-1)$th, and $n$th qubits in the circuit, and graphs $G_{\mathbbm{1},H}^{\left(n,d\right)}$ for performing $\mathbbm{1}$ or $H$ on the $n$th qubit in the circuit.
The graph $G_\textup{S}^{\left(N,d\right)}$ for permuting the $N$ qubits consists of graphs $G_\textup{C}^{\left(n_1,n_2,d\right)}$ for performing a $\textsc{SWAP}$ or identity gate on the $n_1$th and $n_2$th qubits in the circuit.
Definitions of these hypergraphs (which include graphs as special cases) $G_{\mathbbm{1},H,CCZ}^{\left(n,d\right)}$, $G_{\mathbbm{1},H}^{\left(n,d\right)}$, and $G_\textup{C}^{\left(n_1,n_2,d\right)}$ are given in Fig.~\ref{fig:parts}, where some of the vertices are labeled as shown in Fig.~\ref{fig:parts}.
Each unlabeled vertex in Fig.~\ref{fig:parts} is considered to have a different label from any other and is never identified with other vertices when we take the union of hypergraphs.
In the following, we define $G_\textup{S}^{\left(N,d\right)}$, $G_\textup{depth}^{\left(N,d\right)}$, and $G_\mathrm{res}^{\left(N,D\right)}$ using the hypergraphs $G_{\mathbbm{1},H,CCZ}^{\left(n,d\right)}$, $G_{\mathbbm{1},H}^{\left(n,d\right)}$, and $G_\textup{C}^{\left(n_1,n_2,d\right)}$.

Our construction of the graph $G_\textup{S}^{\left(N,d\right)}$ for the permutation is based on sorting networks~\cite{K4}.
As illustrated in Fig.~\ref{fig:sorting_network},
a sorting network is a \textit{fixed} circuit that is composed of comparators and $N$ wires, and rearranges $N$ arbitrary inputs from $\left\{1,\ldots,N\right\}$ in ascending order, where each of the $N$ wires represents each of $1,\ldots,N$, and each comparator is a two-input two-output gate acting as a $\textsc{SWAP}$ or identity gate so that the two inputs are output in ascending order.
Using a sorting network, we can rearrange $N$ inputs in arbitrary order by appropriately performing a $\textsc{SWAP}$ or identity gate at each of the comparators in the sorting network.

While construction of the shallowest sorting network is believed to be computationally hard~\cite{P2},
several explicit constructions of (not necessarily optimal but practical) sorting networks are known to have poly-logarithmic depths, such as the bitonic sorters, the odd-even merging networks, and the pairwise sorting networks.
The bitonic sorters~\cite{B10} are sorting networks defined for $N=2^l$ and suitable for graphics processing unit sorting~\cite{K7}, having depth
\begin{equation}
    \frac{1}{2}\left(\log_2 N\right)\left(\log_2 N +1\right),
\end{equation}
and the number of comparators is
\begin{equation}
    \frac{1}{4}N\left(\log_2^2 N + \log_2 N\right).
\end{equation}
The odd-even merging networks~\cite{B10} are also defined for $N=2^l$, having the same depth as the bitonic sorters
\begin{equation}
    \label{eq:depth_sorting_network}
    \frac{1}{2}\left(\log_2 N\right)\left(\log_2 N +1\right),
\end{equation}
and a smaller number of comparators than that of the bitonic sorters for $N\geqq 4$
\begin{equation}
    \label{eq:comparators_sorting_network}
    \frac{1}{4}N\left(\log_2^2 N - \log_2 N + 4 \right)-1.
\end{equation}
The pairwise networks~\cite{P3} have the same depth and number of comparators as the odd-even merging networks, but are different construction.
As shown in Fig.~\ref{fig:sort_arbitrary}, given any sorting network of $2^{\left\lceil\log_2 N\right\rceil}$ wires, a sorting network of $N$ wires can be obtained by removing $2^{\left\lceil\log_2 N\right\rceil}-N$ wires and the comparators that involve these removed wires.
Although any sorting network  of $N$ wires can be used for our construction of $G_\textup{S}^{\left(N,d\right)}$,
in the following, we specifically use the odd-even merging network of $N$ wires obtained from that of $2^{\left\lceil\log_2 N\right\rceil}$ wires by removing the last wires in the same way as the one in Fig.~\ref{fig:sort_arbitrary}.
For brevity, a sorting network may refer to this construction of sorting network,

\begin{figure*}[t]
  \centering
  \includegraphics[width=7.0in]{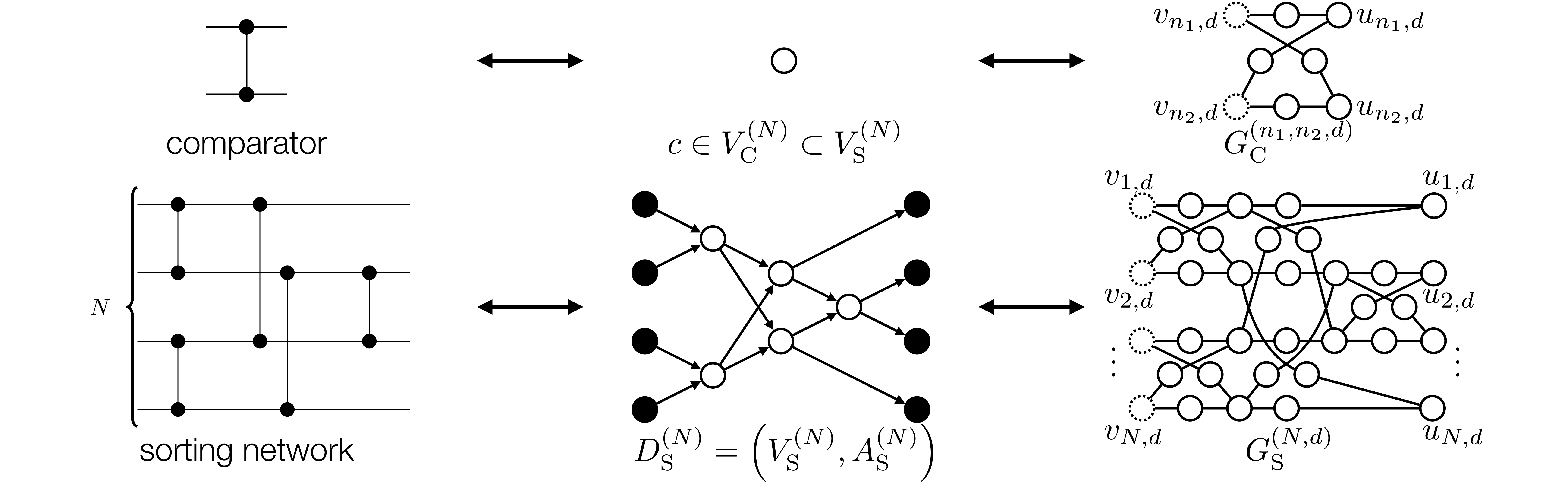}
  \caption{\label{fig:sort_graph}A graph $G_\textup{S}^{\left(N,d\right)}$ defined as~\eqref{eq:G_S_N_d} using $G_\textup{C}^{\left(n_1,n_2,d\right)}$. For a sorting network of $N$ wires such as the odd-even merging network, let $D_\textup{S}^{\left(N\right)}=\left(V_\textup{S}^{\left(N\right)},A_\textup{S}^{\left(N\right)}\right)$ be a directed acyclic graph representing the sorting network, where the white vertices $c\in V_\textup{C}^{\left(N\right)}\subset V_\textup{S}^{\left(N\right)}$ represent the comparators, and the black vertices represent the inputs and the outputs of the sorting network. Using $D_\textup{S}^{\left(N\right)}$, we define $G_\textup{S}^{\left(N,d\right)}$ in the main text as the graph obtained by replacing each comparator acting of the $n_1$th and $n_2$th wires of the sorting network with $G_\textup{C}^{\left(n_1,n_2,d\right)}$.}
\end{figure*}

\begin{figure}[t]
    \centering
    \includegraphics[width=3.4in]{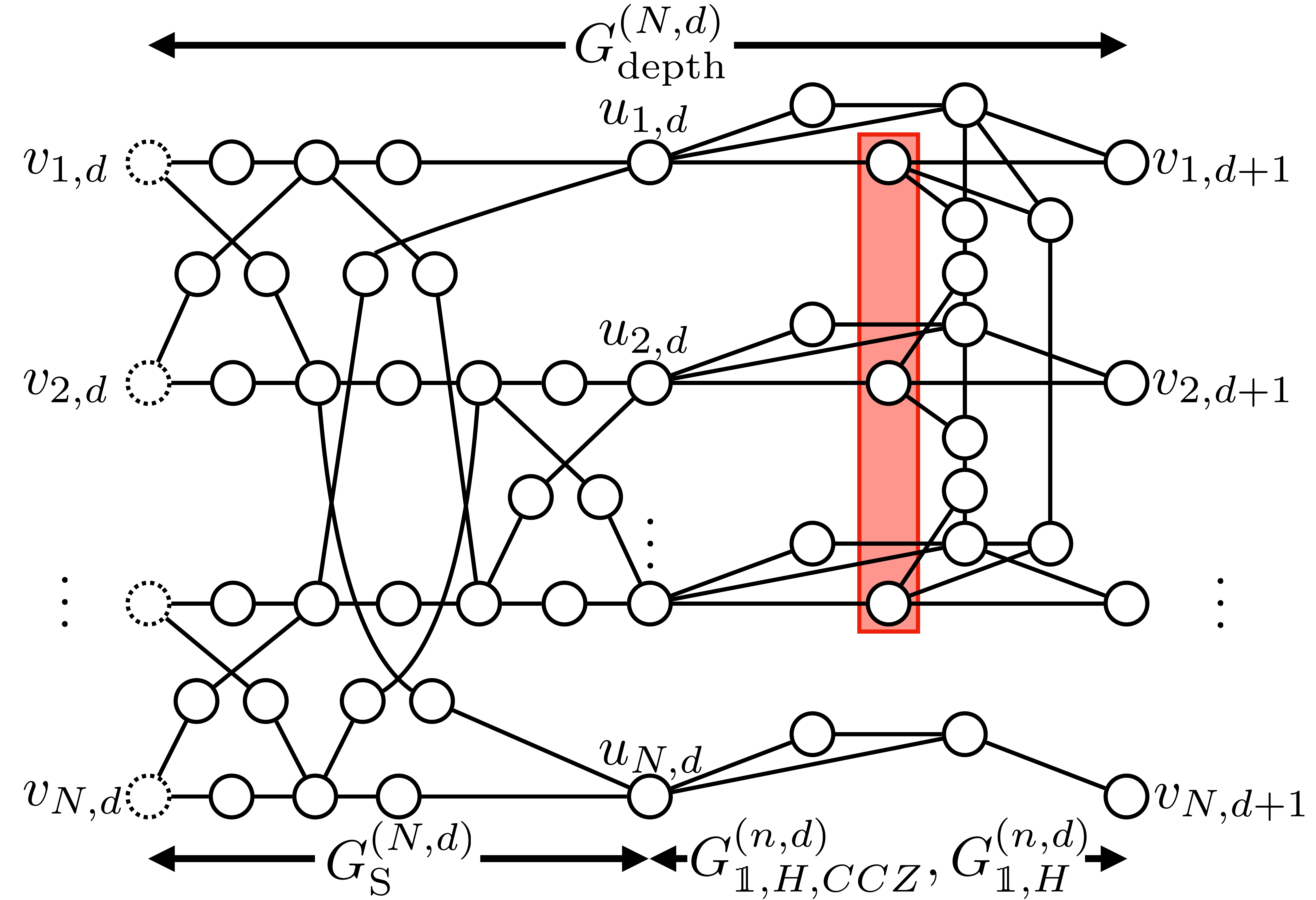}
    \caption{\label{fig:one_depth}A hypergraph $G_\textup{depth}^{\left(N,d\right)}$ defined as~\eqref{eq:one_depth} using $G_\textup{S}^{\left(N,d\right)}$, $G_{\mathbbm{1},H,CCZ}^{\left(n,d\right)}$, and $G_{\mathbbm{1},H}^{\left(n,d\right)}$. Based on Fig.~\ref{fig:circuit_sorted} on rewriting any $N$-qubit $1$-depth quantum circuit composed of $\left\{H,CCZ\right\}$ in terms of permutation and geometrically local $CCZ$ gates, the permutation part in the rewritten circuit is replaced with $G_\textup{S}^{\left(N,d\right)}$ for a sorting network, the $H$ and $CCZ$ part is replaced with as many $G_{\mathbbm{1},H,CCZ}^{\left(n,d\right)}$s as possible, and the remainder is replaced with $G_{\mathbbm{1},H}^{\left(n,d\right)}$s, where $G_{\mathbbm{1},H,CCZ}^{\left(n,d\right)}$ for each $n=3m,\, m\in\left\{1,\ldots,\left\lfloor\frac{N}{3}\right\rfloor\right\}$ corresponds to an arbitrary $3$-qubit $1$-depth circuit composed of $\left\{H,CCZ\right\}$ on the $(3m-2)$th, $(3m-1)$th, and $(3m)$th qubits in the rewritten circuit, and $G_{\mathbbm{1},H}^{\left(n,d\right)}$ for each $n\in\left\{3\left\lfloor\frac{N}{3}\right\rfloor+1,\ldots,N\right\}$ corresponds to $H$ or $\mathbbm{1}$ on the $n$th qubit.}
\end{figure}

\begin{figure*}[t]
    \centering
    \includegraphics[width=7.0in]{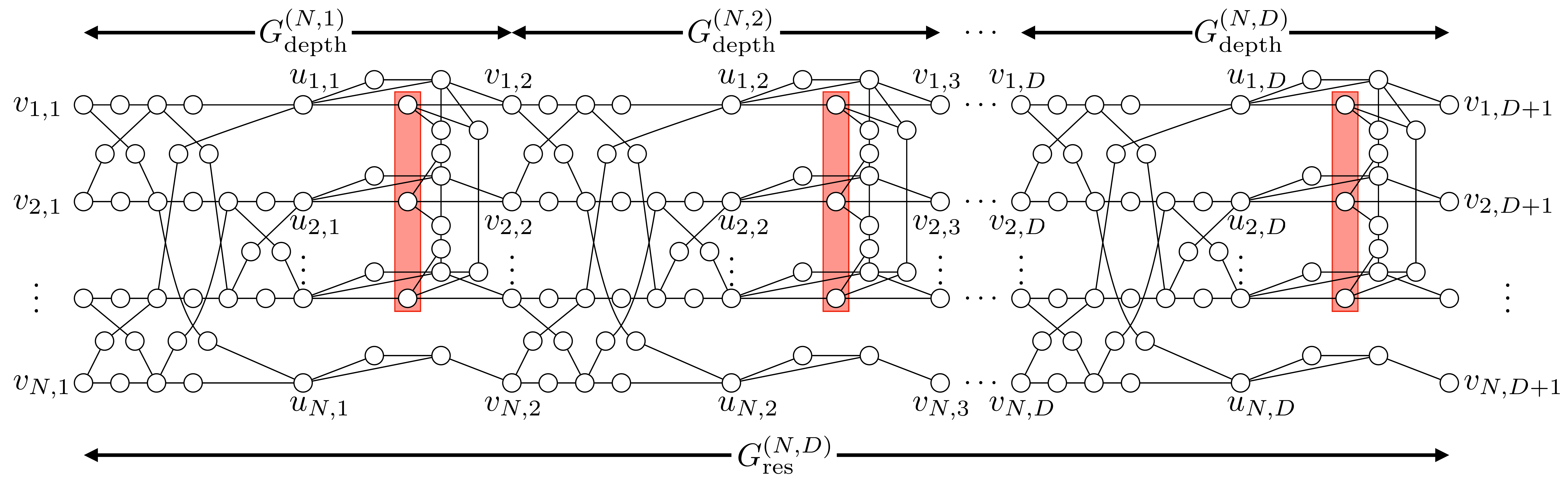}
    \caption{\label{fig:resource}A hypergraph $G_\mathrm{res}^{\left(N,D\right)}$ defined as~\eqref{eq:g_res} using $G_\textup{depth}^{\left(N,d\right)}$. An arbitrary $N$-qubit $D$-depth quantum circuit composed of $\left\{H,CCZ\right\}$ can be performed by repeating $N$-qubit $1$-depth quantum circuits $D$ times, and $G_\mathrm{res}^{\left(N,D\right)}$ is defined as a hypergraph where $G_\textup{depth}^{\left(N,d\right)}$ is repeated $D$ times.}
    \includegraphics[width=7.0in]{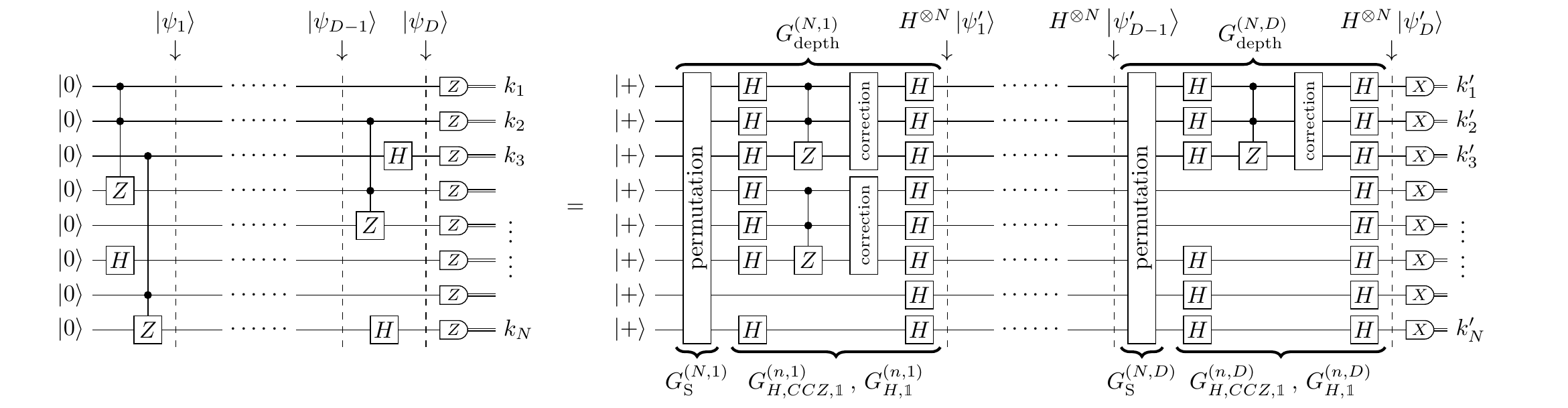}
    \caption{\label{fig:circuit_n_qubit_d_depth}A circuit identity between an arbitrary $N$-qubit $D$-depth geometrically nonlocal quantum circuit composed of $\left\{H,CCZ\right\}$ on the left-hand side and its equivalent circuit on the right-hand side illustrating our MBQC protocol using the resource hypergraph state $\Ket{G_\mathrm{res}^{\left(N,D\right)}}$, where $G_\mathrm{res}^{\left(N,D\right)}$ is a hypergraph in Fig.~\ref{fig:resource} representing $\Ket{G_\mathrm{res}^{\left(N,D\right)}}$. For each depth $d\in\left\{1,\ldots,D\right\}$, let $\Ket{\psi_d}$ denote an $N$-qubit state that can be generated from $\Ket{0}^{\otimes N}$ by the first $d$-depth part of the quantum circuit illustrated on the left-hand side. In the corresponding part of the equivalent circuit on the right-hand side, we obtain an $N$-qubit state $H^{\otimes N}\Ket{\psi_d^\prime}$ from $\Ket{+}^{\otimes N}$, where $\Ket{\psi_d^\prime}$ coincides with $\Ket{\psi_d}$ up to permutation of its $N$ qubits. Our MBQC protocol implements arbitrary permutations using resources corresponding to $G_\textup{S}^{\left(N,d\right)}$, followed by implementing geometrically local $CCZ$ gates, $H$ gates, and identity gates $\mathbbm{1}=H^2$ using resources represented by $G_{\mathbbm{1},H,CCZ}^{\left(n,d\right)}$ and $G_{\mathbbm{1},H}^{\left(n,d\right)}$, so that the $N$ qubits corresponding to $v_{1,d},\ldots,v_{N,d}$ of $G_\mathrm{res}^{\left(N,D\right)}$ can be prepared in $H^{\otimes N}\Ket{\psi_{d-1}^\prime}$. After repeating these implementations $D$ times to prepare $H^{\otimes N}\Ket{\psi_D^\prime}$, we perform $X$ measurements to obtain classical outcomes $k_1^\prime,\ldots,k_N^\prime$, as illustrated on the right-hand side of the circuit identity. These outcomes in MBQC are equivalent to the outcomes $k_1,\ldots,k_N$ of the $Z$ measurements in the circuit model on the left-hand side of the circuit identity, up to permutation.}
\end{figure*}

Given a sorting network of $N$ wires,
$G_\textup{S}^{\left(N,d\right)}$ is defined in such a way that $G_\textup{C}^{\left(n_1,n_2,d\right)}$ plays the role of each comparator in the sorting network,
as illustrated in Fig.~\ref{fig:sort_graph}.
To provide a formal definition of $G_\textup{S}^{\left(N,d\right)}$, represent the given sorting network of $N$ wires as a directed acyclic graph $D_\textup{S}^{\left(N\right)}=\left(V_\textup{S}^{\left(N\right)},A_\textup{S}^{\left(N\right)}\right)$, where each vertex $v\in V_\textup{S}^{\left(N\right)}$ represents an input, an output, or a comparator of the sorting network, and each arc (\textit{i.e.}, directed edge) $a\in A_\textup{S}^{\left(N\right)}$ represents a wire connecting these components of the sorting network.
The vertices that are incident with two incoming arcs and two outgoing arcs represent the comparators, colored white in $D_\textup{S}^{\left(N\right)}$ of Fig.~\ref{fig:sort_graph}.
The $N$ vertices that are incident with no incoming arc and one outgoing arc represent the $N$ inputs, and the $N$ vertices that are incident with one incoming arc and no outgoing arc represent the $N$ outputs, colored black in Fig.~\ref{fig:sort_graph}.
The set of the vertices of $D_\textup{S}^{\left(N\right)}$ representing the comparators is denoted by $V_\textup{C}^{\left(N\right)}\subset V_\textup{S}^{\left(N\right)}$.
For each comparator $c\in V_\textup{C}^{\left(N\right)}$, let $a_0,a_1\in A_\textup{S}^{\left(N\right)}$ be the two incoming arcs, and $a_2,a_3\in A_\textup{S}^{\left(N\right)}$ be the two outgoing arcs.
Then, define ${G^\prime}_\textup{C}^{\left(c\right)}$ as a graph $G_\textup{C}^{\left(n_1,n_2,d\right)}$ with its vertex labels $v_{n_1,d}$, $v_{n_2,d}$, $u_{n_1,d}$, and $u_{n_2,d}$ relabeled as $w_{a_0}$, $w_{a_1}$, $w_{a_2}$, and $w_{a_3}$, respectively.
Using the union to equate vertices labeled $w_a$ for the same arc $a\in A_\textup{S}^{\left(N\right)}$, we define
\begin{equation}
  \label{eq:G_S_N_d}
  {G^\prime}_\textup{S}^{\left(N,d\right)}\coloneqq \bigcup_{c\in V_\textup{C}^{\left(N\right)}}{G^\prime}_\textup{C}^{\left(c\right)},
\end{equation}
and define
$G_\textup{S}^{\left(N,d\right)}$
as a graph ${G^\prime}_\textup{S}^{\left(N,d\right)}$ with its labels $w_{a}$ for outgoing arc $a$ of the $N$ inputs relabeled as $v_{1,d},\ldots,v_{N,d}$, and with its labels $w_{a}$ for incoming arc $a$ of the $N$ outputs relabeled as $u_{1,d},\ldots,u_{N,d}$.
The labels $w_{a}$ for any other arc $a$ are to be removed.

For any $N$ and $d$,
the hypergraph $G_\textup{depth}^{\left(N,d\right)}$ is defined so that any $N$-qubit $1$-depth quantum circuit composed of $\left\{H,CCZ\right\}$ can be simulated by MBQC using the corresponding hypergraph state.
This $N$-qubit $1$-depth quantum circuit can be rewritten as a quantum circuit using permutation and geometrically local $CCZ$ gates, as illustrated in Fig.~\ref{fig:circuit_sorted}.
Suppose that this permutation is performed so that these geometrically local $CCZ$ gates act only on $(3m-2)$th, $(3m-1)$th, and $(3m)$th qubits in the circuit for some $m\in\left\{1,\ldots,\left\lfloor \frac{N}{3}\right\rfloor\right\}$.
On the remaining $(N\mod 3)$ qubits, \textit{i.e.}, $n$th qubits for $n\in\left\{3\left\lfloor \frac{N}{3}\right\rfloor + 1,\ldots,N\right\}$,
$H$ and $\mathbbm{1}$ may act but $CCZ$ does not act.
As illustrated in Fig.~\ref{fig:one_depth},
$G_\textup{depth}^{\left(N,d\right)}$ corresponding to this rewritten quantum circuit is defined as
\begin{align}
    \label{eq:one_depth}
    &G_\textup{depth}^{\left(N,d\right)}\coloneqq\nonumber\\
    &\quad G_\textup{S}^{\left(N,d\right)}\cup\left(\bigcup_{m=1}^{\left\lfloor\frac{N}{3}\right\rfloor}G_{\mathbbm{1},H,CCZ}^{\left(3m,d\right)}\right)\cup\left(\bigcup_{n=3\left\lfloor \frac{N}{3}\right\rfloor + 1}^{N}G_{\mathbbm{1},H}^{\left(n,d\right)}\right).
\end{align}

For any $N$ and $D$,
the hypergraph $G_\mathrm{res}^{\left(N,D\right)}$ is defined so that any $N$-qubit $D$-depth quantum circuit composed of $\left\{H,CCZ\right\}$ can be simulated by MBQC using the corresponding hypergraph state.
As illustrated in Fig.~\ref{fig:resource},
$G_\textup{res}^{\left(N,D\right)}$ is a hypergraph where $G_\textup{depth}^{\left(N,d\right)}$ for performing $N$-qubit $1$-depth quantum circuits is repeated $D$ times, defined as
\begin{equation}
    \label{eq:g_res}
    G_\mathrm{res}^{\left(N,D\right)}\coloneqq\bigcup_{d=1}^{D}G_\textup{depth}^{\left(N,d\right)}.
\end{equation}
In the following subsections, we may use a recursive relation of $G_\mathrm{res}^{\left(N,D\right)}$, \textit{i.e.},
\begin{equation}
    \label{eq:g_res_inductive}
    \begin{alignedat}{2}
        G_\mathrm{res}^{\left(N,1\right)}&= G_\textup{depth}^{\left(N,1\right)}\quad&&\left(D=1\right),\\
        G_\mathrm{res}^{\left(N,D\right)}&= G_\mathrm{res}^{\left(N,D-1\right)}\cup G_\textup{depth}^{\left(N,D\right)}\quad&&\left(D>1\right).
    \end{alignedat}
\end{equation}

\subsection{\label{sec:preparation_complexity}Complexity of resource state preparation}

We analyze the required number of operations for preparing the hypergraph state $\Ket{G_\mathrm{res}^{\left(N,D\right)}}$ that serves as a resource for our MBQC protocol.
To analyze this preparation complexity,
we count the numbers of vertices and hyperedges of the hypergraph $G_\mathrm{res}^{\left(N,D\right)}$ defined as~\eqref{eq:g_res}, where hyperedges may refer to edges between two vertices and hyperedges among three.
This counting yields the following proposition.

\begin{proposition}[\label{prp:num_vertices_hyperedges}The numbers of vertices and hyperedges]
    The number of vertices of the hypergraph $G_\mathrm{res}^{\left(N,D\right)}$ defined as~\eqref{eq:g_res} is
    \begin{equation}
        \label{eq:num_vertices}
      \frac{3}{2}{DN\left\lceil \log_2 N\right\rceil}^2 + O\left(DN\right).
    \end{equation}
    The number of hyperedges of $G_\mathrm{res}^{\left(N,D\right)}$ is
    \begin{equation}
        \label{eq:num_hyperedges}
      2{DN\left\lceil \log_2 N\right\rceil}^2 + O\left(DN\right).
    \end{equation}
    In particular, for any $N$ and $D$, the number of hyperedges among three vertices of $G_\mathrm{res}^{\left(N,D\right)}$ is
    \begin{equation}
        \label{eq:num_ccz}
        D\left\lfloor\frac{N}{3}\right\rfloor.
    \end{equation}
\end{proposition}

Using this proposition on the number of vertices and edges of $G_\mathrm{res}^{\left(N,D\right)}$, the preparation complexity of $\Ket{G_\mathrm{res}^{\left(N,D\right)}}$ is shown in the following theorem.
We also discuss the efficiency of the preparation of these resource states in terms of the required number of $CCZ$ gates, after showing this theorem.

\begin{theorem}[\label{thm:preparation_complexity}Preparation complexity]
    For any $N$ and $D$,
    if we use a gate set $\left\{H,S,CZ,CCZ\right\}$ as shown in~\eqref{eq:H_S_CZ_CCZ},
    the preparation complexity of the hypergraph state $\Ket{G_\mathrm{res}^{\left(N,D\right)}}$ corresponding to the hypergraph $G_\mathrm{res}^{\left(N,D\right)}$ defined as~\eqref{eq:g_res} is
    \begin{equation}
      \frac{7}{2}{DN\left\lceil \log_2 N\right\rceil}^2 + O\left(DN\right),
    \end{equation}
    in terms of the required number of $H$, $CZ$, and $CCZ$ gates.
\end{theorem}

\begin{proof}[\textbf{Proof of Proposition~\ref{prp:num_vertices_hyperedges}}]
    We count vertices and hyperedges of $G_\mathrm{res}^{\left(N,D\right)}$ in the following.

    As for vertices, consider first a case where $N$ is a power of $2$ so that $\log_2 N$ is an integer, while general cases of any $N$ are discussed later.
    The hypergraph $G_\mathrm{res}^{\left(N,D\right)}$ consists of $D$ sub-hypergraphs $G_\textup{depth}^{\left(N,d\right)}$, as shown in~\eqref{eq:g_res}.
    Each $G_\textup{depth}^{\left(N,d\right)}$ consists of one subgraph $G_\textup{S}^{\left(N,d\right)}$, $\left\lfloor\frac{N}{3}\right\rfloor$ sub-hypergraphs $G_{\mathbbm{1},H,CCZ}^{\left(n,d\right)}$, and $\left(N-3\left\lfloor \frac{N}{3}\right\rfloor\right)$ subgraphs $G_{\mathbbm{1},H,CCZ}^{\left(n,d\right)}$, as shown in~\eqref{eq:one_depth}.
    The graph $G_\textup{S}^{\left(N,d\right)}$ corresponding to a sorting network consists of subgraphs $G_\textup{C}^{\left(n_1,n_2,d\right)}$ corresponding to its comparators, as shown in~\eqref{eq:G_S_N_d},
    and when an odd-even merging network is used as the sorting network,
    there are $\left(\frac{1}{4}N\left(\log_2^2 N - \log_2 N + 4 \right)-1\right)$ subgraphs $G_\textup{C}^{\left(n_1,n_2,d\right)}$ as shown in~\eqref{eq:comparators_sorting_network}.
    The number of vertices of $G_\textup{C}^{\left(n_1,n_2,d\right)}$ other than $v_{n_1,d}$ and $v_{n_2,d}$ is 6,
    and the numbers of vertices of $G_{\mathbbm{1},H,CCZ}^{\left(n,d\right)}$ and $G_{\mathbbm{1},H}^{\left(n,d\right)}$ other than $u_{n-2,d}$, $u_{n-1,d}$, and $u_{n,d}$ are 18 and 3, respectively, where the counted vertices are illustrated as solid circles in Fig.~\ref{fig:parts}, and those illustrated as dotted circles in Fig.~\ref{fig:parts} are not counted for avoiding double counting.
    Counting $N$ vertices $v_{1,1},\ldots,v_{N,1}$ of $G_\mathrm{res}^{\left(N,D\right)}$ in addition to vertices of these sub-hypergraphs,
    for any $N$ that is a power of $2$,
    we obtain the number of vertices of $G_\mathrm{res}^{\left(N,D\right)}$
    \begin{equation}
        \begin{aligned}
            &N+D\left[\left(\frac{1}{4}N\left(\log_2^2 N- \log_2 N + 4 \right)-1\right)\cdot 6\right.\\
            &\quad\left.+ \left\lfloor\frac{N}{3}\right\rfloor\cdot 18 + \left(N-3\left\lfloor \frac{N}{3}\right\rfloor\right)\cdot 3\right].
        \end{aligned}
    \end{equation}
    In the general cases, let $\left|V_{G_\mathrm{res}^{\left(N,D\right)}}\right|$ denote the number of vertices of $G_\mathrm{res}^{\left(N,D\right)}$ for any $N$ and $D$,
    and then $\left|V_{G_\mathrm{res}^{\left(N,D\right)}}\right|$ satisfies
    \begin{equation}
      \label{eq:num_vertices_proof}
        \begin{aligned}
            &\left|V_{G_\mathrm{res}^{\left(N,D\right)}}\right|\\
            &\quad\leqq N+D\left[\left(\frac{1}{4}N\left({\left\lceil\log_2 N\right\rceil}^2- \left\lceil\log_2 N\right\rceil + 4 \right)-1\right)\cdot 6\right.\\
             &\qquad\left.+ \left\lfloor\frac{N}{3}\right\rfloor\cdot 18 + \left(N-3\left\lfloor \frac{N}{3}\right\rfloor\right)\cdot 3\right]\\
             &\quad=\frac{3}{2}{DN\left\lceil \log_2 N\right\rceil}^2+O\left(DN\right),
        \end{aligned}
    \end{equation}
    which yields~\eqref{eq:num_vertices}.

    As for hyperedges, the number of hyperedges of $G_\textup{C}^{\left(n_1,n_2,d\right)}$, $G_{\mathbbm{1},H,CCZ}^{\left(n,d\right)}$, and $G_{\mathbbm{1},H}^{\left(n,d\right)}$ are 8, 34, and 4, respectively,
    and in a similar way to the above counting of vertices,
    for any $N$ that is a power of $2$,
    we obtain the number of hyperedges of $G_\mathrm{res}^{\left(N,D\right)}$
    \begin{equation}
        \begin{aligned}
            &D\left[\left(\frac{1}{4}N\left(\log_2^2 N- \log_2 N + 4 \right)-1\right)\cdot 8\right.\\
            &\quad\left.+ \left\lfloor\frac{N}{3}\right\rfloor\cdot 34 + \left(N-3\left\lfloor \frac{N}{3}\right\rfloor\right)\cdot 4\right].
        \end{aligned}
    \end{equation}
    In the general cases, let $\left|E_{G_\mathrm{res}^{\left(N,D\right)}}\right|$ denote the number of hyperedges of $G_\mathrm{res}^{\left(N,D\right)}$ for any $N$ and $D$,
    and then $\left|E_{G_\mathrm{res}^{\left(N,D\right)}}\right|$ satisfies
    \begin{equation}
      \label{eq:num_edges_proof}
        \begin{aligned}
            &\left|E_{G_\mathrm{res}^{\left(N,D\right)}}\right|\\
            &\quad\leqq D\left[\left(\frac{1}{4}N\left({\left\lceil\log_2 N\right\rceil}^2- {\left\lceil\log_2 N\right\rceil} + 4 \right)-1\right)\cdot 8\right.\\
                           &\qquad\left.+ \left\lfloor\frac{N}{3}\right\rfloor\cdot 34 + \left(N-3\left\lfloor \frac{N}{3}\right\rfloor\right)\cdot 4\right]\\
                           &\quad=2{DN\left\lceil \log_2 N\right\rceil}^2 + O\left(DN\right),
        \end{aligned}
    \end{equation}
    which yields~\eqref{eq:num_hyperedges}.

    The number of hyperedges among three vertices of $G_\mathrm{res}^{\left(N,D\right)}$ is given by~\eqref{eq:num_ccz} because each $G_{\mathbbm{1},H,CCZ}^{\left(n,d\right)}$ has one hyperedge among three vertices, and $G_\textup{C}^{\left(n_1,n_2,d\right)}$ and $G_{\mathbbm{1},H,CCZ}^{\left(n,d\right)}$ have no hyperedges among three vertices.
\end{proof}

\begin{proof}[\textbf{Proof of Theorem~\ref{thm:preparation_complexity}}]
    The hypergraph state $\Ket{G_\mathrm{res}^{\left(N,D\right)}}$ can be prepared from $\Ket{0}\otimes\Ket{0}\otimes\cdots$ by applying $H$ gates to all the qubits, which yields $\Ket{+}\otimes\Ket{+}\otimes\cdots$, followed by applying $CZ$ and $CCZ$ gates corresponding to all the hyperedges of $G_\mathrm{res}^{\left(N,D\right)}$.
    The required number of $H$ gates equals the number of vertices of $G_\mathrm{res}^{\left(N,D\right)}$ given by~\eqref{eq:num_vertices_proof} in the proof of Proposition~\ref{prp:num_vertices_hyperedges}.
    In addition, the required number of $CZ$ and $CCZ$ gates equals the number of hyperedges of $G_\mathrm{res}^{\left(N,D\right)}$ given by~\eqref{eq:num_edges_proof} in the proof of Proposition~\ref{prp:num_vertices_hyperedges}.
    In total, the required number of $H$, $CZ$, and $CCZ$ gates are
    \begin{equation}
      \frac{7}{2}{DN\left\lceil \log_2 N\right\rceil}^2 + O\left(DN\right),
    \end{equation}
    which yields the conclusion.
\end{proof}

\begin{remark}[Saving of required number of qubits in MBQC]
  The hypergraph state $\Ket{G_\mathrm{res}^{\left(N,D\right)}}$ is designed so that the number of qubits for $\Ket{G_\mathrm{res}^{\left(N,D\right)}}$ should be small.
  As will be shown in Sec.~\ref{sec:quantum_complexity}, we can use $\Ket{G_\mathrm{res}^{\left(N,D\right)}}$ as a resource for MBQC simulating an $N$-qubit $D$-depth quantum circuit composed of a gate set $\left\{H,CCZ\right\}$.
  The saving of the number of qubits for $\Ket{G_\mathrm{res}^{\left(N,D\right)}}$ stems from two factors.
  First, the use of the sorting network in $\Ket{G_\mathrm{res}^{\left(N,D\right)}}$ reduces the number of qubits; in particular, in terms of $N$, the existing resource states for MBQC with polynomial overhead require $\Omega(N^{1+\alpha})$ qubits with $\alpha>0$ to simulate the $N$-qubit circuit due to the geometrical constraint as discussed in Sec.~\ref{sec:introduction}, but the required number of qubits for $\Ket{G_\mathrm{res}^{\left(N,D\right)}}$ is reduced to $O(N\log^2 N)$ as shown in Proposition~\ref{prp:num_vertices_hyperedges}.
  Second, we reduce the number of qubits required for implementing each $H$ and $CCZ$ gate in MBQC\@; for example, a hypergraph state $\Ket{G_3^1}$ in Ref.~\cite{Y1} uses $66$ qubits to implement an arbitrary $3$-qubit $1$-depth circuit composed of $\left\{H,CCZ\right\}$ by MBQC, but $\Ket{G_{\mathbbm{1},H,CCZ}^{\left(n,d\right)}}$ in Fig.~\ref{fig:parts} uses as few as $21$ qubits for this implementation.
  A drawback of $\Ket{G_\mathrm{res}^{\left(N,D\right)}}$ may arise from the fact that due to the sorting network, the preparation of $\Ket{G_\mathrm{res}^{\left(N,D\right)}}$ requires interactions that are not necessarily local with respect to some $2$- or $3$-dimensional geometry.
  However, the nonlocal interactions are vital to overcoming the geometrical constraint causing the polynomial overhead,
  and when we use photonic GKP qubits to prepare resource states for MBQC,
  we can move photonic systems in space feasibly to realize the nonlocal interactions.
  In the case of photonic architectures, rather than the nonlocal interactions, the cost of realizing non-Gaussian operations is crucial.
  Since a GKP state is a non-Gaussian state that is technologically costly to generate,
  the saving of the number of qubits in MBQC using $\Ket{G_\mathrm{res}^{\left(N,D\right)}}$ is advantageous to MBQC using the GKP qubits in reducing the cost of non-Gaussianity.
\end{remark}

\begin{remark}[The required number of $CCZ$ gates for resource state preparation]
    As discussed in Sec.~\ref{sec:mbqc},
    when we aim at preparing resource states for MBQC using photonic GKP qubits,
    Clifford gates implemented by Gaussian operations are preferable,
    and to simplify the resource state preparation, reducing the required number of non-Clifford $CCZ$ gates is advantageous.
    As will be shown in Sec.~\ref{sec:quantum_complexity}, the hypergraph state $\Ket{G_\mathrm{res}^{\left(N,D\right)}}$ is used as a resource for MBQC simulating an $N$-qubit $D$-depth quantum circuit composed of a gate set $\left\{H,CCZ\right\}$, and such a circuit may include $\left\lfloor\frac{N}{3}\right\rfloor$ $CCZ$ gates per depth.
    Theorem~\ref{thm:preparation_complexity} shows that the number of $CCZ$ gates that are used for preparing $\Ket{G_\mathrm{res}^{\left(N,D\right)}}$ is $D\left\lfloor\frac{N}{3}\right\rfloor$ in total, \textit{i.e.}, $\left\lfloor\frac{N}{3}\right\rfloor$ per depth of the circuit to be simulated, and all the other required operations than these $CCZ$ gates are Clifford gates that are also used for preparing graph states, that is, $H$ and $CZ$ gates.
    Note that the existing resource hypergraph states in Refs.~\cite{M2,M6,G4,Y1} do not consider this optimization in terms of the number of $CCZ$ gates;
    \textit{e.g.}, to implement an $N$-qubit $D$-depth circuit composed of $\left\{H,CCZ\right\}$, the resource hypergraph state in Ref.~\cite{Y1} uses $O\left(N^3\right)$ $CCZ$ gates per depth, rather than $O\left(N\right)$ in our case.
    Our construction of $\Ket{G_\mathrm{res}^{\left(N,D\right)}}$ is advantageous since no redundant $CCZ$ gate is used; that is, the $\left\lfloor\frac{N}{3}\right\rfloor$ $CCZ$ gates used in the preparation of $\Ket{G_\mathrm{res}^{\left(N,D\right)}}$ per depth are necessary for implementing at most $\left\lfloor\frac{N}{3}\right\rfloor$ $CCZ$ gates at each depth of the $N$-qubit quantum circuit, as long as we use a periodic resource state whose period corresponds to each depth of the circuit.
    Note that a general proof on the optimality of the number of $CCZ$ gates among all possible resource states without assuming the periodicity is unknown.
\end{remark}

\begin{remark}[\label{rem:depth}The depth for resource state preparation]
  While Theorem~\ref{thm:preparation_complexity} on the preparation complexity of the hypergraph state $\Ket{G_\mathrm{res}^{\left(N,D\right)}}$ shows the size of the quantum circuit for preparing $\Ket{G_\mathrm{res}^{\left(N,D\right)}}$,
  the depth of this quantum circuit can be as small as a constant in terms of $D$ and $N$, as shown in Appendix~\ref{sec:depth}.
  This constant-depth resource preparation where the gates are performed in parallel may be beneficial to reducing errors caused by finite decoherence time of quantum systems and loss in optical fibers in photonic MBQC, while the preparation complexity for sequential resource preparation also matters to photonic MBQC illustrated in Fig.~\ref{fig:introduction} where resource states are generated in order.
\end{remark}

\subsection{\label{sec:quantum_complexity}Protocol for universal MBQC and its quantum and classical complexities}

We prove the universality of MBQC using our resource hypergraph state $\Ket{G_\mathrm{res}^{\left(N,D\right)}}$, by constructing an MBQC protocol.
In particular, we show the following theorem on the complexity of our MBQC protocol.
The generality of our construction of universal resource states and the corresponding MBQC protocols are discussed in Remark~\ref{rem:generality} after giving the proof.
Remarkably, this theorem achieves a poly-logarithmic overhead in simulating quantum circuits by MBQC \textit{without assuming} geometrical locality of gates in the circuits to be simulated.
Moreover, this MBQC protocol can be performed only with $Z$ and $X$ measurements.

\begin{theorem}[\label{thm:universality}Universal resources for polylog-overhead MBQC with $Z$ and $X$ measurements]
    There exists an MBQC protocol that uses $Z$ and $X$ measurements for a family of hypergraph states ${\left\{\Ket{G_\mathrm{res}^{\left(N,D\right)}}\right\}}_{N,D}$ defined as~\eqref{eq:resources}, to simulate an arbitrary $N$-qubit $D$-depth quantum circuit composed of a computationally universal gate set $\left\{H,CCZ\right\}$ within the complexity
    \begin{equation}
        \label{eq:complexity}
        O\left(DN\log^2 N\right),
    \end{equation}
  where the complexity includes the preparation, quantum, and classical complexities as discussed in Sec.~\ref{sec:mbqc}.
\end{theorem}

To show the Theorem~\ref{thm:universality},
we recall state transformations of graph and hypergraph states using $Z$ and $X$ measurements, as well as commutation relations used for correcting byproduct operators induced by these $Z$ and $X$ measurements.
In the following,
given any hypergraph (including a graph), vertices are illustrated as circles, and an edge between two vertices and a hyperedge among three vertices are illustrated as a black line and a red rectangle, respectively, in the same way as Fig.~\ref{fig:parts}.
For a hypergraph state represented by the hypergraph, a qubit that is supposed to be measured in the $Z$ basis and that measured in the $X$ basis are represented as a vertex illustrated by a black and gray circle, respectively.
For clarity, a qubit where we input a state, which is not necessarily $\Ket{+}$, is illustrated by a circle surrounded by a square.
For $Z$ measurements in $\left\{\Ket{0},\Ket{1}\right\}$, the outcomes corresponding to $\Ket{0}$ and $\Ket{1}$ are denoted by $k=0$ and $k=1$, respectively.
For $X$ measurements in $\left\{\Ket{+},\Ket{-}\right\}$, the outcomes corresponding to $\Ket{+}$ and $\Ket{-}$ are denoted by $k=0$ and $k=1$, respectively.

Given a hypergraph state, if a $Z$ measurement is performed on a qubit,
the measured qubit is disentangled and removed from the rest of this hypergraph state,
while depending on the measurement outcome, byproduct operators may be applied to the neighboring qubits of this measured qubit on the corresponding hypergraph.
For example, the following circuit identity holds:
\begin{equation}
    \label{eq:z_measurement}
    \includegraphics[width=3.4in]{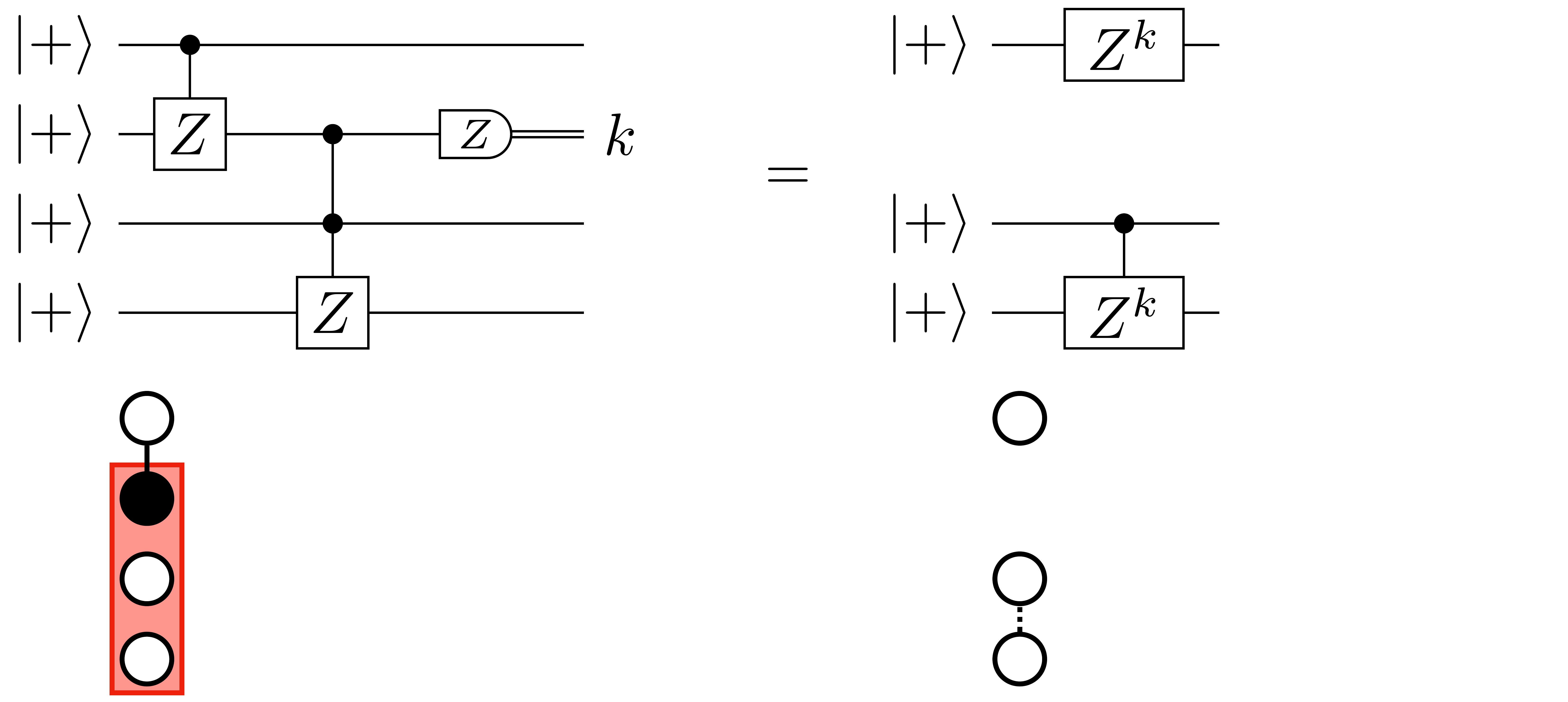}
\end{equation}
where $k$ is the outcome of the $Z$ measurement, the hypergraphs below the circuits correspond to hypergraph states obtained from the circuits, and the dotted edge represents byproduct $CZ^k$ depending on $k$ and acting on the two neighboring qubits of the measured qubit with respect to the initial hyperedge representing $CCZ$, while local Pauli byproduct $Z^k$ acting on another neighboring qubit with respect to the initial edge representing $CZ$ is not explicitly illustrated in the hypergraph.

Apart from this $Z$ measurement, an $X$ measurement performed on a qubit in a graph state can be understood as quantum teleportation~\cite{B13}.
In particular, if a qubit in a graph state is measured in the $X$ basis, the state of the measured qubit is transferred to a qubit represented as a neighbor on the graph, with the Hadamard gate $H$ and a byproduct Pauli operator applied to this state.
For example, the following circuit identity holds:
\begin{equation}
    \label{eq:x_measurement}
    \includegraphics[width=3.4in]{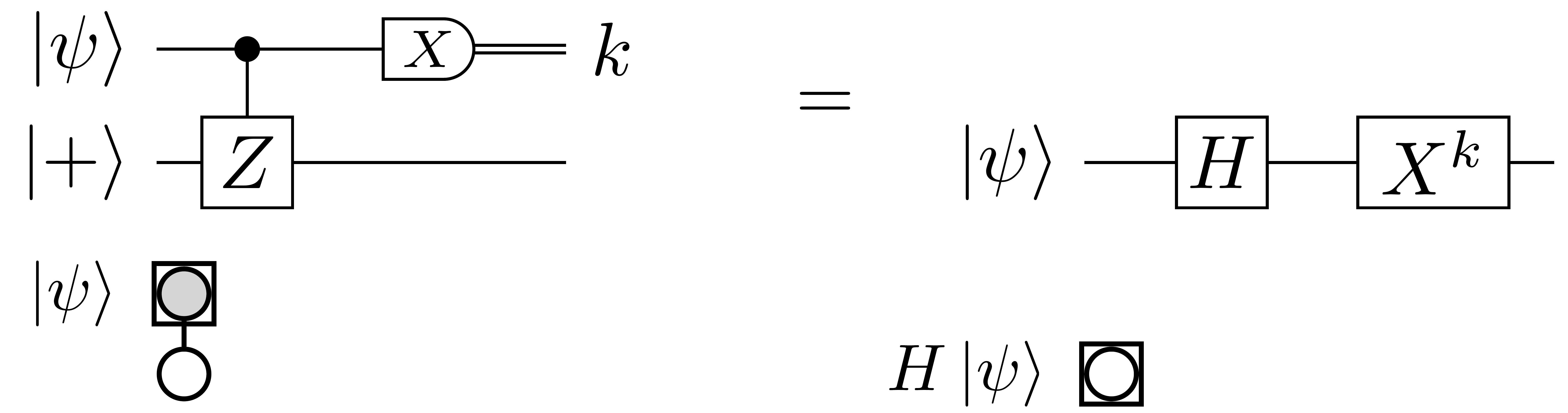}
\end{equation}
where $k$ is the outcome of the $X$ measurement, and local Pauli byproduct $X^k$ acting on the neighboring qubit of the measured qubit is not explicitly illustrated.

\begin{figure*}[t]
    \centering
    \includegraphics[width=7.0in]{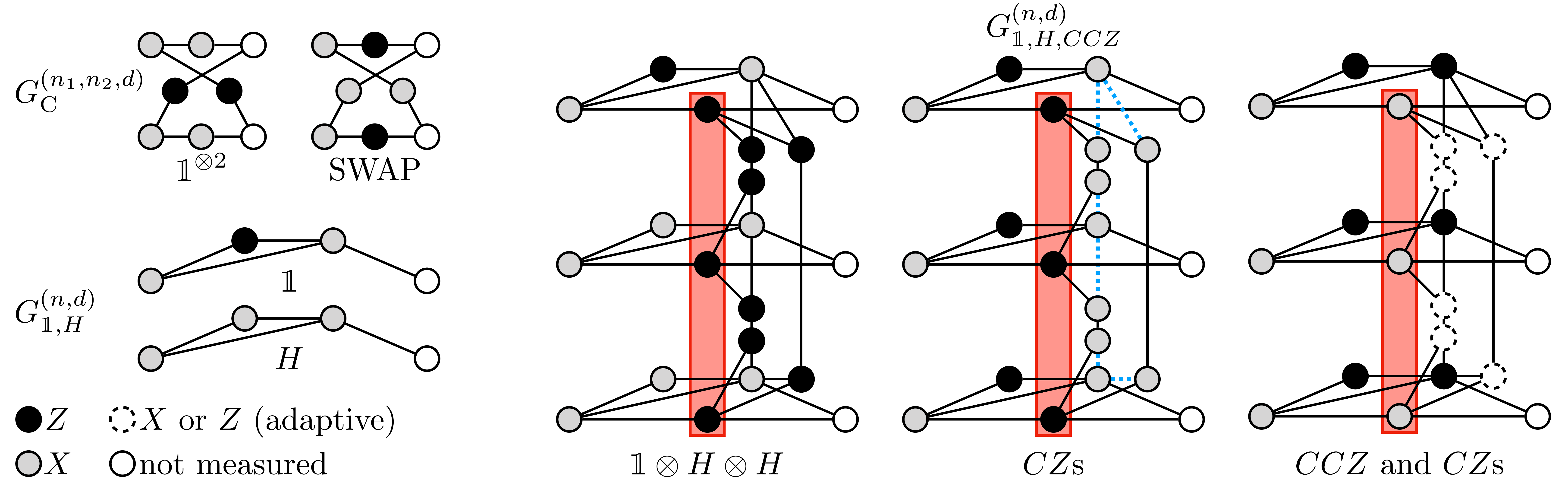}
    \caption{\label{fig:measurement}Measurement patterns of the resource hypergraph state $\Ket{G_\mathrm{res}^{\left(N,D\right)}}$ for implementing quantum gates used in our MBQC protocol in Theorem~\ref{thm:universality} to achieve universal quantum computation, where edges and hyperedges are illustrated in the same way as Fig.~\ref{fig:parts}. The patterns for implementing the identity and $\textsc{SWAP}$ gates using resources corresponding to the hypergraph $G_\textup{C}^{\left(n_1,n_2,d\right)}$ are shown on the top left, those for the identity and $H$ gates using $G_{\mathbbm{1},H}^{\left(n,d\right)}$ are on the bottom left, and those for the identity, $H$, $CZ$, and $CCZ$ gates using $G_{\mathbbm{1},H,CCZ}^{\left(n,d\right)}$ are on the right. For each pattern, the input qubits correspond to the leftmost vertices, and the output qubits that are not measured in the pattern are illustrated as the rightmost white vertices. Black and gray vertices represent qubits that are measured in the $Z$ and $X$ bases, respectively. For implementing $CCZ$ gates, the $CZ$ byproducts can be corrected by adaptively choosing $Z$ or $X$ measurements of qubits represented as dashed white vertices, so that appropriate $CZ$ gates for the correction can be adaptively implemented. As discussed in Appendix~\ref{sec:parallelizability}, to implement multiple $CCZ$ gates in parallel, we can perform these corrections of $CZ$ byproducts collectively later using the blue dotted edges for implementing $CZ$ gates, while these blue dotted edges are optional for the universality.}
\end{figure*}

To achieve deterministic MBQC using our resource hypergraph state $\Ket{G_\mathrm{res}^{\left(N,D\right)}}$,
byproduct operators induced by these $Z$ and $X$ measurements must be corrected appropriately.
For the correction, outcomes of these measurements are memorized temporarily and used in the feed-forward classical processes for calculating commutation relations
\begin{equation}
    \label{eq:cz_z}
    \includegraphics[width=3.4in]{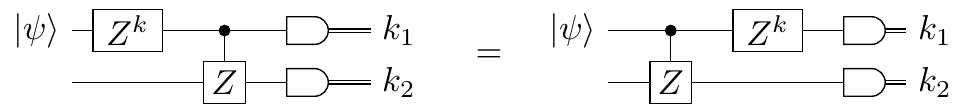}
\end{equation}
\begin{equation}
    \label{eq:cz_x}
    \includegraphics[width=3.4in]{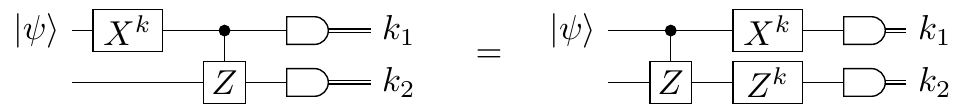}
\end{equation}
\begin{equation}
    \label{eq:ccz_z}
    \includegraphics[width=3.4in]{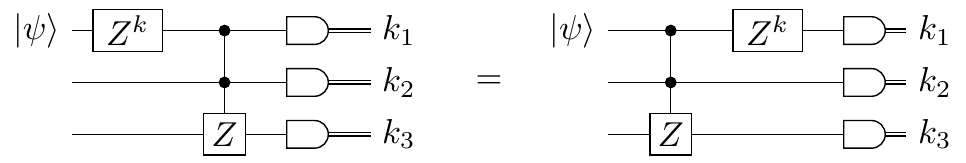}
\end{equation}
\begin{equation}
    \label{eq:ccz_x}
    \includegraphics[width=3.4in]{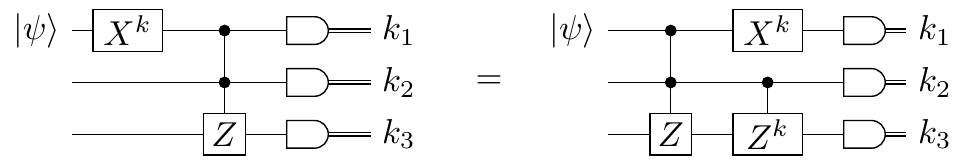}
\end{equation}
where byproduct operators conditioned on a previous measurement outcome $k$ and applied directly to the input state $\Ket{\psi}$ on the left-hand sides are translated into those applied just before any measurements on the right-hand sides by commuting them with the $CZ$ and $CCZ$ gates corresponding to hyperedges, which have been applied to the qubits in the initial preparation of the resource hypergraph state $\Ket{G_\mathrm{res}^{\left(N,D\right)}}$.
As long as $Z$ and $X$ measurements are concerned,
the $Z^k$ and $X^k$ byproduct operators applied just before the $Z$ and $X$ measurements on these right-hand sides only flip the outcomes of the $Z$ and $X$ measurements due to
\begin{alignat}{7}
    &Z\Ket{0}&&=&&\Ket{0},\,&&Z\Ket{1}&&=&-&\Ket{1},\\
    &X\Ket{0}&&=&&\Ket{1},\,&&X\Ket{1}&&=& &\Ket{0},\\
    &Z\Ket{+}&&=&&\Ket{-},\,&&Z\Ket{-}&&=& &\Ket{+},\\
    &X\Ket{+}&&=&&\Ket{+},\,&&X\Ket{-}&&=&-&\Ket{-},
\end{alignat}
and hence according to these equations, these local Pauli byproducts $Z^k$ and $X^k$ can be corrected by appropriately relabeling measurement outcomes that are labeled by $\left\{0,1\right\}$, without adaptively changing the measurement bases.
As for the $CZ^k$ byproduct operator applied before the measurements on the right-hand side of~\eqref{eq:ccz_x}, our MBQC protocol adaptively implement $CZ$ gates conditioned on $k$ to cancel this byproduct $CZ^k$.
This MBQC protocol is described in the following proof of Theorem~\ref{thm:universality}.

\begin{proof}[\textbf{Proof of Theorem~\ref{thm:universality}}]
    Using the resource hypergraph state $\Ket{G_\mathrm{res}^{\left(N,D\right)}}$,
    we construct an MBQC protocol whose complexity of simulating an arbitrary $N$-qubit $D$-depth quantum circuit composed of $\left\{H,CCZ\right\}$ satisfies the bound~\eqref{eq:complexity}.
    Since Theorem~\ref{thm:preparation_complexity} shows that the preparation complexity of $\Ket{G_\mathrm{res}^{\left(N,D\right)}}$ is $O\left(DN\log^2 N\right)$, it suffices to show that the quantum and classical complexities are also $O\left(DN\log^2 N\right)$, which we will show after providing the description of the MBQC protocol.

    \textbf{Sketch of the MBQC protocol}:
    Given an arbitrary $N$-qubit $D$-depth quantum circuit composed of $\left\{H,CCZ\right\}$ illustrated on the left-hand side of Fig.~\ref{fig:circuit_n_qubit_d_depth}, our MBQC protocol implements an equivalent quantum circuit illustrated on the right-hand side of Fig.~\ref{fig:circuit_n_qubit_d_depth}.
    Let $\Ket{\psi_d}$ for each $d\in\left\{1,\ldots,D\right\}$ denote the $N$-qubit state generated from $\Ket{0}^{\otimes N}$ by the first $d$-depth part of the given circuit (see the left-hand side of Fig.~\ref{fig:circuit_n_qubit_d_depth}).
    In the equivalent circuit,
    we obtain an $N$-qubit state $H^{\otimes N}\Ket{\psi_d^\prime}$ from $\Ket{+}^{\otimes N}$ after the corresponding part (see the right-hand side of Fig.~\ref{fig:circuit_n_qubit_d_depth}),
    where $\Ket{\psi_d^\prime}$ is given by permuting the $N$ qubits of $\Ket{\psi_d}$.
    In our MBQC protocol, each quantum gate applied to $\Ket{\psi_d}$ in the given circuit is implemented on $H^{\otimes N}\Ket{\psi_d^\prime}$ by first performing a set of the Hadamard gates $H$ to obtain $\Ket{\psi_d^\prime}$, and then applying the gate to $\Ket{\psi_d^\prime}$, followed by performing another set of Hadamard gates $H$; \textit{e.g.}, a $CCZ$ gate on $\Ket{\psi_d}$ is simulated by implementing $(H^{\otimes 3})CCZ (H^{\otimes 3})$ on $H^{\otimes N}\Ket{\psi_d^\prime}$, as shown in Fig.~\ref{fig:circuit_n_qubit_d_depth}.
    While classical outputs in the circuit model are obtained by $Z$ measurements for the $N$-qubit state $\Ket{\psi_D}$ output from the $D$-depth circuit,
    our MBQC protocol can prepare the state $H^{\otimes N}\Ket{\psi_D^\prime}$ on the $N$ qubits represented as the vertices $v_{1,D+1},\ldots,v_{N,D+1}$ of $G_\mathrm{res}^{\left(N,D\right)}$ illustrated in Fig.~\ref{fig:resource}, by appropriately performing $Z$ and $X$ measurements of all the rest of the qubits that are initially prepared in $\Ket{G_\mathrm{res}^{\left(N,D\right)}}$.
    By performing $X$ measurements for $H^{\otimes N}\Ket{\psi_D^\prime}$, we can obtain the same classical outputs as those of $Z$ measurements for $\Ket{\psi_D}$ in the corresponding circuit model up to permutation.

    Since $G_\mathrm{res}^{\left(N,D\right)}$ is inductively defined using the same hypergraphs $G_\textup{depth}^{\left(N,1\right)},\ldots,G_\textup{depth}^{\left(N,D\right)}$ as shown in~\eqref{eq:g_res_inductive},
    we prove by induction that for any $d\in\left\{1,\ldots,D\right\}$, MBQC using $\Ket{G_\mathrm{res}^{\left(N,d\right)}}$ can prepare $H^{\otimes N}\Ket{\psi_d^\prime}$.
    For this proof, it suffices to show that we can use each $G_\textup{depth}^{\left(N,d\right)}$ given in~\eqref{eq:one_depth} and composed of $G_\textup{S}^{\left(N,d\right)}$, $G_{\mathbbm{1},H,CCZ}^{\left(n,d\right)}$, and $G_{\mathbbm{1},H}^{\left(n,d\right)}$ to implement an arbitrary $N$-qubit $1$-depth quantum circuit up to permutation.
    As illustrated on the right-hand side of Fig.~\ref{fig:circuit_n_qubit_d_depth}, we show that our MBQC protocol  can achieve this implementation of an arbitrary $N$-qubit $1$-depth quantum circuit by performing an arbitrary permutation of $N$ qubits using $G_\textup{S}^{\left(N,d\right)}$, followed by performing identity, $H$, and geometrically local $CCZ$ gates using $G_{\mathbbm{1},H,CCZ}^{\left(n,d\right)}$ and $G_{\mathbbm{1},H}^{\left(n,d\right)}$.
    In the following, we provide patterns of the $Z$ and $X$ measurements for these implementations, followed by analysis of quantum and classical complexities of this MBQC protocol,
    where we refer to the vertices of hypergraphs $G_\textup{C}^{\left(n_1,n_2,d\right)}$, $G_{\mathbbm{1},H}^{\left(n,d\right)}$, and $G_{\mathbbm{1},H,CCZ}^{\left(n,d\right)}$ using the labels shown in Fig.~\ref{fig:parts}.

    \textbf{Implementation of permutations}:
    To implement the permutation using $G_\textup{S}^{\left(N,d\right)}$, by construction of $G_\textup{S}^{\left(N,d\right)}$ using a sorting network as shown in Fig.~\ref{fig:sort_graph},
    it suffices to show that each $G_\textup{C}^{\left(n_1,n_2,d\right)}$ corresponding to each comparator of the sorting network can be used for implementing identity and $\textsc{SWAP}$ gates from the two input qubits represented as vertices $v_{n_1,d}$ and $v_{n_2,d}$, to the two output qubits represented as $u_{n_1,d}$ and $u_{n_2,d}$.
    The measurement pattern of $G_\textup{C}^{\left(n_1,n_2,d\right)}$ for implementing the identity gate $\mathbbm{1}^{\otimes 2}$ for the two qubits is labeled ``$\mathbbm{1}^{\otimes 2}$'' on the top left of Fig.~\ref{fig:measurement}.
    In this case, $Z$ measurements acting as shown in~\eqref{eq:z_measurement} are used for deleting vertices so that the resulting graph consists of a unique path connecting $v_{n_1,d}$ and $u_{n_1,d}$, and that connecting $v_{n_2,d}$ and $u_{n_2,d}$.
    For each path, $X$ measurements acting as shown in~\eqref{eq:x_measurement} are used for teleporting the input state with the $H$ gate applied to the input state.
    Since $H^2=\mathbbm{1}$, repeating the $X$ measurements twice on each path transfers any input state of a qubit represented as $v_{n_1,d}$ (and $v_{n_2,d}$) to the same output state of that represented as $u_{n_1,d}$ (and $u_{n_2,d}$, respectively).
    The measurement pattern of $G_\textup{C}^{\left(n_1,n_2,d\right)}$ for implementing the $\textsc{SWAP}$ gate, which is labeled ``$\textsc{SWAP}$'' on the top left of Fig.~\ref{fig:measurement}, works in a similar way, using two paths connecting $v_{n_1,d}$ and $u_{n_2,d}$, and connecting $v_{n_2,d}$ and $u_{n_1,d}$.
    Appropriately choosing these measurement patterns for implementing identity and $\textsc{SWAP}$ gates for each subgraph $G_\textup{C}^{\left(n_1,n_2,d\right)}$ corresponding to a comparator,
    we can implement any permutation of $N$ qubits using $G_\textup{S}^{\left(N,d\right)}$, where any state can be input from the qubits represented as $v_{1,d},\ldots,v_{N,d}$ of $G_\textup{S}^{\left(N,d\right)}$, and an arbitrarily permuted state of this input state is output to those represented as $u_{1,d},\ldots,u_{N,d}$.

    \textbf{Implementation of identity gates}:
    Similarly to the above implementation of the identity gate using of $G_\textup{S}^{\left(N,d\right)}$,
    the measurement pattern of $G_{\mathbbm{1},H}^{\left(n,d\right)}$ labeled ``$\mathbbm{1}$'' on the bottom left of Fig.~\ref{fig:measurement} can be used for implementing the identity gate, where a $Z$ measurement is used for making a path, and $X$ measurements are used for teleporting the input state twice.
    As for $G_{\mathbbm{1},H,CCZ}^{\left(n,d\right)}$, the pattern labeled ``$\mathbbm{1}\otimes H\otimes H$'' in Fig.~\ref{fig:measurement} works in the same way, where all the qubits illustrated as the black vertices are removed from the resource hypergraph state by $Z$ measurements.

    \textbf{Implementation of $H$ gates}:
    While the above implementation of the identity gate uses the $H$ gates twice,
    we implement the $H$ gate by using it only once rather than three times.
    Using~\eqref{eq:x_measurement} and linear-algebraic calculation,
    we have the following circuit identity:
    \begin{widetext}
      \begin{equation}
        \label{eq:triangle}
        \includegraphics[width=7.0in]{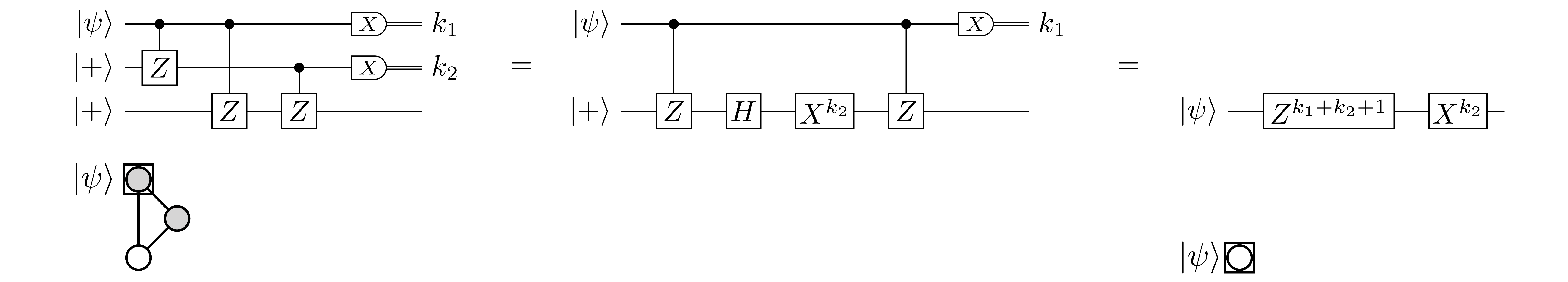}
      \end{equation}
    \end{widetext}
    indicating that given three qubits corresponding to three vertices that are connected as a cycle, $X$ measurements of two of these qubits implement the one-qubit identity gate, up to local Pauli byproduct operators.
    Thus, the measurement pattern of $G_{\mathbbm{1},H}^{\left(n,d\right)}$ that is labeled ``$H$'' on the bottom left of Fig.~\ref{fig:measurement} can be used for implementing the $H$ gate, where the input state of the qubit $u_{n,d}$ is first teleported to the qubit neighboring $v_{n,d}$ without applying $H$, and then teleported to the qubit $v_{n,d}$ with $H$ applied.
    As for $G_{\mathbbm{1},H,CCZ}^{\left(n,d\right)}$, the $H$ gate can be implemented by the same measurement pattern, as shown in the pattern labeled ``$\mathbbm{1}\otimes H\otimes H$'' in Fig.~\ref{fig:measurement}.

    \textbf{Implementation of $CZ$ gates}:
    In addition to these implementations of $\mathbbm{1}$ and $H$ using $G_{\mathbbm{1},H,CCZ}^{\left(n,d\right)}$,
    we can also use $G_{\mathbbm{1},H,CCZ}^{\left(n,d\right)}$ to implement $CZ$ gates by a measurement pattern labeled ``$CZ$s'' in Fig.~\ref{fig:measurement}.
    Note that this implementation of $CZ$ gates using the blue dotted edges in Fig.~\ref{fig:measurement} is optional for universal quantum computation to prove Theorem~\ref{thm:universality}; that is, a hypergraph state corresponding to $G_\mathrm{res}^{\left(N,D\right)}$ without these blue dotted edges still serves as a universal resource.
    However, this implementation of $CZ$ gates is beneficial to implementing multiple $CCZ$ gates in parallel as discussed in Appendix~\ref{sec:parallelizability}, and we will also apply a similar measurement pattern to correcting $CZ$ byproducts in implementing $CCZ$ gates.
    By performing $X$ measurements of the two qubits corresponding to a pair of connected vertices incident with blue dotted edges in Fig.~\ref{fig:measurement}, we can apply a $CZ$ gate to the two qubits represented as vertices connected to the measured qubits by the blue dotted edges, due to the following circuit identity:
    \begin{equation}
        \label{eq:xx}
        \includegraphics[width=3.4in]{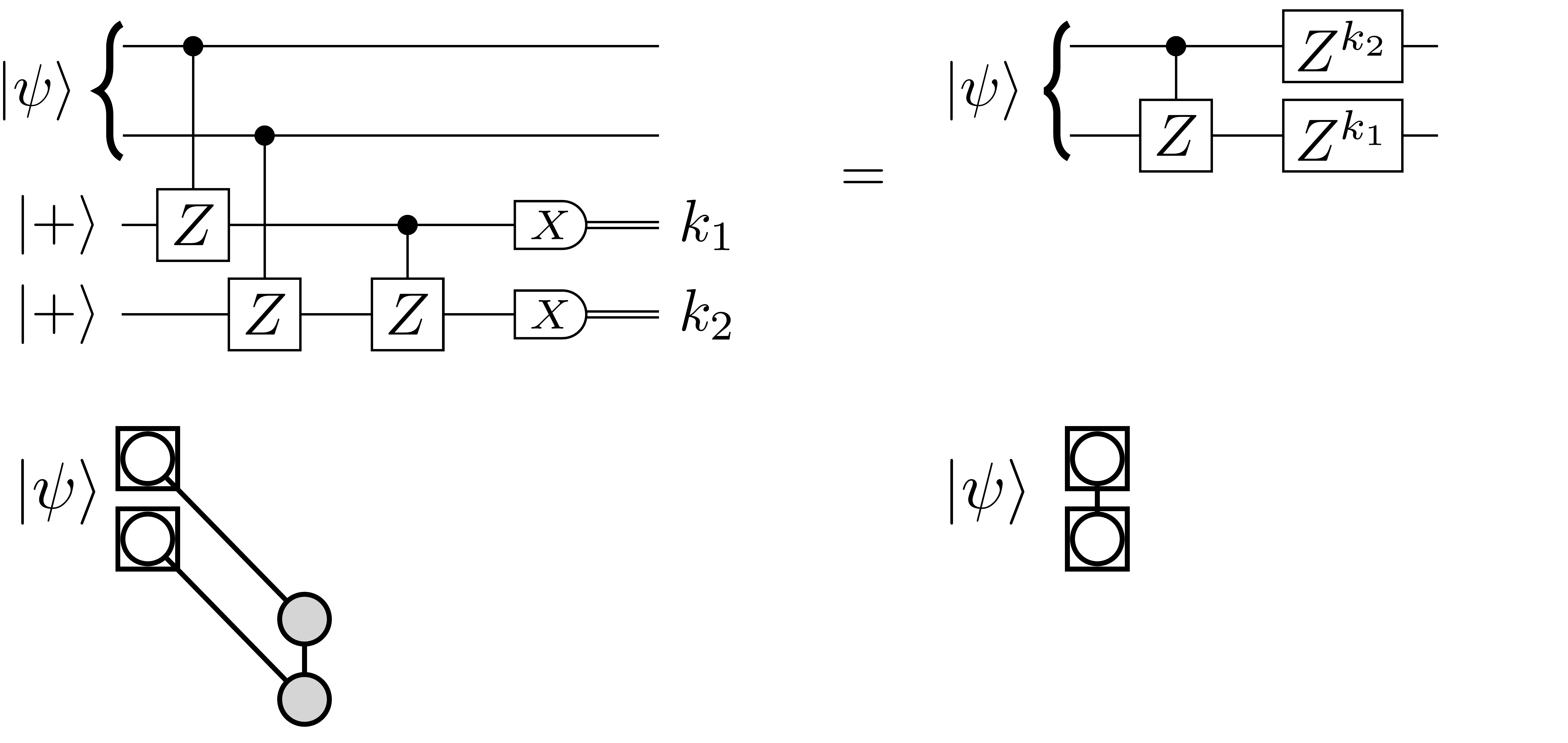}
    \end{equation}
    which indicates that given four qubits corresponding to four vertices connected in the form of a line graph, $X$ measurements of the inner two qubits directly connect the outer two by an edge representing $CZ$, up to local Pauli byproduct operators.
    Thus, the measurement pattern labeled ``$CZ$s'' in Fig.~\ref{fig:measurement} implements three $CZ$ gates acting on all the pairs of qubits in the $3$-qubit $1$-depth circuit corresponding to $G_{\mathbbm{1},H,CCZ}^{\left(n,d\right)}$, while $CZ$ gates for any two qubits can be implemented in the same way by removing unnecessary pairs of vertices by $Z$ measurements rather than the $X$ measurements.

    \textbf{Implementation of $CCZ$ gates}:
    To implement a $CCZ$ gate using $G_{\mathbbm{1},H,CCZ}^{\left(n,d\right)}$, instead of using the paths from $u_{n-2,d}$ to $v_{n-2,d}$, from $u_{n-1,d}$ to $v_{n-1,d}$, and from $u_{n,d}$ to $v_{n,d}$ for implementing identity and $H$ gates, we use the other paths from $u_{n-2,d},u_{n-1,d},u_{n,d}$ to $v_{n-2,d},v_{n-1,d},v_{n,d}$
    as shown in the measurement pattern labeled ``$CCZ$ and $CZ$s'' in Fig.~\ref{fig:measurement}.
    Due to the commutation relation~\eqref{eq:ccz_x}, two-qubit byproduct operators $CZ$ may appear in this measurement pattern, depending on the outcomes of $X$ measurements of the qubits represented as $u_{n-2,d}$, $u_{n-1,d}$, and $u_{n,d}$, as well as local Pauli byproduct operators caused by the previous measurements.
    Using~\eqref{eq:xx},
    we cancel these $CZ$ byproducts by adaptively choosing whether $Z$ or $X$ measurements are performed on two qubits that are represented as two connected dashed vertices in Fig.~\ref{fig:measurement} corresponding to each of these $CZ$ byproducts.

    \textbf{Complexity}:
    We analyze the quantum complexity and the classical complexity of this MBQC protocol, while the preparation complexity is $O(DN\log^2 N)$ as shown in Theorem~\ref{thm:preparation_complexity}.
    Regarding the quantum complexity,
    each $Z$ or $X$ measurement can be implemented by a constant-depth quantum circuit, and since there are $O\left(DN\log^2 N\right)$ qubits to be measured, the quantum complexity is
    \begin{equation}
        O\left(DN\log^2 N\right).
    \end{equation}

    Thus in the following, we analyze the classical complexity of our MBQC protocol.
    Given a classical description of any $N$-qubit $D$-depth quantum circuit composed of $\left\{H,CCZ\right\}$,
    we perform classical preprocessing before beginning measuring the qubits, so that this description of the circuit is converted into the form on the right-hand side of Fig.~\ref{fig:circuit_n_qubit_d_depth}.
    In such a conversion, for each of the $D$ depths, the sorting network is classically performed for identifying which of the comparators act as the identity gates and which as the $\textsc{SWAP}$ gates to achieve an appropriate permutation, taking $O\left(N\log^2 N\right)$ steps per depth for the odd-even merging networks.
    Thus, this classical preprocessing for all the $D$ depths takes $O\left(DN\log^2 N\right)$ steps in total, which determines measurement patterns to make the paths used in our MBQC protocol from each of $\left\{v_{1,d},\ldots,v_{N,d}\right\}$ to one of $\left\{v_{1,d+1},\ldots,v_{N,d+1}\right\}$ for each $d\in\left\{1,\ldots,D\right\}$.
    After this preprocessing, we obtain measurement patterns for the qubits except for those represented as dashed vertices of $G_{\mathbbm{1},H,CCZ}^{\left(n,d\right)}$ in Fig.~\ref{fig:measurement}, which may be used for corrections of $CZ$ byproducts in implementing $CCZ$ gates.

    In addition to this classical preprocessing,
    after performing each measurement of the qubits,
    we need feed-forward classical processes to calculate and memorize local Pauli byproduct operators applied to the $N$ qubits, so that the $CZ$ byproducts can be identified and corrected.
    The calculation of the byproducts can be performed within $O(DN\log^2 N)$ steps in total using the following procedure.
    For memorizing and tracing local Pauli byproducts for all the $O(DN\log^2 N)$ qubits (\textit{i.e.}, $\mathbbm{1}$, $X$, $Z$, or $XZ$ for each qubit), we assume that we can use an array data structure, where initialization of $O\left(DN\log^2 N\right)$ elements corresponding to these byproducts for the qubits takes $O\left(DN\log^2 N\right)$ steps, and random access to each element and rewriting each element take a constant step.
    In our MBQC protocol, before performing each $X$ measurement of a qubit represented by a vertex $v$, we list up all the neighboring vertices of $v$, and if some qubits corresponding to the neighboring vertices are supposed to be measured in the $Z$ basis, then we perform the $Z$ measurements of the neighboring qubits prior to the $X$ measurement for $v$.
    In the case of performing a $Z$ measurement of a qubit represented by a vertex $v'$, we list up all the neighboring vertices of $v'$ before the $Z$ measurement, and if some qubits corresponding to the neighboring vertices are supposed to be measured in the $Z$ basis, then we perform the $Z$ measurements of the neighboring qubits prior to the $Z$ measurement for $v'$.
    Applying this procedure recursively before each $X$ measurement, we can guarantee that each $X$ measurement is performed after  making the path by $Z$ measurements to remove the neighboring qubits.
    Meanwhile, we trace and correct local Pauli byproducts caused by the $X$ and $Z$ measurements, using the patterns~\eqref{eq:z_measurement},~\eqref{eq:x_measurement},~\eqref{eq:triangle}, and~\eqref{eq:xx} combined with the commutation relations~\eqref{eq:cz_z},~\eqref{eq:cz_x},~\eqref{eq:ccz_z}, and~\eqref{eq:ccz_x}.
    For each measurement, the number of qubits that are affected by local Pauli byproducts are \textit{upper bounded by a constant} since the maximal degree of $G_\mathrm{res}^{\left(N,D\right)}$ is five.
    Note that a degree of a vertex of a hypergraph refers to the number of hyperedges that include the vertex, and the maximal degree is the maximal of the degrees among all the vertices.
    Hence, the feed-forward classical processes for the measurements of all the $O\left(DN\log^2 N\right)$ qubits take $O\left(DN\log^2 N\right)$ steps in total.

    Therefore, including both the classical preprocessing and the feed-forward classical processes, the classical complexity is
    \begin{equation}
        O\left(DN\log^2 N\right).
    \end{equation}
    These analyses of quantum and classical complexities, together with the analysis of preparation complexity in Theorem~\ref{thm:preparation_complexity}, yield the conclusion.
\end{proof}

\begin{remark}[\label{rem:generality}Generality of our construction of universal resource states and MBQC protocols]
  Our construction of universal resource states $\Ket{G_\mathrm{res}^{\left(N,D\right)}}$ straightforwardly generalizes to other universal resource states in the following two ways.
  Firstly, we can use different sorting networks from the odd-even merging networks for implementing the qubit permutation.
  In particular, apart from the sorting networks discussed in Sec.~\ref{sec:definition}, namely, the bitonic sorters, the odd-even merging networks, and the pairwise networks, there exist sorting networks whose depth is $O\left(\log N\right)$ and whose number of comparators is $O\left(N\log N\right)$, asymptotically optimal in scaling~\cite{A2,A3,P4}.
  Instead of the odd-even merging networks, we may use such asymptotically optimal sorting networks to construct an MBQC protocol achieving the complexity
  \begin{equation}
    O\left(DN\log N\right),
  \end{equation}
  while the constant factors of the known asymptotically optimal sorting networks are impractically large.
  Note that it is an open problem to construct practical and asymptotically optimal sorting networks for any $N$, while shallower sorting networks than the odd-even merging networks are also known for some special $N$~\cite{K4}.
  In contrast to such large constant factors for the asymptotically optimal sorting networks, the odd-even merging networks yield practically small constant factors as shown in Proposition~\ref{prp:num_vertices_hyperedges} and Theorem~\ref{thm:preparation_complexity}.
  Secondly, in addition to this possibility of changing sorting networks corresponding to $G_\textup{S}^{n_1,n_2,d}$ in the definition of $G_\mathrm{res}^{N,D}$,
  hypergraphs $G_{\mathbbm{1},H,CCZ}^{\left(n,d\right)}$ and $G_{\mathbbm{1},H}^{\left(n,d\right)}$ for implementing each elementary gate in the gate set $\left\{H,CCZ\right\}$ can also be replaced with other entangled states for implementing a different gate set, without increasing the scaling of the complexity in terms of $D$ and $N$.
  Note that if this replacement uses an entangled state that is not a hypergraph state, the resulting resource state can be different from a hypergraph state, and the required measurement bases for universal MBQC are not necessarily the $Z$ and $X$ bases.
\end{remark}

\begin{remark}[\label{rem:verifiability}Verifiability]
  We remark on situations where the resource state $\Ket{G_\mathrm{res}^{\left(N,D\right)}}$ for the MBQC protocol in Theorem~\ref{thm:universality} are prepared by an untrusted quantum source that we need to verify using trusted measurement devices, so that we can guarantee that outcomes of MBQC protocols using these resource states are correct.
  Such situations may arise when an experimentalist constructs a quantum source for generating $\Ket{G_\mathrm{res}^{\left(N,D\right)}}$ and checks whether the source is working correctly, while the verification is also crucial for blind quantum computation~\cite{B14,B15}.
  We here discuss the complexity of verification of the multiqubit hypergraph state $\Ket{G_\mathrm{res}^{\left(N,D\right)}}$ for any $D$ and $N$.
  While the verification would be achievable by quantum state tomography~\cite{P1} at an exponential cost of the required number of states from the source in terms of $D$ and $N$,
  we show that in the verification of $\Ket{G_\mathrm{res}^{\left(N,D\right)}}$, the required number of states from the source can indeed be independent of $D$ and $N$.
  In the task of verification, the source repeatedly provides multiqubit quantum states that may not be independently and identically distributed, and we perform measurements of these states so that measurement outcomes can guarantee that this source provides $\Ket{G_\mathrm{res}^{\left(N,D\right)}}$.
  Each measurement of a state that is supposed to be $\Ket{G_\mathrm{res}^{\left(N,D\right)}}$ is called a test, and the state is said to pass a test if the outcomes of the measurements of all the qubits for this multiqubit state satisfy a certain success condition for passing the test.
  Given a significance level $\delta> 0$, an infidelity level $\epsilon> 0$, and $(t+1)$ states from a source,
  we use $t$ states for $t$ tests to verify whether the last one state $\rho$ is $\Ket{G_\mathrm{res}^{\left(N,D\right)}}$ or not.
  If the source provides $\Ket{G_\mathrm{res}^{\left(N,D\right)}}$ every time, then all the states from the source pass the tests,
  and on the other hand, if the fidelity between $\rho$ and $\Ket{G_\mathrm{res}^{\left(N,D\right)}}$ is smaller than $1-\epsilon$, that is,
  \begin{equation}
    \Braket{G_\mathrm{res}^{\left(N,D\right)}|\rho|G_\mathrm{res}^{\left(N,D\right)}}<1-\epsilon,
  \end{equation}
  then the probability of the states passing all the $t$ tests is at most $\delta$.
  A state-of-the-art protocol for this type of verification for a hypergraph state is given in Refs.~\cite{Z2,Z3,Z4}.
  The cost $t$ of verification in this protocol depends on the vertex coloring of the hypergraph corresponding to the hypergraph state, that is, the number $\chi$ of colors necessary for coloring vertices of the hypergraph so that neighboring vertices of the hypergraph have different colors.
  As shown in Appendix~\ref{sec:verifiability}, the minimal number of colors for the vertex coloring of the hypergraph $G_\mathrm{res}^{\left(N,D\right)}$ corresponding to $\Ket{G_\mathrm{res}^{\left(N,D\right)}}$ is $\chi=3$.
  Then, the protocol in Ref.~\cite{Z2} achieves the verification of $\Ket{G_\mathrm{res}^{\left(N,D\right)}}$ using $t$ tests where
  \begin{equation}
    t=\left\lfloor \frac{4.06}{\epsilon}\ln\frac{1}{\left(1-\epsilon\right)\delta}\right\rfloor,
  \end{equation}
  which is independent of $D$ and $N$.
  Note that this verification is implementable by nonadaptive measurements of the qubits for $\Ket{G_\mathrm{res}^{\left(N,D\right)}}$ only in the $Z$ and $X$ bases, \textit{i.e.}, the same measurement bases as the MBQC protocol using $\Ket{G_\mathrm{res}^{\left(N,D\right)}}$ in Theorem~\ref{thm:universality}.
  The complexity of the MBQC protocol in Theorem~\ref{thm:universality} using $\Ket{G_\mathrm{res}^{\left(N,D\right)}}$, combined with this verification, is in total
  \begin{equation}
    O\left(\left(DN\log^2 N\right)\times \frac{1}{\epsilon}\log\left(\frac{1}{\delta}\right)\right),
  \end{equation}
  and hence including the verification, the polylog overhead cost in this MBQC protocol using $\Ket{G_\mathrm{res}^{\left(N,D\right)}}$ in terms of $D$ and $N$ is still achievable.
  Note that combining this verification
  with the asymptotically optimal sorting networks (with a large constant factor) as discussed in Remark~\ref{rem:generality},
  we can also obtain an MBQC protocol with verification of complexity
  \begin{equation}
    O\left(\left(DN\log N\right)\times \frac{1}{\epsilon}\log\left(\frac{1}{\delta}\right)\right).
  \end{equation}
\end{remark}

\section{\label{sec:fault_tolerant}Fault-tolerant photonic MBQC protocol}

In this section, we construct a fault-tolerant MBQC protocol that achieves a polylog overhead using the multiqubit hypergraph state $\Ket{G_\mathrm{res}^{\left(N,D\right)}}$ given in Sec.~\ref{sec:resource_state} as a resource.
Aiming at implementing fault-tolerant MBQC using photonic architectures, we represent each qubit for $\Ket{G_\mathrm{res}^{\left(N,D\right)}}$ as a logical qubit of the concatenated $7$-qubit code at concatenation level $L$~\cite{PhysRevLett.77.793,S3,N4} using photonic GKP qubits as physical qubits, which we call the \textit{GKP $7$-qubit code}.
For an integer $l\in\left\{0,\ldots,L\right\}$, a level-$l$ logical qubit refers to a logical qubit of the GKP $7$-qubit code at concatenation level $l$, which consists of $7^l$ physical GKP qubits.
In Sec.~\ref{sec:fault_tolerant_mbqc_protocol}, we describe how to perform fault-tolerant MBQC if a sufficiently high-fidelity resource state $\Ket{G_\mathrm{res}^{\left(N,D\right)}}$ encoded by the GKP $7$-qubit code is given.
To prepare this resource state in a fault-tolerant way, in Sec.~\ref{sec:optimzed_sqec}, we introduce an optimized algorithm for reducing variances of GKP qubits, which improves the performance of SQEC summarized in Sec.~\ref{sec:qec_gkp}.
Then in Sec.~\ref{sec:resource_generation}, we provide a protocol for the fault-tolerant resource state preparation shown in Protocol~\ref{alg:preparation}.
This protocol for fault-tolerant resource state preparation is applicable in general; that is, we can use our fault-tolerant state preparation protocol for preparing not only $\Ket{G_\mathrm{res}^{\left(N,D\right)}}$ for our MBQC protocol but also for any multiqubit entangled state generated by Clifford and $T$ gates, while we will describe the protocol using $\Ket{G_\mathrm{res}^{\left(N,D\right)}}$ for clarity.
As for the overhead, we show in Sec.~\ref{sec:complexity_qec} that the overall overhead of the fault-tolerant MBQC protocol for implementing an $N$-qubit $D$-depth quantum circuit within a failure probability $\epsilon$ scales poly-logarithmically, and hence if we use $\Ket{G_\mathrm{res}^{\left(N,D\right)}}$, the polylog overhead of the MBQC protocol in Theorem~\ref{thm:universality} still holds including QEC\@.
As shown in Sec.~\ref{sec:threshold} by numerical calculation,
the threshold for our fault-tolerant MBQC protocol is $7.8$~dB in terms of the squeezing level of the GKP code.

\subsection{\label{sec:fault_tolerant_mbqc_protocol}Fault-tolerant MBQC protocol using encoded resource state}

\begin{figure*}[t]
  \centering
  \includegraphics[width=7.0in]{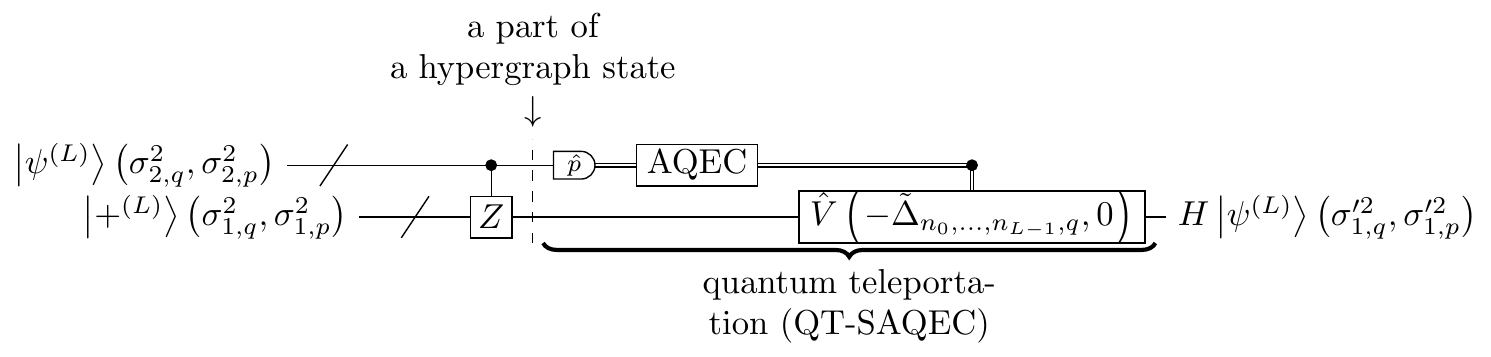}
  \caption{\label{fig:qt_sqec}Quantum-teleportation single-logical-qubit analog quantum error correction (QT-SAQEC) analogous to Knill's quantum-teleportation-based quantum error correction for a given GKP $7$-qubit code state $\Ket{\psi^{\left(L\right)}}$ at a concatenation level $L$. For all the physical GKP qubits labeled $\left(n_0,\ldots,n_{L-1}\right)$ that remains after the quantum teleportation, the variances $\left(\sigma_{n_0,\ldots,n_{L-1},1,q}^2,\sigma_{n_0,\ldots,n_{L-1},1,p}^2\right)$ in the $\hat{q}$ and $\hat{p}$ quadratures are collectively written as $\left(\sigma_{1,q}^2,\sigma_{1,p}^2\right)$. These GKP qubits are prepared in a fixed level-$L$ logical state $\Ket{+^{\left(L\right)}}$ and connected with the input GKP qubits prepared in $\Ket{\psi^{\left(L\right)}}$ by a transversally implemented $CZ$ gate, which can be a part of the resource hypergraph state for our MBQC protocol. For all these input GKP qubits labeled $\left(n_0,\ldots,n_{L-1}\right)$ and measured in the quantum teleportation, the variances $\left(\sigma_{n_0,\ldots,n_{L-1},2,q}^2,\sigma_{n_0,\ldots,n_{L-1},2,p}^2\right)$ are collectively written as $\left(\sigma_{2,q}^2,\sigma_{2,p}^2\right)$. Homodyne detection of the GKP qubits for $\Ket{\psi^{\left(L\right)}}$ in the $\hat{p}$ quadrature achieves the quantum teleportation in the same way as~\eqref{eq:x_measurement}, and at the same time, AQEC illustrated in Fig.~\ref{fig:analog_qec} yields estimates of the deviations of the measurement outcomes for all the measured GKP qubits for $\Ket{\psi^{\left(L\right)}}$. Thus, the QT-SAQEC corrects this deviation using the estimate $\tilde{\Delta}_{n_0,\ldots,n_{L-1},q}$ of the deviation given by~\eqref{eq:qt_saqec} for each GKP qubit. The variances after this correction $\left(\sigma_{n_0,\ldots,n_{L-1},1,q}^{\prime 2},\sigma_{n_0,\ldots,n_{L-1},1,p}^{\prime 2}\right)$ is given in the same way as $\hat{q}$-SQEC~\eqref{eq:sqec_q} illustrated in the upper part of Fig.~\ref{fig:sqec}, collectively written as $\left(\sigma_{1,q}^{\prime 2},\sigma_{1,p}^{\prime 2}\right)$.}
\end{figure*}

We describe a protocol for performing the MBQC protocol in Theorem~\ref{thm:universality} in a fault-tolerant way, using a given hypergraph state $\Ket{G_\mathrm{res}^{\left(N,D\right)}}$ shown in Sec.~\ref{sec:resource_state} that we here assume is represented as a logical state of the GKP $7$-qubit code in a sufficiently high fidelity.
To achieve the high fidelity of the logical state required for this protocol,
we will also show in Sec.~\ref{sec:resource_generation} how to prepare arbitrarily high-fidelity $\Ket{G_\mathrm{res}^{\left(N,D\right)}}$ encoded by the GKP $7$-qubit code from given GKP qubits at a better squeezing level than a threshold.
Throughout the resource state preparation shown in Sec.~\ref{sec:resource_generation},
we keep track of the variances of each physical GKP qubit,
and hence, we here assume that our protocol can use the values of the variances of each GKP qubit in the $\hat{q}$ and $\hat{p}$ quadratures for QEC\@.

Suppose that a sufficiently high-fidelity encoded state of $\Ket{G_\mathrm{res}^{\left(N,D\right)}}$ using the GKP $7$-qubit code is provided, where the variances of each GKP qubit are also given.
Then, we can perform the MBQC protocol in Theorem~\ref{thm:universality} by replacing $Z$ and $X$ measurements in the protocol with logical $Z$ and $X$ measurements implemented by homodyne detection in $\hat{q}$ and $\hat{p}$ quadratures, respectively, using transversal implementations of logical Clifford operations for the $7$-qubit code summarized in Sec.~\ref{sec:qec_gkp}.
Finite squeezing of GKP qubits may cause errors that are analogous to bit- and phase-flip errors at the logical level of the GKP code.
We correct these bit- and phase-flip errors using the analog quantum error correction (AQEC) summarized in Sec.~\ref{sec:qec_gkp}.
In the same way as~\eqref{eq:label},
we label each of the $7^L$ GKP qubit as
\begin{equation}
  \label{eq:label_L}
  (n_0,\ldots,n_{L-1})\in{\{1,\ldots,7\}}^L
\end{equation}
where $n_0$ is a label of a physical GKP qubit, and $n_l$ for each $l\in\left\{1,\ldots,L-1\right\}$ is a label of a level-$l$ logical qubit composed of $7$ level-$(l-1)$ logical qubits.
When performing the logical $Z$ or $X$ measurement of the GKP $7$-qubit code, we estimate a bit-valued logical measurement outcome $\tilde{b}^{(\mathrm{AQEC})}$ from real-valued outcomes of homodyne detection in such a way as~\eqref{eq:b_aqec} of the AQEC, and calculate the deviations of the outcome of homodyne detection
\begin{equation}
  \label{eq:Delta_L_aqec}
  \tilde{\Delta}_{n_0,\ldots,n_{L-1}}^{(\mathrm{AQEC})}
\end{equation}
using~\eqref{eq:deviation_estimate_aqec}.
We use $\tilde{b}^{(\mathrm{AQEC})}$ to perform the MBQC protocol in Theorem~\ref{thm:universality} at the logical level, while we additionally perform correction using the analog information $\tilde{\Delta}_{n_0,\ldots,n_{L-1}}^{(\mathrm{AQEC})}$ at the physical level.

The correction can be performed in analogy to Knill's quantum-teleportation-based quantum error correction~\cite{K5,K6}, but we exploit the analog information $\tilde{\Delta}_{n_0,\ldots,n_{L-1}}^{(\mathrm{AQEC})}$ obtained from the AQEC to improve the performance of QEC\@.
In the fault-tolerant MBQC protocol, after each logical $X$ measurement~\eqref{eq:x_measurement} for quantum teleportation,
we correct the deviation in a similar way as the single-qubit quantum error correction (SQEC) using $\tilde{\Delta}_q$ defined as~\eqref{eq:Delta_q}, as shown in Fig.~\ref{fig:qt_sqec}.
In particular, for each GKP qubit labeled as~\eqref{eq:label_L},
we obtain from~\eqref{eq:Delta_q} the deviation for the correction
\begin{equation}
  \label{eq:qt_saqec}
  \tilde{\Delta}_{n_0,\ldots,n_{L-1},q}\coloneqq\frac{\sigma_{n_0,\ldots,n_{L-1},1,q}^2}{\sigma_{n_0,\ldots,n_{L-1},1,q}^2+\sigma_{n_0,\ldots,n_{L-1},2,p}^2}\tilde{\Delta}_{n_0,\ldots,n_{L-1}}^{(\mathrm{AQEC})},
\end{equation}
where $\left(\sigma_{n_0,\ldots,n_{L-1},1,q}^2,\sigma_{n_0,\ldots,n_{L-1},1,p}^2\right)$ denotes the variances of the GKP qubit labeled $\left(n_0,\ldots,n_{L-1}\right)$ that remains after the quantum teleportation, and $\left(\sigma_{n_0,\ldots,n_{L-1},2,q}^2,\sigma_{n_0,\ldots,n_{L-1},2,p}^2\right)$ denotes the variances of the GKP qubit labeled $\left(n_0,\ldots,n_{L-1}\right)$ and measured in the quantum teleportation.
The variances of the GKP qubit labeled $\left(n_0,\ldots,n_{L-1}\right)$ after the quantum teleportation and the correction are given in the same way as~\eqref{eq:sqec_q}, denoted by $\left(\sigma_{n_0,\ldots,n_{L-1},1,q}^{\prime 2},\sigma_{n_0,\ldots,n_{L-1},1,p}^{\prime 2}\right)$.
We refer to this QEC method as \textit{quantum-teleportation single-logical-qubit analog quantum error correction} (QT-SAQEC).
Note that we do not perform this correction after logical $Z$ measurements of the GKP $7$-qubit code.

After performing each QT-SAQEC, we obtain the deviation $\tilde{\Delta}_{n_0,\ldots,n_{L-1},q}$ given by~\eqref{eq:qt_saqec}
in the $\hat{q}$ quadrature of each of the $7^L$ GKP qubits comprising a logical qubit of the GKP $7$-qubit code where the input state is teleported.
Throughout the protocol, we do not have to explicitly apply displacement transformation for correcting the deviation to physical GKP qubits for the resource state; instead, the deviation can be corrected by appropriately relabeling real-valued measurement outcomes of homodyne detection by subtracting the deviations from the outcomes.
In a $Z$ measurement of a GKP qubit labeled $(n_0,\ldots,n_{L-1})$, we correct the real-valued outcome $\tilde{q}\in\mathbb{R}$ of homodyne detection in the $\hat{q}$ quadrature by subtracting $\tilde{\Delta}_{n_0,\ldots,n_{L-1},q}$ from $\tilde{q}$.
Furthermore, in a $X$ measurement of a GKP qubit connected to the GKP qubit $(n_0,\ldots,n_{L-1})$ by a $CZ$ gate transversally implemented using~\eqref{eq:cz_quadrature}, we correct the real-valued outcome $\tilde{p}\in\mathbb{R}$ of homodyne detection in the $\hat{p}$ quadrature according to~\eqref{eq:cz_quadrature}, that is, by subtracting $\tilde{\Delta}_{n_0,\ldots,n_{L-1},q}$ from $\tilde{p}$.
In contrast to the $CZ$ gate,
a logical $CCZ$ gate of the GKP $7$-qubit code is not implemented transversally.
Although
an outcome $\tilde{p}$ of homodyne detection in the $\hat{p}$ quadrature of a GKP qubit connected to the GKP qubit $(n_0,\ldots,n_{L-1})$ by implementing the logical $CCZ$ gate may also be changed by $\tilde{\Delta}_{n_0,\ldots,n_{L-1},q}$ similarly to the $CZ$ gate,
we ignore the changes in $\tilde{p}$ potentially caused by the logical $CCZ$ gate;
that is, when we perform the $X$ measurement of a logical qubit of the GKP $7$-qubit code connected to two logical qubits by the logical $CCZ$ gate, we do not add or subtract the deviations in $\hat{q}$ of these two neighboring logical qubits in correcting the measurement outcome, but only subtract the deviations that can be calculated from the transversal $CZ$ gates.
In this case, even if we ignore the potential deviation at the physical level caused by the logical $CCZ$ gate of the GKP $7$-qubit code, we can correct the logical bit- and phase-flip errors caused by this ignorance at the logical level of the $7$-qubit code as long as the failure probability in QEC is sufficiently suppressed by preparing the high-fidelity resource state.

In this way, we can combine the MBQC protocol in Theorem~\ref{thm:universality} with the QT-SAQEC to achieve fault-tolerant MBQC by homodyne detection, as long as we can prepare a sufficiently high-fidelity hypergraph state $\Ket{G_\mathrm{res}^{\left(N,D\right)}}$ encoded by the GKP $7$-qubit code.

\subsection{\label{sec:optimzed_sqec}Optimized GKP single-qubit quantum error correction (SQEC)}

Before providing a protocol for preparing a high-fidelity encoded state $\Ket{G_\mathrm{res}^{\left(N,D\right)}}$ in the next subsection, we introduce an optimized algorithm of SQEC, where SQEC is summarized in Sec.~\ref{sec:qec_gkp}.
In the same way as Sec.~\ref{sec:qec_gkp}, consider a data GKP qubit $1$ and an auxiliary GKP qubit $2$ in the following, where the variances in the $\hat{q}$ and $\hat{p}$ quadratures are $\left(\sigma_{1,q}^2,\sigma_{1,p}^2\right)$ for $1$ and $\left(\sigma_{2,q}^2,\sigma_{2,p}^2\right)$ for $2$.
SQEC aims at reducing the variances $\sigma_{1,q}^2$ and $\sigma_{1,p}^2$ of the data GKP qubit that accumulate as we apply quantum gates,
using an auxiliary GKP qubit $2$ with variances $\sigma_{2,q}^2$ and $\sigma_{2,p}^2$ that are expected to be smaller than $\sigma_{1,q}^2$ and $\sigma_{1,p}^2$.
Observe that the variances $\sigma_{1,q}^2$ and $\sigma_{1,p}^2$ of the data GKP qubit eventually cause bit- and phase-flip errors of the GKP qubit, respectively,
and when we use the GKP $7$-qubit code for correcting the errors,
the eventual failure probability of QEC is determined by the larger variance of $\sigma_{1,q}^2$ and $\sigma_{1,p}^2$ since the larger variance causes more errors than the other.
Based on this observation, we here aim to minimize
\begin{equation}
  \max\left\{\sigma_{1,q}^2,\sigma_{1,p}^2\right\}
\end{equation}
by adjusting the variances $\left(\sigma_{2,q}^2,\sigma_{2,p}^2\right)$ of the auxiliary GKP qubit in SQEC\@.

Consider a case where a given data GKP qubit satisfies
\begin{equation}
  \sigma_{1,q}^2\geqq\sigma_{1,p}^2,
\end{equation}
and we optimize $\hat{q}$-SQEC~\eqref{eq:sqec_q} illustrated in the upper part of Fig.~\ref{fig:sqec}.
The case of $\sigma_{1,q}^2<\sigma_{1,p}^2$ for $\hat{p}$-SQEC~\eqref{eq:sqec_p} will be discussed later.
Suppose that we can prepare auxiliary GKP qubits in $\Ket{0}$, where the variances are $\left(\sigma_{2,q}^2,\sigma_{2,p}^2\right)$.
SQEC using this auxiliary GKP qubit would decrease $\sigma_{1,q}^2$ to $\sigma_{1,q}^{\prime 2}$ and increase $\sigma_{1,p}^2$ to $\sigma_{1,p}^{\prime 2}$, where $\sigma_{1,q}^{\prime 2}$ and $\sigma_{1,p}^{\prime 2}$ are given by~\eqref{eq:sqec_q}.
In contrast, we here optimize the variances of the auxiliary GKP qubit corresponding to $\left(\sigma_{2,q}^2,\sigma_{2,p}^2\right)$, so that neither of the variances of the data GKP qubit after SQEC corresponding to $\sigma_{1,q}^{\prime 2}$ and $\sigma_{1,p}^{\prime 2}$ can be dominant; that is, the variances corresponding to $\sigma_{1,q}^{\prime 2}$ and $\sigma_{1,p}^{\prime 2}$ can be as close as possible.
To achieve this optimization by only using auxiliary GKP qubits prepared in $\Ket{0}$ with variances $\left(\sigma_{2,q}^2,\sigma_{2,p}^2\right)$, we prepare $\Ket{0}$ with variances $\left(\frac{1}{m}\sigma_{2,q}^2,m\sigma_{2,p}^2\right)$ or $\left(m\sigma_{2,q}^2,\frac{1}{m}\sigma_{2,p}^2\right)$ for any $m\in\left\{1,2,\ldots\right\}$, which we call \textit{variance adjustment}.

\begin{figure}[t]
  \centering
  \includegraphics[width=3.4in]{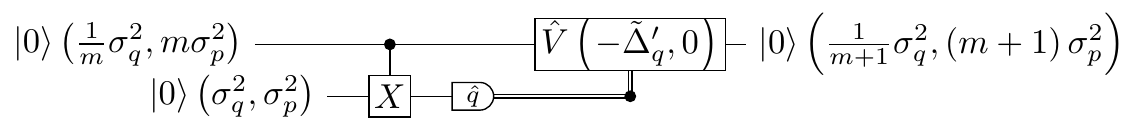}
  \includegraphics[width=3.4in]{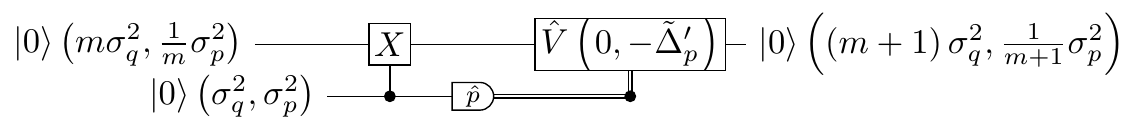}
  \caption{\label{fig:adjustment}Quantum circuits for variance adjustment to optimize the performance of SQEC in Fig~\ref{fig:sqec}. Given GKP qubits prepared in $\Ket{0}$ with variances $\left(\sigma_q^2,\sigma_p^2\right)$, the variance adjustment uses only these GKP qubits to prepare $\Ket{0}$ with variances $\left(\frac{1}{m}\sigma_q^2,m\sigma_p^2\right)$ by the upper circuit and $\left(m\sigma_q^2,\frac{1}{m}\sigma_p^2\right)$ by the lower circuit for any $m\in\left\{1,2,\ldots\right\}$. In the upper circuit, for each $m$, we input a GKP qubit prepared in $\Ket{0}$ with variances $\left(\frac{1}{m}\sigma_q^2,m\sigma_p^2\right)$, and the variance adjustment estimates the deviation in $\hat{q}$ of this input as $\tilde{\Delta}_q^\prime$ given by~\eqref{eq:Delta_q_prime}, followed by the displacement $\hat{V}$ defined as~\eqref{eq:translation}. As a result, the variances of the input changes into $\left(\frac{1}{m+1}\sigma_q^2,\left(m+1\right)\sigma_p^2\right)$ as shown in~\eqref{eq:variance_adjustment}. In the lower circuit, for each $m$, we input a GKP qubit prepared in $\Ket{0}$ with variances $\left(m\sigma_q^2,\frac{1}{m}\sigma_p^2\right)$, and estimate the deviation in $\hat{p}$ of this input as $\tilde{\Delta}_p^\prime$ given by~\eqref{eq:Delta_p_prime}, followed by $\hat{V}$. As a result, the variances of the input changes into $\left(\left(m+1\right)\sigma_q^2,\frac{1}{m+1}\sigma_p^2\right)$ as shown in~\eqref{eq:variance_adjustment_p}.}
\end{figure}

\begin{figure*}[t]
  \centering
  \includegraphics[width=7.0in]{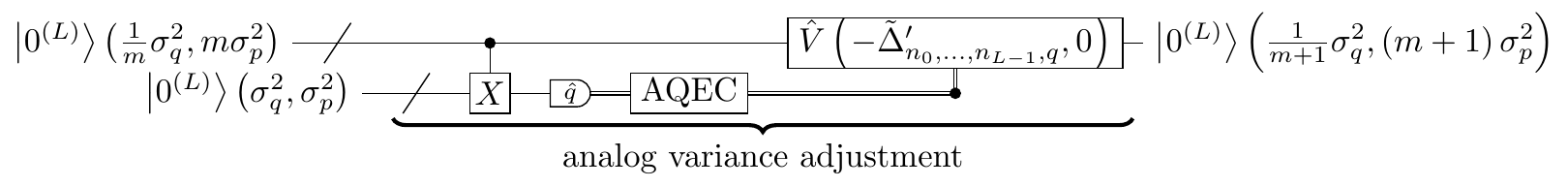}
  \includegraphics[width=7.0in]{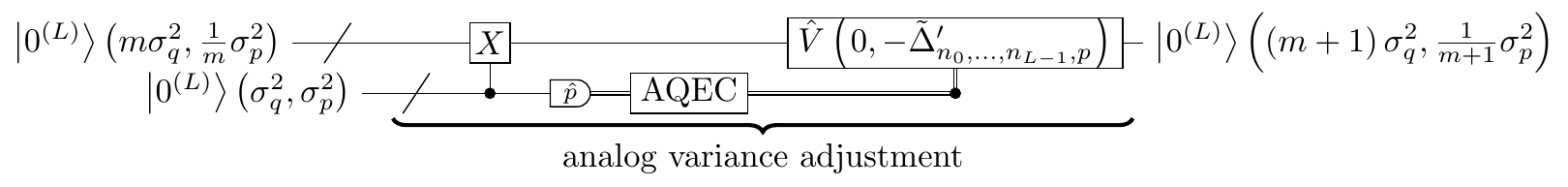}
  \includegraphics[width=7.0in]{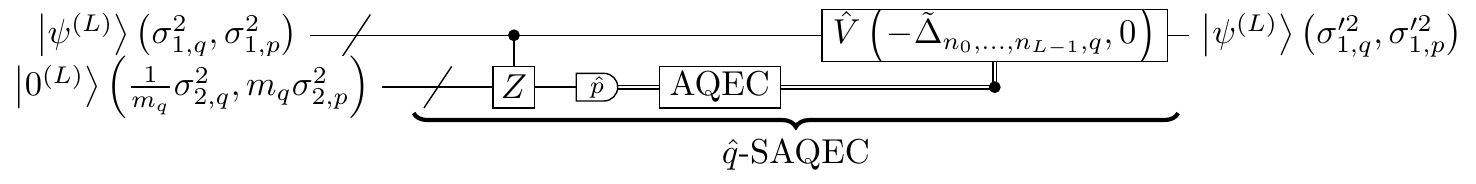}
  \includegraphics[width=7.0in]{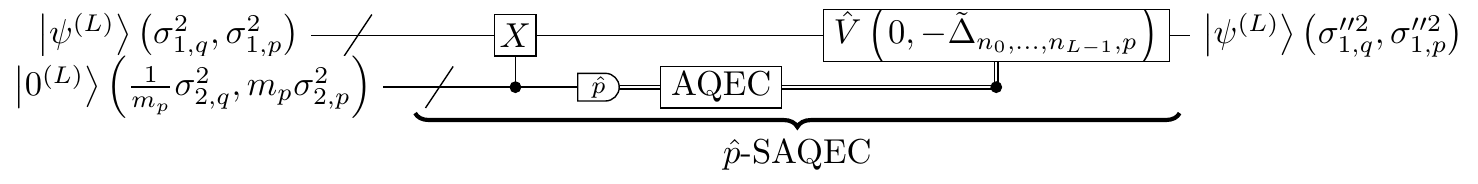}
  \caption{\label{fig:optimized_sqec}Analog variance adjustment in the first and second circuits corresponding to the upper and lower circuits in Fig.~\ref{fig:adjustment}, respectively, single-logical-qubit analog quantum error correction (SAQEC) analogous to Steane's quantum error correction for bit-flip errors ($\hat{q}$-SAQEC) for a given state $\Ket{\psi^{\left(L\right)}}$ of the GKP $7$-qubit code at concatenation level $L$ in the third circuit, and SAQEC analogous to that for phase-flip errors ($\hat{p}$-SAQEC) in the fourth circuit. The notations are the same as those for QT-SAQEC in Fig.~\ref{fig:qt_sqec}. Analog variance adjustment uses AQEC in Fig.~\ref{fig:analog_qec} to apply variance adjustment in Fig.~\ref{fig:adjustment} to quantum error-correcting codes, where the AQEC estimates the deviation as $\tilde{\Delta}_{n_0,\ldots,n_{L-1},q}^\prime$ given by~\eqref{eq:analog_variance_adjustment} and $\tilde{\Delta}_{n_0,\ldots,n_{L-1},p}^\prime$ given by~\eqref{eq:analog_variance_adjustment_p}. In $\hat{q}$- and $\hat{p}$-SAQEC, we use the analog variance adjustment to enhance the performance of SAQEC compared to conventional SQEC in Fig.~\ref{fig:sqec} without variance adjustment; that is, we adjust the variances $\left(\sigma_{2,q}^2,\sigma_{2,p}^2\right)$ of the auxiliary qubit labeled $2$ depending on the variances $\left(\sigma_{1,q}^2,\sigma_{1,p}^2\right)$ of the data qubit labeled $1$. The parameters $m_q$ and $m_p$ for the analog variance adjustment are given by~\eqref{eq:m_q} and~\eqref{eq:m_p}, respectively, where we use the sums of the variances of all the physical qubits instead of the variances in the single-qubit case. The variances after $\hat{q}$-SAQEC and $\hat{p}$-SAQEC, collectively written as $\left(\sigma_{1,q}^{\prime 2},\sigma_{1,p}^{\prime 2}\right)$ and $\left(\sigma_{1,q}^{\prime\prime 2},\sigma_{1,p}^{\prime\prime 2}\right)$, are given by~\eqref{eq:sqec_q} and~\eqref{eq:sqec_p}, respectively, in the same way as the single-qubit cases. In these circuits, the logical $CZ$ and $\textsc{CNOT}$ gates are transversally implemented by Gaussian operations on the $7^L$ physical GKP qubits as discussed in Sec.~\ref{sec:gkp}, and the AQEC decides estimates of the deviations, followed by correcting the deviations using $\tilde{\Delta}_{n_0,\ldots,n_{L-1},q}$ given by~\eqref{eq:qt_saqec} and $\tilde{\Delta}_{n_0,\ldots,n_{L-1},p}$ given by~\eqref{eq:p_saqec}.}
\end{figure*}

For the variance adjustment, we use the circuits shown in Fig.~\ref{fig:adjustment}.
In the upper and lower circuits of Fig.~\ref{fig:adjustment},
for each $m\in\left\{1,2,\ldots\right\}$, we input a GKP qubit prepared in $\Ket{0}$ with variances $\left(\frac{1}{m}\sigma_{q}^2,m\sigma_{p}^2\right)$ and $\left(m\sigma_{q}^2,\frac{1}{m}\sigma_{p}^2\right)$, respectively, and we also prepare $\Ket{0}$ of an auxiliary GKP qubit with variances $\left(\sigma_{q}^2,\sigma_{p}^2\right)$.
Then, using the $\textsc{CNOT}$ gate implemented by~\eqref{eq:cx_quadrature}, we measure $\hat{q}_1+\hat{q}_2$ and $\hat{p}_1+\hat{p}_2$, respectively, where $\hat{q}_1$ and $\hat{q}_2$ denote the $\hat{q}$ quadratures of the control mode and the target mode of the $\textsc{CNOT}$ gate in the upper circuit of Fig.~\ref{fig:adjustment}, and $\hat{p}_1$ and $\hat{p}_2$ denote the $\hat{p}$ quadratures of the control mode and the target mode of the $\textsc{CNOT}$ gate in the lower circuit of Fig.~\ref{fig:adjustment}.
In the same way as SQECs~\eqref{eq:Delta_q} and~\eqref{eq:Delta_p},
to estimate the deviation for the correction in the upper and lower circuits of Fig.~\ref{fig:adjustment},
we use, respectively,
\begin{align}
  \label{eq:Delta_q_prime}
  \tilde{\Delta}_q^\prime&\coloneqq\frac{\frac{1}{m}\sigma_{q}^2}{\frac{1}{m}\sigma_{q}^2+\sigma_{q}^2}\tilde{\Delta},\\
  \label{eq:Delta_p_prime}
  \tilde{\Delta}_p^\prime&\coloneqq\frac{\frac{1}{m}\sigma_{p}^2}{\frac{1}{m}\sigma_{p}^2+\sigma_{p}^2}\tilde{\Delta},
\end{align}
where $\tilde{\Delta}$ is given by~\eqref{eq:deviation_estimate}.
For each variance adjustment, the variances of the input GKP qubit change in the same way as SQECs~\eqref{eq:sqec_q} and~\eqref{eq:sqec_p}, that is, in the upper and lower circuits of Fig.~\ref{fig:adjustment}, respectively,
\begin{align}
  \label{eq:variance_adjustment}
  \left(\frac{1}{m}\sigma_{q}^2,m\sigma_{p}^2\right)&\to\left(\frac{1}{m+1}\sigma_{q}^2,\left(m+1\right)\sigma_{p}^2\right),\\
  \label{eq:variance_adjustment_p}
  \left(m\sigma_{q}^2,\frac{1}{m}\sigma_{p}^2\right)&\to\left(\left(m+1\right)\sigma_{q}^2,\frac{1}{m+1}\sigma_{p}^2\right).
\end{align}

To exploit these variance adjustments for improving the performance of the $\hat{q}$-SQEC,
we optimize $m$ so that we can obtain
\begin{align}
  &m=m_q^\ast\coloneqq\nonumber\\
  &\quad\sup\left\{m:\frac{\sigma_{1,q}^2\cdot m\sigma_{2,p}^2}{\sigma_{1,q}^2+m\sigma_{2,p}^2}\leqq\sigma_{1,p}^2+\frac{1}{m}\sigma_{2,q}^2,\right.\nonumber\\
  &\qquad \left. m\in\mathbb{N}\;\text{or }\frac{1}{m}\in\mathbb{N}\right\},
\end{align}
where $\mathbb{N}=\left\{1,2,\ldots\right\}$, and the right-hand side yields the case where $\sigma_{1,q}^{\prime 2}$ and $\sigma_{1,p}^{\prime 2}$ in~\eqref{eq:sqec_q} are close; that is, the larger of $\sigma_{1,q}^{\prime 2}$ and $\sigma_{1,p}^{\prime 2}$ is minimized.
Note that if $\sigma_{1,q}^2>\sigma_{1,p}^2$, then $m_q^\ast$ is finite, since we have for sufficiently small $m>0$
\begin{equation}
  \frac{\sigma_{1,q}^2\cdot m\sigma_{2,p}^2}{\sigma_{1,q}^2+m\sigma_{2,p}^2}<\sigma_{1,p}^2+\frac{1}{m}\sigma_{2,q}^2,
\end{equation}
and for sufficiently large $m$
\begin{equation}
  \frac{\sigma_{1,q}^2\cdot m\sigma_{2,p}^2}{\sigma_{1,q}^2+m\sigma_{2,p}^2}>\sigma_{1,p}^2+\frac{1}{m}\sigma_{2,q}^2,
\end{equation}
where the left-hand sides and the right-hand sides of these inequalities monotonically increase and decrease, respectively, as $m$ increases.
As $\sigma_{1,q}^2$ and $\sigma_{1,p}^2$ get closer,
$m_q^\ast$ gets larger,
and the larger variance $m_q^\ast\sigma_{2,p}^2$ of $\left(\frac{1}{m_q^\ast}\sigma_{2,q}^2,m_q^\ast\sigma_{2,p}^2\right)$ may cause more logical phase-flip errors.
To avoid the increase of these logical phase-flip errors, let $M$ be an upper bound of the larger variance of the GKP qubit,
and define
\begin{equation}
  \label{eq:m_q}
  m_q\coloneqq\begin{cases}
    m_q^\ast&\text{if }m_q^\ast\leqq M,\\
    M&\text{if }m_q^\ast > M,
  \end{cases}
\end{equation}
where we set in our protocol
\begin{equation}
  M=4.
\end{equation}
Consequently, the optimized $\hat{q}$-SQEC is performed by the same circuit as the upper part of Fig.~\ref{fig:sqec} yet using the auxiliary GKP qubit $2$ prepared in $\Ket{0}$ with variances $\left(\frac{1}{m_q}\sigma_{2,q}^2,m_q\sigma_{2,p}^2\right)$ adjusted by repeating the variance adjustment $m_q$ times if $m_q\geqq 1$ (or $\frac{1}{m_q}$ times if $m_q<1$).

In the same way, consider a case where a given data GKP qubit satisfies
\begin{equation}
  \sigma_{1,q}^2<\sigma_{1,p}^2,
\end{equation}
and we optimize $\hat{p}$-SQEC~\eqref{eq:sqec_p} illustrated in the lower part of Fig.~\ref{fig:sqec}.
Define
\begin{align}
  \label{eq:m_p}
  m_p^\ast&\coloneqq\sup\left\{m:\sigma_{1,q}^2+\frac{1}{m}\sigma_{2,q}^2\geqq \frac{\sigma_{1,p}^2\cdot m\sigma_{2,p}^2}{\sigma_{1,p}^2+m\sigma_{2,p}^2},\right.\nonumber\\
  &\qquad \left. m\in\mathbb{N}\;\text{or }\frac{1}{m}\in\mathbb{N}\right\},\\
  m_p&\coloneqq\begin{cases}
    m_p^\ast&\text{if }m_p^\ast\leqq M,\\
    M&\text{if }m_p^\ast > M,
  \end{cases}
\end{align}
where $m_p^\ast$ yields the case where $\sigma_{1,q}^{\prime\prime 2}$ and $\sigma_{1,p}^{\prime\prime 2}$ in~\eqref{eq:sqec_p} are close, and $M$ is the same as that in~\eqref{eq:m_q}.
Then, the optimized SQEC for $\hat{p}$ is performed by the same circuit as the lower part of Fig.~\ref{fig:sqec} yet using the auxiliary GKP qubit $2$ prepared in $\Ket{0}$ with the adjusted variances $\left(\frac{1}{m_p}\sigma_{2,q}^2,m_p\sigma_{2,p}^2\right)$.

In addition to these optimized SQECs for a single GKP qubit, we generalize SQEC to that for a logical one-qubit state of a quantum error-correcting code concatenated with the GKP code, which is analogous to Steane's quantum error correction~\cite{PhysRevLett.78.2252} but we exploit the AQEC to improve the performance.
We refer to this method for $\hat{q}$ and $\hat{p}$ as $\hat{q}$-SAQEC and $\hat{p}$-SAQEC, respectively, where SAQEC refers to \textit{single-logical-qubit analog quantum error correction}.
We may collectively refer to $\hat{q}$-SAQEC and $\hat{p}$-SAQEC as SAQECs\@.

In the SAQECs, we use the variance adjustment combined with the AQEC, which we call \textit{analog variance adjustment}.
The analog variance adjustment only uses auxiliary logical qubits of the level-$L$ GKP $7$-qubit code prepared in $\Ket{0^{(L)}}$ with variances $\left(\sigma_{n_0,\ldots,n_{L-1},q}^2,\sigma_{n_0,\ldots,n_{L-1},p}^2\right)$ to prepare $\Ket{0^{(L)}}$ with variances
$\left(\frac{1}{m}\sigma_{n_0,\ldots,n_{L-1},q}^2,m\sigma_{n_0,\ldots,n_{L-1},p}^2\right)$ or $\left(m\sigma_{n_0,\ldots,n_{L-1},q}^2,\frac{1}{m}\sigma_{n_0,\ldots,n_{L-1},p}^2\right)$ for any $m\in\left\{1,2,\ldots\right\}$, where $\left(\sigma_{n_0,\ldots,n_{L-1},q}^2,\sigma_{n_0,\ldots,n_{L-1},p}^2\right)$ denotes the variances of physical GKP qubits labeled $\left(n_0,\ldots,n_{L-1}\right)$ in the same way as~\eqref{eq:label_L}.
Corresponding to the upper and the lower circuits of Fig.~\ref{fig:adjustment} for the variance adjustment in the single-qubit case,
we perform the first and second circuits in Fig.~\ref{fig:optimized_sqec}, respectively, where the logical bit value and the deviations are decided by AQEC summarized in Sec.~\ref{sec:qec_gkp}.
Similarly to the estimation~\eqref{eq:qt_saqec} of the deviation in the QT-SAQEC,
we obtain from~\eqref{eq:Delta_q_prime} and~\eqref{eq:Delta_p_prime} the deviations for the correction
\begin{align}
  \label{eq:analog_variance_adjustment}
  \tilde{\Delta}_{n_0,\ldots,n_{L-1},q}^\prime&\coloneqq\frac{\frac{1}{m}\sigma_{n_0,\ldots,n_{L-1},q}^2}{\frac{1}{m}\sigma_{n_0,\ldots,n_{L-1},q}^2+\sigma_{n_0,\ldots,n_{L-1},q}^2}\tilde{\Delta}_{n_0,\ldots,n_{L-1}}^{(\mathrm{AQEC})},\\
  \label{eq:analog_variance_adjustment_p}
  \tilde{\Delta}_{n_0,\ldots,n_{L-1},p}^\prime&\coloneqq\frac{\frac{1}{m}\sigma_{n_0,\ldots,n_{L-1},p}^2}{\frac{1}{m}\sigma_{n_0,\ldots,n_{L-1},p}^2+\sigma_{n_0,\ldots,n_{L-1},p}^2}\tilde{\Delta}_{n_0,\ldots,n_{L-1}}^{(\mathrm{AQEC})},
\end{align}
where $\tilde{\Delta}_{n_0,\ldots,n_{L-1}}^{(\mathrm{AQEC})}$ is defined as~\eqref{eq:Delta_L_aqec}.
By performing the first and second circuits in Fig.~\ref{fig:optimized_sqec},
the variances of the input GKP qubits change in the same way as the single-qubit case~\eqref{eq:variance_adjustment} and~\eqref{eq:variance_adjustment_p}, respectively.

To describe the $\hat{q}$-SAQEC, consider a data logical qubit of the GKP $7$-qubit code at a concatenation level $L$ labeled $1$, and an auxiliary logical qubit of this GKP $7$-qubit code labeled $2$.
Let $\Ket{\psi^{\left(L\right)}}$ be a state of the logical qubit $1$,
and $\sigma_{1,q}^2$ and $\sigma_{1,p}^2$ be the sum of the variances $\left(\sigma_{n_0,\ldots,n_{L-1},1,q}^2,\sigma_{n_0,\ldots,n_{L-1},1,p}^2\right)$ in $\hat{q}$ and $\hat{p}$, respectively, for all the $7^L$ physical GKP qubits $\left(n_0,\ldots,n_{L-1}\right)$ that comprise $\Ket{\psi^{\left(L\right)}}$.
In addition, assume that we can prepare $\Ket{0^{\left(L\right)}}$ of an auxiliary logical qubit labeled $2$,
and let $\sigma_{2,q}^2$ and $\sigma_{2,p}^2$ be the sum of the variances $\left(\sigma_{n_0,\ldots,n_{L-1},2,q}^2,\sigma_{n_0,\ldots,n_{L-1},2,p}^2\right)$ of $\hat{q}$ and $\hat{p}$, respectively, for all the $7^L$ GKP qubits $\left(n_0,\ldots,n_{L-1}\right)$ for $\Ket{0^{\left(L\right)}}$.
Using these sums of the variances, calculate $m_q$ by~\eqref{eq:m_q} in the same way as the single-qubit case.
We perform the $\hat{q}$-SAQEC for $\Ket{\psi^{\left(L\right)}}$ of the data logical qubit using the third circuit of Fig.~\ref{fig:optimized_sqec}, where the variances of the auxiliary GKP qubits in $\Ket{0^{\left(L\right)}}$ are adjusted to $\left(\frac{1}{m_q}\sigma_{n_0,\ldots,n_{L-1},2,q}^2,m_q\sigma_{n_0,\ldots,n_{L-1},2,p}^2\right)$ using the analog variance adjustment repeatedly, and the logical bit value and the deviations are decided by the AQEC, followed by the same corrections as the QT-SAQEC~\eqref{eq:qt_saqec}.
The variances of the GKP qubit labeled $\left(n_0,\ldots,n_{L-1}\right)$ after the $\hat{q}$-SAQEC are given in the same way as~\eqref{eq:sqec_q}, denoted by $\left(\sigma_{n_0,\ldots,n_{L-1},1,q}^{\prime 2},\sigma_{n_0,\ldots,n_{L-1},1,p}^{\prime 2}\right)$.

As for $\hat{p}$-SAQEC, we perform the $\hat{p}$-SAQEC for $\Ket{\psi^{\left(L\right)}}$ using the fourth circuit of of Fig.~\ref{fig:optimized_sqec}, where the variances of the auxiliary GKP qubits are adjusted to $\left(\frac{1}{m_p}\sigma_{2,q}^2,m_p\sigma_{2,p}^2\right)$, and the AQEC is followed by similar corrections to QT-SAQEC~\eqref{eq:qt_saqec} but obtained from~\eqref{eq:Delta_p} rather than~\eqref{eq:Delta_q}, that is,
\begin{equation}
  \label{eq:p_saqec}
  \tilde{\Delta}_{n_0,\ldots,n_{L-1},p}\coloneqq\frac{\sigma_{n_0,\ldots,n_{L-1},1,p}^2}{\sigma_{n_0,\ldots,n_{L-1},1,p}^2+\sigma_{n_0,\ldots,n_{L-1},2,p}^2}\tilde{\Delta}_{n_0,\ldots,n_{L-1}}^{(\mathrm{AQEC})},
\end{equation}
where $\tilde{\Delta}_{n_0,\ldots,n_{L-1}}^{(\mathrm{AQEC})}$ is defined as~\eqref{eq:Delta_L_aqec}.
The variances of the GKP qubit labeled $\left(n_0,\ldots,n_{L-1}\right)$ after the $\hat{p}$-SAQEC are given in the same way as~\eqref{eq:sqec_p}, denoted by $\left(\sigma_{n_0,\ldots,n_{L-1},1,q}^{\prime\prime 2},\sigma_{n_0,\ldots,n_{L-1},1,p}^{\prime\prime 2}\right)$.

\begin{figure}[t]
    \centering
    \includegraphics[width=3.4in]{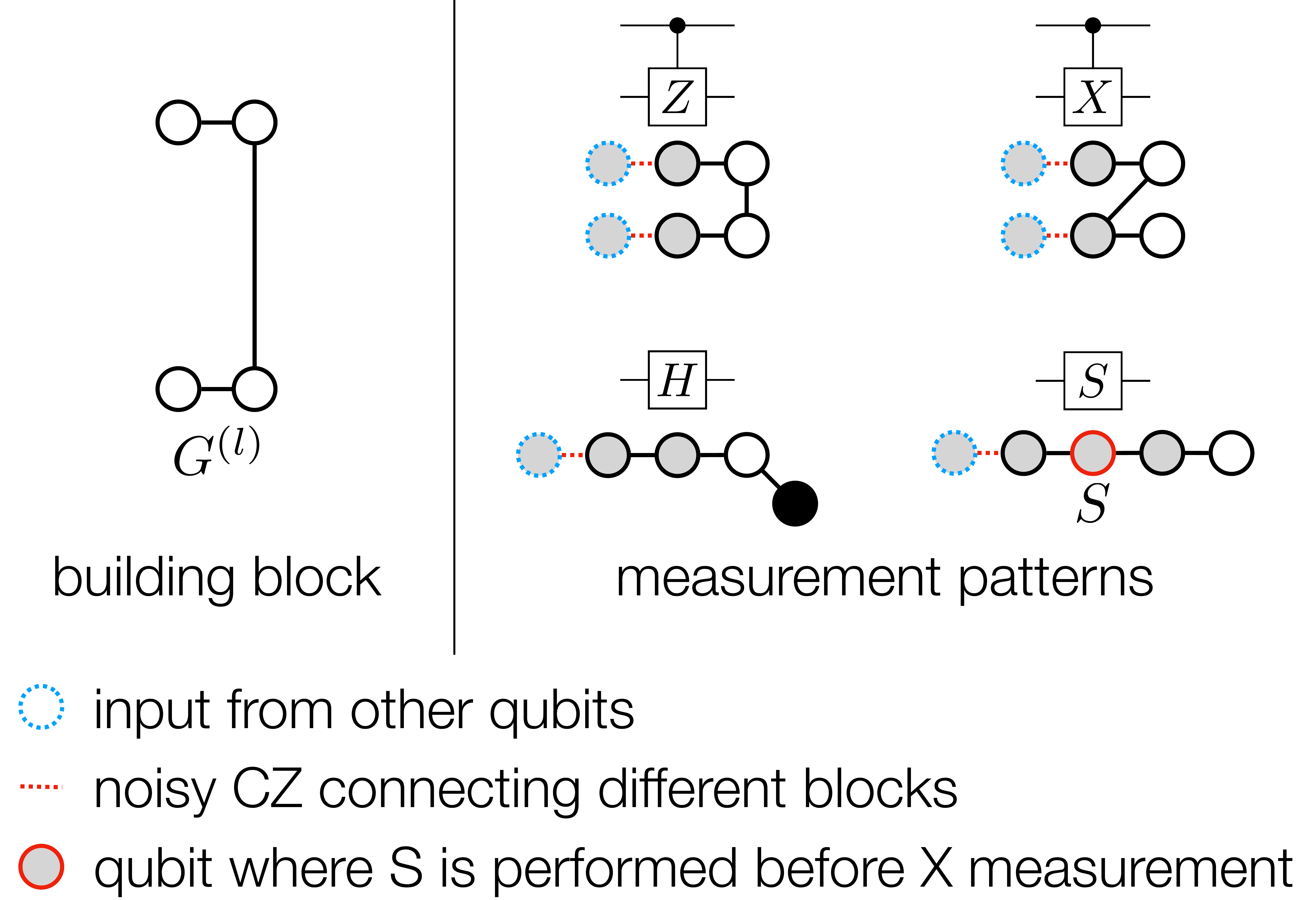}
    \caption{\label{fig:blocks}A graph on the left representing a graph state $\Ket{G^{\left(l\right)}}$ that serves as a building block for fault-tolerant preparation of the resource state $\Ket{G_{\mathrm{res}}^{\left(N,D\right)}}$ for our MBQC protocol, and the measurement patterns on the right that can be used for applying the gates in the figure to any given input state of the logical qubits (blue dotted circles) of the level-$l$ GKP $7$-qubit code. Each vertex represents a level-$l$ logical qubit of the GKP $7$-qubit code. In the same way as Fig.~\ref{fig:measurement}, the color of each vertex represents in which logical basis the qubit is to be measured. To apply a gate in the figure to the given GKP qubits, perform potentially noisy $CZ$ gates (red dotted lines) for connecting the given logical qubits and a building block, and then perform $\hat{p}$-SAQEC and $\hat{q}$-SAQEC in Fig~\ref{fig:optimized_sqec} in this order for each of the GKP qubits where the $CZ$ gates act, followed by performing $Z$ and $X$ measurements of the logical qubits using AQEC in Fig.~\ref{fig:analog_qec} and QT-SAQEC in Fig.~\ref{fig:qt_sqec}\@. In the implementation of the $S$ gate, we perform the $S$ gate on the logical qubit represented as the red vertex by transversal Gaussian operations, and then perform $\hat{p}$-SAQEC and $\hat{q}$-SAQEC in this order, followed by the $X$ measurement of this logical qubit. Note that $S^\dag$ can be implemented in the same way as $S$. Combining $\Ket{\frac{\pi}{8}^{(l)}}$ with $\Ket{G^{\left(l\right)}}$, we can also perform the $T$ gate by implementing each gate of the circuit in  Fig.~\ref{fig:gate_teleportation}, while $T^\dag=S^\dag T$ can be implemented by performing $T$ followed by $S^\dag$.}
\end{figure}

\subsection{\label{sec:resource_generation}Fault-tolerant resource state preparation}

\begin{algorithm*}[t]
  \caption{\label{alg:preparation}Fault-tolerant preparation of a hypergraph state $\Ket{G_\mathrm{res}^{\left(N,D\right)}}$ using the GKP $7$-qubit code.}
  \begin{algorithmic}[1]
    \Require{%
      \Statex{Positive integers $N$ and $D$ for an $N$-qubit $D$-depth quantum circuit composed of $\left\{H,CCZ\right\}$ to be implemented by MBQC\@.}
      \Statex{A parameter $\epsilon>0$ representing the failure probability to be achieved.}
      \Statex{Protocols for optimized single-qubit quantum error correction (SQEC) and single-logical-qubit analog quantum error correction (SAQEC) with variance adjustment that we have introduced in Sec.~\ref{sec:optimzed_sqec}, as shown in Fig.~\ref{fig:optimized_sqec}.}
      \Statex{Protocols for highly-reliable measurements (HRM) shown in Fig.~\ref{fig:hrm}.}
      \Statex{Protocols for analog quantum error correction (AQEC) shown in Fig.~\ref{fig:analog_qec}.}
      \Statex{Protocols using building blocks to implement quantum gates by measurement patterns in Fig.~\ref{fig:blocks} using quantum-teleportation SAQEC (QT-SAQEC) shown in Fig.~\ref{fig:qt_sqec}.}
      \Statex{Light sources that can emit GKP qubits prepared in $\Ket{+}$ and $\Ket{\frac{\pi}{8}}$ at a squeezing level larger than the threshold given in Sec.~\ref{sec:threshold}.}
    }
    \Ensure{%
      \Statex{Returning the resource hypergraph state $\Ket{G_\mathrm{res}^{\left(N,D\right)}}$ defined as~\eqref{eq:g_res} that is represented by logical qubits of the GKP $7$-qubit code at a concatenation level that can achieve fault-tolerant MBQC with a failure probability smaller than $\epsilon$.}
    }
    \Procedure{PrepareResource}{$N,D,\epsilon$}
    \Comment{Prepare $\Ket{G_\mathrm{res}^{\left(N,D\right)}}$ using GKP $7$-qubit code at concatenation level $L$.}
    \State{$L\gets$ a positive integer of an order $L=O\left(\log\left(\polylog\left(\frac{DN}{\epsilon}\right)\right)\right)$ representing the required concatenation level.}
    \State{$L_0\gets$ a constant positive integer, \textit{e.g.}, $L_0=4$, representing the concatenation level prepared with post-selection\@.}
    \For{$l\in\left\{0,\ldots,L_0\right\}$}
    \Comment{Prepare $\Ket{+^{\left(L_0\right)}}$, $\Ket{\frac{\pi}{8}^{\left(L_0\right)}}$, and $\Ket{G^{\left(L_0\right)}}$ with HRM and post-selection.}
  \Statex\Comment{In this loop, measurements are performed by HRMs, and the $7$-qubit code is used as the error-detecting code.}
    \If{$l=0$}
    \State{Prepare quantum states $\Ket{G^{\left(0\right)}}$ from GKP qubits in $\Ket{0}$ with HRM (Fig.~\ref{fig:initial}).}
    \Else%
    \If{$l=L_0$ or $l=L$}
    \State{Prepare quantum states $\Ket{+^{\left(l\right)}}$ from $\Ket{G^{\left(l-1\right)}}$ with HRM, using error detection once (Fig.~\ref{fig:preparation_7qubit_plus_optimized}).}
    \Else%
    \State{Prepare quantum states $\Ket{+^{\left(l\right)}}$ from $\Ket{G^{\left(l-1\right)}}$ with HRM, using error detection twice (Fig.~\ref{fig:preparation_7qubit_plus_optimized2}).}
    \EndIf%
    \If{$l=L_0$ or $l=L$}
    \State{Prepare quantum states $\Ket{\frac{\pi}{8}^{\left(l\right)}}$ from $\Ket{\frac{\pi}{8}^{(l-1)}}$ and $\Ket{G^{(l-1)}}$ with HRM, using error detection once (Fig.~\ref{fig:preparation_magic}).}
    \Else%
    \State{Prepare quantum states $\Ket{\frac{\pi}{8}^{\left(l\right)}}$ from $\Ket{\frac{\pi}{8}^{(l-1)}}$ and $\Ket{G^{(l-1)}}$ with HRM, using error detection twice (Fig.~\ref{fig:preparation_magic2}).}
    \EndIf%
    \State{Prepare quantum states $\Ket{G^{\left(l\right)}}$ from $\Ket{+^{\left(l\right)}}$ and $\Ket{G^{\left(l-1\right)}}$ with HRM (Fig.~\ref{fig:preparation_graph_state}).}
    \EndIf%
    \EndFor%

  \Statex\Comment{In the following, HRM is not used, and the $7$-qubit code is used as the error-correcting code.}

    \If{$L>L_0$}
    \For{$l\in\left\{L_0+1,\ldots,L\right\}$}
    \Comment{Prepare $\Ket{+^{\left(L\right)}}$, $\Ket{\frac{\pi}{8}^{\left(L\right)}}$, and $\Ket{G^{\left(L\right)}}$ without HRM\@.}
    \State{Prepare quantum states $\Ket{+^{\left(l\right)}}$ from $\Ket{G^{\left(l-1\right)}}$, using error detection once (Fig.~\ref{fig:preparation_7qubit_plus_optimized}).}
    \State{Prepare quantum states $\Ket{\frac{\pi}{8}^{\left(l\right)}}$ from $\Ket{\frac{\pi}{8}^{(l-1)}}$ and $\Ket{G^{(l-1)}}$, using error detection once (Fig.~\ref{fig:preparation_magic}).}
    \State{Prepare quantum states $\Ket{G^{\left(l\right)}}$ from $\Ket{+^{\left(l\right)}}$ and $\Ket{G^{\left(l-1\right)}}$ (Fig.~\ref{fig:preparation_graph_state}).}
    \EndFor%
    \EndIf%
    \State{Prepare quantum states $\Ket{CCZ^{\left(L\right)}}$ from $\Ket{\frac{\pi}{8}^{\left(L\right)}}$ and $\Ket{G^{\left(L\right)}}$ (Fig.~\ref{fig:preparation_ccz}).}
    \State{Prepare $\Ket{G_\mathrm{res}^{\left(N,D\right)}}$ from $\Ket{CCZ^{\left(L\right)}}$ and $\Ket{G^{\left(L\right)}}$.}
    \State{\Return$\Ket{G_\mathrm{res}^{\left(N,D\right)}}$.}
    \EndProcedure%
  \end{algorithmic}
\end{algorithm*}

In this section, we show a fault-tolerant resource state preparation protocol for obtaining an arbitrarily high-fidelity encoded state $\Ket{G_\mathrm{res}^{\left(N,D\right)}}$ using the GKP $7$-qubit code at a given concatenation level $L$, from given GKP qubits in $\Ket{+}$ and $\Ket{\frac{\pi}{8}}$ at a squeezing level above the threshold.
We give this protocol in Protocol~\ref{alg:preparation}.
Note that our resource state preparation protocol works in general, that is, also for preparing any other multiqubit entangled state than $\Ket{G_\mathrm{res}^{\left(N,D\right)}}$ that can be generated by Clifford and $T$ gates and is represented as a logical state of the GKP $7$-qubit code.

\textbf{Sketch of the fault-tolerant resource state preparation protocol}:
We outline our fault-tolerant resource state preparation protocol.
Each step will be detailed later.
For each concatenation level $l\in\left\{0,1,\ldots,L\right\}$ where level $0$ corresponds to physical GKP qubits,
we prepare level-$l$ logical qubits in $\Ket{+}$ and $\Ket{\frac{\pi}{8}}$, denoted respectively by
\begin{equation}
  \Ket{+^{\left(l\right)}}\text{and }\Ket{\frac{\pi}{8}^{\left(l\right)}}.
\end{equation}
We use $\Ket{+^{\left(l\right)}}$ for SAQEC introduced in Sec.~\ref{sec:optimzed_sqec}, and $\Ket{\frac{\pi}{8}^{\left(l\right)}}$ for implementing a logical non-Clifford gate of the GKP $7$-qubit code.
Furthermore, our protocol prepares a logical $4$-qubit graph state encoded using level-$l$ concatenation, which is denoted by
\begin{equation}
  \Ket{G^{\left(l\right)}},
\end{equation}
called a \textit{building block}, and represented as a graph illustrated on the left of Fig.~\ref{fig:blocks}.

The building block $\Ket{G^{\left(l\right)}}$ serves as a resource for implementing Clifford gates acting on other level-$l$ logical qubits,
in particular, the $CZ$ gate, the $\textsc{CNOT}$ gate, the $H$ gate, and the $S$ gate,
as shown on the right of Fig.~\ref{fig:blocks} that works in the same way as the MBQC protocol shown in Sec.~\ref{sec:quantum_complexity}.
Combined with $\Ket{G^{\left(l\right)}}$,
$\Ket{\frac{\pi}{8}^{\left(l\right)}}$ serves as a resource for implementing the $T$ gate on a level-$l$ logical qubit by implementing each gate in the circuit for state injection shown in Fig.~\ref{fig:gate_teleportation}.
We can also implement $T^\dag=S^\dag T$ by implementing $T$ followed by $S^\dag$, where $S^\dag$ can be implemented using the same measurement pattern as that of $S$.
Note that the implementation of the $H$ gate in Fig.~\ref{fig:blocks} is not explicitly used in our protocol for resource state preparation, but we include $H$ in Fig.~\ref{fig:blocks} to complete the universal gate set, \textit{i.e.}, the Clifford and $T$ gates.

To implement a gate shown on the right of Fig.~\ref{fig:blocks}, we first connect a building block $\Ket{G^{\left(l\right)}}$ with given level-$l$ logical qubits by directly applying $CZ$ gates,
which we can implement for the GKP $7$-qubit code transversally using Gaussian operations.
Second, to reduce the variances that accumulate due to the $CZ$ gates, for each of the directly applied $CZ$ gates, we perform $\hat{p}$-SAQEC and $\hat{q}$-SAQEC in this order for the two level-$l$ logical qubits on which the $CZ$ gate acts, using $\Ket{0^{\left(l\right)}}=H\Ket{+^{\left(l\right)}}$ with analog variance adjustment shown in Fig.~\ref{fig:optimized_sqec}.
Finally, we follow measurement patterns that are illustrated on the right of Fig.~\ref{fig:blocks} to implement the gate, where these measurement patterns are implemented using QT-SAQEC in Fig.~\ref{fig:qt_sqec} to suppress logical bit- and phase-flip errors.
The directly applied $CZ$ gates in the first step are potentially noisy in the sense that $CZ$ gates increase variances of the GKP qubits causing logical bit- and phase-flip errors eventually,
and hence we reduce the variances by the SAQEC\@.
Our resource state preparation protocol is designed to use reliable gates implemented by the variance-reduced building blocks as much as possible, so that the potentially noisy $CZ$ gate acts on each level-$l$ logical qubit \textit{at most once} per variance reduction by SAQECs.

We use the quantum gates implemented by the level-$(l-1)$ building blocks $\Ket{G^{\left(l-1\right)}}$
to prepare $\Ket{+^{\left(l\right)}}$, $\Ket{\frac{\pi}{8}^{\left(l\right)}}$, and $\Ket{G^{\left(l\right)}}$ from $\Ket{+^{\left(l-1\right)}}$ and $\Ket{\frac{\pi}{8}^{\left(l-1\right)}}$.
Repeating this preparation for each $l\in\left\{1,\ldots,L\right\}$, we obtain level-$L$ states $\Ket{+^{\left(L\right)}}$, $\Ket{\frac{\pi}{8}^{\left(L\right)}}$, and $\Ket{G^{\left(L\right)}}$.
In addition, using $\Ket{G^{\left(L\right)}}$ and $\Ket{\frac{\pi}{8}^{\left(L\right)}}$,
we generate level-$L$ logical three qubits prepared in $\Ket{CCZ}$ denoted by
\begin{equation}
  \Ket{CCZ^{\left(L\right)}}.
\end{equation}
Then, we use $\Ket{G^{\left(L\right)}}$s to add edges representing $CZ$ gates to $\Ket{CCZ^{\left(L\right)}}$s, so that we can deterministically prepare the resource state $\Ket{G_\mathrm{res}^{\left(N,D\right)}}$.
In the following, we describe each step of Protocol~\ref{alg:preparation} in detail.

\textbf{Decision of required concatenation level}:
Protocol~\ref{alg:preparation} in lines $2$ and $3$ begins with deciding the required concatenation level for achieving a given failure probability $\epsilon$.
The total number of elementary operations in implementing the MBQC protocol in Theorem~\ref{thm:universality} is of order
\begin{equation}
  O\left(DN\log^2 N\right).
\end{equation}
When we use the $7$-qubit code, or a concatenated quantum error-correcting code in general, for performing these operations in a fault-tolerant way within $\epsilon$,
it suffices to choose the concatenation level $L$ as shown in~\eqref{eq:overhead_qec}, \textit{i.e.},
\begin{align}
  \label{eq:L}
  L&=O\left(\log\left(\polylog\left(\frac{DN\log^2 N}{\epsilon}\right)\right)\right)\nonumber\\
   &=O\left(\log\left(\polylog\left(\frac{DN}{\epsilon}\right)\right)\right).
\end{align}

\begin{figure}[t]
    \centering
    \includegraphics[width=2.0in]{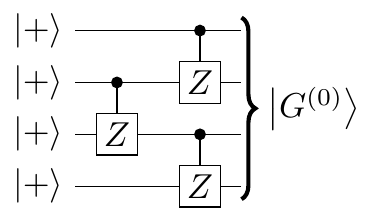}
    \caption{\label{fig:initial}A quantum circuit for preparing a level-$0$ building block $\Ket{G^{\left(0\right)}}$ in Fig.~\ref{fig:blocks} from $\Ket{+}$ of GKP qubits. The GKP qubits are initialized as $\Ket{+}=H\Ket{0}$. Logical Clifford gates on GKP qubits are implemented by Gaussian operations as shown in Sec.~\ref{sec:gkp}. Each gate in the circuit is immediately followed by optimized $\hat{p}$-SQEC and optimized $\hat{q}$-SQEC in this order to reduce the variances, where we introduce these optimized SQECs with variance adjustment in Sec.~\ref{sec:optimzed_sqec}.}
\end{figure}

To improve the threshold, our protocol uses techniques based on post-selection, namely, highly reliable measurements (HRMs) shown in Fig.~\ref{fig:hrm}, and the use of the $7$-qubit code as the error-detecting code rather than the error-correcting code as explained in Sec.~\ref{sec:qec_gkp}.
HRMs decrease misidentification of a bit-valued outcome in measuring a GKP qubit at the expense of probabilistically discarding the generated building blocks.
Let $p\in\left(0,1\right)$ be a lower bound of the success probability of an HRM\@.
If HRMs were used in generating level-$L$ building blocks, the joint probability of succeeding in all the $O\left(7^L\right)$ HRMs would be lower bounded by a probability of order
\begin{equation}
\label{eq:super_poly_small_probability}
  O\left(p^{\polylog\left(\frac{DN}{\epsilon}\right)}\right),
\end{equation}
which would decrease super-polynomially in terms of $D$ and $N$.
As discussed in Sec.~\ref{sec:overhead},
the super-polynomially small success probability~\eqref{eq:super_poly_small_probability} would incur a super-polynomial overhead cost.
Our protocol avoids this super-polynomial overhead by using HRMs only a \textit{constant} number of times at an early stage of concatenation levels $0,\ldots,L_0$ in generating the building blocks,
where $L_0\in\left\{0,1,\ldots\right\}$ is a constant that we set in our numerical simulation in Sec.~\ref{sec:threshold} of our protocol
\begin{equation}
  L_0=4.
\end{equation}
At the rest of the concatenation levels $L_0+1,\ldots,L$, we never use HRMs\@.
In the same way, at concatenation levels $0,\ldots,L_0$, we use the $7$-qubit code as the error-detecting code that can detect up to two errors, by discarding the state if any one or two errors are detected in measuring the stabilizers~\eqref{eq:syndromes} of the $7$-qubit code.
At the rest of the concatenation levels $L_0+1,\ldots,L$, we use the $7$-qubit code as the error-correcting code that can correct up to one error.
Since the parameter $L_0$ is a constant chosen independently of $N$, $D$, and $\epsilon$,
the number of the post-selections in the protocol at the early stage of concatenation levels $0,\ldots,L_0$ is bounded by a constant.
Thus, in order for the post-selections at levels $0,\ldots,L_0$ to succeed at least once, it suffices to repeat them a constant number of times in expectation; that is, the overhead arising from these post-selections is bounded by a constant.
Note that $L_0$ is chosen depending on the squeezing level of the GKP code so that the threshold of QEC can be improved compared to the protocol without the post-selections, as we will discuss in Sec.~\ref{sec:threshold}.

\textbf{Preparation of $\Ket{G^{(0)}}$ from GKP codewords} (Fig.~\ref{fig:initial}):
At the concatenation level $0$, Protocol~\ref{alg:preparation} in lines~6 prepares level-$0$ building-block states $\Ket{G^{\left(0\right)}}$ of physical GKP qubits using Gaussian operations and GKP qubits prepared initially in $\Ket{0}$.
The graph state $\Ket{G^{\left(0\right)}}$ is generated using a circuit illustrated in Fig.~\ref{fig:initial}.
This circuit begins with initializing GKP qubits in $\Ket{+}=H\Ket{0}$ using the initial GKP state $\Ket{0}$ and the Gaussian operations summarized in Sec.~\ref{sec:gkp}.
Then, $CZ$ gates are implemented by Gaussian operations.
Each application of the gates in the circuit is immediately followed by the optimized $\hat{p}$-SQEC and the optimized $\hat{q}$-SQEC in this order to reduce the variances, where we have introduced these optimized SQECs with variance adjustment in Sec.~\ref{sec:optimzed_sqec}.

Notice that in the circuit of Fig.~\ref{fig:initial}, the second and third qubits are more prone to bit- and phase-flip errors of GKP qubits caused by the $CZ$ gates compared to the first and the fourth qubits.
This is because two $CZ$ gates act on the second and third qubits, while only one $CZ$ gate act on the first and the fourth qubits.
Thus, to decrease propagation of these errors of the GKP qubits, we perform the $CZ$ gate on the second and third qubits before that on the first and second qubits and that on the third and fourth qubits, as in the circuit of Fig.~\ref{fig:initial}.

\begin{figure}[t]
    \centering
    \includegraphics[width=3.4in]{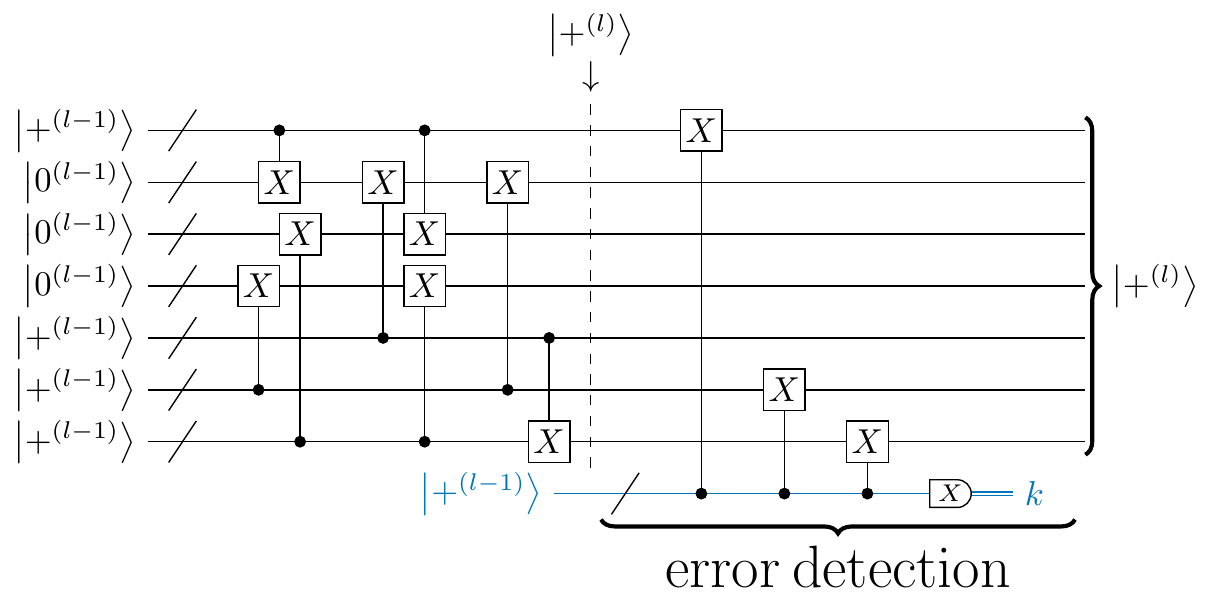}
    \includegraphics[width=3.4in]{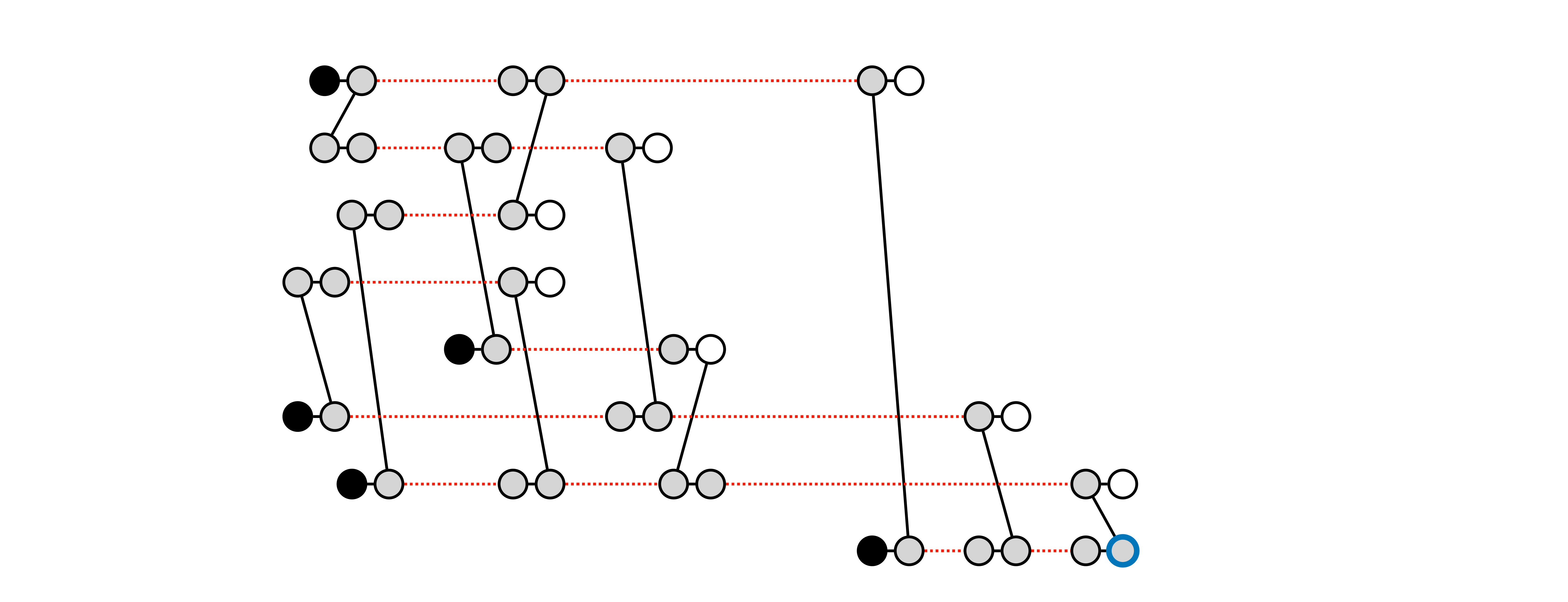}
  \caption{\label{fig:preparation_7qubit_plus_optimized}A quantum circuit using error detection once for preparing a level-$l$ logical qubit in $\Ket{+}$, denoted by $\Ket{+^{\left(l\right)}}$ for $l\in\left\{L_0,\ldots,L\right\}$. This circuit is implemented using multiple level-$\left(l-1\right)$ building blocks $\Ket{G^{\left(l-1\right)}}$ later prepared in Fig.~\ref{fig:preparation_graph_state}, by the measurement patterns in Fig.~\ref{fig:blocks}. The blue wire in the circuit is an auxiliary qubit whose measurement outcome is used for error detection; that is, if the measurement outcome is $k=1$, then the output state of this circuit is discarded. The circuit is designed so that even if any one of the gates suffers from an $X$ or $Z$ error, the output state in the case of $k=0$ should include at most one error. For clarity, the measurement patterns in this case are shown in the lower part, where the notations are the same as those in Fig.~\ref{fig:blocks}, and the blue-bordered vertex is the qubit whose logical bit value of the measurement outcome decided by the AQEC in Fig.~\ref{fig:analog_qec} is used for the error detection. Note that a part of $\Ket{G^{\left(l-1\right)}}$s can be used as $\Ket{+^{\left(l-1\right)}}$s and $\Ket{0^{\left(l-1\right)}}$s in the circuit to be implemented.}
\end{figure}

\begin{figure}[t]
    \centering
    \includegraphics[width=3.4in]{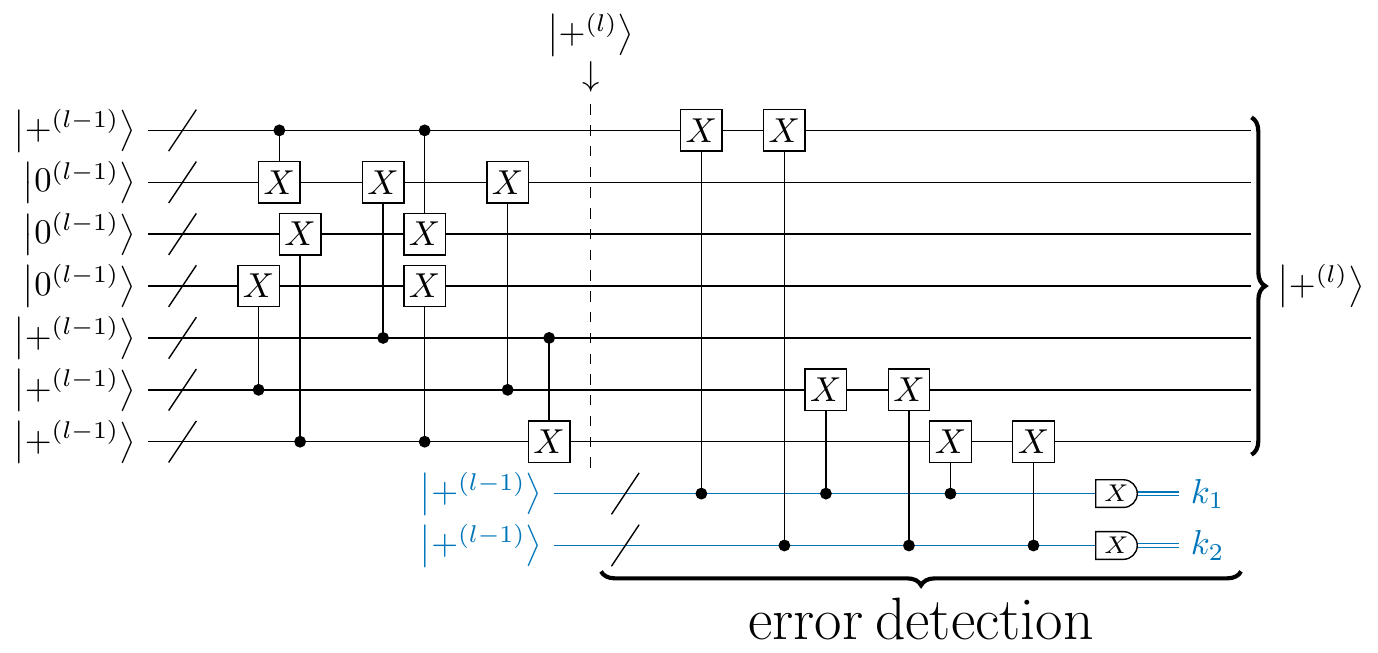}
    \caption{\label{fig:preparation_7qubit_plus_optimized2}A quantum circuit using error detection twice for preparing a level-$l$ logical qubit in $\Ket{+}$, denoted by $\Ket{+^{\left(l\right)}}$ for $l\in\left\{1,\ldots,L_0-1\right\}$. The circuit is implemented in the same way as Fig.~\ref{fig:preparation_7qubit_plus_optimized} by the measurement patterns in Fig.~\ref{fig:blocks}. The blue wires in the circuit are auxiliary qubits whose measurement outcomes are used for error detection; that is, if $k_1=1$ or $k_2=1$, then the output state of this circuit is discarded. The circuit is designed so that even if any two of the gates suffer from $X$ or $Z$ errors, the output state in the case of $k_1=k_2=0$ should include at most two errors.}
\end{figure}

\textbf{Preparation of $\Ket{+^{\left(l\right)}}$} (Fig.~\ref{fig:preparation_7qubit_plus_optimized} using error detection once, and Fig.~\ref{fig:preparation_7qubit_plus_optimized2} using error detection twice):
At the concatenation level $l\in\left\{L_0,L_0+1\ldots,L\right\}$, Protocol~\ref{alg:preparation} in lines~9 and~23 prepares a level-$l$ state $\Ket{+^{\left(l\right)}}$, using error detection \textit{once}, from $\Ket{G^{\left(l-1\right)}}$, which we here assume has already been prepared in the protocol.
We use $\Ket{+^{\left(l\right)}}$ for SAQECs in performing the measurement patterns in Fig.~\ref{fig:blocks}, while $\Ket{+^{\left(l\right)}}$ will also be used for preparing $\Ket{G^{(l)}}$ in Fig.~\ref{fig:preparation_graph_state}.
Based on a low-overhead circuit for preparing a codeword of the $7$-qubit code in Ref.~\cite{G6}, $\Ket{+^{\left(l\right)}}$ is prepared using a circuit illustrated in Fig.~\ref{fig:preparation_7qubit_plus_optimized}, which is implemented as shown in the lower part of Fig.~\ref{fig:preparation_7qubit_plus_optimized} using the measurement patterns in Fig.~\ref{fig:blocks}.
While the circuit starts from $\Ket{+^{\left(l-1\right)}}$s and $\Ket{0^{\left(l-1\right)}}$s,
we can use a part of $\Ket{G^{\left(l-1\right)}}$s as $\Ket{+^{\left(l-1\right)}}$s and $\Ket{0^{\left(l-1\right)}}$s in the circuit to be implemented.
Also at the concatenation level $l\in\left\{0,\ldots,L_0-1\right\}$, Protocol~\ref{alg:preparation} in line~11 prepares a level-$l$ state $\Ket{+^{\left(l\right)}}$, using error detection \textit{twice}, from $\Ket{G^{\left(l-1\right)}}$.
To achieve the preparation of $\Ket{+^{\left(l\right)}}$ using error detection twice,
we construct a circuit illustrated in Fig.~\ref{fig:preparation_7qubit_plus_optimized2}.
Whereas the circuit in Fig.~\ref{fig:preparation_7qubit_plus_optimized2} is similar to that in Fig.~\ref{fig:preparation_7qubit_plus_optimized} based on Ref.~\cite{G6}, our contribution here is to design the error detection in Fig.~\ref{fig:preparation_7qubit_plus_optimized2} so that not only one but also two errors can be detected, as explained in the following.

In the circuits of Figs.~\ref{fig:preparation_7qubit_plus_optimized} and~\ref{fig:preparation_7qubit_plus_optimized2}, we discard the output state if any error is detected; that is, $k=1$ in the circuit of Fig.~\ref{fig:preparation_7qubit_plus_optimized}, and $k_1=1$ or $k_2=1$ in the circuit of Fig.~\ref{fig:preparation_7qubit_plus_optimized2}.
Each gate in the circuits of Figs.~\ref{fig:preparation_7qubit_plus_optimized} and~\ref{fig:preparation_7qubit_plus_optimized2}, implemented using the measurement patterns in Fig.~\ref{fig:blocks}, may suffer from level-$(l-1)$ bit- and phase-flip errors, \textit{i.e.}, $X$ and $Z$ errors.
If an $X$ or $Z$ error occurs on a qubit in the circuits, the error may propagate to other qubits through the two-qubit gates in the circuits according to the commutation relations
\begin{align}
  \label{eq:x1}
  \includegraphics[width=3.4in]{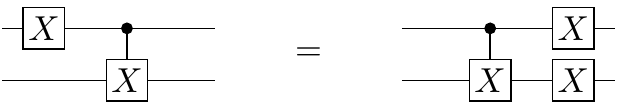}\\
  \label{eq:x2}
  \includegraphics[width=3.4in]{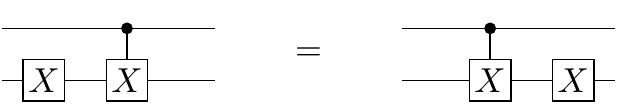}\\
  \label{eq:x3}
  \includegraphics[width=3.4in]{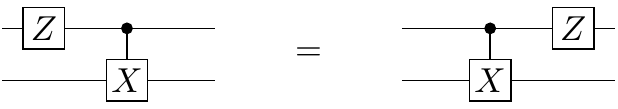}\\
  \label{eq:x4}
  \includegraphics[width=3.4in]{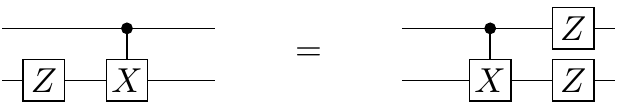},
\end{align}
which may cause $X$ and $Z$ errors on multiple qubits at the end of the circuits.
We use the circuit in Fig.~\ref{fig:preparation_7qubit_plus_optimized} when the level-$l$ $7$-qubit code for the output state $\Ket{+^{(l)}}$ is to be used as the error-correcting code that can correct up to one $X$ error and one $Z$ error.
In the circuit of Fig.~\ref{fig:preparation_7qubit_plus_optimized},
the preparation of $\Ket{+^{(l)}}$ is followed by one error detection, so that even if the gates in the circuit suffer from at most one $X$ error and at most one $Z$ error, the state at the end of the circuit should include at most one-qubit $X$ and $Z$ errors, which we can correct by the $7$-qubit code.
Similarly,
we use the circuit in Fig.~\ref{fig:preparation_7qubit_plus_optimized2} when the level-$l$ $7$-qubit code for the output state $\Ket{+^{(l)}}$ is to be used as the error-detecting code that can detect up to two $X$ errors and two $Z$ errors, by discarding the state when any error is detected as summarized in Sec.~\ref{sec:qec_gkp}.
In the circuit of Fig.~\ref{fig:preparation_7qubit_plus_optimized2},
the preparation of $\Ket{+^{(l)}}$ is followed by two error detections, so that even if the gates in the circuit suffer from at most \textit{two} $X$ errors and at most \textit{two} $Z$ errors in total,
the state at the end of the circuit should include at most \textit{two-qubit} $X$ and $Z$ errors, which we can \textit{detect} by the $7$-qubit code.
Note that as long as the squeezing level of the GKP qubits is larger than the threshold, the probability of detecting an error is expected to approach zero doubly exponentially as the concatenation level $l$ increases, and hence, the post-selection in the error detection of Figs.~\ref{fig:preparation_7qubit_plus_optimized} and~\ref{fig:preparation_7qubit_plus_optimized2} does not affect asymptotic scaling of the overhead.

In particular, to show these fault-tolerant properties of the circuits in Figs.~\ref{fig:preparation_7qubit_plus_optimized} and~\ref{fig:preparation_7qubit_plus_optimized2} mathematically,
we use the following definition of fault tolerance~\cite{G,Chamberland2019faulttolerantmagic}.
\begin{definition}[\label{def:fault_tolerance_correction}Fault tolerance in terms of quantum error correction]
In terms of the error correction,
a state-preparation protocol is said to be $1$-fault-tolerant for $X$ errors if the following conditions are satisfied:
\begin{itemize}
  \item if there are $s$ $X$ errors in the protocol for $s\leqq 1$, then the errors in the state prepared by the protocol have at most weight $s$;
  \item if there are $s$ $X$ errors in the protocol for $s\leqq 1$, then after the protocol, by performing an error-recovery (\textit{i.e.}, decoding) operation that is free from errors, the state prepared by the protocol with the errors is transformed into the error-free state that is supposed to be prepared.
\end{itemize}
The $1$-fault tolerance for $Z$ errors is defined in the same way as the $1$-fault tolerance for $X$ errors by replacing $X$ with $Z$.
If a state-preparation protocol is $1$-fault-tolerant for $X$ errors and also $1$-fault-tolerant for $Z$ errors, then the protocol is said to be $1$-fault-tolerant.
\end{definition}
In terms of this definition, the circuit in Fig.~\ref{fig:preparation_7qubit_plus_optimized} is $1$-fault-tolerant as shown in Ref.~\cite{G6}.
In the same way, the circuit in Fig.~\ref{fig:preparation_7qubit_plus_optimized2} is also $1$-fault-tolerant,
and moreover,
the circuit in Fig.~\ref{fig:preparation_7qubit_plus_optimized2} followed by the error detection using the $7$-qubit code as the error-detecting code can detect \textit{two} errors in a fault-tolerant way, as we show in Appendix~\ref{sec:fault_tolerance_preparation}.
In particular, the circuit in Fig.~\ref{fig:preparation_7qubit_plus_optimized2} has the following fault-tolerant property for two errors in terms of the error detection.
\begin{definition}[\label{def:fault_tolerance_detection}Fault-tolerant property in terms of quantum error detection]
In terms of the error detection,
a state-preparation protocol is said to have a fault-tolerant property for two errors if the following conditions are satisfied for both $X$ errors and $Z$ errors:
\begin{itemize}
  \item if there are $s$ $X$ ($Z$) errors in the protocol for $s\leqq 2$, then the errors in the state prepared by the protocol have at most weight $s$;
  \item if there are $s$ $X$ ($Z$) errors in the protocol for $s\leqq 2$, then after the protocol, by performing an error-detecting operation that is free from errors, the state prepared by the protocol with the errors is transformed into the error-free state that is supposed to be prepared, which should be satisfied conditioned on the success of the post-selection in the error-detecting operation.
\end{itemize}
\end{definition}

\begin{figure*}[t]
    \centering
    \includegraphics[width=7.0in]{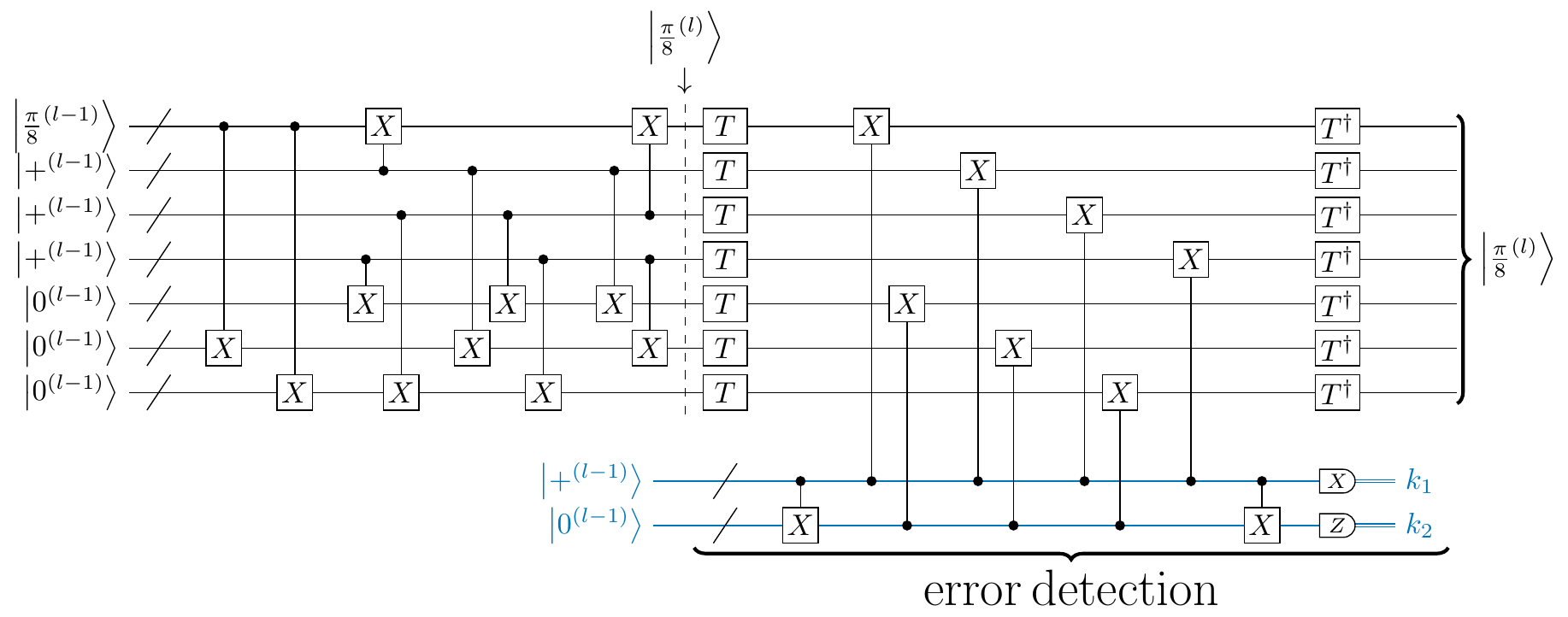}
  \caption{\label{fig:preparation_magic}A quantum circuit using error detection once for preparing a level-$l$ logical qubit in $\Ket{\frac{\pi}{8}}$, denoted by $\Ket{\frac{\pi}{8}^{\left(l\right)}}$ for $l\in\left\{L_0,\ldots,L\right\}$. This circuit is implemented using multiple level-$\left(l-1\right)$ building blocks $\Ket{G^{\left(l-1\right)}}$ later prepared in Fig.~\ref{fig:preparation_graph_state}, and $\Ket{\frac{\pi}{8}^{(l-1)}}$, by the measurement patterns in Fig.~\ref{fig:blocks}, where $\Ket{\frac{\pi}{8}^{(l-1)}}$ is used to perform $T$ gates and $T^\dag$ gates by implementing the circuit in Fig.~\ref{fig:gate_teleportation}. Note that a part of $\Ket{G^{\left(l-1\right)}}$ can be used as $\Ket{+^{\left(l-1\right)}}$s and $\Ket{0^{\left(l-1\right)}}$s in this circuit, in the same way as the implementation shown in the lower part of Fig.~\ref{fig:preparation_7qubit_plus_optimized}. Each blue wire in the circuit is an auxiliary qubit for error detection in the same way as Fig.~\ref{fig:preparation_7qubit_plus_optimized}; that is, if at least one of the measurement outcomes $k_1$ and $k_2$ is $1$, then the output state of this circuit is discarded. The circuit is designed so that even if any one of the gates suffers from an $X$ or $Z$ error, the output state in the case of $k_1=k_2=0$ should include at most one error.}
\end{figure*}

\begin{figure*}[t]
    \centering
    \includegraphics[width=7.0in]{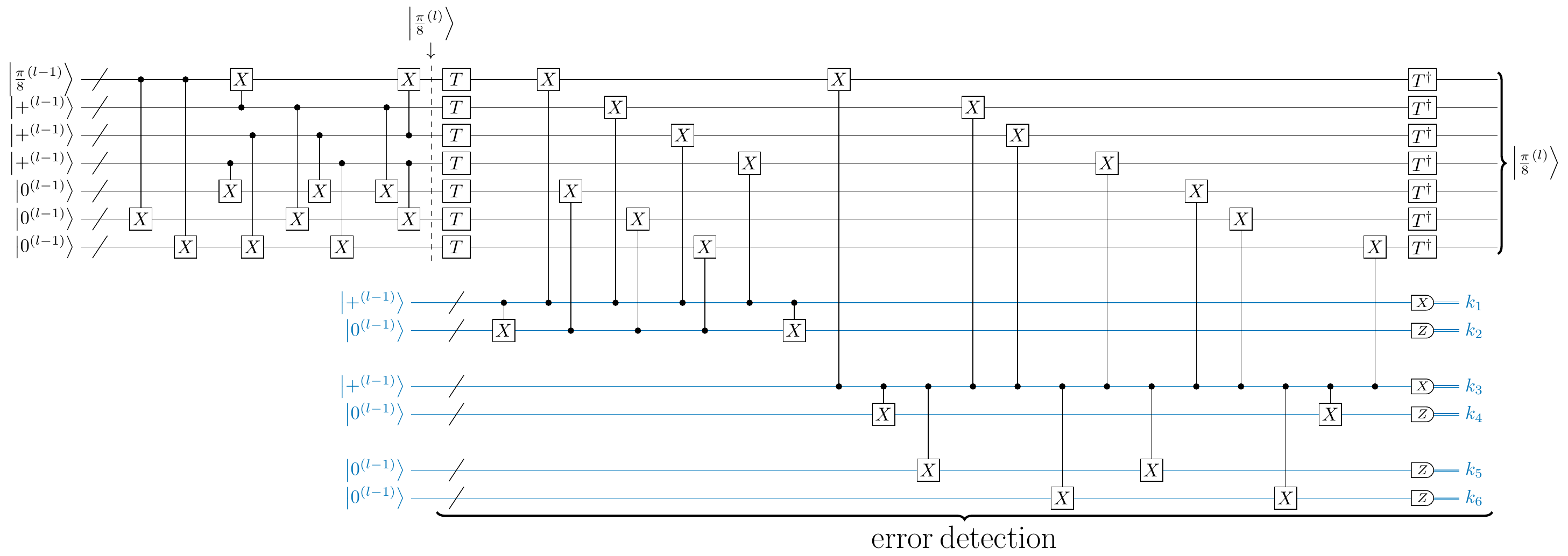}
    \caption{\label{fig:preparation_magic2}A quantum circuit using error detection twice for preparing a level-$l$ logical qubit in $\Ket{\frac{\pi}{8}}$, denoted by $\Ket{\frac{\pi}{8}^{\left(l\right)}}$ for $l\in\left\{1,\ldots,L_0-1\right\}$. This circuit is implemented in the same way as Fig.~\ref{fig:preparation_magic}, by the measurement patterns in Fig.~\ref{fig:blocks}. Each blue wire in the circuit is an auxiliary qubit for error detection in the same way as Fig.~\ref{fig:preparation_magic}; that is, if at least one of the measurement outcomes $k_1,\ldots,k_6$ is $1$, then the output state of this circuit is discarded. The circuit is designed so that even if any two of the gates suffer from $X$ or $Z$ errors, the output state in the case of $k_1=\cdots=k_6=0$ should include at most two errors.}
\end{figure*}

\textbf{Preparation of $\Ket{\frac{\pi}{8}^{\left(l\right)}}$} (Fig.~\ref{fig:preparation_magic} using error detection once, and Fig.~\ref{fig:preparation_magic2} using error detection twice):
At the concatenation level $l\in\left\{L_0,L_0+1\ldots,L\right\}$, Protocol~\ref{alg:preparation} in lines~14 and~24 prepares a level-$l$ magic state $\Ket{\frac{\pi}{8}^{\left(l\right)}}$ of the GKP $7$-qubit code, using error detection \textit{once}, from $\Ket{G^{\left(l-1\right)}}$ and $\Ket{\frac{\pi}{8}^{\left(l-1\right)}}$, which we here assume have already been prepared in the protocol.
Based on a low-overhead circuit for preparing a magic state of the $7$-qubit code in Ref.~\cite{G6}, we prepare $\Ket{\frac{\pi}{8}^{\left(l\right)}}$ using a circuit illustrated in Fig.~\ref{fig:preparation_magic}, which is implemented using the measurement patterns in Fig.~\ref{fig:blocks} in the same way as the preparation of $\Ket{+^{\left(l\right)}}$ in Fig.~\ref{fig:preparation_7qubit_plus_optimized}.
Note that the circuit in Fig.~\ref{fig:preparation_magic} has the same number of gates as an equivalent circuit used in Refs.~\cite{Chamberland2019faulttolerantmagic,Chamberland2020} based on the technique of flag qubits~\cite{C3,C4,R9,Chao2019}, but has a shorter depth than that in Ref.~\cite{Chamberland2019faulttolerantmagic}, which contributes to reducing errors.
While the circuit starts with $\Ket{+^{\left(l-1\right)}}$s and $\Ket{0^{\left(l-1\right)}}$s,
we can use a part of $\Ket{G^{\left(l-1\right)}}$s as $\Ket{+^{\left(l-1\right)}}$s and $\Ket{0^{\left(l-1\right)}}$s in the circuit to be implemented.
The $T$ gate in the circuit is implemented using the circuit for state injection shown in Fig.~\ref{fig:gate_teleportation}, where $\Ket{\frac{\pi}{8}^{(l-1)}}$ is used in addition to $\Ket{G^{(l-1)}}$.
Similarly, at the concatenation level $l\in\left\{0,\ldots,L_0-1\right\}$, Protocol~\ref{alg:preparation} in line~16 prepares a level-$l$ GKP magic state $\Ket{\frac{\pi}{8}^{\left(l\right)}}$, using error detection \textit{twice}, from $\Ket{G^{\left(l-1\right)}}$ and $\Ket{\frac{\pi}{8}^{\left(l-1\right)}}$.
To achieve the preparation of $\Ket{\frac{\pi}{8}^{\left(l\right)}}$ using error detection twice,
we construct a circuit illustrated in Fig.~\ref{fig:preparation_magic2}.
Whereas the circuit in Fig.~\ref{fig:preparation_magic2} is similar to that in Fig.~\ref{fig:preparation_magic} based on Ref.~\cite{G6}, our contribution here is to design the error detection in Fig.~\ref{fig:preparation_magic2}, so that not only one but also two errors can be detected as explained in the following.

In the circuits of Figs.~\ref{fig:preparation_magic} and~\ref{fig:preparation_magic2}, we discard the output state if any error is detected; that is, $k_1=1$ or $k_2=1$ in the circuit of Fig.~\ref{fig:preparation_magic}, and $k_1=1,k_2=1,\ldots$, or $k_6=1$ in the circuit of Fig.~\ref{fig:preparation_magic2}.
In the same way as Fig.~\ref{fig:preparation_7qubit_plus_optimized},
the circuit in Fig.~\ref{fig:preparation_magic} is designed so that even if the gates in the circuit suffer from at most one $X$ error and at most one $Z$ error, the state at the end of the circuit should include at most one-qubit errors, which we can correct using the $7$-qubit code as the error-correcting code; in particular, the circuit in Fig.~\ref{fig:preparation_magic2} is $1$-fault-tolerant in the sense of Definition~\ref{def:fault_tolerance_correction}~\cite{G6,Chamberland2019faulttolerantmagic}.
Also,
in the same way as Fig.~\ref{fig:preparation_7qubit_plus_optimized2},
the circuit in Fig.~\ref{fig:preparation_magic2} is designed so that even if the gates in the circuit suffer from at most two $X$ errors and at most two $Z$ errors, the state at the end of the circuit should include at most two-qubit errors, which we can detect using the $7$-qubit code as the error-detecting code;
in particular, the circuit in Fig.~\ref{fig:preparation_magic2} is also $1$-fault-tolerant, and moreover,
the circuit in Fig.~\ref{fig:preparation_magic2} followed by the error detection using the $7$-qubit code has a fault-tolerant property in detecting \textit{two} errors in the sense of Definition~\ref{def:fault_tolerance_detection}, as we show in Appendix~\ref{sec:fault_tolerance_preparation}.
Note that for the same reason as the error detection of $\Ket{+^{\left(l\right)}}$, the post-selection in the error detection of Figs.~\ref{fig:preparation_magic} and~\ref{fig:preparation_magic2} does not affect asymptotic scaling of the overhead, as long as the squeezing level of the GKP qubits is larger than the threshold.

\begin{figure}[t]
    \centering
    \includegraphics[width=2.0in]{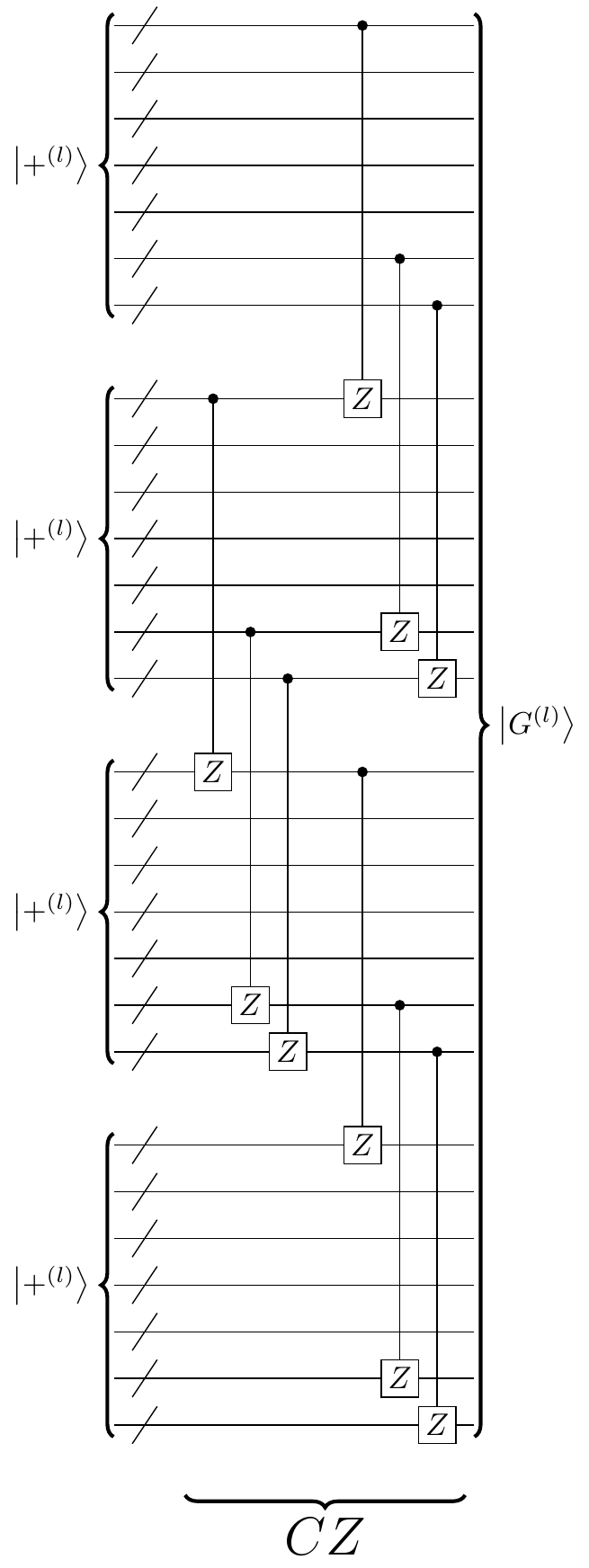}
    \caption{\label{fig:preparation_graph_state}A quantum circuit for preparing a level-$l$ building block $\Ket{G^{\left(l\right)}}$ for $l\in\left\{1,\ldots,L\right\}$. This circuit is implemented using level-$l$ states $\Ket{+^{\left(l\right)}}$ prepared in Figs.~\ref{fig:preparation_7qubit_plus_optimized} and~\ref{fig:preparation_7qubit_plus_optimized2}, and level-$\left(l-1\right)$ building blocks $\Ket{G^{\left(l-1\right)}}$, by the measurement patterns in Fig.~\ref{fig:blocks}. Note that at the logical level of the $7$-qubit code, this circuit implements the circuit in Fig.~\ref{fig:initial}.}
\end{figure}

\begin{figure}[t]
    \centering
    \includegraphics[width=3.4in]{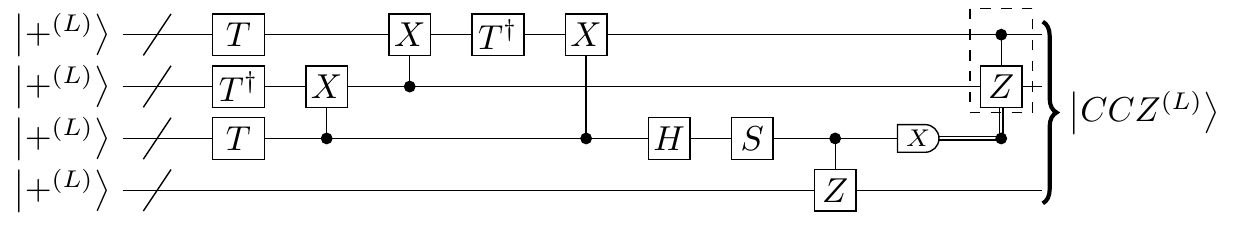}
    \caption{\label{fig:preparation_ccz}A quantum circuit for preparing level-$L$ logical qubits in $\Ket{CCZ}$, denoted by $\Ket{CCZ^{\left(L\right)}}$. This circuit is implemented using level-$L$ states $\Ket{\frac{\pi}{8}^{\left(L\right)}}$ prepared in Fig.~\ref{fig:preparation_magic} and level-$L$ building blocks $\Ket{G^{\left(L\right)}}$ prepared in Fig.~\ref{fig:preparation_graph_state}, by the measurement patterns in Fig.~\ref{fig:blocks}.}
\end{figure}

\textbf{Preparation of $\Ket{G^{\left(l\right)}}$} (Fig.~\ref{fig:preparation_graph_state}):
At a nonzero concatenation level $l\in\left\{1,\ldots,L_0\right\}$,
Protocol~\ref{alg:preparation} in lines~18 and~25 prepares $\Ket{G^{\left(l\right)}}$ from $\Ket{+^{\left(l\right)}}$ and $\Ket{G^{\left(l-1\right)}}$, which we here assume have already been prepared in the protocol.
The graph state $\Ket{G^{\left(l\right)}}$ is prepared from $\Ket{+^{(l)}}$s using a circuit illustrated in Fig.~\ref{fig:preparation_graph_state} by the transversal implementations of logical Clifford gates of the $7$-qubit code shown in Sec.~\ref{sec:qec_gkp}, which we implement using the measurement patterns in Fig.~\ref{fig:blocks}.

\textbf{Preparation of $\Ket{CCZ^{(L)}}$} (Fig.~\ref{fig:preparation_ccz}) \textbf{and $\Ket{G_\mathrm{res}^{(N,D)}}$}:
At the concatenation level $L$, Protocol~\ref{alg:preparation} in line~28 prepares level-$L$ states $\Ket{CCZ^{\left(L\right)}}$ from $\Ket{\frac{\pi}{8}^{\left(L\right)}}$ and $\Ket{G^{\left(L\right)}}$ given above.
Based on a circuit for implementing a $CCZ$ gate using four $T$ gates in Refs.~\cite{S2,C1}, $\Ket{CCZ^{\left(L\right)}}$ is generated by a circuit illustrated in Fig.~\ref{fig:preparation_ccz}, which we implement using the measurement patterns in Fig.~\ref{fig:blocks}.
Note that in the same way as Fig.~\ref{fig:preparation_7qubit_plus_optimized} for preparing $\Ket{+^{(l)}}$,
we can use a part of $\Ket{G^{\left(L\right)}}$s as $\Ket{+^{\left(L\right)}}$s in this circuit.
The $\Ket{CCZ^{\left(L\right)}}$s are used to generate the resource state $\Ket{G_\mathrm{res}^{(N,D)}}$, which can be obtained by adding edges representing $CZ$ gates to hyperedges of $\Ket{CCZ^{\left(L\right)}}$s using the building blocks $\Ket{G^{(L)}}$ and the measurement patterns in Fig.~\ref{fig:blocks}.
Note that we can use $\Ket{\frac{\pi}{8}^{\left(L\right)}}$ and $\Ket{G^{\left(L\right)}}$ to implement the universal gate set, the Clifford and $T$ gates, for preparing not only $\Ket{G_\mathrm{res}^{(N,D)}}$ but an arbitrarily multiqubit entangled state in a fault-tolerant way.

\subsection{\label{sec:complexity_qec}Complexity of fault-tolerant MBQC protocol and resource state preparation}

In this subsection, we show that the overhead cost of the fault-tolerant MBQC protocol given in this section, including the QEC, scales poly-logarithmically in implementing an $N$-qubit $D$-depth quantum circuit within a failure probability $\epsilon$.
In the following, we analyze the number of computational steps in Protocol~\ref{alg:preparation} for the resource state preparation shown in Sec.~\ref{sec:resource_generation}, and that in the fault-tolerant MBQC protocol shown in Sec.~\ref{sec:fault_tolerant_mbqc_protocol}.

First, we analyze the number of computational steps in Protocol~\ref{alg:preparation} for the resource state preparation shown in Sec.~\ref{sec:resource_generation}.
For implementing the MBQC protocol in Theorem~\ref{thm:universality} in a fault-tolerant way within $\epsilon$, the required concatenation level $L$ given by~\eqref{eq:L} shows that the number of physical GKP qubit for representing each level-$L$ logical qubit is
\begin{equation}
  \label{eq:building_block_qubits}
  7^L=O\left(\polylog\left(\frac{DN}{\epsilon}\right)\right).
\end{equation}
A constant number of the level-$L$ logical qubits comprise each level-$L$ building block $\Ket{G^{\left(L\right)}}$, and hence, the number of physical GKP qubits comprising each of $\Ket{+^{(L)}}$, $\Ket{\frac{\pi}{8}^{(L)}}$, and $\Ket{G^{\left(L\right)}}$ is of the same order as~\eqref{eq:building_block_qubits}.
Since the post-selections in Protocol~\ref{alg:preparation} are designed to incur at most a constant overhead cost as discussed in Sec.~\ref{sec:resource_generation}, Protocol~\ref{alg:preparation} can prepare each of $\Ket{+^{(L)}}$, $\Ket{\frac{\pi}{8}^{(L)}}$, and $\Ket{G^{\left(L\right)}}$ within a complexity per level-$L$ state in expectation
\begin{equation}
  O\left(\polylog\left(\frac{DN}{\epsilon}\right)\right).
\end{equation}
Note that the AQEC for a level-$L$ logical qubit can be performed efficiently within a complexity of the same order as~\eqref{eq:building_block_qubits}, as shown in~\eqref{eq:aqec_complexity}.
Due to Theorem~\ref{thm:preparation_complexity}, the required number of $\Ket{\frac{\pi}{8}^{(L)}}$ and $\Ket{G^{\left(L\right)}}$ for preparing the level-$L$ encoded resource state $\Ket{G_{\mathrm{res}}^{\left(N,D\right)}}$ is
\begin{equation}
  O\left(DN\log^2 N\right).
\end{equation}
Therefore, the preparation complexity of the fault-tolerant resource state preparation by Protocol~\ref{alg:preparation} is in expectation
\begin{align}
  &O\left(DN\log^2\left(N\right)\times \polylog\left(\frac{DN}{\epsilon}\right)\right)\nonumber\\
  &=O\left(DN\times \polylog\left(\frac{DN}{\epsilon}\right)\right).
\end{align}

Next,
we analyze the number of computational steps in the fault-tolerant MBQC protocol shown in Sec.~\ref{sec:fault_tolerant_mbqc_protocol}.
In this fault-tolerant MBQC protocol, each homodyne measurement in the MBQC protocol in Theorem~\ref{thm:universality} is replaced with AQEC for a level-$L$ logical qubit, whose cost is of the same order as~\eqref{eq:building_block_qubits}.
Thus, the complexity of this fault-tolerant MBQC protocol is again
\begin{align}
  &O\left(DN\log^2\left(N\right)\times \polylog\left(\frac{DN}{\epsilon}\right)\right)\nonumber\\
  &=O\left(DN\times \polylog\left(\frac{DN}{\epsilon}\right)\right).
\end{align}
Consequently, the overall fault-tolerant MBQC protocol including the fault-tolerant resource state preparation and this fault-tolerant computation using the resource has this complexity, the overhead of which scales poly-logarithmically.

\subsection{\label{sec:threshold}Threshold}

\begin{figure}[t]
    \centering
    \includegraphics[width=3.4in]{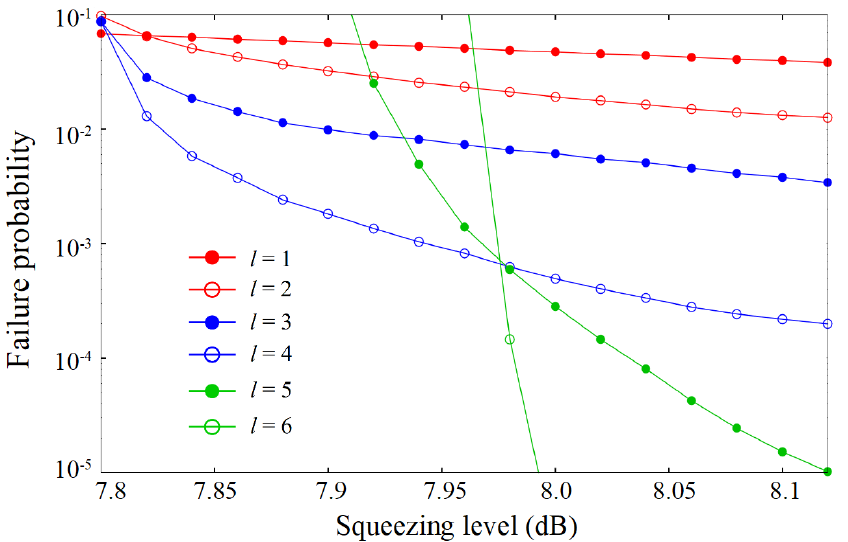}
  \caption{\label{fig:threshold}Simulation results of the failure probability of QEC for performing Clifford gates in the fault-tolerant MBQC protocol shown in Sec.~\ref{sec:fault_tolerant_mbqc_protocol} using the resource state prepared by Protocol~\ref{alg:preparation} in Sec.~\ref{sec:resource_generation}, where the parameter $L_0$ representing the concatenation level using post-selection is chosen as $L_0=4$. The failure probability is plotted as a function of the squeezing level $-10\log_{10}\left(2\sigma^2\right)$ in decibel of the GKP qubits at a concatenation level $l$, where $l$ is $1$ (filled red circles), $2$ (open red circles), $3$ (filled blue circles), $4$ (open blue circles), $5$ (filled green circles), and $6$ (open green circles). The threshold of our protocol is $7.8$~dB, where the parameter $L_0$ of the protocol should be chosen appropriately depending on the squeezing level of the GKP qubit. In particular, the figure shows that it suffices to choose $L_0=4$ when we can realize the GKP code at a squeezing level larger than $7.98$~dB.}
\end{figure}

In this subsection, we show the improvement of our fault-tolerant MBQC protocol in fault tolerance over the existing protocols by numerically calculating the threshold of our fault-tolerant MBQC protocol.

We simulate QEC in the fault-tolerant MBQC protocol shown in Sec.~\ref{sec:fault_tolerant_mbqc_protocol} using the resource state generated by Protocol~\ref{alg:preparation} in Sec.~\ref{sec:resource_generation}, by means of a Monte Carlo method in the same way as Refs.~\cite{F2,F6};
in particular, we simulate the QEC for Clifford gates.
The noise model is summarized in Sec.~\ref{sec:qec_gkp}.
In cases of the $7$-qubit code~\cite{G6,Chamberland2019faulttolerantmagic},
errors of the Clifford gates are usually dominant compared to those of the preparation of the magic state for implementing non-Clifford gates,
and in this paper, following the previous works~\cite{M4,F2,F6,N3} to compare with our work, we use the QEC for the Clifford gates to estimate the threshold.
Note that numerical simulation of AQEC in the preparation of the magic state $\Ket{\frac{\pi}{8}^{(L)}}$ of the level-$L$ GKP $7$-qubit code is computationally hard due to the non-Clifford gates, and hence, further research is needed to establish techniques for numerical simulation to confirm that the threshold for the preparation of the magic state using AQEC is indeed better than that for the Clifford gates.

A feature of our fault-tolerant MBQC protocol is that it switches between two protocols, that is, one with post-selection at earlier concatenation levels $0,\ldots,L_0$ to improve the threshold, and the other without post-selection at later concatenation levels $L_0+1,\ldots,L$ to achieve low overhead.
In contrast to fault-tolerant protocols using post-selection at all the concatenation levels such as Knill's error-correcting $C_4 / C_6$ architecture~\cite{K5,K6}, our protocol can control the trade-off between the failure probability and the overhead of QEC by choosing a parameter $L_0$ of the protocol.
As discussed in Sec.~\ref{sec:qec_gkp},
the threshold of our protocol is defined in terms of the squeezing level of initial GKP qubits;
by definition, as long as the squeezing level of the initial GKP qubits is strictly larger than the threshold, a fault-tolerant protocol can reduce the failure probability to arbitrarily small $\epsilon>0$ by increasing the concatenation level $L$.
The post-selection at earlier levels $0,\ldots,L_0$ in Protocol~\ref{alg:preparation} contributes to reducing errors, and thus to improving the threshold.
By contrast, although achieving lower overhead, a protocol without post-selection at later levels than $L_0$ would have a worse threshold than that with post-selection.
In Protocol~\ref{alg:preparation}, the parameter $L_0$ should be chosen appropriately so that the post-selection at the earlier levels $0,\ldots,L_0$ can reduce the failure probability of QEC at the level $L_0+1$ to this worse threshold of the protocol without post-selection.
Owing to this error reduction at the earlier levels, we can use the low-overhead protocol without post-selection at later levels $L_0+1,\ldots,L$ to reduce the failure probability arbitrarily.
As shown in~\eqref{eq:L},
the required level $L$ of all the concatenations may get large as we decrease $\epsilon$ and increase the size of the quantum circuit to be implemented;
by contrast, the parameter $L_0$ of Protocol~\ref{alg:preparation} is independent of $\epsilon$ and the size of the circuit to be implemented, determined depending on the squeezing level of the GKP code.
Note that as the squeezing level of the initial GKP qubits gets larger owing to progress on quantum technologies, we can choose smaller $L_0$ to reduce the constant factor of the overhead; in the smallest case, \textit{i.e.}, $L_0=0$, the overhead of our protocol reduces to that of the state-of-the-art low-overhead fault-tolerant protocols in Refs.~\cite{G6,Chamberland2019faulttolerantmagic} by construction.
Consequently, with the appropriate choice of the parameter $L_0$,
the threshold of our fault-tolerant protocol is given by that at the earlier concatenation level with post-selection.

In Fig.~\ref{fig:threshold}, we plot the failure probabilities of the QEC at concatenation levels
\begin{equation}
  l=1,2,3,4,5,6
\end{equation}
as a function of the squeezing level~\eqref{eq:squeezing_level}, $-10\log_{10}\left(2\sigma^2\right)$, in decibel of the initial physical GKP qubit~\eqref{eq:gkp_approximate_codeword_0} and~\eqref{eq:gkp_approximate_codeword_1}, where the parameter $L_0$ of the protocol is chosen as $L_0=4$.
The numerical result shows that the threshold value of the squeezing level is
\begin{equation}
  7.8\,\mathrm{dB}, (\sigma=0.288),
\end{equation}
and moreover, if we can realize the squeezing level of GKP qubits larger than $7.98$~dB, ($\sigma=0.282$), then it suffices to choose $L_0=4$ to suppress the failure probability arbitrarily by increasing the concatenation level.
Compared to the threshold value $8.3$ dB of the best existing protocol~\cite{F6} in the same noise model, the threshold value of our protocol is improved by $\approx 0.5$ dB.
This improvement arises from two factors.
First, logical bit- and phase-flip errors in preparing the building blocks of the resource state are reduced by the techniques using the post-selection at the earlier concatenation levels, that is, the HRM and the use of the $7$-qubit code as the error-detecting code.
Second, throughout the preparation of the resource state for our MBQC protocol, the variances of GKP qubits are maintained small by the optimized SQEC algorithm with the variance adjustment that we have developed in Sec.~\ref{sec:optimzed_sqec},
and also by our protocol design for reducing variance accumulation in preparing the resource state deterministically from the building blocks;
in particular, the protocol in Ref.~\cite{F6} applies \textit{two} potentially noisy $CZ$ gates to each GKP qubit to obtain the resource $3$-dimensional cluster state, but our protocol applies only \textit{one} to each in the measurement patterns of Fig.~\ref{fig:blocks}.
Consequently, our protocol can achieve the better threshold and at the same time can control the trade-off between the failure probability and the overhead of QEC,
serving as a starting point of future research on optimizing the parameters of the protocol to balance the fault tolerance and the overhead based on advances of quantum technologies in experiments.

\section{\label{sec:conclusion}Discussion and Conclusion}

We have constructed a protocol for measurement-based quantum computation (MBQC) that can implement any $N$-qubit $D$-depth quantum circuit composed of a computationally universal gate set, Hadamard ($H$) and controlled-controlled-$Z$ ($CCZ$) gates, at a poly-logarithmic (polylog) overhead cost in terms of $N$ and $D$, even if no geometrical constraint is imposed on the $CCZ$ gates in the original circuit to be implemented.
This polylog-overhead MBQC protocol is advantageous over the existing MBQC protocol incurring polynomial overhead in implementing the geometrically nonlocal circuit as discussed in Sec.~\ref{sec:introduction}.
Furthermore, we have developed a fault-tolerant MBQC protocol for implementing this MBQC protocol on photonic architectures at a polylog overhead cost including quantum error correction (QEC), where the threshold of QEC in our protocol outperforms the existing state-of-the-art protocol~\cite{F6} for continuous-variable (CV) quantum computation\@.
Our protocols can be implemented by homodyne detection and Gaussian entangling operations with light sources that can emit a one-qubit computational-basis state $\Ket{0}$ or a magic state $\Ket{\frac{\pi}{8}}$ encoded in optical modes using the Gottesman-Kitaev-Preskill (GKP) code.
Our MBQC protocols establish a new way for implementing quantum algorithms at a polylog overhead cost in a fault-tolerant way, where we can realize not only exponential but also any polynomial quantum speedups, without canceling out the speedups.

To optimize the overhead cost of MBQC in Sec.~\ref{sec:resource_state}, we discover a structure of multipartite quantum entanglement for implementing an arbitrary qubit permutation within a polylog overhead, introducing the idea of the sorting networks to MBQC\@.
The entanglement structures that we have introduced are applicable not only to our resource state but to a wide class of other resource states for MBQC, as discussed in Remark~\ref{rem:generality}.
In recent theoretical progress on MBQC, useful entanglement structures for MBQC, including the hypergraph states~\cite{M2}, have found in states of many-body quantum systems characterized by local-interaction Hamiltonians that are potentially relevant to condensed matter physics, but geometrical constraints of the interactions inevitably cause polynomial overheads in implementing quantum computation.
We open up another perspective on MBQC\@; rather than obtaining the entanglement structures from Hamiltonians with local interactions, we arbitrarily design as optimized entanglement structures as we can.
This discloses what complexity is allowed in MBQC within the law of quantum mechanics and bridges a gap between the theoretical progress on MBQC and the experimental implementation of MBQC especially by means of photonic architectures.
In this direction, it remains an open question whether it is possible or not to achieve constant-overhead MBQC, \textit{i.e.}, a further improvement of our results,
while constant-space-overhead QEC has recently been attracting growing interest~\cite{F9,Gottesman2013} as another aspect of reducing overhead in implementing quantum computation.

To achieve the improved threshold of QEC in Sec.~\ref{sec:fault_tolerant}, we optimize the conventionally used single-qubit QEC for reducing variances of a GKP-encoded qubit by introducing variance-adjustment techniques, and also develop techniques for post-selecting high-fidelity building-block states that can be transformed into our resource state for MBQC\@.
Using our protocol concatenating the GKP code with the $7$-qubit code,
we can prepare any multiqubit entangled state at a logical level, and hence can implement MBQC using any multiqubit resource state in a fault-tolerant way, in contrast to existing fault-tolerant MBQC protocols using specific resource states~\cite{R8,R7,R5,R6,N1,N5}.
The $7$-qubit code has been gaining in importance due to its potential applicability to small- and intermediate-scale quantum computation at a few concatenation levels by means of flag qubits~\cite{C3,C4,R9,Chao2019}.
In contrast, the QEC techniques that we have introduced indicate that also on the large scale, the performance of QEC using the $7$-qubit code in MBQC can outperform that of topological QEC using the surface code combined with the GKP code~\cite{F2,F6,N3,V4,Hanggli2020} in terms of the threshold.
The threshold of our protocol is $7.8$~dB in terms of the required squeezing level of the GKP code, outperforming $8.3$~dB of the best existing protocol~\cite{F6} in the same noise model.
Comparable to this threshold, experiments on generating GKP codewords have demonstrated $7.4$--$9.5$ dB using a superconducting cavity~\cite{C8} and $5.5$--$7.3$ dB using a trapped-ion mechanical oscillator~\cite{F7},
while it remains a technological challenge to transform the GKP codewords of these matter-based systems into those of flying qubits.
Another possible way to realize the GKP code in optical modes is all-optical GKP state preparation using photon-number-resolving (PNR) detectors~\cite{Eaton2019,Tzitrin2019,Yamasaki2019}, circumventing the state transformation from matter to the flying CV mode.
As the squeezing level of the GKP code realizable in experiments gets larger,
we can reduce the constant factor in the overhead of QEC by controlling a parameter of our fault-tolerant MBQC protocol, while further research is needed to optimize the parameters of the protocol based on advances in quantum technologies.
Our results open up a route of taking flight from the geometrically constrained matter that can prepare GKP codewords, and of advancing techniques for implementing MBQC using flying photonic systems, towards establishing architectures for implementing quantum computation without canceling out not only exponential but also polynomial quantum speedups.

\acknowledgments{This work was supported by CREST (Japan Science and Technology Agency) JPMJCR1671 and Cross-ministerial Strategic Innovation Promotion Program (SIP) (Council for Science, Technologyand Innovation (CSTI)). KF acknowledges financial support by donations from Nichia Corporation. YT thanks Yasuhiro Takahashi for helpful discussions. YT is supported by MEXT Quantum Leap Flagship Program (MEXT Q-LEAP) Grant Number JPMXS0118067394.}

\appendix

\section{\label{sec:depth}Constant-depth resource state preparation}

In this appendix, we show that there exists a constant-depth quantum circuit for preparing the hypergraph state $\Ket{G_\mathrm{res}^{\left(N,D\right)}}$ for any $N$ and $D$ represented as a hypergraph $G_\mathrm{res}^{\left(N,D\right)}$ shown in Sec.~\ref{sec:resource_state}, which serves as a resource for the MBQC protocol in Theorem~\ref{thm:universality}.
In particular, we show that any part of $\Ket{G_\mathrm{res}^{\left(N,D\right)}}$ corresponding to a sub-hypergraph of $G_\mathrm{res}^{\left(N,D\right)}$ can be prepared by a constant-depth quantum circuit.
In the context of the photonic MBQC, this constant-depth preparation is preferable because the loss caused by an optical fiber can be upper bounded by a constant.

Given a hypergraph state (including a graph state as a special case), the \textit{preparation depth} refers to the required depth of the circuit for preparing the state, which is characterized by edge coloring~\cite{B22} of the hypergraph representing the hypergraph state.
In the context of the edge coloring, a hypergraph $G=\left(V,E\right)$ is said to be $k$-edge-colorable if $k$ different colors can be assigned to its hyperedges in such a way that no vertex is incident with two different hyperedges $e_1,e_2\in E$ of the same color; \textit{i.e.},
\begin{equation}
    \nexists v\in V\; \text{such that } v\in e_1\; \text{and }v\in e_2.
\end{equation}
Note that given a hypergraph $G$ in general, it is \textsf{NP}-hard (and hence computationally hard) to compute the optimal (\textit{i.e.}, minimal) $k$ such that $G$ is $k$-edge-colorable~\cite{B22}.
A hypergraph state corresponding to a $k$-edge-colorable hypergraph can be prepared by a $(k+1)$-depth quantum circuit from $\Ket{0}\otimes\Ket{0}\otimes\cdots$ consisting of the Hadamard gate and the generalized controlled-$Z$ gate, where the first of the $k+1$ depths corresponds to converting $\Ket{0}\otimes\Ket{0}\otimes\cdots$ to $\Ket{+}\otimes\Ket{+}\otimes\cdots$ by the Hadamard gates, and each of the other $k$ depths corresponds to applying the generalized controlled-$Z$ gates represented as hyperedges having the same color.
In particular, if a hypergraph corresponding to an $M$-qubit hypergraph state is $K$-edge-colorable where $K$ is a constant that is independent of $M$, the preparation depth of this hypergraph state is constant even if $M$ is arbitrarily large.

We show that $\Ket{G_\mathrm{res}^{\left(N,D\right)}}$ can be prepared by a $6$-depth quantum circuit from $\Ket{0}\otimes\Ket{0}\otimes\cdots$.
In the proof, we exploit the periodicity of $G_\mathrm{res}^{\left(N,D\right)}$ for proving that the hypergraph $G_\mathrm{res}^{\left(N,D\right)}$ is $5$-edge-colorable for any $N$ and $D$, while not $4$-edge-colorable, which yields the following.

\begin{theorem}[\label{thm:preparation_depth}Constatnt-depth preparation of resouce hypergraph states]
    For any $N$, $D$, any sub-hypergraph of the hypergraph $G_\mathrm{res}^{\left(N,D\right)}$ defined as~\eqref{eq:g_res} is $5$-edge-colorable but not necessarily $4$-edge-colorable,
    and if we use a gate set $\left\{H,S,CZ,CCZ\right\}$ as shown in~\eqref{eq:H_S_CZ_CCZ},
    a hypergraph state corresponding this sub-hypergraph can be prepared from a fixed product state $\Ket{0}\otimes\Ket{0}\otimes\cdots$ by a $6$-depth quantum circuit.
\end{theorem}

\begin{figure}[t]
    \centering
    \includegraphics[width=3.4in]{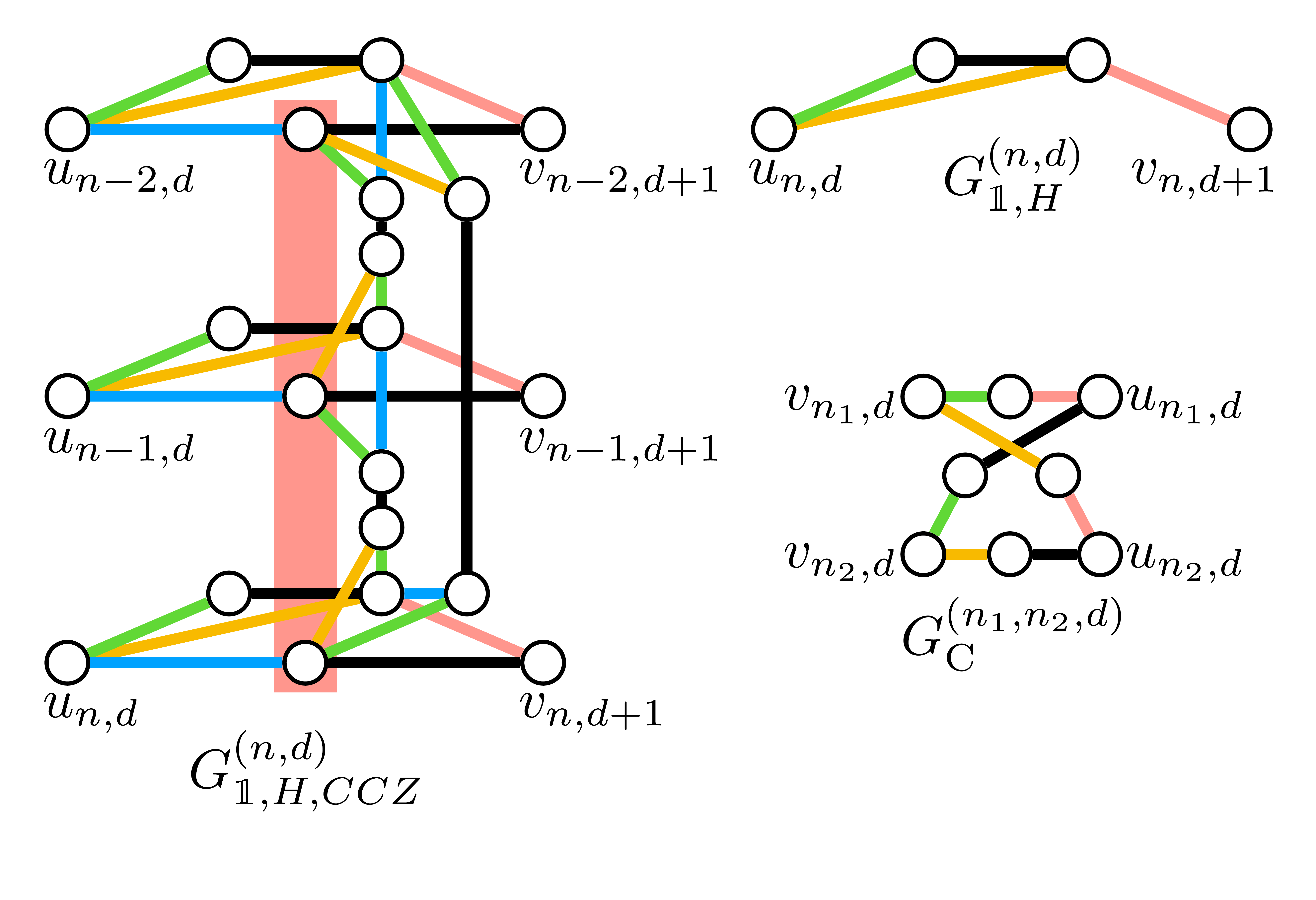}
    \caption{\label{fig:edge_coloring}Edge coloring of sub-hypergraphs of the hypergraph $G_\mathrm{res}^{\left(N,D\right)}$ in Fig.~\ref{fig:resource}. Hyperedges of these hypergraphs $G_{\mathbbm{1},H,CCZ}^{\left(n,d\right)}$, $G_{\mathbbm{1},H}^{\left(n,d\right)}$, and $G_\textup{C}^{\left(n_1,n_2,d\right)}$ in Fig.~\ref{fig:parts} can be colored using at most five colors, namely, green, yellow, blue, pink, and black, and due to the periodicity of $G_\mathrm{res}^{\left(N,D\right)}$, we show in the main text that hyperedges of $G_\mathrm{res}^{\left(N,D\right)}$ can also be colored using these five colors. As shown in Protocol~\ref{alg:const_depth_preparation}, this edge coloring provides a protocol represented by a $6$-depth quantum circuit for preparing the corresponding resource hypergraph state, which serves as a resource for our MBQC protocol in Theorem~\ref{thm:universality}.}
\end{figure}

\begin{algorithm*}[t]
  \caption{\label{alg:const_depth_preparation}Preparation of the hypergraph state $\Ket{G_\mathrm{res}^{\left(N,D\right)}}$ based on the edge coloring of the hypergraph $G_\mathrm{res}^{\left(N,D\right)}$ obtained from Fig.~\ref{fig:edge_coloring}.}
  \begin{algorithmic}[1]
    \State{Apply the Hadamard gate $H$ to each of the $M$ qubits initialized as $\Ket{0}$ in parallel to obtain $\Ket{+}^{\otimes M}$, where $M$ is the number of qubits for the hypergraph state $\Ket{G_\mathrm{res}^{\left(N,D\right)}}$.}
    \For{$c\in\left\{\text{green, yellow, blue, pink, black}\right\}$, where $c$ represents each of the five colors used in the edge coloring of $G_\mathrm{res}^{\left(N,D\right)}$,}
    \Statex{\Comment{Repeated five times in total.}}
    \State{In parallel apply $CZ$ and $CCZ$ gates corresponding to the hyperedges of $G_\mathrm{res}^{\left(N,D\right)}$ colored $c$.}
    \EndFor{}
  \end{algorithmic}
\end{algorithm*}

\begin{proof}[\textbf{Proof}]
  Since the maximum degree of the vertices of the hypergraph $G_\mathrm{res}^{\left(N,D\right)}$, which is $5$, implies that $G_\mathrm{res}^{\left(N,D\right)}$ cannot be $4$-edge-colorable, it suffices to show that $G_\mathrm{res}^{\left(N,D\right)}$ is $5$-edge-colorable.
  Correspondingly, we obtain Protocol~\ref{alg:const_depth_preparation} for preparing $\Ket{G_\mathrm{res}^{\left(N,D\right)}}$, which can be represented by a $6$-depth quantum circuit.
  Note that it suffices to consider preparation of the whole state $\Ket{G_\mathrm{res}^{\left(N,D\right)}}$ since a circuit for this preparation can also be used for any hypergraph state corresponding to a sub-hypergraph of $G_\mathrm{res}^{\left(N,D\right)}$.

  To show that $G_\mathrm{res}^{\left(N,D\right)}$ is $5$-edge-colorable,
  observe the periodicity of $G_\mathrm{res}^{\left(N,D\right)}$ composed of three fixed hypergraphs $G_{\mathbbm{1},H,CCZ}^{\left(n,d\right)}$, $G_{\mathbbm{1},H}^{\left(n,d\right)}$, and $G_\textup{C}^{\left(n_1,n_2,d\right)}$ shown in Fig.~\ref{fig:parts}.
  Hyperedges of these hypergraphs $G_{\mathbbm{1},H,CCZ}^{\left(n,d\right)}$, $G_{\mathbbm{1},H}^{\left(n,d\right)}$, and $G_\textup{C}^{\left(n_1,n_2,d\right)}$ can be colored by five colors, namely, green, yellow, blue, pink, and black, as illustrated in Fig.~\ref{fig:edge_coloring}.
  In this edge coloring, edges of $G_{\mathbbm{1},H,CCZ}^{\left(n,d\right)}$ with which $u_{n-2,d}$, $u_{n-1,d}$, and $u_{n,d}$ are incident are colored green, yellow, and blue, and those with which $v_{n-2,d+1}$ $v_{n-1,d+1}$, and $v_{n,d+1}$ are incident are colored pink and black.
  In the same way, edges of $G_{\mathbbm{1},H}^{\left(n,d\right)}$ with which $u_{n,d}$ is incident are colored green, yellow, and blue, and those with which $v_{n,d+1}$ is incident are colored pink and black.
  As for $G_\textup{C}^{\left(n_1,n_2,d\right)}$, edges of $G_\textup{C}^{\left(n_1,n_2,d\right)}$ with which $v_{n_1,d}$ and $v_{n_2,d}$ are incident are colored green and yellow, and those with which $u_{n_1,d}$ and $u_{n_2,d}$ are incident are colored pink and black.
  By repeating $G_\textup{C}^{\left(n_1,n_2,d\right)}$ having this edge coloring to construct $G_\textup{S}^{\left(N,d\right)}$ in the same way as shown in~\eqref{eq:G_S_N_d} and Fig.~\ref{fig:sort_graph},
  we obtain an edge coloring of $G_\textup{S}^{\left(N,d\right)}$ where edges with which $v_{1,d},\ldots,v_{N,d}$ are incident can be colored green and yellow, and those with which $u_{1,d},\ldots,u_{N,d}$ are incident can be colored pink and black.
  By combining these edge colorings of $G_{\mathbbm{1},H,CCZ}^{\left(n,d\right)}$, $G_{\mathbbm{1},H}^{\left(n,d\right)}$, and $G_\textup{S}^{\left(N,d\right)}$ to construct $G_\textup{depth}^{\left(N,d\right)}$ in the same way as shown in~\eqref{eq:one_depth} and Fig.~\ref{fig:one_depth},
  we obtain an edge coloring of $G_\textup{depth}^{\left(N,d\right)}$ where edges with which $v_{1,d},\ldots,v_{N,d}$ are incident are colored green and yellow, and those with which $v_{1,d+1},\ldots,v_{N,d+1}$ are incident are colored pink and black.
  By repeating $G_\textup{depth}^{\left(N,d\right)}$ with this edge coloring to construct $G_\mathrm{res}^{\left(N,D\right)}$ in the same way as shown in~\eqref{eq:g_res} and Fig.~\ref{fig:resource},
  we obtain an edge coloring of $G_\mathrm{res}^{\left(N,D\right)}$ by the five colors, which proves that $G_\mathrm{res}^{\left(N,D\right)}$ is $5$-edge-colorable.

  This edge coloring of $G_\mathrm{res}^{\left(N,D\right)}$ yields  Protocol~\ref{alg:const_depth_preparation} for preparing $\Ket{G_\mathrm{res}^{\left(N,D\right)}}$.
  This protocol consists of $6$ computational depths, and hence can be represented by a $6$-depth quantum circuit.
\end{proof}

\section{\label{sec:parallelizability}Parallelizability}

\begin{figure*}[t]
    \centering
    \includegraphics[width=7.0in]{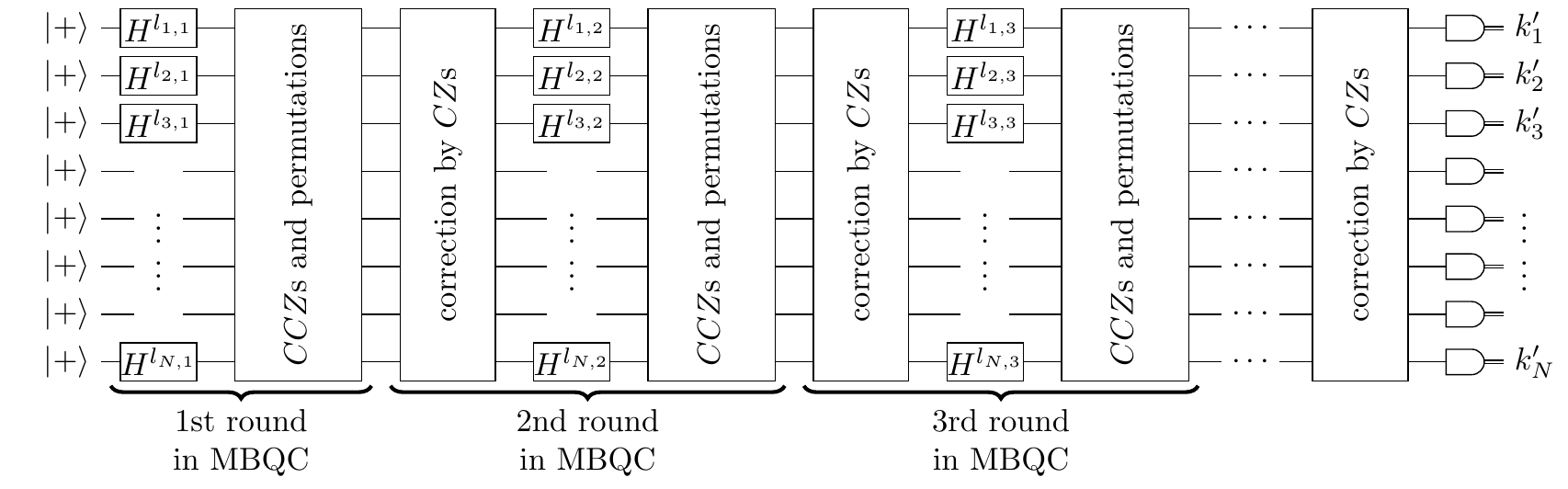}
    \caption{\label{fig:parallelizablity}A quantum circuit representing implementation of a quantum circuit by MBQC using parallel measurements of the resource hypergraph state $\Ket{G_\mathrm{res}^{\left(N,D^\prime\right)}}$ defined as~\eqref{eq:resources} and shown in Fig.~\ref{fig:resource}. In the circuit, $l_{n,d}\in\left\{0,1\right\}$ represents whether the $H$ gate on the $n$th qubit in the depth corresponding to the $d$th measurement round in MBQC is applied or not. In particular, the $H$ gate is applied if $l_{n,d}=1$, and not if $l_{n,d}=0$. We can rewrite an arbitrary $N$-qubit $D$-depth quantum circuit composed of a gate set $\left\{H,CCZ\right\}$ into this form. The required number of rounds of parallel measurements in this MBQC protocol can be smaller than $D$, while the required number of qubits for $\Ket{G_\mathrm{res}^{\left(N,D^\prime\right)}}$ can be larger; \textit{i.e.}, it may hold that $D^\prime> D$, which exhibits the trade-off between spatial and temporal resources in MBQC\@.}
\end{figure*}

We discuss parallelizability of MBQC protocols using our construction of the resource hypergraph state $\Ket{G_\mathrm{res}^{\left(N,D\right)}}$ shown in Sec.~\ref{sec:resource_state}.
In contrast to the MBQC protocol shown in Theorem~\ref{thm:universality} where qubits for $\Ket{G_\mathrm{res}^{\left(N,D\right)}}$ are measured sequentially, we here consider a situation where measurements on different qubits in MBQC can be performed in parallel, as long as the measurements in each measurement round are not conditioned on the outcomes of any other measurements that are performed simultaneously in the same round or will be performed in future rounds.
In this situation of parallel MBQC, there can be a trade-off relation between spatial and temporal resources as we recall in Sec.~\ref{sec:def_parallelizability}.
We discuss in Sec.~\ref{sec:parallel} how this trade-off arises in the case of MBQC using $\Ket{G_\mathrm{res}^{\left(N,D\right)}}$ as a resource.

\subsection{\label{sec:def_parallelizability}Definition of spatial and temporal resources in MBQC}

In MBQC, there are two types of computational resources analogous to those in the context of the space-time trade-off, as discussed in Ref.~\cite{M8}.
One is the number of qubits for a resource state to simulate any $N$-qubit $D$-depth quantum circuit, called a spatial resource, and the other is the required rounds of parallel measurements for MBQC to simulate the circuit, called a temporal resource.
These spatial and temporal resources are in a trade-off relation in the sense that in some cases, protocols for MBQC may use a larger number of qubits than those for a quantum circuit to be simulated,
so that the rounds of parallel measurements in MBQC can be smaller than the depth of the circuit to be simulated~\cite{B5}.
This compression of the measurement rounds in MBQC is possible because a part of quantum circuit can be simulated in MBQC without adaptively changing measurement bases, and corrections of byproduct operators can be collectively performed later.
For example, MBQC using the cluster state for simulating a quantum circuit composed of Clifford gates and a one-qubit non-Clifford gate can simulate any part of the circuit consisting only of Clifford gates without any adaptation of measurement bases, and hence, Clifford gates can be parallelized in this case~\cite{R2}.
Similar parallelizability is also shown for MBQC using a hypergraph state for simulating circuits consisting of $\left\{H,CCZ\right\}$, where $CCZ$ gates can be parallelized~\cite{G4}.
Parallelizability in MBQC is useful for potentially decreasing the required time steps for obtaining outcomes of a computation~\cite{B4,H1,G4} compared to the circuit model, while parallelizable gates depend on resource states for MBQC\@.

\subsection{\label{sec:parallel}Parallelizability and trade-off between spatial and temporal resources in MBQC}

We show a trade-off between spatial and temporal resources in MBQC using $\Ket{G_\mathrm{res}^{\left(N,D\right)}}$.
As discussed in Sec.~\ref{sec:def_parallelizability},
MBQC may exhibit a trade-off between the spatial and temporal resources, \textit{i.e.}, the number of qubits for resources and the number of rounds of measurements in MBQC protocols using these resources.
In cases of MBQC using $\Ket{G_\mathrm{res}^{\left(N,D\right)}}$, we show such a trade-off in parallel implementation of multiple $CCZ$ gates.

While the MBQC protocol that we show in the proof of Theorem~\ref{thm:universality} corrects byproducts after each measurement,
some of these corrections can be delayed and collectively performed later, so that multiple measurements can be performed in parallel without any adaptation conditioned on the other measurement outcomes.
To use the parallel measurements in MBQC for simulating a given $N$-qubit $D$-depth quantum circuit composed of $\left\{H,CCZ\right\}$, instead of rewriting this circuit into the form on the right-hand side in Fig.~\ref{fig:circuit_n_qubit_d_depth} that is used for the MBQC protocol in Theorem~\ref{thm:universality},
we rewrite it into a form shown in Fig.~\ref{fig:parallelizablity}, which consists of parts composed only of the $H$ gates and parts composed only of the $CCZ$ gates.

When $CCZ$ gates are implemented in MBQC using $\Ket{G_\mathrm{res}^{\left(N,D\right)}}$ and the measurement pattern labeled ``$CCZ$ and $CZ$s'' in Fig.~\ref{fig:measurement}, it is possible to delay the corrections of $CZ$ byproducts by performing $Z$ measurements for removing all the qubits represented as the dashed vertices in Fig.~\ref{fig:measurement}, and in such a case, the corrections are performed later by the measurement pattern for adaptively implementing $CZ$ gates, which is labeled ``$CZ$s'' in Fig.~\ref{fig:measurement}.
While the set of local Pauli byproduct operators generated by $\left\{X,Z\right\}$ is closed under the conjugation with $H$ in the sense of the group theory
\begin{align}
    \label{eq:h_z}
    HXH&=Z,\\
    \label{eq:h_x}
    HZH&=X,
\end{align}
the set of all the possible byproduct operators generated by $\left\{X,Z,CZ\right\}$ is not,
and hence, the $CZ$ byproducts are needed to be corrected before applying $H$ gates.
Thus, a part of the circuit composed only of $CCZ$ gates, which can be of an arbitrarily large depth and may include geometrically nonlocal gates that are implemented as geometrically local gates combined with permutations, can be simulated in MBQC using $\Ket{G_\mathrm{res}^{\left(N,D\right)}}$ by only one round of parallel measurements for implementing $CCZ$ gates as well as permutations.
After implementing this part composed of the $CCZ$ gates,
we perform a collective correction of $CZ$ byproducts in the next round of parallel measurements for implementing $CZ$ gates that may be geometrically nonlocal and hence are implemented as geometrically local $CZ$ gates combined with permutations.
Note that the permutations can be implemented in the same way as the MBQC protocol in Theorem~\ref{thm:universality} using a sorting network.
Due to commutation relations shown in~\eqref{eq:cz_z},~\eqref{eq:cz_x},~\eqref{eq:h_z}, and~\eqref{eq:h_x}, corrections of one-qubit Pauli byproducts in implementing $CZ$ and $H$ gates can also be delayed.
Hence, a part of the circuit composed only of $H$ gates after the part composed only of the $CCZ$ gates can be simulated by MBQC in the same round of parallel measurements as the round for correcting $CZ$ byproducts.

Therefore, in rewriting the given circuit into the form shown in Fig.~\ref{fig:parallelizablity},
each layer of the rewritten circuit consists of the corrections of the $CZ$ byproducts conditioned on measurement outcomes in the previous round of implementing the $CCZ$ gates, and the implementations of the $H$ gates, followed by the implementations of the $CCZ$ gates, where MBQC uses one round of parallel measurements per implementation of each layer.
Then using this rewritten circuit in Fig.~\ref{fig:parallelizablity}, we can simulate the original $D$-depth quantum circuit by MBQC that may require a smaller number of rounds of parallel measurements than $D$, where the number of these rounds is determined by that of the layers in the rewritten circuit as shown in Fig.~\ref{fig:parallelizablity}.
Note that similar parallelization of $CCZ$ gates and $\textsc{SWAP}$ gates is also possible in MBQC using other hypergraph states than $\Ket{G_\mathrm{res}^{\left(N,D\right)}}$~\cite{G4}.

As for the required number of resource qubits for this parallelized simulation of the $N$-qubit $D$-depth quantum circuits,
the MBQC protocol in this section may require a hypergraph state $\Ket{G_\mathrm{res}^{\left(N,D^\prime\right)}}$ for some $D^\prime>D$.
This hypergraph state is larger than $\Ket{G_\mathrm{res}^{\left(N,D\right)}}$ in terms of the number of qubits, where the extra qubits corresponding to $\left(D^\prime-D\right)$ repetitions of $\Ket{G_\textup{depth}^{\left(N,d\right)}}$ are used for corrections of $CZ$ byproducts.
Note that in terms of the complexity of MBQC using sequential measurements, $D^\prime$ can be too large to achieve the poly-logarithmic overhead, while we assume in this section that an arbitrary number of measurements can be performed in parallel.
In this regard, this parallelized MBQC protocol requires more spatial resources, that is, the number of the qubits, than that in Theorem~\ref{thm:universality}, for reducing temporal resources in terms of the number of rounds of parallel measurements, exhibiting the trade-off between spatial and temporal resources in MBQC using $\Ket{G_\textup{depth}^{\left(N,d\right)}}$.

\section{\label{sec:verifiability}Vertex coloring}

\begin{figure}[t]
    \centering
    \includegraphics[width=3.4in]{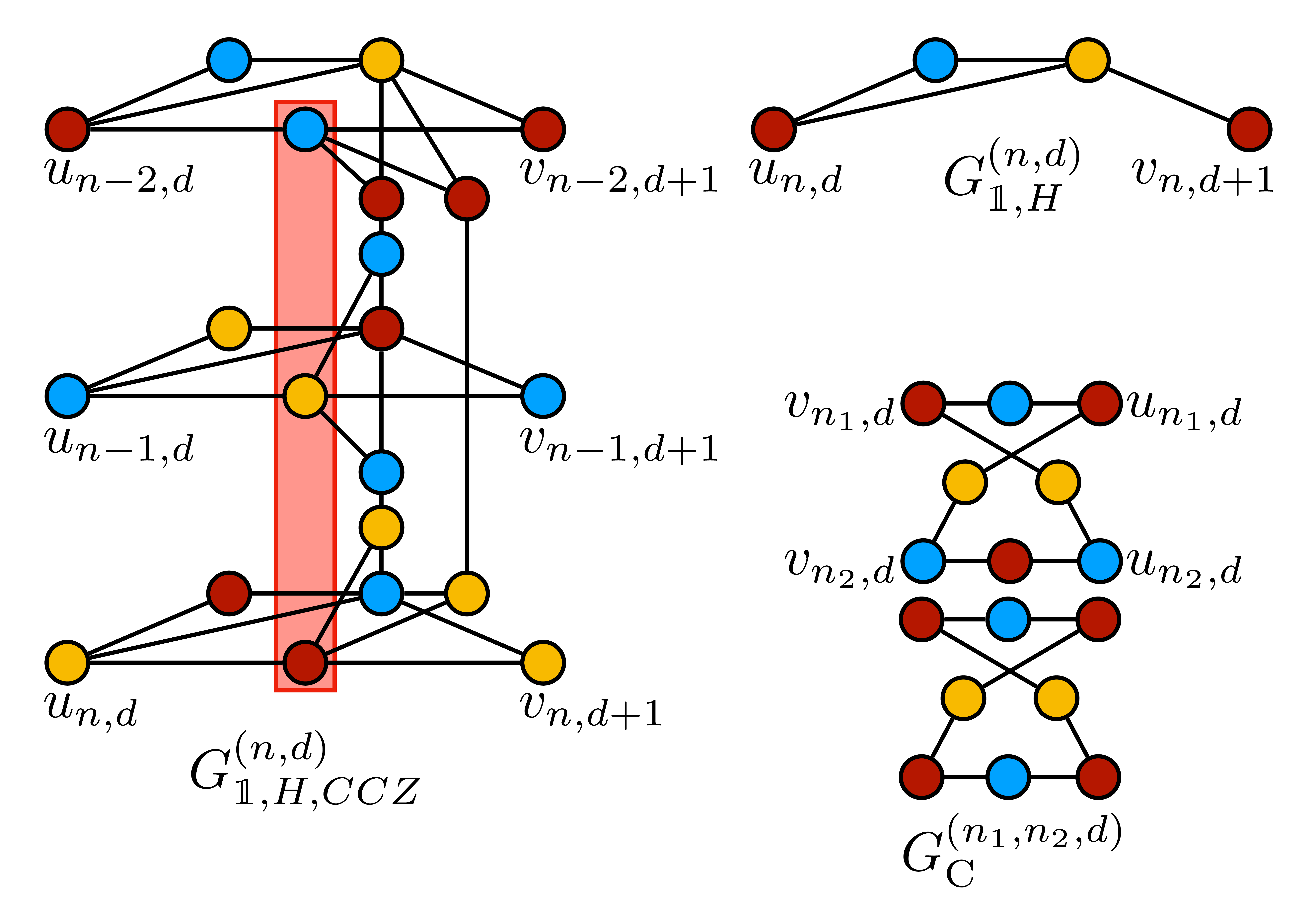}
    \caption{\label{fig:vertex_coloring}Vertex coloring of sub-hypergraphs of the hypergraph $G_\mathrm{res}^{\left(N,D\right)}$ in Fig.~\ref{fig:resource}. Vertex coloring is used in verification of MBQC as discussed in Remark~\ref{rem:verifiability}. Vertices of these hypergraphs $G_{\mathbbm{1},H,CCZ}^{\left(n,d\right)}$, $G_{\mathbbm{1},H}^{\left(n,d\right)}$, and $G_\textup{C}^{\left(n_1,n_2,d\right)}$ in Fig.~\ref{fig:parts} can be colored using three colors, namely, red, blue, and yellow.
      The figure shows a vertex coloring of $G_{\mathbbm{1},H}^{\left(n,d\right)}$ with $u_{n,d}$ and $v_{n,d+1}$ colored red, but in the same way as this vertex coloring, we can obtain vertex colorings of $G_{\mathbbm{1},H}^{\left(n,d\right)}$ with $u_{n,d}$ and $v_{n,d+1}$ colored in the same color, for any color such as blue and yellow.
    In the same way as the vertex colorings of $G_\textup{C}^{\left(n_1,n_2,d\right)}$ in the figure, we can obtain vertex colorings of $G_\textup{C}^{\left(n_1,n_2,d\right)}$ with $v_{n_1,d}$ and $u_{n_1,d}$ colored in the same color and with $v_{n_2,d}$ and $u_{n_2,d}$ colored in the same color, for any colors regardless of whether these two colors are the same or not. Then in the main text, using the periodicity of $G_\mathrm{res}^{\left(N,D\right)}$, we show the vertex coloring of $G_\mathrm{res}^{\left(N,D\right)}$ with these three colors.}
\end{figure}

In this appendix, we show that three colors suffice for the vertex coloring~\cite{B22} of the hypergraph $G_\mathrm{res}^{\left(N,D\right)}$ shown in Sec.~\ref{sec:resource_state} for any $N$ and $D$.
The vertex coloring is used for the verification of MBQC as discussed in Remark~\ref{rem:verifiability}.
A hypergraph $G=\left(V,E\right)$ is said to be $k$-vertex-colorable, or $k$-colorable for short, if $k$ different colors can be assigned to its vertices so that no two vertices $v_1$ and $v_2$ of the same color are the neighbor of each other; \textit{i.e.},
\begin{equation}
  \nexists e\in E\; \text{such that } v_1\in e\; \text{and }v_2\in e.
\end{equation}
The minimum $k$ for which $G$ is $k$-colorable is called the chromatic number of $G$, denoted by $\chi$.
Note that given a hypergraph $G$, it is \textsf{NP}-hard in general (and hence computationally hard) to compute the chromatic number $\chi$ of $G$~\cite{B22}.

We here show that for any $N$ and $D$, the chromatic number $\chi$ of $G_\mathrm{res}^{\left(N,D\right)}$ is $\chi=3$.
The proof exploits the periodicity of $G_\mathrm{res}^{\left(N,D\right)}$ for constructing a vertex coloring of $G_\mathrm{res}^{\left(N,D\right)}$ using three colors, while $G_\mathrm{res}^{\left(N,D\right)}$ cannot be $2$-colorable, which yields the following.

\begin{theorem}[Vertex coloring of hypergraphs for resource hypergraph states]
    For any $N$ and $D$, the chromatic number $\chi$ of the hypergraph $G_\mathrm{res}^{\left(N,D\right)}$ defined as~\eqref{eq:g_res} is
    \begin{equation}
      \chi=3.
    \end{equation}
\end{theorem}

\begin{proof}[\textbf{Proof}]
    It suffices to prove that $G_\mathrm{res}^{\left(N,D\right)}$ is $3$-colorable, that is, $\chi\leqq 3$, since the hyperedges of $G_\mathrm{res}^{\left(N,D\right)}$ among three vertices yields $\chi\geqq 3$.
    We show a construction of the vertex coloring of $G_\mathrm{res}^{\left(N,D\right)}$ using three colors, namely, red, blue, and yellow, where for any $n=3m$, the vertices labeled 
    \begin{equation}
      v_{3m-2,1},u_{3m-2,1},v_{3m-2,2},u_{3m-2,2},\ldots,v_{3m-2,D+1}
    \end{equation}
    are colored red,
    \begin{equation}
      v_{3m-1,1},u_{3m-1,1},v_{3m-1,2},u_{3m-1,2},\ldots,v_{3m-1,D+1}
    \end{equation}
    are colored blue, and
    \begin{equation}
      v_{3m,1},u_{3m,1},v_{3m,2},u_{3m,2},\ldots,v_{3m,D+1}
    \end{equation}
    are colored yellow.
    In the same way as the proof of Theorem~\ref{thm:preparation_depth},
    recall that $G_\mathrm{res}^{\left(N,D\right)}$ is composed of three fixed hypergraphs $G_{\mathbbm{1},H,CCZ}^{\left(n,d\right)}$, $G_{\mathbbm{1},H}^{\left(n,d\right)}$, and $G_\textup{C}^{\left(n_1,n_2,d\right)}$ shown in Fig.~\ref{fig:parts}.
    Vertices of $G_{\mathbbm{1},H,CCZ}^{\left(n,d\right)}$, $G_{\mathbbm{1},H}^{\left(n,d\right)}$, and $G_\textup{C}^{\left(n_1,n_2,d\right)}$ can be colored by three colors, namely, yellow, blue, and red, as illustrated in Fig.~\ref{fig:vertex_coloring}.
    In the vertex coloring of $G_{\mathbbm{1},H,CCZ}^{\left(n,d\right)}$ in Fig.~\ref{fig:vertex_coloring}, the vertices are colored so that $v_{3m-2,d}$ and $u_{3m-2,d+1}$ can be colored red, $v_{3m-1,d}$ and $u_{3m-1,d+1}$ colored blue, and $v_{3m,d+1}$ and $u_{3m,d+1}$ colored yellow, for any $m$ and $d$.
    The vertex coloring of $G_{\mathbbm{1},H}^{\left(n,d\right)}$ in Fig.~\ref{fig:vertex_coloring} yields those with $u_{n,d}$ and $v_{n,d+1}$ colored in the same color, for any $n$, $d$, and color.
    The vertex colorings of $G_{\mathrm{C}}^{\left(n_1,n_2,d\right)}$ in Fig.~\ref{fig:vertex_coloring} shows those with $v_{n_1,d}$ and $u_{n_1,d}$ colored in the same color and with $v_{n_2,d}$ and $u_{n_2,d}$ colored in the same color, for any $n_1$, $n_2$, $d$, and colors, where these two colors can either be the same or different.
    Then, using the periodicity of $G_\mathrm{res}^{\left(N,D\right)}$ in the same way as the edge coloring in the proof of Theorem~\ref{thm:preparation_depth},
    we obtain a vertex coloring of $G_\mathrm{res}^{\left(N,D\right)}$ by the three colors, which proves that $G_\mathrm{res}^{\left(N,D\right)}$ is $3$-colorable.
\end{proof}

\section{\label{sec:fault_tolerance_preparation}Fault tolerance of state preparation}

\begin{figure*}[t]
    \centering
    \includegraphics[width=7.0in]{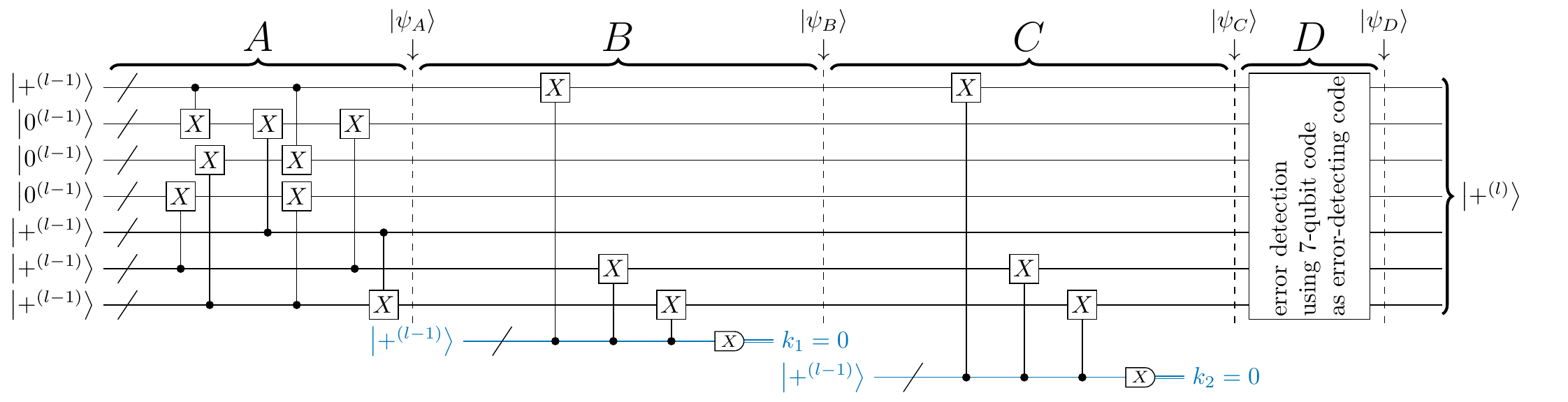}
    \caption{\label{fig:preparation_7qubit_explanation}A quantum circuit equivalent to the circuit in Fig.~\ref{fig:preparation_7qubit_plus_optimized2} followed by the error detection using the 7-qubit code as the error-detecting code. We decompose the circuit into four parts, namely, $A$, $B$, $C$, and $D$, where the states obtained from the parts $A$, $AB$, $ABC$, and $ABCD$ of the circuit are denoted by $\Ket{\psi_A}$, $\Ket{\psi_B}$, $\Ket{\psi_C}$, and $\Ket{\psi_D}$, respectively.}
    \includegraphics[width=7.0in]{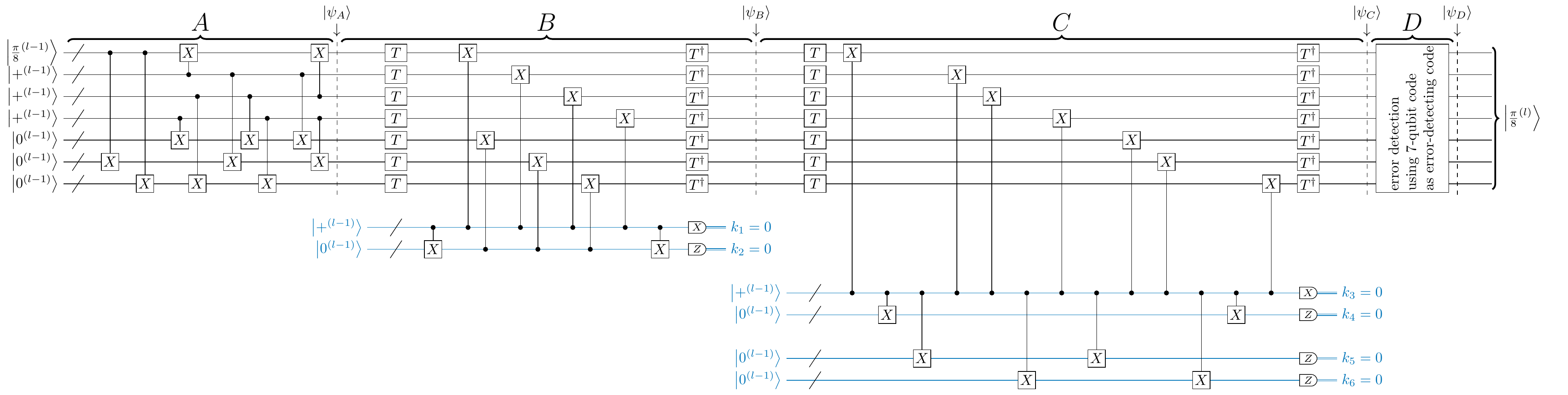}
    \caption{\label{fig:preparation_magic_explanation}A quantum circuit equivalent to the circuit in Fig.~\ref{fig:preparation_magic2} followed by the error detection using the 7-qubit code as the error-detecting code. We decompose the circuit into four parts, namely, $A$, $B$, $C$, and $D$, where the states obtained from the parts $A$, $AB$, $ABC$, and $ABCD$ of the circuit are denoted by $\Ket{\psi_A}$, $\Ket{\psi_B}$, $\Ket{\psi_C}$, and $\Ket{\psi_D}$, respectively. Note that the circuit in Fig.~\ref{fig:preparation_magic2} and that in the figure are equivalent since ${T^\dag}^{\otimes 7}$ in the part $B$ and ${T}^{\otimes 7}$ in the part $C$ cancel out due to $T^\dag T=\mathbbm{1}$.}
\end{figure*}

In this appendix, we analyze fault tolerance of preparation of $\Ket{+^{(l)}}$ using error detection twice (Fig.~\ref{fig:preparation_7qubit_plus_optimized2}), and that of preparation of $\Ket{\frac{\pi}{8}^{(l)}}$ using error detection twice (Fig.~\ref{fig:preparation_magic2}).
We show that the circuits in Figs.~\ref{fig:preparation_7qubit_plus_optimized2} and~\ref{fig:preparation_magic2} are $1$-fault-tolerant in terms of the error correction (Definition~\ref{def:fault_tolerance_correction}), and also have the fault-tolerant property for two errors in terms of the error detection (Definition~\ref{def:fault_tolerance_detection}).
In particular,
our proof will mainly focus on showing the fault-tolerant property for two errors in terms of the error detection,
since the $1$-fault tolerance in terms of the error correction can be shown in the same way as our analysis as well as the established arguments in Refs.~\cite{G6,Chamberland2019faulttolerantmagic}.
We prove the following two propositions.

\begin{proposition}
  [\label{prp:preparation_7qubit_plus_optimized2}Fault tolerance of preparation of $\Ket{+^{(l)}}$ using error detection twice]
  The circuit in Fig.~\ref{fig:preparation_7qubit_plus_optimized2} is $1$-fault-tolerant in terms of the error correction (Definition~\ref{def:fault_tolerance_correction}) and has the fault-tolerant property for two errors in terms of the error detection (Definition~\ref{def:fault_tolerance_detection}).
\end{proposition}

\begin{proposition}
[\label{prp:preparation_magic2}Fault tolerance of preparation of $\Ket{\frac{\pi}{8}^{(l)}}$ using error detection twice]
  The circuit in Fig.~\ref{fig:preparation_magic2} is $1$-fault-tolerant in terms of the error correction (Definition~\ref{def:fault_tolerance_correction}) and has the fault-tolerant property for two errors in terms of the error detection (Definition~\ref{def:fault_tolerance_detection}).
\end{proposition}

\begin{proof}[Proof of Proposition~\ref{prp:preparation_7qubit_plus_optimized2}]
  We analyze the fault-tolerant property for two errors in terms of the error detection,
  since the $1$-fault tolerance in terms of the error correction can be shown essentially in the same way as our analysis and Ref.~\cite{G6}.
  To analyze the circuit in Fig.~\ref{fig:preparation_7qubit_plus_optimized2} followed by the error detection using the 7-qubit code as the error-detecting code,
  we decompose the circuit into four parts, namely, $A$, $B$, $C$, and $D$ in Fig.~\ref{fig:preparation_7qubit_explanation},
  where the states obtained from the parts $A$, $AB$, $ABC$, and $ABCD$ of the circuit are denoted by $\Ket{\psi_A}$, $\Ket{\psi_B}$, $\Ket{\psi_C}$, and $\Ket{\psi_D}$, respectively.
  Since the $7$-qubit code is a Calderbank-Shor-Steane (CSS) code, the analyses of the fault-tolerant property for two $X$ errors and that for two $Z$ errors will be similar, and in the following, errors may refer to either $X$ and $Z$ errors unless stated otherwise.

  We consider the following exhaustive cases of two errors on $ABC$ in Fig.~\ref{fig:preparation_7qubit_explanation}, which corresponds to the circuit in Fig.~\ref{fig:preparation_7qubit_plus_optimized2}.
  \begin{enumerate}
    \item\label{case:1} Two errors occur in $A$ or $B$.
    \item\label{case:2} One error occurs in $A$ or $B$, and one error occurs in $C$.
    \item\label{case:4} Two errors occur in $C$.
  \end{enumerate}

  \textbf{Case~\ref{case:1}}:
  After the parts $A$ and $B$ with errors, we obtain a state $\Ket{\psi_B}$.
  Using stabilizers~\eqref{eq:syndromes} of the $7$-qubit code,
  we can reduce the $X$ and $Z$ errors in $\Ket{\psi_B}$ to a combination of a one-qubit $X$ error, a one-qubit $Z$ error, a logical $X$ error, and a logical $Z$ error.
  For $n\in\{1,\ldots,7\}$, let $X_n$ and $Z_n$ denote the one-qubit $X$ error and the one-qubit $Z$ error, respectively, on the $n$th qubit of the $7$-qubit code.
  Then, we can write $\Ket{\psi_B}$ suffering from the errors as
  \begin{equation}
    \label{eq:state_error_plus}
    \Ket{\psi_B}=E_n^{(X)} E_m^{(Z)} E_\mathrm{L}^{(X)} E_\mathrm{L}^{(Z)}\Ket{+^{(l)}},
  \end{equation}
  where $E_n^{(X)}\in\{\mathbbm{1},X_n\}$ for $n\in\{1,\ldots,7\}$ represents the one-qubit $X$ error on the $n$th qubit, $E_m^{(Z)}\in\{\mathbbm{1},Z_m\}$ for $m\in\{1,\ldots,7\}$ represents the one-qubit $Z$ error on the $m$th qubit, and $E_\mathrm{L}^{(X)}\in\left\{\mathbbm{1}^{\otimes 7},X_\mathrm{L}\right\}$ and $E_\mathrm{L}^{(Z)}\in\left\{\mathbbm{1}^{\otimes 7},Z_\mathrm{L}\right\}$ represent logical $X$ and $Z$ errors in terms of the logical operators $X_\mathrm{L}$ and $Z_\mathrm{L}$ defined as~\eqref{eq:logical}.
  Since $\Ket{+^{(l)}}$ is stabilized by $X_\mathrm{L}$
  \begin{equation}
    \label{eq:stabilizer_plus}
    X_\mathrm{L}\Ket{+^{(l)}}=\Ket{+^{(l)}},
  \end{equation}
  $\Ket{\psi_B}$ in~\eqref{eq:state_error_plus} can indeed be written as
  \begin{equation}
    \label{eq:state_error_plus_no_x}
    \Ket{\psi_B}=E_n^{(X)} E_m^{(Z)} E_\mathrm{L}^{(Z)}\Ket{+^{(l)}}.
  \end{equation}
  In the part $C$, we detect the logical $Z$ error by measuring whether the state is stabilized by $X_\mathrm{L}$.
  In the part $D$, we detect up to two-qubit $X$ and $Z$ errors.
  Hence, after $C$ and $D$, the state~\eqref{eq:state_error_plus_no_x} is projected onto
  \begin{equation}
    \label{eq:state_error_plus_no_error}
    \Ket{\psi_D}=\Ket{+^{(l)}},
  \end{equation}
  which includes no error.

  \textbf{Case~\ref{case:2}}:
  Since the parts $A$ and $B$ are the same as the circuit in Fig.~\ref{fig:preparation_7qubit_plus_optimized},
  one error in $A$ and $B$ causes at most one-qubit error in $\Ket{\psi_B}$ as shown in Ref.~\cite{G6}; that is, we have
  \begin{equation}
    \Ket{\psi_B}=E_n^{(X)} E_m^{(Z)}\Ket{+^{(l)}}.
  \end{equation}
  Since one error occurs in $C$, we have the following cases, which show that $\Ket{\psi_C}$ is $\Ket{+^{(l)}}$ with at most two-qubit errors.
  \begin{itemize}
    \item If the one-qubit $X$ error and the one-qubit $Z$ error in $C$ occur on the $n^\prime$th qubit and the $m^\prime$th qubit of the $7$-qubit code, respectively, then we have
      \begin{equation}
        \label{eq:psi_c_1}
        \Ket{\psi_C}=E_{n^\prime}^{(X)} E_{m^\prime}^{(Z)} \left(E_n^{(X)} E_m^{(Z)}\Ket{+^{(l)}}\right).
      \end{equation}
    \item If the one-qubit $X$ error in $C$ occurs in the auxiliary qubit for the error detection in $C$ with measurement outcome $k_2$, the $X$ error propagates to the qubits for $\Ket{+^{(l)}}$ according to the commutation relation~\eqref{eq:x2}, and in the same way as~\eqref{eq:state_error_plus}, we have for some $n^\prime$
      \begin{align}
        \label{eq:psi_c_3}
        \Ket{\psi_C}&=E_{n^\prime}^{(X)} E_\mathrm{L}^{(X)} \left(E_n^{(X)} E_m^{(Z)}\Ket{+^{(l)}}\right)\nonumber\\
                    &\propto E_{n^\prime}^{(X)} \left(E_n^{(X)} E_m^{(Z)}\Ket{+^{(l)}}\right),
      \end{align}
      where we use~\eqref{eq:stabilizer_plus} to ignore the logical $X$ error.
    \item If the one-qubit $Z$ error in $C$ occurs in the auxiliary qubits for the error detection in $C$, the $Z$ error never propagates to the qubits for $\Ket{+^{(l)}}$.
  \end{itemize}
  Since we can detect up to two-qubit errors in $D$, we obtain from any of~\eqref{eq:psi_c_1} and~\eqref{eq:psi_c_3}
  \begin{equation}
    \Ket{\psi_D}=\Ket{+^{(l)}}.
  \end{equation}

  \textbf{Case~\ref{case:4}}:
  We have
  \begin{equation}
    \Ket{\psi_B}=\Ket{+^{(l)}}.
  \end{equation}
  In the same way as Case~\ref{case:2}, due to the fact that the logical $X$ error in $\Ket{+^{(l)}}$ can be ignored, two errors in $C$ cause at most two-qubit errors
  \begin{equation}
    \Ket{\psi_C}= E_{n^\prime}^{(X)} E_n^{(X)} E_{m^\prime}^{(Z)} E_m^{(Z)}\Ket{+^{(l)}}.
  \end{equation}
  Then, we can detect these two-qubit errors in $D$
  \begin{equation}
    \Ket{\psi_D}=\Ket{+^{(l)}}.
  \end{equation}

  Therefore, even if the circuit in Fig.~\ref{fig:preparation_7qubit_explanation} includes two errors, $\Ket{\psi_C}$ in any of these cases includes at most two-qubit errors that can be detected in $D$.
  Thus, we obtain the conclusion.
\end{proof}

\begin{proof}[Proof of Proposition~\ref{prp:preparation_magic2}]
  In the same way as the proof of Proposition~\ref{prp:preparation_7qubit_plus_optimized2},
  we analyze the fault-tolerant property for two errors in terms of the error detection,
  since the $1$-fault tolerance in terms of the error correction can be shown essentially in the same way as our analysis and Refs.~\cite{G6,Chamberland2019faulttolerantmagic}.
  To analyze the circuit in Fig.~\ref{fig:preparation_magic2} followed by the error detection using the 7-qubit code as the error-detecting code,
  we consider an equivalent circuit in Fig.~\ref{fig:preparation_magic_explanation} that is decomposed into four parts, namely, $A$, $B$, $C$, and $D$ in Fig.~\ref{fig:preparation_magic_explanation},
  where the states obtained from the parts $A$, $AB$, $ABC$, and $ABCD$ of the circuit are denoted by $\Ket{\psi_A}$, $\Ket{\psi_B}$, $\Ket{\psi_C}$, and $\Ket{\psi_D}$, respectively.
  Note that the circuits in Figs.~\ref{fig:preparation_magic2} and~\ref{fig:preparation_magic_explanation} are equivalent since ${T^\dag}^{\otimes 7}$ in the part $B$ of Fig.~\ref{fig:preparation_magic_explanation} and ${T}^{\otimes 7}$ in the part $C$ of Fig.~\ref{fig:preparation_magic_explanation} cancel out due to $T^\dag T=\mathbbm{1}$.
  The circuit in Fig.~\ref{fig:preparation_magic_explanation} includes non-Clifford gates $T$ and $T^\dag$ following commutation relations
  \begin{align}
    TX&=\mathrm{e}^{\mathrm{i}\frac{\pi}{4}}XS^\dag T,\\
    T^\dag X&=\mathrm{e}^{-\mathrm{i}\frac{\pi}{4}}XS T^\dag,\\
    TZ&=ZT,\\
    T^\dag Z&=ZT^\dag.
  \end{align}
  Thus, even if $X$ and $Z$ errors occur on qubits in the circuit,
  the state prepared by this circuit may include errors represented by Clifford gates including $S$ and $S^\dag$, not only $X$ and $Z$.
  However, we can expand the Clifford gates representing the errors using
  \begin{align}
    S&=\frac{1+\mathrm{i}}{2}\mathbbm{1}+\frac{1-\mathrm{i}}{2}Z,\\
    S^\dag &=\frac{1-\mathrm{i}}{2}\mathbbm{1}+\frac{1+\mathrm{i}}{2}Z,
  \end{align}
  so that each term can be represented in terms of $X$ and $Z$.
  Based on this expansion, we analyze the $X$ and $Z$ errors in the following of this proof,
  and in the same way as the proof of Proposition~\ref{prp:preparation_7qubit_plus_optimized2}, errors may refer to either $X$ and $Z$ errors unless stated otherwise.

  We consider the following exhaustive cases of two errors on $ABC$ in Fig.~\ref{fig:preparation_magic_explanation}, which corresponds to the circuit in Fig.~\ref{fig:preparation_magic2}.
  \begin{enumerate}
    \item\label{case:1magic} Two errors occur in $A$ or $B$.
    \item\label{case:2magic} One error occurs in $A$, and one error occurs in $C$.
    \item\label{case:4magic} One error occurs in $B$, and one error occurs in $C$.
    \item\label{case:6magic} Two errors occur in $C$.
  \end{enumerate}

  \textbf{Case~\ref{case:1magic}}:
  We here show that given $\Ket{\psi_B}$ with errors,
  after $C$ and $D$ without error, the state is projected onto
  \begin{equation}
    \Ket{\psi_D}=\Ket{\frac{\pi}{8}^{(l)}},
  \end{equation}
  which includes no error.
  In $B$ and $C$ for the error detection of $\Ket{\frac{\pi}{8}^{(l)}}$,
  similarly to the magic state preparation in Refs.~\cite{G6,Chamberland2019faulttolerantmagic},
  we exploit the fact that a one-qubit state $\Ket{\frac{\pi}{8}}$ is stabilized by a Clifford gate $TXT^\dag$
  \begin{equation}
    TXT^\dag\Ket{\frac{\pi}{8}}=\Ket{\frac{\pi}{8}}.
  \end{equation}
  Due to the transversal implementation~\eqref{eq:transversal_implementation} of Clifford gates for the $7$-qubit code,
  a logical Clifford gate $TXT^\dag=\mathrm{e}^{-\mathrm{i}\frac{\pi}{4}}SX$ of the $7$-qubit code at the concatenation level $l$ is implemented by ${\left(T^\dag X T\right)}^{\otimes 7}={\left(\mathrm{e}^{\mathrm{i}\frac{\pi}{4}}S^\dag X\right)}^{\otimes 7}$ at the concatenation level $l-1$; \textit{e.g.}, a level-$1$ state $\Ket{\frac{\pi}{8}^{(1)}}$ of the $7$-qubit code is stabilized by
  \begin{equation}
    {\left(T^\dag XT\right)}^{\otimes 7}\Ket{\frac{\pi}{8}^{(1)}}=\Ket{\frac{\pi}{8}^{(1)}}.
  \end{equation}
  Without any error,
  $\Ket{\psi_B}$ after $B$ would be projected onto a $+1$ eigenstate of ${\left(T^\dag XT\right)}^{\otimes 7}$ in the same way as Ref.~\cite{G6},
  and in Case~\ref{case:1}, $\Ket{\psi_C}$ after $C$ is projected onto a $+1$ eigenstate of ${\left(T^\dag XT\right)}^{\otimes 7}$.
  Note that if no error occurs, only the first auxiliary qubit in $C$ with measurement outcome $k_3$ is necessary for this projection, while we will use the other three auxiliary qubits in $C$ with measurement outcomes $k_4$, $k_5$, and $k_6$ later to detect errors in the auxiliary qubits.
  The obtained $+1$ eigenstate $\Ket{\psi_C}$ of ${\left(T^\dag XT\right)}^{\otimes 7}$ may not be $\Ket{\frac{\pi}{8}^{(l)}}$ since this $+1$ eigenstate of ${\left(T^\dag XT\right)}^{\otimes 7}$ is not necessarily stabilized by the stabilizers of the $7$-qubit code.
  After $D$ without error, the state is projected onto a $+1$ eigenstate of all the stabilizers~\eqref{eq:syndromes} of the $7$-qubit code.
  Since ${\left(T^\dag XT\right)}^{\otimes 7}$ is a logical operator commuting with all the stabilizers of the $7$-qubit code,
  we obtain the $+1$ eigenstate of ${\left(T^\dag XT\right)}^{\otimes 7}$ in the code space of the $7$-qubit code, \textit{i.e.},
  \begin{equation}
  \Ket{\psi_D}=\Ket{\frac{\pi}{8}^{(l)}}.
  \end{equation}

  \textbf{Case~\ref{case:2magic}}:
  In the same way as Case~\ref{case:1magic},
  given $\Ket{\psi_A}$ with errors,
  after $B$ without error, the state is projected onto a $+1$ eigenstate of ${\left(T^\dag XT\right)}^{\otimes 7}$.
  Then, since one error occurs in $C$, we have the following cases, which show that $\Ket{\psi_C}$ is a $+1$ eigenstate of ${\left(T^\dag XT\right)}^{\otimes 7}$ with at most one-qubit errors.
  \begin{itemize}
    \item If the one-qubit $X$ error and the one-qubit $Z$ error in $C$ occur on one of the qubits for $\Ket{\frac{\pi}{8}^{(l)}}$, then after $C$, we have a $+1$ eigenstate of ${\left(T^\dag XT\right)}^{\otimes 7}$ with one-qubit errors.
    \item If the one-qubit $X$ error in $C$ occurs in the first auxiliary qubit for the error detection in $C$ with measurement outcome $k_3$, the $X$ error propagates to the qubits for $\Ket{\frac{\pi}{8}^{(l)}}$ according to the commutation relation~\eqref{eq:x2}. However, the $X$ error also propagates to the second auxiliary qubit with measurement outcome $k_4$, and the measurement outcome of the second auxiliary qubit becomes $k_4=1$; that is, the state is discarded.
    \item If the one-qubit $X$ error in $C$ occurs in one of the second, third, and fourth auxiliary qubits with measurement outcomes $k_4$, $k_5$, and $k_6$, respectively, the $X$ error never propagates to the qubits for $\Ket{\frac{\pi}{8}^{(l)}}$.
    \item If the one-qubit $Z$ error in $C$ occurs in the auxiliary qubits for the error detection in $C$, the $Z$ error never propagates to the qubits for $\Ket{\frac{\pi}{8}^{(l)}}$.
  \end{itemize}
  Note that a $+1$ eigenstate $\Ket{\psi_C}$ of ${\left(T^\dag XT\right)}^{\otimes 7}$ with one-qubit errors means that by applying one-qubit operators to $\Ket{\psi_C}$ for correcting the one-qubit errors, we could obtain one of the $+1$ eigenstates of ${\left(T^\dag XT\right)}^{\otimes 7}$.
  Since the $7$-qubit code has distance $3$,
  after applying $D$ to $\Ket{\psi_C}$ even with one-qubit errors, we obtain the $+1$ eigenstate of ${\left(T^\dag XT\right)}^{\otimes 7}$ projected onto the code space of the $7$-qubit code, \textit{i.e.},
  \begin{equation}
    \Ket{\psi_D}=\Ket{\frac{\pi}{8}^{(l)}}.
  \end{equation}

  \textbf{Case~\ref{case:4magic}}:
  We have
  \begin{equation}
    \Ket{\psi_A}=\Ket{\frac{\pi}{8}^{(l)}}.
  \end{equation}
  As shown in Ref.~\cite{G6}, if $B$ includes one error, then after $B$, we have a $+1$ eigenstate of ${\left(T^\dag XT\right)}^{\otimes 7}$ with at most one-qubit errors.
  Then after $C$ with another error, due to the same analysis as Case~\ref{case:2magic},
  we have a $+1$ eigenstate of ${\left(T^\dag XT\right)}^{\otimes 7}$ with at most two-qubit errors.
  In this case, we can detect these two-qubit errors in $D$
  \begin{equation}
    \Ket{\psi_D}=\Ket{\frac{\pi}{8}^{(l)}}.
  \end{equation}

  \textbf{Case~\ref{case:6magic}}:
  We have
  \begin{equation}
    \Ket{\psi_B}=\Ket{\frac{\pi}{8}^{(l)}}.
  \end{equation}
  Then, since two errors occur in $C$, we have the following cases, which show that $\Ket{\psi_C}$ is a $+1$ eigenstate of ${\left(T^\dag XT\right)}^{\otimes 7}$ with at most two-qubit errors.
  \begin{itemize}
    \item For $n=1,2$, if $n$ errors in $C$ occur on $n$ out of the $7$ qubits for $\Ket{\frac{\pi}{8}^{(l)}}$, then after $C$, we have a $+1$ eigenstate of ${\left(T^\dag XT\right)}^{\otimes 7}$ with $n$-qubit errors.
    \item If one or two $X$ errors in $C$ occur in the first auxiliary qubit for the error detection in $C$ with measurement outcome $k_3$, the $X$ errors propagate to the qubits for $\Ket{\frac{\pi}{8}^{(l)}}$ according to the commutation relation~\eqref{eq:x2}. However, each $X$ error also propagates to at least one of the second, third, and fourth auxiliary qubits with measurement outcomes $k_4$, $k_5$, and $k_6$. The part $C$ of the circuit in Fig.~\ref{fig:preparation_magic_explanation} is designed so that at least one of the measurement outcomes of these auxiliary qubits should become $1$ in this case; that is, the state is discarded.
    \item If one or two $X$ errors in $C$ occur in one of the second, third, and fourth auxiliary qubits with measurement outcomes $k_4$, $k_5$, and $k_6$, respectively, the $X$ errors never propagate to the qubits for $\Ket{\frac{\pi}{8}^{(l)}}$.
    \item If one or two $Z$ errors in $C$ occur in the auxiliary qubits for the error detection in $C$, the $Z$ errors never propagate to the qubits for $\Ket{\frac{\pi}{8}^{(l)}}$.
  \end{itemize}
  Then, we can detect up to two-qubit errors in $D$
  \begin{equation}
    \Ket{\psi_D}=\Ket{+^{(l)}}.
  \end{equation}

  Therefore, even if the circuit in Fig.~\ref{fig:preparation_magic_explanation} includes two errors, $\Ket{\psi_D}$ in any of these cases is $\Ket{\frac{\pi}{8}^{(l)}}$.
  Also note that since the $7$-qubit code has distance $3$, errors in $\Ket{\psi_C}$ can be reduced to at most two-qubit errors.
  Thus, we obtain the conclusion.
\end{proof}

\bibliography{citation_bibtex}

\end{document}